\documentclass[
  twoside,
  12pt,
  a4paper,
  openright,
  mathcomp, %
  english, %
  marginleft=1in,
  marginright=1in,
  margintop=1in,
  marginbottom=1in,
  bindingoffset=6mm
]{hdphdthesis}

\usepackage{afterpage}
\usepackage{lmodern}  %
\newcommand{\neighbors}{N}%

\usepackage{amsmath}
\newcommand{\AlgName}[1]{\ensuremath{\text{{\sf #1}}}}
\usepackage{afterpage}
\usepackage{float}
\usepackage{placeins}
\usepackage{graphicx}
\usepackage{balance}  %
\newif\ifFull
\Fullfalse

\usepackage{caption}
\usepackage{subcaption}
\usepackage[utf8]{inputenc}
\usepackage{url}
\usepackage{array}
\usepackage{algorithm}
\usepackage[noend]{algpseudocode}
\usepackage{numprint}
\npdecimalsign{.} %
\usepackage{amsthm}
\newtheorem{theorem}{Theorem}[chapter]
\newtheorem{corollary}[theorem]{Corollary}
\newtheorem{lemma}[theorem]{Lemma}

\numberwithin{figure}{chapter}
\numberwithin{table}{chapter}
\numberwithin{algorithm}{chapter}

\DeclareMathOperator*{\argmax}{argmax}
\DeclareMathOperator*{\argmin}{argmin}

\usepackage{hyperref}
\usepackage{xspace}
\usepackage{wrapfig}
\usepackage{graphics}
\usepackage{todonotes}
\usepackage{booktabs}
\usepackage{microtype}%
\definecolor {infocolor} {rgb} {0.6,0.6,0.6}

\usepackage{rotating}
\usepackage{lscape} 
\usepackage{threeparttable}
\usepackage{hvfloat}

\usepackage{multirow}
\usepackage[numbers,square,sort&compress]{natbib}
\usepackage{doi}
\usepackage{multicol}%
\usepackage{tikz}
\usetikzlibrary{shapes.arrows,chains}
\usetikzlibrary{decorations.pathreplacing}
\usepackage{framed,enumitem} 

\makeatletter
\def\@listarabic{\@listarabic}
\makeatother

\DeclareRobustCommand{\frcshape}{\fontfamily{frc}\selectfont}
\DeclareTextFontCommand{\textfrc}{\frcshape}
\newcommand{\mathfrc}[1]{{ #1}}
\renewcommand{\mathcal}[1]{{ #1}}
\def\Is{\ensuremath{=}}

\def\MdR{\ensuremath{\mathbb{R}}}
\def\MdN{\ensuremath{\mathbb{N}}}

\newcommand{\Id}[1]{{\detokenize{#1}}}

\newcommand{\ie}{i.\,e.,\xspace}
\newcommand{\eg}{e.\,g.,\xspace}
\newcommand{\etal}{et~al.\xspace}

\usepackage{color}

\def\comment#1{}

\def\withcomments{
	\newcounter{mycommentcounter}
	\def\comment##1{\refstepcounter{mycommentcounter}%
		\ifhmode%
		\unskip%
		{\dimen1=\baselineskip \divide\dimen1 by 2 %
			\raise\dimen1\llap{\tiny\bfseries \textcolor{red}{-\themycommentcounter-}}}\fi%
		\marginpar[{\renewcommand{\baselinestretch}{0.8}%
			\hspace*{3em}\begin{minipage}{5em}\footnotesize [\themycommentcounter]: \raggedright ##1\end{minipage}}]{\renewcommand{\baselinestretch}{0.8}%
			\begin{minipage}{5em}\footnotesize [\themycommentcounter]: \raggedright ##1\end{minipage}}}
}

\newcommand{\Xcomment}[1]{}

\usepackage{amsfonts}

\newcommand{\set}[1]{\left\{ #1\right\}}
\newcommand{\gilt}{:}
\newcommand{\sodass}{\,:\,}
\newcommand{\setGilt}[2]{\left\{ #1\sodass #2\right\}}

\newcommand{\realrange}[2]{\left[#1, #2\right]}

\newcommand{\unitrange}[2]{\realrange{0}{1}}

\newcommand{\discussionsize}{\small}

\newcommand{ \scaleFactorSmall} {0.485}
\newcommand{ \imgScaleFactorSmall} {0.95}
\newcommand{ \capPositionSmall} {-.45cm}
\newcommand{ \afterCapSmall} {-.45cm}

\newcommand{\MyWhile}    {{\bf while\ }}

\newcommand{\Do}       {{\bf do\ }}

\newdimen\endofsize\endofsize=0.5em

\newcommand{ \scaleFactor} {0.62}
\newcommand{ \imgScaleFactor} {0.9}
\newcommand{ \capPosition} {-.45cm}
\newcommand{ \afterCap} {-.45cm}

\usepackage[
    type={CC},
    modifier={by-nc-nd},
    version={3.0},
]{doclicense}

\begin{document}
\pagenumbering{roman}

\captionsetup[figure]{labelfont={bf},name={Figure},labelsep=period,singlelinecheck=false}
\captionsetup[table]{labelfont={bf},name={Table},labelsep=period,singlelinecheck=false}

\author{Marcelo Fonseca Faraj}
\bornin{Belo Horizonte, Brasilien}
\examdate{}
\title{\textbf{Streaming, Local, and Multi-Level (Hyper)Graph Decomposition}}
\setreferees{
    Prof. Dr. Christian Schulz,
}
\maketitle

\begin{coverpage}{Abstract}
(Hyper)Graph decomposition is a family of problems that aim to break down large (hyper)graphs into smaller sub(hyper)graphs for easier analysis.
The importance of this lies in its ability to enable efficient computation on large and complex (hyper)graphs, such as social networks, chemical compounds, and computer networks.
This dissertation explores several types of (hyper)graph decomposition problems, including graph partitioning, hypergraph partitioning, local graph clustering, process mapping, and signed graph clustering.
Our main focus is on streaming algorithms, local algorithms and multilevel algorithms. 
In terms of streaming algorithms, we make contributions with highly efficient and effective algorithms for (hyper)graph partitioning and process mapping.
In terms of local algorithms, we propose sub-linear algorithms which are effective in detecting high-quality local communities around a given seed node in a graph based on the distribution of a given motif.
In terms of multilevel algorithms, we engineer high-quality multilevel algorithms for process mapping and signed graph clustering.
We provide a thorough discussion of each algorithm along with experimental results demonstrating their superiority over existing state-of-the-art techniques.
The results show that the proposed algorithms achieve improved performance and better solutions in various metrics, making them highly promising for practical applications.
Overall, this dissertation showcases the effectiveness of advanced combinatorial algorithmic techniques in solving challenging (hyper)graph decomposition problems.
\end{coverpage}

\begin{coverpage}{Zusammenfassung}
(Hyper-)Graphenzerlegung fasst mehrere Probleme zusammen, bei denen große Graphen oder Hypergraphen zur einfacheren Analyse in kleinere Subgraphen oder Subhypergraphen zerlegt werden.
Dies ist bedeutsam, denn es befähigt dazu Berechnungen auf großen und komplexen Graphen und Hypergraphen, wie soziale Netzwerke, chemische Verbindungen und Computernetzwerke, effizient zu machen.
Diese \linebreak\hfill Dissertation befasst sich mit der Untersuchung verschiedener Arten von \linebreak\hfill (Hyper-)Graphenzerlegungsproblemen, darunter Graphpartitionierung, Hypergraphpartitionierung, Prozesszuweisung, Signed-Graph-Clustering und Local-Graph-Clustering.
Unser Hauptaugenmerk liegt auf Streaming-Algorithmen,lokalen Algorithmen und Multi-Level-Algorithmen.
Im Bereich der Streaming-Algorithmen leisten \linebreak\hfill wir Beiträge mit hocheffizienten und effektiven Algorithmen für die \linebreak\hfill(Hyper-)Graphpartitionierung und Prozesszuweisung.
Bei den lokalen Algorithmen schlagen wir sub-lineare Algorithmen vor, die auf Grundlage der Verteilung eines bestimmten Motivs hochwertige lokale Cluster um einen gegebenen Startknoten in einem Graphen bestimmen.
Was die Multi-Level-Algorithmen betrifft, so entwickeln wir hochwertige Multi-Level-Algorithmen für Prozesszuweisung und Signed-Graph-Clustering.
Wir liefern zu jedem Algorithmus eine ausführliche Diskussion zusammen mit experimentellen Ergebnissen, die ihre Überlegenheit gegenüber bestehenden Stand der Technik zeigen.
Die Ergebnisse legen dar, dass die vorgeschlagenen Algorithmen eine bessere Laufzeit und bessere Lösungen in verschiedenen Metriken erzielen, womit sie für praktische Anwendungen sehr vielversprechend sind.
Insgesamt zeigt diese Dissertation die Effektivität fortschrittlicher kombinatorischer algorithmischer Techniken bei der Lösung anspruchsvoller (Hyper-)Graphenzerlegungsprobleme.
\end{coverpage}

\begin{coverpage}{Acknowledgements}

I feel incredibly grateful and indebted to my advisor, Prof. Dr. Christian Schulz.
Thanks to him, I was introduced to the fascinating world of algorithm engineering and (hyper)graph decomposition problems, which quickly became a passion.
His guidance and support throughout my Ph.D. journey were invaluable, and I couldn't have asked for a more nurturing mentor.
Under his supervision, I learned the habits and behaviors that make a scientist successful and a leader effective.
Working at the University of Vienna and Heidelberg University, thanks to him, allowed me to publish and present many papers at top conferences across the globe, an opportunity I cherish.
Most importantly, Christian is not just an excellent mentor but also a kind and friendly person who always motivated and inspired me to pursue my ambitions.
Additionally, I am grateful for the chance to meet and interact with some of the giants of Computer Science, such as Peter Sanders, Monika Henzinger, and Robert Tarjan, \hbox{all thanks to him}.

I cannot express enough gratitude towards my wife, Yara Dalilla Xavier Faraj, who has been my rock throughout this journey.
Yara's unwavering love and support have been instrumental in helping me navigate the challenges of pursuing a Ph.D. abroad.
She has been with me every step of the way, from being my girlfriend to becoming my wife and greatest partner in life.
Without her, those four years would have been much more arduous and less happy.
I am thankful for Yara's caring and nurturing presence, which provided me with a sense of family during my time away from home.
Yara's encouragement and unwavering support have helped me stay motivated and focused on my goals.
I feel blessed to have such a wonderful person in my life, and I am excited to continue our journey together as \hbox{we embark on new adventures}.

I would like to express my deep gratitude to two wonderful couples who made my time in Europe even more special.
Firstly, my cousin Alfred Lenzner and his wife Marilia Simão dos Santos, with whom Yara and I shared some great moments during our stay in Germany.
Secondly, I would like to thank Alexander Noe and Anique-Marie Cabardos, with whom we had the pleasure of spending unforgettable moments in both Vienna and Mannheim.
Their company, warmth, and positive energy made our experiences in Europe \hbox{truly unique and enriching}.

I want to express my gratitude to all the amazing people who have been a part of my academic journey.
First, I would like to thank my colleagues from Heidelberg University, Alexander Noe, Konrad von Kirchbach, Stefan Neumann, Gramoz Goranci, Alexander Svozil, Wolfgang Ost, and Ernestine Großmann, for being wonderful co-Ph.D. students and for making my time in Heidelberg so enjoyable.
I would also like to express my sincere thanks to Kathrin Hanauer, Rudolf Hürner, Sagar Kale, Ami Paz, Xiaowei Wu (Eric), Kamer Kaya, Ulrike Frolik-Steffan, Christina Licayan, Catherine Proux-Wieland, and Henrik Reinstädtler, for being such amazing colleagues and making my Ph.D. journey an enriching experience.
I also want to thank all my co-authors, including Ümit V. Çatalyürek, Karen D. Devine, Lars Gottesbüren, Tobias Heuer, Alexander van der Grinten, Henning Meyerhenke, Jesper Larsson Träff, Peter Sanders, Sebastian Schlag, Christian Schulz, Daniel Seemaier, Dorothea Wagner, Adil Chhabra, Felix Hausberger, and Kamal Eyubov, for collaborating \hbox{with me on joint papers}.

Furthermore, I would like to thank all the students I had the pleasure of supervising in their scientific practica or theses.
I would also like to express my appreciation to Bora Uçar for the wonderful opportunity to work together on various projects and share an office with him in Lyon and Heidelberg.
I would like to extend my thanks to Peter Bastian for the incredible opportunity to participate in the IGCM 2023 conference in Bangalore.
Additionally, I would like to thank Monika Henzinger for the opportunity to be a part of her research group in Vienna.
I cannot forget to thank Jesper Larsson Träff, who worked with me in Vienna and with whom I had some amazing philosophical discussions.
Finally, I want to thank my dear friend and former (co)advisor João Fernando Machry Sarubbi, who has been a source of inspiration and guidance since my Bachelor's thesis.
I am grateful for his continued support and advice as I navigate through \hbox{my academic and personal life}.

I want to express my deepest gratitude to my family.
First and foremost, I want to thank my wonderful parents, who have been an unwavering source of love and support throughout my life.
Even though they were miles away from me, they never failed to show their constant presence and care in a way that only Brazilian parents could do or even imagine.
I also want to thank my amazing siblings, Nara Carolina, Priscila, Sarah, Valéria, and Mário, along with their partners and children.
I am blessed to have such a loving and caring family, and I am grateful for all the lessons and values they have instilled in me.
Additionally, I want to express my appreciation to my aunts, uncles, cousins, and in-laws for their support, \hbox{encouragement, and inspiration}.

In conclusion, I want to express my deepest gratitude to God for His continuous blessings and guidance in my life.
As Aristotle said, "To know thyself is the beginning of all wisdom."
However, our knowledge of ourselves is limited, and it is only through God's infinite wisdom and love that we can truly know our purpose and potential.
I am grateful for all the opportunities and experiences that God has presented me with, and for the strength and wisdom He has granted me to face challenges \hbox{along the way}.

\end{coverpage}

\begin{coverpage}{Related Publications}
Several contributions in this thesis are already published or accepted for publication in conference and journal papers as well as technical reports. Here we list all these~publications. \\

\section*{Accepted Papers}       

\begin{enumerate}

\item[\refstepcounter{enumi}{[\theenumi]}]
Marcelo Fonseca Faraj, Alexander van der Grinten, Henning Meyerhenke, Jesper Larsson Träff, and Christian Schulz.
\newblock {High-Quality Hierarchical Process Mapping}.
\newblock In {\em Proceedings of the 18th Symposium on Experimental Algorithms (SEA)}, volume 160 of LIPIcs, pages 4:1--4:15, 2020.

\item[\refstepcounter{enumi}{[\theenumi]}]
Marcelo Fonseca Faraj and Christian Schulz.
\newblock {Buffered Streaming Graph Partitioning}.
\newblock {\em ACM Journal of Experimental Algorithmics}, Volume 27, pages 1.10:1--1.10:26, 2022.

\item[\refstepcounter{enumi}{[\theenumi]}]
Marcelo Fonseca Faraj and Christian Schulz.
\newblock {Recursive Multi-Section on the Fly: Shared-Memory Streaming Algorithms for Hierarchical Graph Partitioning and Process Mapping}.
\newblock In {\em IEEE International Conference on Cluster Computing (CLUSTER)}, volume 9411, pages 473--483, 2022.

\item[\refstepcounter{enumi}{[\theenumi]}]
{\"{U}}mit V. {\c{C}}ataly{\"{u}}rek, Karen D. Devine, Marcelo Fonseca Faraj, Lars Gottesb{\"{u}}ren, Tobias Heuer, Henning Meyerhenke, Peter Sanders, Sebastian Schlag, Christian Schulz, Daniel Seemaier, and Dorothea Wagner
\newblock {More Recent Advances in (Hyper)Graph Partitioning}.
\newblock {\em ACM Computing Surveys}, Volume 55, Issue 12, Article No 253, pages 1--38, 2022.

\item[\refstepcounter{enumi}{[\theenumi]}]
Adil Chhabra, Marcelo Fonseca Faraj, and Christian Schulz.
\newblock {Local Motif Clustering via (Hyper)Graph Partitioning (Extended Abstract)}.
\newblock In {\em Proceedings of the 15th International Symposium on Combinatorial Search (SoCS)}, pages 261--263, AAAI Press,~2022.

\item[\refstepcounter{enumi}{[\theenumi]}]
Adil Chhabra, Marcelo Fonseca Faraj, and Christian Schulz.
\newblock {Local Motif Clustering via (Hyper)Graph Partitioning}.
\newblock In {\em Proceedings of the Symposium on Algorithm Engineering and Experiments (ALENEX)}, pages 96--109. SIAM, 2023.

\item[\refstepcounter{enumi}{[\theenumi]}]
Felix Hausberger, Marcelo Fonseca Faraj, and Christian Schulz.
\newblock {A Distributed Multilevel Memetic Algorithm for Signed Graph Clustering (Short Paper)}.
\newblock In {\em Genetic and Evolutionary Computation Conference Companion (GECCO ’23 Companion)}, to appear, 2023.

\item[\refstepcounter{enumi}{[\theenumi]}]
Kamal Eyubov, Marcelo Fonseca Faraj, and Christian Schulz.
\newblock {FREIGHT: Fast Streaming Hypergraph Partitioning}.
\newblock In {\em Proceedings of the 21st Symposium on Experimental Algorithms (SEA)}, to appear, 2023.

\end{enumerate}

\section*{Papers under Peer Review}       

\begin{enumerate}

\item[\refstepcounter{enumi}{[\theenumi]}]
Felix Hausberger, Marcelo Fonseca Faraj, and Christian Schulz.
\newblock {A Distributed Multilevel Memetic Algorithm for Signed Graph Clustering}.
\newblock 2023.

\item[\refstepcounter{enumi}{[\theenumi]}]
Adil Chhabra, Marcelo Fonseca Faraj, and Christian Schulz.
\newblock {Faster Local Motif Clustering via Maximum Flows}.
\newblock 2023.

\end{enumerate}

\section*{Technical Reports}       

\begin{enumerate}

\item[\refstepcounter{enumi}{[\theenumi]}]
Marcelo Fonseca Faraj, Alexander van der Grinten, Henning Meyerhenke, Jesper Larsson Träff, and Christian Schulz.
\newblock {High-Quality Hierarchical Process Mapping}.
\newblock Technical Report, University of Vienna, Humboldt Universtät zu Berlin, and Technical University of Vienna, 2020. (arXiv:2001.07134v2)

\item[\refstepcounter{enumi}{[\theenumi]}]
Marcelo Fonseca Faraj and Christian Schulz.
\newblock {Buffered Streaming Graph Partitioning}.
\newblock Technical Report, Heidelberg University, 2021. (arXiv:2102.09384)

\item[\refstepcounter{enumi}{[\theenumi]}]
Marcelo Fonseca Faraj and Christian Schulz.
\newblock {Recursive Multi-Section on the Fly: Shared-Memory Streaming Algorithms for Hierarchical Graph Partitioning and Process Mapping}.
\newblock Technical Report, Heidelberg University, 2021. (arXiv:2202.00394)

\item[\refstepcounter{enumi}{[\theenumi]}]
{\"{U}}mit V. {\c{C}}ataly{\"{u}}rek, Karen D. Devine, Marcelo Fonseca Faraj, Lars Gottesb{\"{u}}ren, Tobias Heuer, Henning Meyerhenke, Peter Sanders, Sebastian Schlag, Christian Schulz, Daniel Seemaier, and Dorothea Wagner
\newblock {More Recent Advances in (Hyper)Graph Partitioning}.
\newblock Technical Report, 2022. (arXiv:2205.13202)

\item[\refstepcounter{enumi}{[\theenumi]}]
Adil Chhabra, Marcelo Fonseca Faraj, and Christian Schulz.
\newblock {Local Motif Clustering via (Hyper)Graph Partitioning}.
\newblock Technical Report, Heidelberg University, 2022. (arXiv:2205.06176)

\item[\refstepcounter{enumi}{[\theenumi]}]
Felix Hausberger, Marcelo Fonseca Faraj, and Christian Schulz.
\newblock {A Distributed Multilevel Memetic Algorithm for Signed Graph Clustering}.
\newblock Technical Report, Heidelberg University, 2022. (arXiv:2208.13618)

\item[\refstepcounter{enumi}{[\theenumi]}]
Kamal Eyubov, Marcelo Fonseca Faraj, and Christian Schulz.
\newblock {FREIGHT: Fast Streaming Hypergraph Partitioning}.
\newblock Technical Report, Heidelberg University, 2023. (arXiv:2302.06259)

\item[\refstepcounter{enumi}{[\theenumi]}]
Adil Chhabra, Marcelo Fonseca Faraj, and Christian Schulz.
\newblock {Faster Local Motif Clustering via Maximum Flows}.
\newblock Technical Report, Heidelberg University, 2023. (arXiv:2301.07145)

\end{enumerate}

\end{coverpage}

\cleardoublepage\phantomsection\addcontentsline{toc}{chapter}{\contentsname}\tableofcontents\newpage

\cleardoublepage
\pagenumbering{arabic}
\chapter{Introduction}
\label{chap:Introduction}

\section{Motivation}
\label{sec:Motivation}

(Hyper)graphs are incredibly versatile mathematical structures that have proven to be useful for modeling and analyzing a wide range of phenomena.
These can include anything from social networks and chemical compounds to computer networks and communication systems.
As these systems grow in size and complexity, it becomes increasingly difficult to process their underlying (hyper)graphs on a single computer.
This is where (hyper)graph decomposition comes in.
The idea behind (hyper)graph decomposition is to break down a large (hyper)graph into smaller subgraphs or partitions, which can be processed separately.
By doing this, we can enable more efficient computation on large and complex (hyper)graphs, which has applications in many areas of research.

In fact, (hyper)graph decomposition is becoming increasingly important in a wide range of fields.
Its potential applications are vast, including scientific simulations, routing, community detection, computational linear algebra, and more.
With the ability to efficiently process large and complex (hyper)graphs, we can gain deeper insights into the structure and behavior of the systems they model.
This can lead to breakthroughs in fields ranging from social science to engineering.
There are several types of (hyper)graph decomposition problems, each of which has its own constraints, objective functions, challenges, and applications.
In this thesis, we specifically explore the following problems: graph partitioning, hypergraph partitioning, local graph clustering, process mapping, and signed graph clustering.
We now explain the \hbox{problems in detail}.

The \emph{(hyper)graph partitioning} problem is a very important and well-studied problem.
It consists of partitioning the nodes of a given (hyper)graph into a predetermined number of blocks with roughly equal cardinality or weight in order to minimize some metric such as the number of cut (hyper)edges.
Specifically in the case of hypergraph partitioning, an additional optimization metric is the \emph{connectivity}, which measures the difference between the number of blocks that have non-empty intersection with cut hyperedges and the number of hyperedges.
A natural application for this problem is the distributed processing of (hyper)graphs, in which multiple processors operate on unique nodes of the (hyper)graph and communicate with one another using message-passing in case of (hyper)edges shared across processors.
(Hyper)graph partitioning is NP-hard \cite{Garey1974} and there can be no approximation algorithm with a constant ratio for general (hyper)graphs~\cite{BuiJ92}. 
Thus, heuristics are used in practice.
A current trend for partitioning huge (hyper)graphs quickly and using low computational resources are \hbox{streaming algorithms~\cite{tsourakakis2014fennel,awadelkarim2020prioritized,jafari2021fast,mayer2018adwise,hoang2019cusp,alistarh2015streaming,tacsyaran2021streaming,HeiStream,StreamMultiSection,freight_paper}}. 

The \emph{process mapping} problem can be seen as an application and a generalization of the graph partitioning problem.
In this problem, we are given a \emph{communication graph} and a \emph{topology} containing processing elements (PEs) alongside with their pair-wise distances. 
The goal of the problem consists of mapping the nodes of the communication graph onto the topology in such a way that roughly the same (weighted) number of nodes are mapped to each PE and the \emph{total communication cost} is minimized.
The total communication cost is usually defined as the sum of weights of the cut edges multiplied by the distance between the PEs containing their respective endpoints.
A natural application of process mapping consists of assigning interdependent tasks to high-performance (HPC) systems in order to minimize the total execution time.
It can also be useful in other areas, such as manufacturing, where optimizing the flow of materials and resources is critical.
All hardness results associated with graph partitioning also apply for process mapping, hence heuristics are used in practice.
A special case of topology, known as \emph{hierarchical topology}, has been the subject of widespread research. 
In this topology, PEs are arranged in a multi-layered hierarchy of modules and sub-modules, and the distance between any two PEs depends exclusively on their nearest \hbox{shared module within the topology}.

The \emph{local graph clustering} problem involves a graph and a seed node, where the objective is to obtain a well-characterized cluster that contains the seed node. 
Conceptually, a well-characterized cluster is a subgraph that consists of many internal edges and few external edges.
More specifically, the quality of a cluster can be quantified by metrics such as \emph{conductance}~\cite{kannan2004clusterings}. %
Applications of this problem include those that only require analyzing a small, localized portion of a graph rather than the entire graph.
This is the case for community-detection on Web~\cite{epasto2014reduce} and social~\cite{jeub2015think} networks as well as structure-discovery in bioinformatics~\cite{voevodski2009spectral} networks, among other real-world problems.
Since minimizing conductance is NP-hard~\cite{wagner1993between}, approximative and heuristic approaches are used in practice.
Furthermore, the nature and scale of this problem necessitates sub-linear methodologies, i.e., which entail time and memory utilization that is only dependent on the size of the discovered cluster, rather than the entire graph.
While traditional approaches to local clustering typically consider the edge distribution when evaluating the quality of a local community~\cite{andersen2006local,leskovec2009community,chung2013solving,li2015uncovering,fountoulakis2020flowbased}, novel methods~\cite{yin2017local,zhang2019local,meng2019local,murali2020online} have shifted focus to finding local communities based on the distribution of \emph{motifs}, higher-order structures within the graph.
Empirical evidence shows that this approach, which can be called \emph{local motif clustering}, is effective at detecting \hbox{high-quality local communities~\cite{yin2017local}}.

The \emph{signed graph clustering} problem involves a signed graph, i.e., a graph in which each edge is associated with a weight that can be either positive or negative. 
The goal is to partition the nodes of the graph into an unspecified number of \emph{well-characterized} clusters.
In the context of signed graphs, a well-characterized cluster is defined as being densely connected by edges with positive weight and sparsely connected by edges with negative weight.
Conversely, distinct clusters that are well-characterized exhibit dense inter-connections through edges with negative weight and sparse inter-connections through edges with positive weight.
Signed graph clustering necessitates distinct metrics and approaches as opposed to those for traditional, unsigned graph clustering.
On the one hand, traditional metrics such as \emph{conductance}~\cite{kannan2004clusterings} and \emph{modularity}~\cite{brandes2007modularity} are insufficient to address negative edge weights.
On the other hand, negative edge weights can make clustering structure more explicit.
As a consequence, metrics like edge-cut, which simply measures the sum of edge weights between clusters, can be used for evaluating signed clustering.
In signed graphs, the edge-cut is bounded by the sum of negative edge weights.
Hence, achieving this edge-cut value through clustering will separate all positive edges within clusters and cut all negative edges, fulfilling the purpose of the problem.
In many real-world applications, interactions between two entities can be accurately represented by \emph{signed graphs}, i.e., the sign associated with an edge can indicate when the nature of an interaction between nodes is positive (e.g., attraction, similarity, friendship) or negative (e.g., repulsion, difference, animosity).
With that being said, graph clustering has practical applications in areas such as criminology, public health, politics, and analysis of social networks~\cite{application_review}.
The problem of finding a clustering of signed graphs with minimum edge-cut is NP-hard~\cite{DBLP:phd/dnb/Wakabayashi86}, hence \hbox{heuristic algorithms are used in practice}.

\section{Main Contributions}
\label{sec:Main Contributions}

In this section we list our main algorithmic contributions in the context of (hyper)graph decomposition.
Our algorithms utilize various algorithmic techniques and data structures to achieve improved performance and better solutions compared to state-of-the-art methods. 
To make the exposition clear, we split our contributions in three groups of algorithms: streaming algorithms, local algorithms and multilevel algorithms. 
A summary of each contribution is presented along with a selection of experimental results, which illustrate their superiority in various metrics when \hbox{compared to the state-of-the-art}.

\subsection{Streaming Algorithms}
\label{subsec:Streaming Algorithms}

We develop three (buffered) streaming algorithms for (hyper)graph decomposition problems.
More specifically, we propose algorithms for (hyper)graph partitioning and \hbox{process mapping}.

We start by proposing a buffered streaming algorithm for graph partitioing.
Our algorithm loads a batch of nodes and then builds a model that represents the loaded subgraph as well as the already present partition structure. 
This model enables us to apply multilevel algorithms and in turn compute much higher quality solutions of huge graphs on cheap machines than previously possible. 
To partition the model, we develop a multilevel algorithm that optimizes an objective function that has previously shown to be effective for the streaming setting. 
Surprisingly, this also removes the dependency on the number of blocks from the running time compared to the previous state-of-the-art. 
Our algorithm computes considerably better solutions than the state-of-the-art using a very small buffer size. 
In addition, for large numbers of blocks, our algorithm becomes \hbox{faster than the state-of-the-art}. 

Our second streaming algorithm is a shared-memory parallel streaming algorithm for the process mapping problem.
It is designed to map a streamed communication graph onto a hierarchical topology by performing recursive multi-sections on the fly.
If a hierarchy is not specified as an input, our approach can also be used as a general tool to solve the graph partitioning problem.  
Our approach is the first streaming algorithm for the process mapping problem.
Furthermore, in the context of non-buffered streaming graph partitioning, it has a considerably lower running time complexity in comparison with state-of-the-art.
Our experiments indicate that our algorithm is both faster and produces better process mappings than competing streaming tools.  
In case of graph partitioning, our framework is up to two orders of magnitude faster at the cost of~$5\%$ more cut edges compared to \hbox{a state-of-the-art algorithm}.

Our third streaming algorithm solves the hypergraph partitioning problem by extending from a state-of-the-art streaming algorithm for graph partitioning.
By using an efficient data structure, we make the overall running of our algorithm linearly dependent on the pin-count of the hypergraph and the memory consumption linearly dependent on the numbers of nets and blocks.
The results of our extensive experimentation showcase the promising performance of our algorithm as a highly efficient and effective solution for streaming hypergraph partitioning. 
Our algorithm demonstrates competitive running time with the Hashing algorithm, with a difference of a maximum factor of four observed on three fourths of the instances.
Significantly, our findings highlight the superiority of our algorithm over all existing (buffered) streaming algorithms and even an in-memory algorithm HYPE, with respect to both weighted number of \hbox{cut hyperedges and connectivity measures}. 

\vfill

\subsection{Local Algorithms}
\label{subsec:Local Algorithms}

We develop two local algorithms for graph decomposition.
In particular, both algorithms are designed to solve the local motif clustering problem.
Our algorithms starts by building a (hyper)graph model which represents the motif-distribution around the seed node on the original graph.
While the graph model is exact for motifs of size at most three, the hypergraph model works for arbitrary motifs and is designed such that an optimal solution in the (hyper)graph model minimizes the motif conductance in \hbox{the original network}.

In our first algorithm, the (hyper)graph model is then partitioned using a powerful multi-level hypergraph or graph partitioner in order to directly minimize the motif conductance of the corresponding partition in the original graph.
Extensive experiments evaluate the trade-offs between the two different models. 
Moreover, when using the graph model for triangle motifs, our algorithm computes communities that have on average one third of the motif conductance value than communities computed by the state-of-the-art while being faster on average and removing the necessity of a preprocessing motif-enumeration \hbox{on the whole network}.

In our second algorithm, we transform the hypergraph model into a flow model based on the fast and effective algorithm \emph{max-flow quotient-cut improvement}~(\AlgName{MQI})~\cite{mqipaper2004}.
We show that a non-trivial maximum flow exists if and only if a superior solution exists, which is obtained automatically.
In our experiments with the triangle motif, our flow-based algorithm produces better communities than the state-of-the-art, while also being up to multiple \hbox{orders of magnitude faster}.

\subsection{Multilevel Algorithms}
\label{subsec:Multilevel Algorithms}

We develop two multilevel algorithms for graph decomposition.
Multilevel algorithms consist of three main phases: coarsening, construction, and uncoarsening.
In coarsening, the graph is recursively contracted into a sequence of smaller graphs which maintain some general structure.
In construction, an initial solution is computed on the smallest graph.
In uncoarsening, the contractions are recursively undone, and local search methods are used to refine the solution induced by each level of contraction.
We propose algorithms for process mapping and signed \hbox{graph clustering}.

We propose and engineer multiple setups of a multilevel algorithm for the process mapping problem.
Important ingredients of our algorithm include fast label propagation, more localized local search, initial partitioning, as well as a compressed data structure to compute processor distances without storing a distance matrix. 
Moreover, our algorithm is able to exploit a given hierarchical structure of the distributed system under consideration. 
Experiments indicate that our algorithm speeds up the overall mapping process and, due to the integrated multilevel approach, also finds much better solutions in practice.  
For example, one configuration of our algorithm yields similar solution quality as the previous state-of-the-art in terms of mapping quality for large numbers of partitions while being a factor~9.3 faster.  
Compared to the currently fastest iterated multilevel mapping algorithm Scotch, we obtain~16\% better solutions while investing \hbox{slightly more running time}.

Our last proposed algorithms are designed to solve the signed graph clustering problem by leveraging some of the most effective techniques from graph partitioning that minimize edge-cut.
We engineer all the details of a multilevel algorithm, which encompasses a coarsening-uncoarsening process and efficient local search methods.
We also introduce a memetic algorithm that utilizes our multilevel algorithm and further enhances it with natural multilevel recombination and mutation operations.
We also parallelize our approach using a scalable coarse-grained island-based strategy which has already shown to be scalable in practice. 
Experimental results demonstrate that our memetic algorithm outperforms the state-of-the-art with respect to edge-cut, producing \hbox{significantly better solutions}.

\section{Outline}
\label{sec:Outline}

This dissertation is organized as follows.
We begin in Chapter~\ref{chap:Preliminaries} by presenting preliminaries and basic concepts that are used throughout this thesis. 
We continue by elaborating related work in Chapter~\ref{chap:Related Work}.
Our algorithmic contributions are presented in Chapter~\ref{chap:Streaming Algorithms}, Chapter~\ref{chap:Local Algorithms}, and Chapter~\ref{chap:Multilevel Algorithms}, which correspond to streaming algorithms, local algorithms, and multilevel algorithms, respectively.
Specific conclusions are given in respective chapters dedicated to each family of algorithms, and a general conclusion is provided in Chapter~\ref{chap:Discussion}.

\chapter{Preliminaries}
\label{chap:Preliminaries}

In this chapter, we present the basic concepts used in this dissertation.

\section{Graphs and Hypergraphs}
\label{sec:Graphs and Hypergraphs}

Let $G=(V=\{0,\ldots, n-1\},E)$ be an \emph{undirected graph} with no multiple or self edges allowed, such that $n = |V|$ and $m = |E|$.
Let $c: V \to \MdR_{\geq 0}$ be a node-weight function, and let $\omega: E \to \MdR_{>0}$ be an edge-weight function.
We generalize $c$ and $\omega$ functions to sets, such that $c(V') = \sum_{v\in V'}c(v)$ and $\omega(E') = \sum_{e\in E'}\omega(e)$.
Let $N(v) = \setGilt{u}{\set{v,u}\in E}$ be the \emph{open neighborhood} of $v$, and let $N[v]=N(v) \cup \{v\}$ be the \emph{closed neighborhood} of $v$.
We generalize the notations $N(.)$ and $N[.]$ to sets, such that $N(V') = \cup_{v\in V'}N(v)$ and $N[V'] = \cup_{v\in V'}N[v]$.
A graph $G'=(V', E')$ is said to be a \emph{subgraph} of $G=(V, E)$ if $V' \subseteq V$ and $E' \subseteq E \cap (V' \times V')$. 
When $E' = E \cap (V' \times V')$, $G'$ is the subgraph \emph{induced} in $G$ by $V'$.
Let $\overline{V'} = V \setminus V'$ be the \emph{complement} of a set $V' \subseteq V$ of nodes. 
Let a \emph{motif} $\mu$ be a connected graph.
\emph{Enumerating} the motifs $\mu$ in a graph $G$ consists building the collection $M$ of all occurrences of $\mu$ as a subgraph of $G$.
Let $d(v)$ be the \emph{degree} of node $v$ and $\Delta$ be the maximum degree of $G$.
Let $d_\omega(v)$ be the \emph{weighted degree} of a node $v$ and $\Delta_\omega$ be the maximum weighted degree of $G$.
Let $d_\mu(v)$ be the \emph{motif~degree} of a node~$v$, i.e., the number of motifs $\mu \in M$ which contain $v$.
We generalize the notations $d(.)$, $d_\omega(.)$, and $d_\mu(.)$ to sets, such that the \emph{volume} of~$V'$ is~$d(V') = \sum_{v\in V'}d(v)$, the \emph{weighted~volume} of~$V'$ is~$d_\omega(V') = \sum_{v\in V'}d_\omega(v)$, and the \emph{motif~volume} of~$V'$ is~$d_\mu(V') = \sum_{v\in V'}d_\mu(v)$.
Let a \emph{spanning forest} of $G$ be an acyclic subgraph of $G$ containing all its nodes.
Let the \emph{arboricity} of $G$ be the minimum number of spanning forests of $G$ \hbox{necessary to cover all its edges}.

Let an \emph{undirected signed graph} be a graph $G=(V=\{0,\ldots, n-1\},E)$ which expands the above definition to permit edges with both positive and negative weights.
More specifically, let $\omega: E \to \MdR \setminus \{0\}$ be a \emph{signed} edge-weight function associated with a signed graph.
Let $E^-$ denote the set of edges with negative weight and $E^+$ denote the set of edges with positive weight, such that $E^- \cup E^+ = E$, $E^- \cap E^+ = \emptyset$, ${m}^- = |E^-|$, and ${m}^+ = |E^+|$.
Let $N^-(v) = N(v) \cap E^-$ denote the neighbors of $v$ {which} are connected to $v$ by an edge with negative weight.
Let $N^+(v) = N(v) \cap E^+$ denote the neighbors of $v$  {which} are connected to $v$ by an edge with positive weight.

Let $H=(\mathcal{V}=\{0,\ldots, \mathfrc{n}-1\},\mathcal{E})$ be an \emph{undirected hypergraph} with no multiple or self hyperedges allowed, with $\mathfrc{n} = |\mathcal{V}|$ \emph{nodes} and $\mathfrc{m} = |\mathcal{E}|$ \emph{hyperedges} (or \emph{nets}).
A net is defined as a subset of $\mathcal{V}$.
The nodes that compose a net are called its \emph{pins}.
Let $\mathfrc{c}: \mathcal{V} \to \MdR_{\geq 0}$ be a node-weight function, and let $\mathfrc{w}: \mathcal{E} \to \MdR_{>0}$ be a net-weight function.
We generalize $\mathfrc{c}$ and $\mathfrc{w}$ functions to sets, such that $\mathfrc{c}(\mathcal{V}') = \sum_{v\in \mathcal{V}'}\mathfrc{c}(v)$ and $\mathfrc{w}(\mathcal{E}') = \sum_{e\in \mathcal{E}'}\mathfrc{w}(e)$.
A node $v\in\mathcal{V}$ is \emph{incident} to a net $e\in\mathcal{E}$ if $v \in e$.
Let~$\mathcal{I}(v)$ be the set of incident nets of~$v$, let~$\mathfrc{d}(v) \Is |\mathcal{I}(v)|$ be the \emph{degree} of~$v$, and let $\mathfrc{d}_{\mathfrc{w}}(v) \Is \mathfrc{w}(\mathcal{I}(v))$ be the \emph{weighted degree} of $v$.
We generalize the notations $\mathfrc{d}(.)$ and  $\mathfrc{d}_{\mathfrc{w}}(.)$ to sets, such that the \emph{volume} of~$\mathcal{V}'$~is~$\mathfrc{d}(\mathcal{V}') = \sum_{v\in \mathcal{V}'}\mathfrc{d}(v)$ and the \emph{weighted~volume} of~$\mathcal{V}'$ is~$\mathfrc{d}_{\mathfrc{w}}(\mathcal{V}') = \sum_{v\in \mathcal{V}'}\mathfrc{d}_{\mathfrc{w}}(v)$.
Two nodes are \emph{adjacent} if both are incident to the same net.
Let the number of pins~$|e|$ in a net~$e$ be the \emph{size} of~$e$.
Given a cluster~$\mathcal{V}^\prime~\subseteq~\mathcal{V}$, the \emph{cut} or \emph{cut-net}~$cut(\mathcal{V}^\prime)$ of~$\mathcal{V}^\prime$ consists of the total weight of the nets crossing the cluster, i.e., $cut(\mathcal{V}^\prime) = \sum_{e \in \mathcal{E}^\prime}\mathfrc{w}(\mathcal{E}^\prime)$, in which $\mathcal{E}^\prime \Is $ $\big\{e \in \mathcal{E} : \exists i,j \mid e \cap \mathcal{V^\prime} \neq \emptyset, e \cap \overline{\mathcal{V^\prime}} \neq \emptyset , i\neq j\big\}$.

\section{Partitions and Clusterings}
\label{sec:Partitions and Clusterings}

The \emph{$k$-way (hyper)graph partitioning} problem consists of assigning each node of a (hyper)graph to exactly one of $k$ distinct \emph{blocks} respecting a balancing constraint in order to minimize the weight of the (hyper)edges running between the blocks, i.e., the edge-cut (resp. cut-net).
More precisely, it partitions $V$ into $k$ blocks $V_1$,\ldots,$V_k$ (i.e., $V_1\cup\cdots\cup V_k=V$ and $V_i\cap V_j=\emptyset$ for $i\neq j$), which is called a \emph{\mbox{$k$-partition}} of the (hyper)graph.
The \emph{edge-cut} (resp. \emph{cut-net}) of a $k$-partition consists of the total weight of the \emph{cut edges} (resp. \emph{cut nets}), i.e., edges (resp. nets) crossing blocks.
More formally, let the edge-cut (resp. cut-net) be $\sum_{i<j}\omega(E')$, in which $E' \Is $ $\big\{e\in E, \exists \set{u,v} \subseteq e : u\in V_i,v\in V_j, i \neq j\big\}$ is the~\emph{cut-set} (i.e.,~the set of all cut nets).
The \emph{balancing constraint} demands that the sum of node weights in each block does not exceed a threshold associated with some allowed \emph{imbalance}~$\epsilon$.
More specifically, $\forall i~\in~\{1,\ldots,k\} \gilt$ $c(V_i)\leq L_{\max}\Is \big\lceil(1+\epsilon) \frac{c(V)}{k} \big\rceil$.
For each net $e$ of a hypergraph, $\Lambda(e) := \{V_i~|~V_i \cap e \neq \emptyset\}$ denotes the \emph{connectivity set} of $e$.
The \emph{connectivity} $\lambda(e)$ of a net~$e$ is the cardinality of its connectivity set, i.e., $\lambda(e) := |\Lambda(e)|$.
The so-called \emph{connectivity} metric ($\lambda$-1) is computed as $\sum_{e\in E'} (\lambda(e) -1)~\omega(e)$, where $E'$ is the cut-set.

In the \emph{local graph clustering} problem, a graph $G=(V,E)$ and a seed node $u \in V$ are taken as input and the goal is to detect a \emph{well-characterized cluster} (or \emph{community}) $C \subset V$ containing~$u$.
A high-quality cluster $C$ usually contains nodes that are densely connected to one another and sparsely connected to $\overline{C}$.
There are many functions to quantify the quality of a cluster, such as \emph{modularity}~\cite{brandes2007modularity} and \emph{conductance}~\cite{kannan2004clusterings}.
The conductance metric is defined as $\phi(C)=|E'|/\min(d(C),d(\overline{C}))$, where $E'=E \cap (C \times \overline{C})$ is the set of edges shared by a cluster $C$ and its complement.
\emph{Local motif graph clustering} is a generalization of local graph clustering where a motif $\mu$ is taken as an additional input and the computed cluster optimizes a clustering metric based on $\mu$.
In particular, the \emph{motif conductance} $\phi_\mu(C)$ of a cluster $C$ is defined by 
\hbox{\citet{benson2016higher}} as a generalization of the conductance in the following way:
$\phi_\mu(C)=|M'|/min(d_\mu(C),d_\mu(\overline{C}))$, where $M'$ are all the motifs $\mu$ which contain at least one node in $C$ and one node in $\overline{C}$.
Note that, if the motif under consideration is simply an edge, then $|M'|$ is the \hbox{edge-cut and $\phi_\mu(C)=\phi(C)$}.

Let a \emph{clustering} of a graph $G=(V,E)$ be any partition of $V$, i.e., a set of \emph{blocks} or \emph{clusters} $V_1$,\ldots,$V_t \subset V$ such that $V_1\cup\cdots\cup V_t=V$ and $V_i\cap V_j=\emptyset$, where $t\in[1,n]$.
An abstract view of the clustering is a \emph{quotient graph} $\mathcal{Q}$, in which nodes represent clusters and edges are induced by the connectivity between clusters.
We call \emph{neighboring clusters} a pair of clusters that is connected by an edge in the quotient graph.
A node $v \in V_i$ that has a neighbor $w \in V_j, i\neq j$, is a boundary node.
The \emph{edge-cut} of a clustering consists of the total weight of the edges crossing clusters (also called \emph{cut edges}), i.e., $\sum_{i<j}\omega(E_{ij})$, where $E_{ij}\Is\setGilt{\set{u,v}\in E}{u\in V_i,v\in V_j}$.
Note that $\frac{1}{2}d^-_\omega(V)$ is an absolute lower bound for the edge-cut of any clustering.
Let a signed graph be \emph{balanced} if there exists a clustering of it with edge-cut equal to $\frac{1}{2}d^-_\omega(V)$, i.e., where all edges with positive weight are inside clusters and all edges with negative weight are cut.
If such clustering does not exist, e.g., a triangle with only one edge with negative weight, we say that the signed graph is \emph{unbalanced}.
Let the \emph{signed graph clustering} (SGC) consist of obtaining a clustering of an undirected signed graph $G$ \hbox{in order to minimize the edge-cut}.

\section{Process Mapping}
\label{sec:Process Mapping}

For process mapping applications of hierarchical partitions, assume that we have $n$ processes and a topology containing~$k$ PEs.
Let $\mathcal{C}\in \MdR^{n \times n}$ denote the communication matrix and let $\mathcal{D}\in \MdR^{k \times k}$ denote the (implicit) topology matrix or distance matrix.
In particular, $\mathcal{C}_{i,j}$ represents the required amount of communication between processes $i$ and $j$, while $\mathcal{D}_{x,y}$ represents the cost of each communication between PEs $x$ and $y$.
Hence, if processes $i$ and $j$ are respectively assigned to PEs $x$ and $y$, or vice-versa, the communication cost between $i$ and $j$ will be $\mathcal{C}_{i,j}\mathcal{D}_{x,y}$.
Throughout this thesis, we assume that $\mathcal{C}$ and $\mathcal{D}$ are symmetric -- otherwise one can create equivalent problems with symmetric inputs \cite{brandfass2013rank}.

In particular, for process mapping applications tackled in this paper, we assume that topologies are organized as homogeneous hierarchies. 
In this case $\mathcal{S}=a_1: a_2: ...:a_\ell$ is a
sequence describing the hierarchy of a supercomputer. The sequence
should be interpreted as each processor having $a_1$ cores, each node
$a_2$ processors, each rack $a_3$ nodes, and so~forth,
such that the total number of PEs is $k=\Pi_{i=1}^{\ell}a_i$.
Without loss of generality, we assume that $a_i \geq 2, \forall i \in \{1, \ldots, \ell\}$.
Let $D = d_1:d_2:\ldots:d_\ell$ be a sequence describing the distance between PEs within each hierarchy level, meaning that the distance between two cores in the same processor is $d_1$, the distance between two cores in the same node but in different processors is $d_2$, the distance between two cores in the same rack but in different nodes is $d_3$, and~so~forth.  
The \emph{process mapping} problem consists of assigning the nodes of a graph to PEs in a communication topology while respecting a balancing constraint in order to minimize the emph{total communication cost}.
Let \mbox{$\Pi: \{1, \ldots, n\} \mapsto \{1, \ldots, k\}$} be the function that maps a node onto its PE.
The  objective of process mapping is to minimize Equation~(\ref{eq:process_mapping_obj}). 

\begin{equation}
	J(\mathcal{C},\mathcal{D}, \Pi) := \sum_{i,j} \mathcal{C}_{i, j}\mathcal{D}_{\Pi(i),\Pi(j)}.
	\label{eq:process_mapping_obj}
\end{equation}

The exact total communication cost of a given mapping of processes onto PEs depends on a combination of \emph{bandwidth} and \emph{latency}, which depends on a multitude of factors in practice.
The total communication cost in latency-based topologies is also known as the \emph{Coco(.)} or \emph{hop-byte} objective function.
An alternative objective function to the total communication cost is the bandwidth-based metric \emph{maximum congestion}, which is defined as the maximum number of message exchanges through any link of the topology graph.
Another possible objective is the \emph{maximum dilation}, which is defined as the maximum communication cost directly associated with a pair of PEs \hbox{for a given mapping}.

\section{Flows}
\label{sec:Flows}

Let~$G_f=(V_f,E_f)$ be a flow graph.
A flow graph has one source node~$s \in V_f$, one sink node~$t \in V_f$, and a set of remaining nodes $V \setminus \{s,t\}$.
All edges~$e=(u,v)$ in a flow graph are directed and associated with a nonnegative capacity~$cap(u,v)$.
An \hbox{s-t flow} is a function \hbox{$f:V_f\times{V_f}\rightarrow\MdR_{>0}$} which satisfies a \emph{capacity} constraint, \ie $f(u,v) \leq cap(u,v)$, a \emph{symmetry} constraint, \ie $\forall{u,v}\in{V_f}: f(u,v)=-f(v,u)$, and a \emph{flow conservation} constraint, \ie \hbox{$\forall{u}\in{V_f}\setminus\{s,t\}:$} \hbox{$\sum_{v\in V_f}f(u,v)=0$}.
An edge~$(u,v)$ is called~\emph{saturated} if \hbox{$cap(u,v)=f(u,v)$};
The total amount of flow moved from~$s$~to~$t$ is defined as the \emph{value}~$|f|$ of~$f$ and is computed as follows: \hbox{$|f|=\sum_{u\in V_f}f(u,t)=\sum_{v\in V_f}f(s,v)$}.
A given \hbox{s-t~flow~$f$} in~$G_f$ is \emph{maximum} if, for any \hbox{s-t~flow}~$f'$ in~$G_f$,~$|f'|\leq|f|$.
Let~$G_r=(V_f,E_r)$ be the \emph{residual graph} associated with a given flow~$f$~on~$G_f$, such that $E_r=$ {$\{(u,v) \in V_f \times V_f: cap(u,v)-f(u,v)>0\}$}.
According to the Max-Flow Min-Cut Theorem~\cite{ford_fulkerson_1956}, the value~$|f|$ of a maximum s-t flow~$f$ on~$G_f$ equals the weight of a minimum s-t cut on~$G_f$, \ie a 2-way partition of~$G_f$ where edge weights equal edge capacities, $s$~and~$t$ are in distinct blocks, and the total weight of the cut edges is minimum.
To find the sink side of the minimum cut associated with a maximum flow in $G_f$, a reverse breadth-first search can be performed~on~$G_f$ \hbox{starting at the~sink~node~$t$}.

For each node~$u$ in a flow graph~${V_f}$, let~\hbox{$exc(u)=\sum_{v\in V_f}f(u,v)$} be its \emph{excess} value and~$d(u)$ be its potential.
A node~$u$ is called \emph{active} if \hbox{$exc(u)>0$}.
An edge~$(u,v)$ is called~\emph{admissible} if \hbox{$cap(u,v)-f(u,v)>0$} and~\hbox{$d(u)=d(v)+1$}.
The push-relabel~\cite{push_relabel88} algorithm builds a maximum flow by computing a succession of \emph{preflows}, \ie~flows where the flow conservation constraint is substituted by \hbox{$\forall{u}\in{V_f}\setminus\{s,t\}:$} \hbox{$exc(u)\geq0$}. 
In the initial preflow, all out-edges of~$s$ are saturated, $\forall{u}\in{V_f}\setminus\{s\}:d(u)=0$, and $d(s)=|V_f|$.
The initial preflow is evolved via operations \emph{push}, \ie sending as much flow as possible from an active node through an admissible edge, and \emph{relabel}, \ie increasing the potential of a node until it~becomes~active.
Preflows induce minimum sink-side cuts, so a maximum flow and a minimum cut are obtained once \hbox{no node is active.}

A common technique to solve flow and cut problems on hypergraphs consists of transforming them in directed graphs and then applying traditional graph-based techniques on~them.
Among the existing transformations~\cite{veldt2022hypergraph,lawler1973cutsets}, we highlight \emph{clique expansion}, \emph{star expansion}, and \emph{Lawler expansion}.
In the \emph{clique expansion}, each net is represented by a clique, i.e., a set of edges connecting each pair of its pins in both directions.
In this approach, the weight of each edge is equal to weight of the corresponding net~$e$ divided by~$|e|-1$ and parallel edges are substituted by a single edge whose weight is the sum of the weights of the removed edges.
In the \emph{star expansion}, each net is represented by an auxiliary artificial node connected to its pins by edges in both directions. 
In this expansion, the edges have the same weight as the corresponding net.
In the \emph{Lawler expansion}, each net~$e$is represented by two auxiliary artificial nodes~$w_1$~and~$w_1$ and a collection of edges.
In particular, there is a directed edge $(w_1,w_2)$ which has the same weight as the corresponding net.
Additionally, each pin of the corresponding net has an out-edge to~$w_1$ and an in-edge from~$w_2$, each of them with weight infinity.
The three transformation approaches are exemplified in Figure~\ref{fig:social_hyperedge_expansion}.

\begin{figure}[t]
	\centering
	\includegraphics[width=0.9\textwidth]{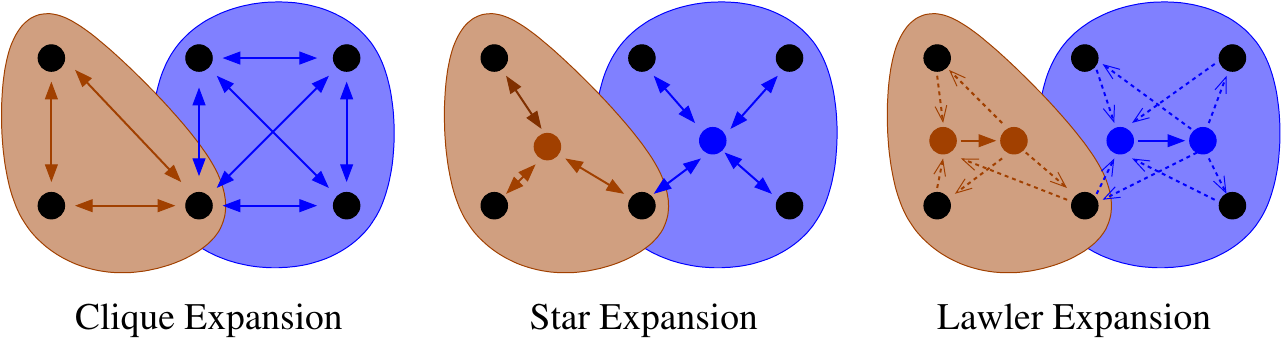}
	\caption{Three existing approaches to represent a hypergraph using a directed graph. Nodes and nets of the original hypergraph are respectively represented by black circles and colored areas around them. Auxiliary artificial nodes and edges are respectively represented by circles and arrows, both with the same color as the corresponding net. Bidirectional arrows represent a pair of edges in both directions. Solid edges have finite weight while dashed edges have \hbox{infinite weight}.}
	\label{fig:social_hyperedge_expansion}
\end{figure}

\section{Multilevel Scheme}
\label{sec:Multilevel Scheme}

In this section, we characterize the \emph{multilevel} scheme, which is a successful heuristic for clustering or partitioning large (hyper)graphs.
Before describing the multilevel scheme, we need to define the terms contraction and uncontraction. 
\emph{Contracting} an edge $e=\set{u,v}$ consists of replacing the nodes $u$ and $v$ by a new node $x$ connected to the former neighbors of $u$ and $v$ and setting $c(x)=c(u)+c(v)$.
If replacing edges of the form $\set{u,w}$, $\set{v,w}$ would generate two parallel edges $\set{x,w}$, a single edge with
$\omega(\set{x,w})=\omega(\set{u,w})+\omega(\set{v,w})$ is inserted.
\emph{Contracting} a cluster of nodes $C=\set{u_1, \ldots, u_{\ell}}$ involves replacing them with a new node~$v$ whose weight is the sum of the weights of the clustered nodes and is connected to all elements  $w \in \bigcup_{i=1}^{\ell} N(u_i)$, \hbox{$\omega(\set{v,w})=\sum_{i=1}^{\ell}\omega(\set{u_i,w})$}.
This ensures the transfer of partition from a coarser to a finer level maintains the edge-cut, as exemplified in Figure~\ref{fig:sgc_clustercontraction}.
We define the \emph{contraction} operator as~$\big/$ such that $G \big/ \mathcal{V}^\prime$, with $\mathcal{V}^\prime \subseteq \mathcal{V}$, is the (hyper)graph obtained by contracting the nodes from~$\mathcal{V}^\prime$ on~$G$.
The \emph{uncontraction} of a node \hbox{undoes the contraction}.

We describe the multilevel scheme within the scope of graph partitioning, although the basic idea is also extensible to other problems, such as process mapping, hypergraph partitioningm and graph clustering.
A \emph{multilevel approach} consists of three main phases.
In the \emph{contraction} (coarsening) phase, successive approximations of an original input graph are created. 
The contractions quickly reduce the size of the graph and stop as soon as it becomes sufficiently small to be partitioned by an expensive algorithm.
In the construction phase, an \emph{initial partitioning} is obtained by partitioning the coarsest graph.
Due to the way we define contraction, every partition of the coarsest level implies a corresponding partition of the input graph with equal edge-cut and balance.
In the \emph{local improvement} (or \emph{uncoarsening}) phase, we uncontract previously contracted nodes to go back through each level, from the coarsest approximation to the original graph. 
After each uncoarsening, local improvement algorithms move nodes between blocks in order to improve the objective function or balance. 
Local search moves nodes between blocks to reduce the objective, as \hbox{exemplified in Figure~\ref{fig:sgc_localsearch}}.  

\begin{figure}[t]
	\centering
	\includegraphics[width=0.75\textwidth]{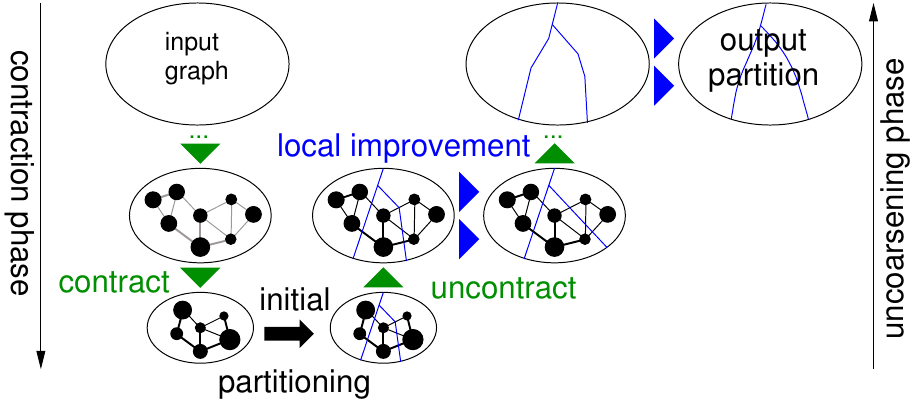}
	\caption{Multilevel scheme (adapted from~\cite{kaffpa}).}
	\label{fig:sgc_MSGC}
\end{figure}

\begin{figure}[t]
	\centering
	\includegraphics[width=.8\textwidth]{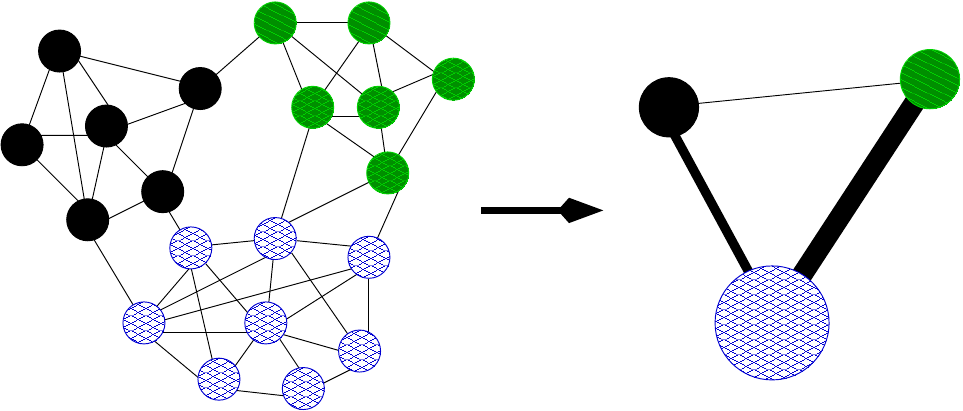}
	\caption{Contraction of a clustering~\cite{pcomplexnetworksviacluster}. Each cluster is represented by a different color on the left hand side graph. Each cluster on the left-hand side is contracted to a single node on the right hand side.}
	\label{fig:sgc_clustercontraction}
\end{figure}

\begin{figure}[t]
	\centering
	\includegraphics[width=.9\textwidth]{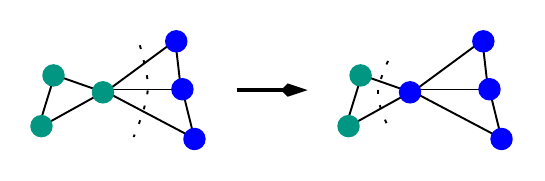}
	\caption{Typical step of a local search algorithm. Cluster assignments are indicated by colors. In this example, a single node is moved to another cluster in order to decrease the overall edge-cut.}
	\label{fig:sgc_localsearch}
\end{figure} 

% \vfill

\section{Evolutionary Algorithms}
\label{sec:Evolutionary Algorithms}

\emph{Evolutionary} or \emph{Memetic algorithms}~\cite{evolutionary_book} are population-based heuristics that mimic natural evolution to optimize a problem.
They use a convenient \emph{notation} to represent the decision variables, called \emph{individuals}, and evaluate their quality through a \emph{fitness} function (simulation or mathematical~\cite{sastry2014genetic}). 
A \emph{population} of individuals evolves during the algorithm. 
\emph{Recombination}~\cite{goldberg1989messy} (crossover) \emph{exploits} characteristics of previous individuals to create new and better solutions. 
\emph{Mutation}~\cite{Michalewicz:1996:GAD:229930} introduces random variations to \emph{explore} the search space and escape local optima. 
Evolutionary algorithms do not guarantee optimality, but are effective in \hbox{exploring and exploiting the solution space}.

% \vfill

\section{Computational Models}
\label{sec:Computational Models}

Streaming algorithms usually follow a load-compute-store logic shown in Figure~\ref{fig:streaming}.
The classic streaming model is the \emph{one-pass} model, in which the nodes are loaded one at a time alongside with their (hyper)edges, then some logic is applied to permanently assign them to blocks. %
This logic can be as simple as a \AlgName{Hashing}  function or as complex as scoring all blocks based on some objective and then assigning the node to the block with highest score.
When assignment decisions of an algorithm for the current node depend on the previous decisions, an algorithm in the model has to store the assignment of the previous loaded nodes and hence needs $\Omega(n)$ space.
An extended version of this model is called the \emph{buffered streaming} model.
More precisely, a $\delta$-sized \emph{buffer} or \emph{batch} of input nodes with their neighborhood is repeatedly loaded.
Partition/block assignment decisions have to be made after the whole buffer is loaded. 
While we investigate the dependence of our algorithm on this parameter, in practice the parameter  will depend on the amount of available memory on a machine. 
The parameter can be dynamically chosen such that the buffer is ``full'' if $\Theta(n)$ space has been loaded from the disk.
Hence the buffered streaming model asymptotically does not need more space than a one-pass streaming algorithm if this setting is used. 
This holds true even in the worst case: when a node has \hbox{degree close to $n$}.

\begin{figure}[t]
	\centering
	\includegraphics[width=1.\textwidth]{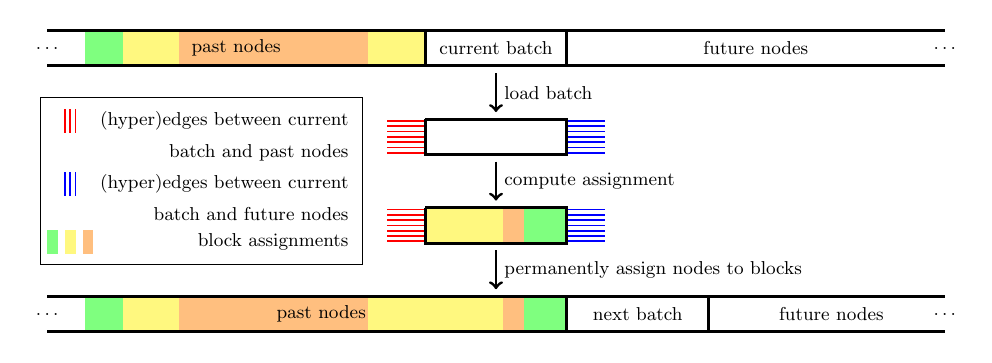}
	\caption{Typical layout of streaming algorithm for (hyper)graph partitioning.}
	\label{fig:streaming}
\end{figure}

\section{Instances}
\label{sec:Instances}

Throughout this thesis we present experiments on various kinds of (signed) (hyper)graphs. 
In this section, we summarize the main properties, the source, and the area of application of the graphs. 

% \vfill

\subsection{Graphs}
\label{subsec:Graphs}

\paragraph*{Experiments with Streaming Algorithms.}
In our experiments with streaming algorithms, we use graphs from various sources~\cite{snapnets,nr-aaai15,benchmarksfornetworksanalysis,kappa,funke2019communication}.
Most of the considered graphs were used for benchmark in previous works on graph partitioning.
The graphs wiki-Talk and web-Google, as well as most networks of co-purchasing, roads, social, web, autonomous systems, citations, circuits, similarity, meshes, and miscellaneous are publicly available either in~\cite{snapnets} or in \cite{nr-aaai15}.
We also use graphs such as eu-2005, in-2004, uk-2002, and uk-2007-05, which are available at the 10$^{th}$ DIMACS Implementation Challenge website~\cite{benchmarksfornetworksanalysis}. 
Finally, we include some artificial random graphs.
We use the name \Id{rggX} for \emph{random geometric graph} with
$2^{X}$ nodes where nodes represent random points in the unit square and edges connect nodes whose Euclidean distance is below $0.55 \sqrt{ \ln n / n }$. 
We use the name \Id{delX} for a graph based on a Delaunay triangulation of $2^{X}$ random points in the unit square~\cite{kappa}.
We use the name \Id{RHGX} for random hyperbolic graphs~\cite{funke2019communication,DBLP:journals/corr/abs-2003-00736} with $10^8$~nodes and $X\times 10^9$~edges.
Basic properties of the graphs under consideration can be found in Table~\ref{tab:heistream_graphs}.
For our experiments, we split the graphs in three disjoint sets.
A \emph{tuning} set for the parameter study experiments, a \emph{test} set for the comparisons against the state-of-the-art, and a set of \emph{huge graphs} for special larger~scale~tests. 
In any case, when streaming the graphs we use the natural given \hbox{order of the nodes}.

\begin{table}[t]
	\scriptsize
	\centering
	\setlength{\tabcolsep}{3pt}
	\begin{tabular}[t]{| l  r  r  r | }
		\hline
		Graph & $n$& $m$ & Type\\
		\hline  \hline
		\multicolumn{4}{|c|}{Tuning Set} \\
		\hline 
		coAuthorsCiteseer & \numprint{227320}    & \numprint{814134}  & Citations \\
		citationCiteseer & \numprint{268495}    & \numprint{1156647}  & Citations \\
		amazon0312 & \numprint{400727}  & \numprint{2349869} & Co-Purch. \\
		amazon0601 & \numprint{403364} & \numprint{2443311} & Co-Purch. \\
		amazon0505 & \numprint{410236} & \numprint{2439437} & Co-Purch. \\
		roadNet-PA & \numprint{1087562} & \numprint{1541514} & Roads \\
		com-Youtube & \numprint{1134890}  & \numprint{2987624} & Social \\
		soc-lastfm & \numprint{1191805} & \numprint{4519330}   & Social \\
		roadNet-TX & \numprint{1351137} & \numprint{1879201} & Roads \\
		in-2004 & \numprint{1382908}    & \numprint{13591473}  & Web \\
		G3\_circuit & \numprint{1585478}    & \numprint{3037674}  & Circuit \\
		soc-pokec & \numprint{1632803} & \numprint{22301964}   & Social \\
		as-Skitter & \numprint{1694616}    & \numprint{11094209}  & Aut.Syst. \\
		wiki-topcats & \numprint{1791489} & \numprint{28511807} & Social \\
		roadNet-CA & \numprint{1957027}  & \numprint{2760388} & Roads \\
		wiki-Talk & \numprint{2388953}     & \numprint{4656682}  & Web \\
		soc-flixster & \numprint{2523386} & \numprint{7918801}   & Social \\
		del22 & \numprint{4194304}  & \numprint{12582869} & Artificial \\
		rgg22 & \numprint{4194304} & \numprint{30359198} & Artificial \\
		del23 & \numprint{8388608}  & \numprint{25165784} & Artificial \\
		rgg23 & \numprint{8388608} & \numprint{63501393} & Artificial \\
		\hline
		\multicolumn{4}{|c|}{Huge Graphs} \\
		\hline 
		uk-2005 &   \numprint{39459923} & \numprint{783027125} & Web \\
		twitter7 & \numprint{41652230} & \numprint{1202513046} & Social \\
		sk-2005 & \numprint{50636154} & \numprint{1810063330} & Web \\
		soc-friendster & \numprint{65608366}  & \numprint{1806067135} & Social \\
		er-fact1.5s26 & \numprint{67108864}  & \numprint{907090182} & Artificial \\
		RHG1 & \numprint{100000000}  & \numprint{1000913106} & Artificial \\
		RHG2 & \numprint{100000000}  & \numprint{1999544833} & Artificial \\
		uk-2007-05 & \numprint{105896555} & \numprint{3301876564}   & Web \\
		\hline
	\end{tabular}
	\begin{tabular}[t]{| l  r  r  r | }
		\hline
		Graph & $n$& $m$ & Type\\
		\hline  \hline	
		\multicolumn{4}{|c|}{Test Set} \\
		\hline 
		Dubcova1 & \numprint{16129} & \numprint{118440} & Meshes \\
		hcircuit & \numprint{105676}  & \numprint{203734} & Circuit \\
		coAuthorsDBLP & \numprint{299067}     & \numprint{977676}  & Citations \\
		Web-NotreDame & \numprint{325729}     & \numprint{1090108}  & Web \\
		Dblp-2010 & \numprint{326186}     & \numprint{807700}  & Citations \\
		ML\_Laplace & \numprint{377002} & \numprint{13656485} & Meshes \\
		coPapersCiteseer & \numprint{434102}     & \numprint{16036720}  & Citations \\
		coPapersDBLP & \numprint{540486}     & \numprint{15245729}  & Citations \\
		Amazon-2008 & \numprint{735323}  & \numprint{3523472} & Similarity \\
		eu-2005 & \numprint{862664}    & \numprint{16138468}  & Web \\
		web-Google & \numprint{916428}    & \numprint{4322051}  & Web \\
		ca-hollywood-2009 & \numprint{1087562} & \numprint{1541514} & Roads \\
		Flan\_1565 & \numprint{1564794} & \numprint{57920625} & Meshes \\
		Ljournal-2008 & \numprint{1957027}  & \numprint{2760388} & Social \\
		HV15R & \numprint{2017169}  & \numprint{162357569} & Meshes \\
		Bump\_2911 & \numprint{2911419}  & \numprint{62409240} & Meshes \\
		del21 & \numprint{2097152}  & \numprint{6291408} & Artificial \\	
		rgg21 & \numprint{2097152} & \numprint{14487995} & Artificial \\
		FullChip & \numprint{2987012} & \numprint{11817567} & Circuit \\
		soc-orkut-dir & \numprint{3072441} & \numprint{117185083} & Social \\
		patents & \numprint{3750822}     & \numprint{14970766}  & Citations \\
		cit-Patents & \numprint{3774768}     & \numprint{16518947}  & Citations \\
		soc-LiveJournal1 & \numprint{4847571} & \numprint{42851237}   & Social \\
		circuit5M & \numprint{5558326} & \numprint{26983926} & Circuit \\
		italy-osm & \numprint{6686493}  & \numprint{7013978} & Roads \\
		great-britain-osm & \numprint{7733822} & \numprint{8156517} & Roads \\
		\hline
	\end{tabular}
	\caption{Graphs for experiments with streaming algorithms.}
	\label{tab:heistream_graphs}
\end{table}

\paragraph*{Experiments with Multilevel Process Mapping.}
\label{par:Process Mapping Paper}

The instances used in our experiments with multilevel process mapping come from various sources to test our algorithm.
We use the largest six graphs from Chris Walshaw's benchmark archive~\cite{soper2004combined}.
Graphs derived from sparse matrices have been taken from the SuiteSparse Matrix Collection~\cite{davis2011university}. 
We also use graphs from the 10th DIMACS Implementation Challenge~\cite{benchmarksfornetworksanalysis} website. 
Here, \Id{rggX} and \Id{delX} are defined as before.
The graphs \Id{af_shell9}, \Id{thermal2},  and \Id{nlr} are from the matrix and the numeric section of the DIMACS benchmark set.
The graphs \Id{eur} and \Id{deu} are large road networks of Europe and Germany taken from~\cite{DSSW09}. 
Basic properties of the graphs under consideration can be \hbox{found in Table~\ref{tab:test_instances_walshaw}}.

\begin{table}[t]
	\scriptsize
	\centering
	\setlength{\tabcolsep}{5pt}
	\begin{tabular}[t]{| l r r r | }
			\hline
			Graph & $n$& $m$ & Type \\
		 	\hline \hline
		 	 \multicolumn{4}{|c|}{Tuning Graphs}\\
			\hline
			  fe\_rotor                                    & \numprint{99617} & \numprint{662431} & Miscellaneous \\
			  598a                                         & \numprint{110971} & \numprint{741934} & Miscellaneous \\
			  ecology2                                     & \numprint{999999}  & \numprint{1997996} & Circuit \\
			  G3\_circuit                                  & \numprint{1585478} & \numprint{3037674} & Circuit \\
			  del22                                        & \numprint{4194304} & \numprint{12582869} & Artificial \\
			  rgg22                                        & \numprint{4194304} & \numprint{30359198} & Artificial \\
                          \hline
		 	 \multicolumn{4}{|c|}{UF Graphs}\\
			\hline
			  cop20k\_A                                    & \numprint{99843}  & \numprint{1262244} & Miscellaneous \\
			  2cubes\_sphere                               & \numprint{101492} & \numprint{772886} & Miscellaneous \\
			  thermomech\_TC                               & \numprint{102158} & \numprint{304700} & Miscellaneous \\
			  cfd2                                         & \numprint{123440} & \numprint{1482229} & Miscellaneous \\
			  boneS01                                      & \numprint{127224} & \numprint{3293964} & Miscellaneous \\
			  Dubcova3                                     & \numprint{146689} & \numprint{1744980} & Miscellaneous \\
			  bmwcra\_1                                    & \numprint{148770} & \numprint{5247616} & Numerical \\
			  G2\_circuit                                  & \numprint{150102} & \numprint{288286}  & Circuit \\
			  shipsec5                                     & \numprint{179860} & \numprint{4966618} & Miscellaneous \\
			  cont-300                                     & \numprint{180895} & \numprint{448799}   & Miscellaneous \\
                          \hline
	\end{tabular}
	\begin{tabular}[t]{| l r r r | }
			\hline
			Graph & $n$& $m$ & Type \\
		 	\hline \hline
		 	  \multicolumn{4}{|c|}{ Large Walshaw Graphs}  \\
                          \hline
			  598a                                           & \numprint{110971} & \numprint{741934}    & Meshes \\
			  fe\_ocean                                      & \numprint{143437} & \numprint{409593}    & Miscellaneous \\
			  144                                            & \numprint{144649} & \numprint{1074393}   & Meshes \\
			  wave                                           & \numprint{156317} & \numprint{1059331}  & Meshes \\
			  m14b                                           & \numprint{214765} & \numprint{1679018}   & Meshes \\
			  auto                                           & \numprint{448695} & \numprint{3314611}   & Meshes \\
                          \hline
		 	   \multicolumn{4}{|c|}{ Large Other Graphs}\\
                          \hline
			  af\_shell9                                   & \numprint{504855} & \numprint{8542010} & Miscellaneous \\
			  thermal2                                     & \numprint{1227087} & \numprint{3676134} & Miscellaneous \\
			  nlr                                          & \numprint{4163763} & \numprint{12487976} & Meshes \\
			  deu                                          & \numprint{4378446} & \numprint{5483587} & Roads \\
			  del23                                        & \numprint{8388608} & \numprint{25165784} & Artificial \\
			  rgg23                                        & \numprint{8388608} & \numprint{63501393} & Artificial \\
			  del24                                        & \numprint{16777216} & \numprint{50331601} & Artificial \\
			  rgg24                                        & \numprint{16777216} & \numprint{132557200} & Artificial \\
			  eur                                          & \numprint{18029721} & \numprint{22217686} & Roads \\
		 	\hline
	\end{tabular}
 	\caption{Graphs for multilevel process mapping experiments.}
 	\label{tab:test_instances_walshaw}
\end{table}

\begin{table}[t]
	\centering
	\scriptsize
	\setlength{\tabcolsep}{3pt}
	\begin{tabular}[t]{| l  r  r  r | }
		\hline
		Graph & $n$& $m$ & Triangles\\
		\hline  \hline
		\multicolumn{4}{|c|}{Tuning Set} \\
		\hline 
		citationCiteseer & \numprint{268495}    & \numprint{1156647}  & \numprint{847420} \\
		coAuthorsCiteseer & \numprint{227320}    & \numprint{814134}  & \numprint{2713298} \\
		amazon0312 & \numprint{400727}  & \numprint{2349869} & \numprint{3686467} \\
		amazon0505 & \numprint{410236} & \numprint{2439437} & \numprint{3951063} \\
		amazon0601 & \numprint{403364} & \numprint{2443311} & \numprint{3986507} \\
		del22 & \numprint{4194304}  & \numprint{12582869} & \numprint{8436672} \\
		del23 & \numprint{8388608}  & \numprint{25165784} & \numprint{16873359} \\
		soc-pokec & \numprint{1632803} & \numprint{22301964}   & \numprint{32557458} \\
		rgg22 & \numprint{4194304} & \numprint{30359198} & \numprint{85962754} \\
		rgg23 & \numprint{8388608} & \numprint{63501393} & \numprint{188022664} \\
		in-2004 & \numprint{1382908}    & \numprint{13591473}  & \numprint{464257245} \\
		\hline
	\end{tabular}
	\begin{tabular}[t]{| l  r  r  r | }
		\hline
		Graph & $n$& $m$ & Triangles\\
		\hline  \hline
		\multicolumn{4}{|c|}{Test Set} \\
		\hline 
		com-amazon & \numprint{334863}  & \numprint{925872} & \numprint{667129} \\
		com-dblp & \numprint{317080}  & \numprint{1049866} & \numprint{2224385} \\
		com-youtube & \numprint{1134890}  & \numprint{2987624} & \numprint{3056386} \\
		com-livejournal & \numprint{3997962}  & \numprint{34681189} & \numprint{177820130} \\
		com-orkut & \numprint{3072441}  & \numprint{117185083} & \numprint{627584181} \\
		com-friendster & \numprint{65608366} & \numprint{1806067135} & \numprint{4173724142} \\		
		\hline
	\end{tabular}
	\caption{Graphs for local motif clustering experiments.}
	\label{tab:lmcvhgp_graphs}
\end{table}

\paragraph*{Experiments with Local Motif Clustering.}
\label{par:Experiments with Local Motif Clustering}

In our experiments with local motif clustering, we use graphs from various sources~\cite{snapnets,nr-aaai15,benchmarksfornetworksanalysis}.
Most of the considered graphs were used for benchmark in previous works in the area.
Basic properties of the graphs under consideration can be found in Table~\ref{tab:lmcvhgp_graphs}.
For our experiments, we split the graphs in two disjoint sets: a \emph{tuning} set for the parameter study experiments and a \emph{test} set for the comparisons against the state-of-the-art. 
The graphs in the test set are exactly the graphs used in~\cite{yin2017local}.

% \vfill

\subsection{Hypergraphs}
\label{subsec:Hypergraphs}

In our experiments, we consider hypergraphs that have been used for benchmark in previous works on hypergraph partitioning.
We use the same benchmark as in~\cite{schlag2016k}.
This consists of 310~hypergraphs from three benchmark sets: 
18 hypergraphs from the ISPD98 Circuit Benchmark Suite~\cite{alpert1998ispd98}, 
192 hypergraphs based on the University of Florida Sparse Matrix Collection~\cite{davis2011university}, and 
100 instances from the international SAT Competition 2014~\cite{Belov2014}. 
The SAT instances were converted into hypergraphs by mapping each boolean variable and its complement to a node and each clause to a net. 
From the Sparse Matrix Collection, one matrix was selected for each application area that had between \numprint{10000} and \numprint{10000000} columns. 
The matrices were converted into hypergraphs using the row-net model, in which each row is treated as a net and \hbox{each column as a node}.

\subsection{Signed Graphs}
\label{subsec:Signed Graphs}

For experiments with signed graphs, we use the real-world signed graphs listed in Table~\ref{tab:sgc_graphs}.
They were all obtained from the public graph collections SNAP~\cite{snap} and KONECT~\cite{konect}.

\begin{table}[t]
	\centering
	\scriptsize
	\setlength{\tabcolsep}{5pt}
	\begin{tabular}{|lrrr|}
		\hline
		Graph          & $n$  & $m$ & Type   \\ \hline\hline
		{bitcoinalpha}   & \numprint{3783}   & \numprint{14081}   & Commerce \\
		{bitcoinotc}     & \numprint{5881}   & \numprint{21434}   & Commerce \\
		{elec}           & \numprint{7118}   & \numprint{100355}  & Miscellaneous \\
		{chess}          & \numprint{7301}   & \numprint{32650}   & Miscellaneous \\
		{slashdot081106} & \numprint{77357}  & \numprint{466666}  & Social \\
		{slashdot-zoo}   & \numprint{79116}  & \numprint{465840}  & Social \\
		{slashdot090216} & \numprint{81871}  & \numprint{495666}  & Social \\
		{slashdot090221} & \numprint{82144}  & \numprint{498532}  & Social \\
		{wikiconflict}   & \numprint{118100} & \numprint{1461058} & Social \\
		{epinions}       & \numprint{131828} & \numprint{708507}  & Social \\
		{wikisigned-k2}  & \numprint{138592} & \numprint{712337}  & Social \\
		\hline
	\end{tabular}
	\caption{Signed real-world graphs for signed graph clustering experiments.}
	\label{tab:sgc_graphs}
\end{table}

\section{Machines}
\label{sec:Machines}

We now describe the five machines that are used in the following chapters.
With hyperthreading, all used machine are capable of handling a number of threads equal to twice the number of their cores. 
\textbf{Machine A} has a two six-core Intel Xeon  E5-2630 processor running at $2.8$ GHz, $64$ GB of main memory, and $3$ MB of L2-Cache. 
It runs Ubuntu GNU/Linux 20.04.1 and Linux kernel version 5.4.0-48. %
\textbf{Machine B} has a four-core Intel Xeon E5420 processor running at $2.5$ GHz, $16$ GB of main memory, and $24$ MB of L2-Cache. 
The machine runs Ubuntu GNU/Linux 20.04.1 and Linux kernel version 5.4.0-65. %
\textbf{Machine C} has a sixteen-core Intel Xeon  Silver 4216 processor running at $2.1$ GHz, $100$ GB of main memory, $16$ MB of L2-Cache, and $22$ MB of L3-Cache running Ubuntu 20.04.1. The machine can handle 32 threads with hyperthreading. %
\textbf{Machine D} has a sixty-four-core AMD EPYC 7702P processor running at $2.0$ GHz, $1$ TB of main memory, $32$ MB of L2-Cache, and $256$ MB of L3-Cache. %
\textbf{Machine E} has four sixteen-core Intel Xeon Haswell-EX E7-8867 processors running at $2.5$ GHz, $1$ TB of main memory, and $32768$ KB of L2-Cache.
The machine runs Debian GNU/Linux 10 and Linux kernel version 4.19.67-2. %

% \vfill

\begin{figure*}[tb]
	\captionsetup[subfigure]{justification=centering}
	\centering
	\includegraphics[width=0.48\textwidth]{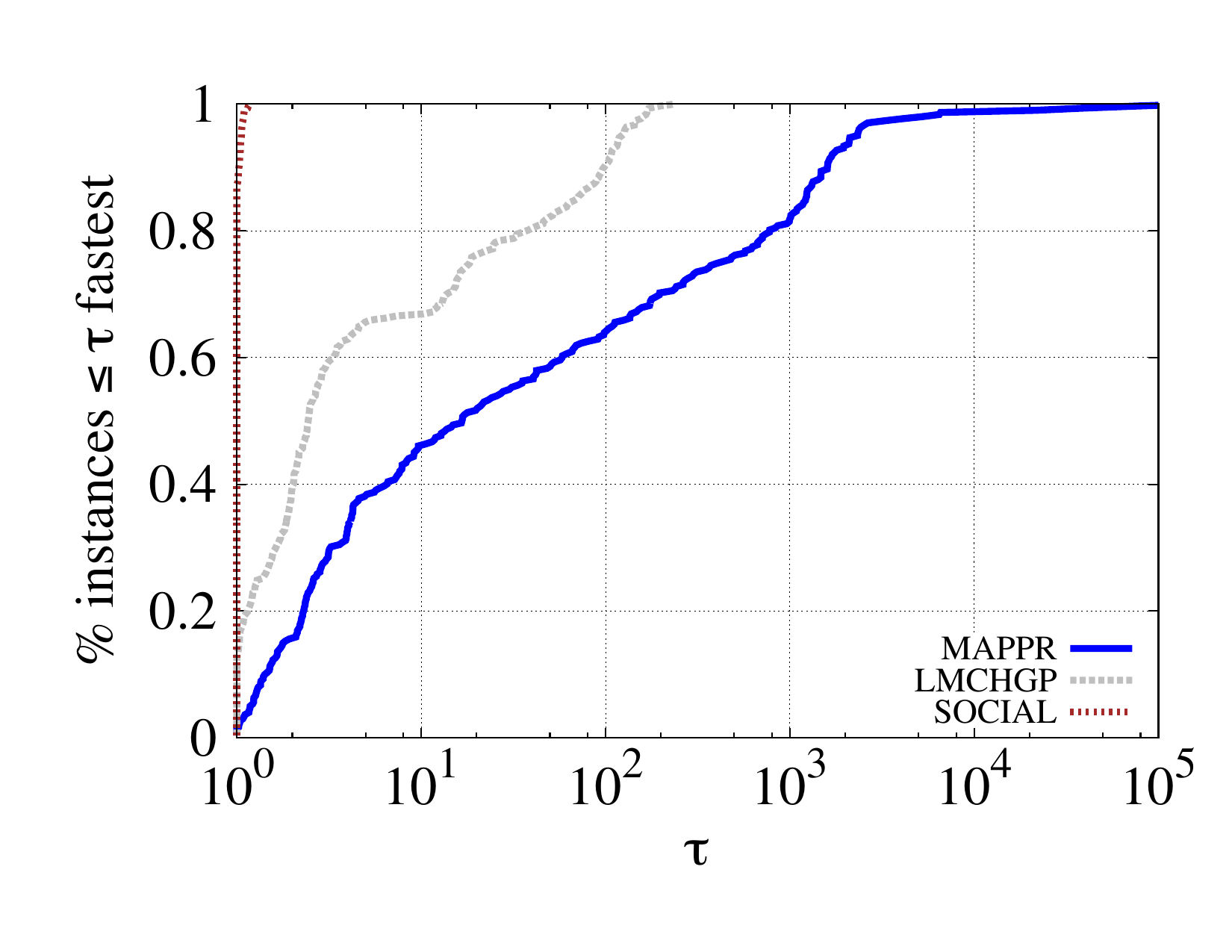}
	\vspace*{-.6cm}
	\caption{Performance profile comparing running times of three algorithms.}
	\label{fig:prel_SIMPLEstateoftheart_graph_pp}
\end{figure*}

\section{Methodology}
\label{sec:Methodology}

In this section, we describe our experimental methodology. 
The concepts and details presented here apply throughout the thesis, unless explicitly stated otherwise.
Depending on the focus of the experiment, we measure running time, memory consumption and/or solution quality.
Depending on the problem under study, solution quality is measured in terms of edge-cut, cut-net, connectivity, communication cost, and/or motif conductance.
In general, we perform ten repetitions per algorithm and instance using different random seeds for initialization, and we compute the arithmetic average of the computed objective functions and running time per instance.
When further averaging over multiple instances, we use the geometric mean in order to give every instance the same influence on the \textit{final score}.

For a solution generated by an algorithm $A$, we express its score $\sigma_A$ (which can measure running time or solution quality) using one or more of the following tools:
\emph{improvement} over an algorithm~$B$, computed as $\big(\frac{\sigma_B}{\sigma_A}-1\big)*100\%$;
\emph{ratio}, computed as $\big(\frac{\sigma_A}{\sigma_{max}}\big)$ with $\sigma_{max}$ being the maximum score for a given instance or for a given x-axis parameter among all competitors including $A$;
\emph{relative} value over an algorithm~$B$, computed as $\big(\frac{\sigma_A}{\sigma_{B}}\big)$.
Bar charts and boxplots are also employed to represent our findings.
We use bar charts to visualize the average value of an objective function in relation to an x-axis parameter, where each algorithm is represented by vertical bars of a given color with origin on the x-axis. 
The bars for every value of this parameter have a common origin and are arranged in terms of their height, allowing all heights to be visible. 
We use boxplots to give a clear picture of the dataset distribution by displaying the minimum, maximum, median, first and third quartiles, \hbox{while disregarding outliers}.

We also present \emph{performance profiles} which  
relate the running time (resp. solution quality) of a group of algorithms to the fastest (resp. best) one on a per-instance basis.
Their x-axis shows a factor~$\tau$ while their y-axis shows the percentage of instances for which A has up to~$\tau$~times the running time (resp. solution quality) of the fastest (resp. best)~algorithm.
Achieving higher fractions at smaller $\tau$ values is considered better.
As an example, Figure~\ref{fig:prel_SIMPLEstateoftheart_graph_pp} compares the running time of three different algorithms using a performance profile.
In this plot, the algorithm \AlgName{SOCIAL} is the fastest one for $87\%$ of the instances, while the algorithms \mbox{\AlgName{LMCHGP}} and \mbox{\AlgName{MAPPR}} are the fastest ones for $12\%$~and~$1\%$ of the instances, respectively (see $\tau = 10^0$).
Note also that the running time of \AlgName{SOCIAL} is within a factor $1.18$ of the running times of the fastest competitors for all instances (see y-axis $= 1$).
Furthermore, the running time of \mbox{\AlgName{LMCHGP}} and \mbox{\AlgName{MAPPR}} are within a factor $10$ of the running times of the fastest competitors for $67\%$ and $46\%$ of the instances, respectively (see $\tau = 10^1$).

\chapter{Related Work}
\label{chap:Related Work}

In this chapter, we give an overview of previous work that has been done on streaming (hyper)graph partitioning, local motif clustering, process mapping, and signed graph clustering.

\section{(Hyper)Graph Partitioning}
\label{sec:(Hyper)Graph Partitioning}

There is massive research on (hyper)graph partitioning in the literature.
The most prominent tools to partition (hyper)graphs in memory include
\AlgName{PaToH}~\cite{ccatalyurek2011patoh},  \AlgName{Metis}~\cite{parmetis-conference},  \AlgName{hMetis}~\cite{hMetis},  \AlgName{Scotch}~\cite{Pellegrini96experimentalanalysis}, \AlgName{HYPE}~\cite{HYPE2018},  \AlgName{KaHIP}~\cite{kabapeE},  \AlgName{KaMinPar}~\cite{gottesburen2021deep},  \AlgName{KaHyPar}~\cite{schlag2016k},\linebreak\hfill \AlgName{Mt-KaHyPar}~\cite{mt-kahypar-d}, and  \AlgName{mt-KaHIP}~\cite{DBLP:conf/europar/Akhremtsev0018}.
The readers are referred to~\cite{more_recent_advances_hgp,SPPGPOverviewPaper,DBLP:reference/bdt/0003S19} for extensive material and references.
Here, we focus on the results specifically related to the scope of this paper.
In particular, we provide an exhaustive account of the state-of-the-art solutions available for the (hyper)graph partitioning problem in the context of the (buffered) (re)streaming model.

\subsection{Streaming Graph Partitioning}
\label{subsec:Streaming Graph Partitioning}
Stanton and Kliot~\cite{stanton2012streaming}  propose heuristics to tackle the graph partitioning problem in the streaming model.
Among their most prominent heuristic is the one-pass method \emph{linear deterministic greedy}~(\AlgName{LDG}) which produces solutions with the best overall edge-cut.
In this algorithm, node assignments prioritize blocks containing more neighbors and use a penalty multiplier to control imbalance.
Particularly, a node $v$ is assigned to the block $V_i$ that maximizes
$|V_i \cap \neighbors(v)|\Phi(i)$ with $\Phi(i)$ being a multiplicative degrading factor defined as $(1-\frac{|V_i|}{L_\text{max}})$
The intuition is that the degrading factor avoids to overload blocks that are already very heavy.
In case of ties on the objective function, \AlgName{LDG} assigns the node to the block with fewer nodes.
Overall, \AlgName{LDG} partitions a graph in $O(m+nk)$ time.
Moreover, the authors also propose a simple one-pass methods based on \emph{hashing}, which has running time $O(n)$ and produces a poor edge-cut.

% \vfill

Later, Stanton~\cite{stanton2014streaming} studies the streaming graph partitioning problem from a more theoretical perspective.
The author proves that no algorithm can obtain an $o(n)$-approximation with a random or adversarial stream ordering.
Next, two variants of a randomized greedy algorithm are analyzed by using a novel coupling to finite Polya Urn~\cite{chung2003generalizations} processes, which intuitively explains the performance of the compared algorithms.

Tsourakakis~et~al.~\cite{tsourakakis2014fennel} propose \AlgName{Fennel},
a simple one-pass partitioning algorithm which adapts the widely-known clustering objective function \emph{modularity}~\cite{brandes2007modularity}.
Roughly speaking, \AlgName{Fennel} assigns a node $v$ to a block $V_i$ in order to maximize an expression of type $|V_i\cap \neighbors(v)|-f(|V_i|)$, where $f(|V_i|)$ is an additive degrading factor.
More specifically, the authors defined the \AlgName{Fennel} objective function by using
$f(|V_i|) = \alpha\gamma |V_i|^{\gamma-1}$, in which $\gamma$ is a free parameter and $\alpha = m \frac{k^{\gamma-1}}{n^{\gamma}}$.
After a parameter tuning made by the authors, \AlgName{Fennel} uses $\gamma=\frac{3}{2}$, which implies $\alpha=\sqrt{k}\frac{m}{n^{3/2}}$.
Although the objective function penalizes imbalanced partitions, the authors define the possibility of a hard constraint to enforce balancing.
In their experiments, \AlgName{Fennel} cuts fewer edges than \AlgName{LDG}~\cite{stanton2012streaming} and, for some instances, it cuts roughly the same number of edges as an offline partitioning algorithm.

Zhang~et~al.~\cite{zhang2018akin} propose \AlgName{AKIN}, a streaming graph partitioning algorithm for distributed graph storage systems.
\AlgName{AKIN} is able to partition
graphs where the number of nodes $n$ is not known in advance by allowing the migration of nodes between blocks over time.
The assignment decisions are mainly based on the \emph{similarity} between nodes, which is evaluated with the \emph{Jaccard similarity coefficient}~\cite{hamers1989similarity}.
Given the (partial) neighborhoods of two nodes, this coefficient is defined as the ratio of their intersection over their union.
Initially, \AlgName{AKIN} assumes a \emph{base block} for each node, which is given by a hash  function.
This base block is taken as a preliminary node assignment as well as a constant-time index for reaching information in the distributed graph storage.
More specifically, for each node, \AlgName{AKIN} stores in its base block a fixed-length list containing its loaded neighbors with largest degree.
This list is used for computing the similarity between nodes.
As soon as an edge is loaded, \AlgName{AKIN} assigns and migrates it and both of its endpoints to the block which maximizes a similarity-based heuristic. 
In the experimental evaluation, the version of \AlgName{AKIN} keeping up to 100 neighbors per node cuts fewer edges compared to \AlgName{Fennel} while maintaining equivalent imbalance and spending $10\%$ more running time.

\subsection{Restreaming Graph Partitioning}
\label{subsec:Restreaming Graph Partitioning}
Nishimura~and~Ugander~\cite{nishimura2013restreaming} introduce a restreaming approach to partition the nodes of a graph.
Their approach is motivated by scenarios where the same graph is
streamed multiple times.
In their model, a one-pass partitioning algorithm can pass multiple times through the entire input while the edge-cut is iteratively reduced.
The authors propose \AlgName{ReLDG} and \AlgName{ReFennel}, which are respective restreaming adaptations of linear deterministic greedy~\cite{stanton2012streaming} (\AlgName{LDG}) and \AlgName{Fennel}~\cite{tsourakakis2014fennel}.
On the one hand, \AlgName{ReLDG} modifies the objective of \AlgName{LDG} to account only for node assignments performed during the current pass when computing block weights.
On the other hand, \AlgName{ReFennel} uses the same objective as \AlgName{Fennel} during restreaming, but its additive balancing degrading factor is increased after each pass in order to enforce balance.
Additionally, the authors prove that \AlgName{ReFennel} converges after a finite number of restreams even without increasing the degrading factor.
Their experiments confirm that their restreaming methods can iteratively reduce edge-cut.

Awadelkarim and Ugander~\cite{awadelkarim2020prioritized} investigate how the order in which nodes are streamed influences one-pass graph partitioning.
The authors introduce the notion of \emph{prioritized streaming}, where (re)streamed nodes are statically or dynamically reordered based on some predefined priority.
Their approach, which is a prioritized version of \AlgName{ReLDG}, uses multiplicative weights of restreaming algorithms and adapts the ordering of the streaming process inspired by balanced label propagation.
In their experiments, the authors  consider a wide range of stream orderings.
The minimum overall edge-cut is obtained using a dynamic node ordering based on their own metric \emph{ambivalence}.
This approach is closely followed by a static ordering based on node degree.

\subsection{Buffered Streaming Graph Partitioning}
\label{subsec:Buffered Streaming Graph Partitioning}
Patwary et al.~\cite{patwary2019window} propose \AlgName{WStream}, a simple streaming graph partitioning algorithm that keeps a sliding window in memory.
The authors allow a few hundred nodes in the sliding window in order to obtain more information about a node before it is permanently assigned to a block based on a greedy function.
As soon as a node is allocated to a block, one more node is loaded from the input stream into the sliding window, which keeps the window size constant.
In their experiments, \AlgName{WStream} cuts fewer edges than \AlgName{LDG} and more edges than offline multilevel partitioning for most~tested~graphs.

Jafari et al.~\cite{jafari2021fast}
perform graph partitioning using a buffered streaming computational model.
The authors propose a shared-memory algorithm which repeatedly loads a batch of nodes from the stream input, partitions it using a multilevel scheme, and then permanently assigns the nodes to blocks.
Their multilevel 
scheme is based on a simplified structure where the one-pass algorithm \AlgName{LDG} is used for coarsening, computing an initial partition, and refining it.
They parallelize \AlgName{LDG} in a node-centric way by simply splitting nodes among processors, which yields a parallelization of the three steps of their multilevel scheme.
In their experiments, their algorithms cuts fewer edges than
\AlgName{LDG} while scaling better than \hbox{offline partitioning algorithms}.

% \vfill

\subsection{Streaming Hypergraph Partitioning}
\label{subsec:Streaming Hypergraph Partitioning}
Alistarh~et~al.~\cite{alistarh2015streaming} propose \AlgName{Min-Max}, a one-pass streaming
algorithm to assign the nodes of a hypergraph to blocks.
For each block, this algorithm keeps track of nets which contain pins in it.
This implies a memory consumption of $O(mk)$, which is more than the typical memory consumption of a streaming algorithm for graph partitioning.
When a node is loaded, \AlgName{Min-Max} allocates it to the block containing the largest intersection with its nets while respecting a hard constraint for load balance.
The authors theoretically prove that their algorithm is able to recover a hidden \emph{co-clustering}
with high probability, where a co-clustering is defined as a simultaneous clustering of nodes and hyperedges.
In the experimental evaluation, \AlgName{Min-Max} outperforms five intuitive streaming approaches with respect to load imbalance, while producing solutions up to five times more imbalanced than internal-memory \hbox{algorithms such as \AlgName{hMetis}}.

Ta{\c{s}}yaran~et~al.~\cite{tacsyaran2021streaming} propose improved versions of the algorithm  \AlgName{Min-Max}~\cite{alistarh2015streaming}.
The authors present \AlgName{Min-Max-N2P}, a modified version of \AlgName{Min-Max} that stores blocks containing each net's pins instead of storing nets per block, as done in \AlgName{Min-Max}.
In their experiments, \AlgName{Min-Max-N2P} is three orders of magnitude faster than \AlgName{Min-Max} while keeping the same cut-net. %
The authors also introduce three algorithms with reduced memory usage compared to \AlgName{Min-Max}: \AlgName{Min-Max-L$\ell$}, a modification of \AlgName{Min-Max-N2P} that employs an upper-bound $\ell$ to limit memory consumption per net, \AlgName{Min-Max-BF} which utilizes Bloom filters for membership queries, and \AlgName{Min-Max-MH} that uses hashing functions to replace the connectivity information between blocks and nets.
In their experiments, their three algorithms reduce the running time in comparison to  \AlgName{Min-Max}, especially \AlgName{Min-Max-L$\ell$} and \AlgName{Min-Max-MH}, which are up to four orders of magnitude faster.
On the other hand, the three algorithms generate solutions with worse cut-net than  \AlgName{Min-Max}, especially \AlgName{Min-Max-MH}, which increases the cut-net by up to an order of magnitude.
Moreover, the authors propose a technique to improve the partitioning decision in the streaming setting by including a buffer to store some nodes and their net sets. 
This approach operates similarly to \AlgName{Min-Max-N2P}, but with the added ability to revisit buffered nodes and adjust their partition assignment based on the connectivity metric.
The authors propose three algorithms using this buffered approach: \AlgName{REF} that buffers every incoming node but only reassigns those that may improve connectivity, \AlgName{REF\_RLX} that buffers all nodes and reassigns all nodes in the buffer, and \AlgName{REF\_RLX\_SV} that only buffers nodes with small net sets and reassigns all nodes in the buffer.
Their experimental results show that the use of buffered approaches leads to a $5$-$20\%$ improvement in partitioning quality compared to non-buffered approaches, but with a trade-off \hbox{of increased runtime}.

% \vfill

\section{Local Motif Clustering}
\label{sec:Local Motif Clustering}

Motif-based clustering has been widely studied in the literature, with works such as~\cite{benson2015tensor, yin2017local,klymko2014using,prvzulj2007biological,tsourakakis2017scalable} partitioning all the nodes of a graph into clusters based on motifs.
We also address the topic of clustering based on motifs, but our focus is on identifying clusters in the immediate vicinity of a specific seed node, rather than on the entire graph.
Several works~\cite{kloster2014heat,li2015uncovering,mahoney2012local,cui2014local,sozio2010community} propose local clustering algorithms, but they do not focus on optimizing for motif-based metrics like this thesis. 
Instead, they use metrics based on edges, like conductance and modularity.
In this section, we review previous work on local clustering based on motifs, which is the focus~of~this~thesis.

\citet{rohe2013blessing} propose a local clustering algorithm based on triangle motifs.
Their algorithm starts by initializing a cluster containing only the seed node, and iteratively grows this cluster.
Particularly, the algorithm greedily inserts nodes contained in at least a predefined number of cut triangles.
\hbox{\citet{huang2014querying}} recover local communities containing a seed node in online and dynamic setups based on higher-order graph structures named Trusses~\cite{cohen2008trusses}.
They define the $k$-truss of a graph as its largest subgraph whose edges are all contained in at least $(k-2)$ triangle motifs, hence trusses are a graph structure based on the frequency of triangles.
The authors use indexes to search for $k$-truss communities in time proportional to the size of the \hbox{recovered~community}. 

\citet{yin2017local}~propose \AlgName{MAPPR}, a local motif clustering algorithm based on the Approximate Personalized PageRank~(\AlgName{APPR}) method.
In a preprocessing phase, \AlgName{MAPPR} enumerates the motif of interest in the entire input graph and constructs a weighted graph~$W$, in which edges only exist between nodes that appear in at least one instance of the motif, and their edge weight is equal to the number of occurrences of the motif containing these two endpoints.
Afterward, \AlgName{MAPPR} uses an adapted version of the \AlgName{APPR} method to find local communities in the weighted graph constructed in the preprocessing phase.
\AlgName{MAPPR} is able to extract local communities from directed input graphs, something that cannot be \hbox{done using~\AlgName{APPR}~alone}.

\citet{zhang2019local} propose \AlgName{LCD-Motif}, an algorithm that addresses the local motif clustering problem using a modified version of the spectral method.
\AlgName{LCD-Motif} has two main differences in comparison to the traditional spectral motif clustering method.
First, instead of computing singular vectors, the algorithm performs random walks to identify potential members of the searched cluster.
They use the span of a few dimensions of vectors, obtained through random walks, as an approximation for the local motif spectra.
Second, Instead of using $k$-means for clustering, \AlgName{LCD-Motif} searches for the minimum 0-norm vector within the previously mentioned span, which must contain the seed nodes in its~support~vector.

\citet{meng2019local} propose \AlgName{FuzLhocd}, a local motif clustering algorithm that uses fuzzy arithmetic to optimize a modified version of modularity.
Given seed node, \AlgName{FuzLhocd} starts by detecting probable core nodes of the targeted local community using fuzzy membership.
After identifying the probable core nodes of the target local community using fuzzy membership, the algorithm expands these nodes using another fuzzy membership to form~a~cluster.

\citet{zhou2021high} propose \AlgName{HOSPLOC}, a local motif clustering algorithm that uses a motif-based random walk to compute a distribution vector, which is then truncated and used in a vector-based partitioning method.
The algorithm begins by approximately estimating the distribution vector through a motif-based random walk. 
To further refine the computation and focus on the local region, \AlgName{HOSPLOC} sets all small vector entries to 0. 
After this preprocessing step, the algorithm applies a vector-based partitioning method~\cite{spielman2013local} on the resulting distribution vector in order to \hbox{identify~a~local~cluster}.

\citet{shang2022local} propose \AlgName{HSEI}, a local motif clustering algorithm that uses motif and edge information to grow a cluster from a seed node.
The algorithm begins by creating an initial cluster consisting of only the seed node.
It then adds nodes to the cluster from the seed's neighborhood, selecting them based on their motif degree.
The cluster is expanded using a motif-based extension of the~modularity~function.

\section{Process Mapping}
\label{sec:process_mapping}

The algorithms for process mapping can be categorized in two groups.
On the one hand, the \emph{single-phase} or \emph{integrated} algorithms
combine process mapping with graph partitioning  \cite{DBLP:journals/fgcs/WalshawC01,DBLP:conf/hpcn/PellegriniR96}, such that the objective of the partitioning -- commonly edge-cut -- is typically replaced by a function that measures communication cost.
On the other hand, the two-phase algorithms decouple partitioning and mapping \cite{schulz2017better,brandfass2013rank,heider1972computationally,muller2013optimale,glantz2015algorithms}.
In a \emph{two-phase} or \emph{decoupled} algorithm, a default graph partitioning algorithm is used to partition a communication graph into $k$ blocks while typically minimizing edge-cut.
Afterwards, the quotient graph of the partitioned communication graph is mapped onto PEs in order to minimize the total communication cost $J$ (Equation~(\ref{eq:process_mapping_obj})) or other objective function.
This step is also known as one-to-one (process) mapping.
To the best of our knowledge there is no streaming algorithm specifically designed adapted to optimize for the process mapping problem.
We now describe the state-of-the-art for the process mapping problem.
The reader is referred to~\cite{GPOverviewBook,SPPGPOverviewPaper} for older works on process mapping.

\AlgName{Jostle} and \AlgName{Scotch} are integrated algorithms to solve the process mapping problem.
\AlgName{Jostle} integrates local search into a multilevel scheme to partition the model of
computation and communication. 
In this scheme, it solves the problem on the coarsest level and afterwards performs refinements based on the user-supplied network communication model. 
\AlgName{Scotch} performs dual recursive bipartitioning to compute a mapping.
More precisely, it starts the recursion considering all given processes and PEs. 
At each recursion level, it bipartitions the communication graph and also the distance graph with a graph bipartitioning algorithm. 
The first (resp., second) block of the communication graph is then assigned to the first (resp., second) block of the distance graph.

Müller-Merbach~\cite{muller2013optimale} propose a greedy construction method to obtain an initial permutation for the one-to-one process mapping problem.
The method roughly works as follows:
Initially compute the total communication volume for each process and also the sum of distances from each core to all the others.
Afterwards, the algorithm proceeds in rounds. In each round, the process with the largest communication volume is assigned to the core with the smallest total distance. 

A method to improve an already given solution for one-to-one process mapping was proposed in \cite{heider1972computationally}. 
The method repeatedly tries to perform swaps in the assignment in a pair-exchange neighborhood $N(\Pi)$ that contains all permutations that can be reached by swapping two elements in $\Pi$. 
Here, swapping two elements means that $\Pi^{-1}(i)$ will be assigned to processor $j$ and $\Pi^{-1}(j)$ will be assigned to processor $i$ after the swap is done. 
The algorithm then looks at the neighborhood in a cyclic manner. 
A swap is performed if it reduces the objective. 
To reduce the runtime, Brandfass~et~al.~\cite{brandfass2013rank} introduced a couple of modifications to speed up the algorithm, such as only considering pairs for swapping that can reduce the objective or partitioning the search space into $s$ consecutive blocks and only performing swaps inside those blocks.

Deveci~et~al.~\cite{DBLP:conf/ipps/DeveciKUC15} propose a greedy
construction algorithm and two refinement algorithms to map processes to PEs based on a topology graph.
The construction algorithm iteratively picks the process with highest connectivity to the already mapped processes.
Then, this process is assigned to the module which minimizes hop-bytes, which is found by performing a breadth-first search on the topology graph.
Both their refinement methods are based on process swapping, but one of them aims at minimizing hop-bytes while the other one aims at minimizing
congestion.
The authors experimentally compare their algorithms on a torus-based system against the default mapper of the system as well as \AlgName{Scotch}~\cite{Scotch} and LibTopoMap~\cite{Walshaw07}.
Their algorithms induce performance increases of $43\%$ and $23\%$ for a communication-only application and a sparse matrix vector multiplication, respectively, compared to the default mapper of the system, and
have the best overall results \hbox{among the competitors}.

Deveci~\etal~\cite{8666156} also propose a scheme that exploits geometric
partitioning to map processors to PEs.  The application data and
PEs are partitioned separately using the MultiJagged geometric partitioner~\cite{DBLP:journals/tpds/DeveciRDC16}.
Blocks from the data partition are then mapped to
corresponding blocks in the PE partition, effectively assigning interdependent
data to ``nearby'' PEs in the network.  The method is appropriate for
applications that have geometric coordinates as well as graph data.

Vogelstein~et~al.~\cite{vogelstein2015fast} solve an application-agnostic problem where there are two graphs with the same number of nodes which are bijectively
mapped onto one another in order to minimize the number of induced edge disagreements.
The authors call this problem \emph{graph matching problem} and formulate it as a quadratic assignment problem~(QAP).
They solve it with a non-linear approximation algorithm based on the gradient vector.
Their algorithm has complexity $O(k^3)$ and performs better than the previous state-of-the-art~\cite{umeyama1988eigendecomposition,singh2007pairwise,zaslavskiy2008path} regarding running time and objective function for over $80\%$ of the QAPLIB benchmark library~\cite{burkard1997qaplib}.

Glantz~et~al.~\cite{glantz2015algorithms}
propose two greedy algorithms to bijectively map blocks onto PEs assuming a communication graph.
Their most successful algorithm, \AlgName{GreedyAllC}, is an adaptation of the greedy mapping algorithm proposed
by M{\"u}ller-Merbach~\cite{muller2013optimale}.
Glantz~et~al.~\cite{glantz2015algorithms} modify this algorithm by scaling the distance with the amount of communication to be done.
This modification improves the overall mapping quality with respect to the quality measures maximum congestion and maximum dilation.

Glantz~et~al.~\cite{glantz2018topology} propose a local improvement algorithm for one-to-one process mapping in which the hardware topology is a partial cube, i.e., an isometric subgraph of a hypercube.
The authors exploit the regularity of these topologies to label PEs as well as processes with bit-strings along convex cuts.
These bit-strings permit fast computation of distances between PEs and the implementation of effective
hierarchical refinement methods to improve the mapping induced by the labels.
Their experimental results show that their algorithm reduces the total communication cost of the mappings produced by state-of-the-art mapping algorithms ~\cite{Scotch,glantz2015algorithms} in a range from $6\%$ to $34\%$ while the mapping time stays within the same order of magnitude.

Schulz~and~Träff~\cite{schulz2017better} solve the process mapping for a hierarchical topology using a two-phase approach.
Their algorithm uses the state-of-the-art partitioner \AlgName{KaHIP}~\cite{kaHIPHomePage} for the partitioning phase and then a multi-section for the one-to-one mapping phase.
The referred multi-section algorithm recursively partitions the quotient graph of the previously obtained partition throughout the layers of the communication topology in a top-down direction.
After this algorithm obtains a one-to-one mapping, a light-weight modification of
the swap-based refinement method proposed by Brandfass~et~al.~\cite{brandfass2013rank} is executed to further minimize the communication cost.
In particular, the authors experimentally show that $N^{10}_\mathcal{C}$, which restricts swapping to processes that have a distance smaller than $10$ in the communication graph, is an adequate choice to obtain good solutions with a moderate running time.
Their experiments also show that their approach produces mappings with lower total communication cost in comparison with alternative approaches which combine the same \AlgName{KaHIP}-based first phase with state-of-the-art algorithms for one-to-one mapping~\cite{muller2013optimale,Walshaw07,glantz2015algorithms}.
Kirchbach~et~al.~\cite{GlobalMultisection} further improve the two-pass approach by Schulz~and~Träff~\cite{schulz2017better}.
In particular, the authors specialize the partitioning phase to apply \AlgName{KaHIP} multiple times, namely throughout the hierarchical topology.
Kirchbach~et~al experimentally show that their best algorithm is faster while also decreasing the total communication cost in comparison to the algorithm from Schulz~and~Träff~\cite{schulz2017better}.

Predari~et~al.~\cite{PredariProcessMapping21} propose a single-phase distributed algorithm to solve the process mapping for a hierarchical topology.
The authors model the system hierarchy as an implicit labeled tree and label the nodes of the communication graph in order to implicitly induce the mapping.
The algorithm optimizes the mapping by using an
adapted version of parallel label propagation where the labeling scheme is used for quick gain computations.
In their experiments, the proposed algorithm has good scalability for up to thousands of PEs while achieving a total communication cost smaller than the distributed algorithms ParMetis~\cite{parmetis-conference} and ParHIP~\cite{DBLP:journals/tpds/MeyerhenkeSS17} and comparable cost to the state-of-the-art sequential mapping algorithm.%

Kirchbach~et~al.~\cite{von2020efficient}
address the process mapping problem assuming processes that communicate in a sparse stencil pattern and PEs organized in Cartesian grids.
First, the authors prove that the Cartesian mapping is already NP-hard for a
two-dimensional Cartesian grid and a one-dimensional stencil.
Then three fast construction algorithms which exploit the regularity of the problem are proposed.
Their algorithm with best overall results is the \emph{stencil strips} algorithm, which partitions the grid into strips of lengths close to the scaled length of an optimal bounding rectangle of the coordinates of the target stencil.
In their experimental evaluation using the MPI\_Neighbor\_alltoall routine of MPI,
their algorithms are up to two orders of magnitude faster than general purpose graph mapping tools such as \AlgName{VieM}~\cite{viemaWebsite} while resulting in a similar communication performance.
Moreover, their algorithms are up to three times faster than other Cartesian grid mapping algorithms while resulting in a much better communication performance.

\section{Signed Graph Clustering}
\label{sec:Signed Graph Clustering}

There is a huge body of research on signed graph clustering.
Roughly the same problem is solved under different names and terms in the literature such as \emph{community detection (or mining) in signed networks}~\cite{yang2007community,memeticsiggraphcommdetec2020,traag2009community}, \emph{correlation clustering}~\cite{correlationclustering_overall,demaine2006correlation}, and \emph{clique partitioning problem}~\cite{grotschel1989cutting,dinh2015toward,miyauchi2015redundant,miyauchi2018exact}.
Among the most frequent approach for solving the problem, we mention spectral clustering~\cite{kunegis2010spectral,gallier2016spectral,knyazev2018spectral,mercado2019spectral} and meta-heuristics~\cite{he2011ant,zhu2018novel,brusco2019partitioning,memeticsiggraphcommdetec2020,abdulrahman2020enhanced,chen2022iterated}.
The readers are referred to the surveys~\cite{gallier2016spectral,surveyTangCAL16,tomasso2021advances} for extensive material and references.
Here, we focus on the results specifically related to {the scope of this thesis}.

% \vfill

\subsection{Contraction-Based}
\label{subsec:Contraction-Based}
Keuper~et~al.~\cite{keuper2015efficient} propose an algorithm named Greedy Additive Edge Contraction (\AlgName{GAEC}).
\AlgName{GAEC} is an adaptation of the greedy agglomeration algorithm where the criterion to evaluate the strength of interactions between clusters is the summed weight of the edges shared by them. 
Bailoni~et~al.~\cite{bailoni2019generalized} propose \AlgName{GASP}, a framework for hierarchical agglomerative clustering on weighted signed graphs. 
The authors prove that their framework is a generalization to many existing clustering algorithms such as \AlgName{GAEC}~\cite{keuper2015efficient} and introduce new algorithms based on yet unexplored special instances of their framework, such as \AlgName{HCC-Sum}.
\AlgName{GASP} is a bottom-up approach where nodes are initially assigned to their own clusters, which are then iteratively merged in a pair-wise fashion. 
The framework should be configured to use a specific criterion to evaluate the interaction between a pair of clusters.
Some of these criteria are the (average) weighted edge-cut between the clusters and the weight of their shared edge with highest absolute weight.
The unconstrained variety of \AlgName{GASP} starts by merging clusters with the largest positive interaction and stops once the remaining clusters share only negative interactions. 
The constrained variety of \AlgName{GASP} introduces cannot-merge constraints between pairs of clusters and terminates when all the remaining clusters are constrained against their neighbors.
This approach greedily selects the pair of clusters with highest absolute interaction, which can be either positive or negative. 
When it is positive, the clusters are merged, otherwise a constraint is added to prohibit the clusters from being merged until the end of the algorithm.
Experimentally, the best algorithms contemplated by \AlgName{GASP} with respect to the minimization of edge-cut are \AlgName{GAEC}~\cite{keuper2015efficient} and \AlgName{HCC-Sum}.
Similarly to \AlgName{GAEC} and \AlgName{HCC-Sum}, our multilevel algorithm also constructs a clustering based on successive contractions of the graph.

Hua~et~al.~\cite{HuaRandomWalk2020} propose \AlgName{FCSG}, an algorithm for signed graph clustering which combines a random walk gap (RWG) mechanism with a greedy shrinking method.
The RWG mechanism makes a random walk on a version of the graph where negative edges are removed and another random walk on a version of the graph where negative edges are made positive. 
Using the two random walk graphs that are constructed, the gap between the two graphs is calculated to give information on the natural clustering structure of the signed graph. 
This information is then used to build a new signed graph whose edge weights more accurately reflect the natural clustering structure of the input signed graph.
A clustering is then computed on this new signed graph based on a greedy \hbox{shrinking approach}.

\subsection{Evolutionary}
\label{subsec:Evolutionary}
Che~et~al.~\cite{memeticsiggraphcommdetec2020} propose \AlgName{MACD-SN}, an evolutionary algorithm for signed graph clustering optimizing for \emph{signed modularity}. 
In their algorithm, an individual represents a clustering and the \emph{signed modularity}~\cite{gomez2009analysis} metric is used as fitness function. 
An individual is denoted by an $n$-sized string where each node is assigned to a cluster.
Individuals for the initial population are built in two steps.
First, randomly assign all nodes to clusters.
Second, go through nodes in order to assign each node $u$ to its neighboring cluster $V_i$ which minimizes the expression $NID(u,V_i) = |N^-\cap V_i| + |N^+\cap \overline{V_i}|$. 
Classical approaches are implemented for selection, recombination, and mutation. 
Namely, a tournament selection, a randomized two-way crossover recombination, and a uniform random mutation.
The authors also implement a mutation operator which finds the cluster whose nodes contain the highest average value of $NID(u,V_i)$ and then simply assign each of its nodes $u$ to the cluster which minimizes $NID(u,V_i)$.
Finally, a local search algorithm is executed on the best mutated offspring in each generation.
While the operators of \AlgName{MACD-SN} are rather simple and traditional, the operators of our evolutionary algorithm are based on a sophisticated multilevel algorithm which we propose in this thesis. 
However, the approach only scales to networks \hbox{with a few thousand nodes}.

\subsection{Integer Linear Programming}
\label{subsec:Integer Linear Programming}
Many integer linear  programming (ILP) formulations were proposed for solving the signed graph clustering problem in order to optimize for edge-cut, as we do here.
Gr{\"{o}}tschel~and~Wakabayashi~\cite{grotschel1989cutting} propose an ILP formulation consisting of $\Theta(n^2)$ binary variables representing pairs of nodes and $\Theta(n^3)$ transitivity constraints.
These constraints ensure that, if two nodes $u,v$ are clustered together, then any other node $w$ will be clustered together with $u$ if, and only if, it is also clustered together with $v$.
Dinh~and~Thai~\cite{dinh2015toward} remove a set of redundant constraints from the ILP by Gr{\"{o}}tschel~and~Wakabayashi~\cite{grotschel1989cutting} for a special case of the signed clustering problem and Miyauchi~and~Sukegawa~\cite{miyauchi2015redundant} extend the ILP by Dinh~and~Thai~\cite{dinh2015toward} to signed graph clustering in general.
The formulation proposed in~\cite{miyauchi2015redundant} removes transitivity constraints associated with two of more edges with negative weight, which reduces the number of constraints to $O(n(n^2-m^-))$ constraints.

Miyauchi~et~al.~\cite{miyauchi2018exact} obtain an improved number $O(nm^+)$ of constraints by running the ILP from~\cite{miyauchi2015redundant} on a new signed graph where artificial edges are introduced between pairs of unconnected nodes from the input graph.
In particular, these artificial edges have a negative weight which is small enough to ensure that every optimal solution in this new graph corresponds to an optimal solution in the input graph.
Recently Koshimura~et~al.~\cite{koshimura2022concise} proposed a further improvement for the ILP by Miyauchi~et~al.~\cite{miyauchi2018exact} which removes around $50\%$ of its constraints.
Although there have been many improved ILP formulations for signed graph clustering in the last years, this approach does not scale well for large instances (in particular for the size of the instances considered in this paper).
In particular, signed graphs with more than a few thousand nodes become intractable in practice, which is the reason why most instances in the works cited above have \hbox{no more than a few hundred nodes}.

\chapter{Streaming Algorithms}
\label{chap:Streaming Algorithms}

In this chapter, we present our contributions in the field of (buffered) streaming (hyper)graph decomposition.
Firstly, we introduce the algorithm \AlgName{HeiStream} that solves the graph partitioning problem using the buffered streaming model.
The algorithm operates by partitioning batches or buffers of nodes with a multilevel algorithm after building a model for each loaded buffer.
\AlgName{HeiStream} is shown to produce solutions of better quality than previous state-of-the-art algorithms while also being faster when partitioning into large numbers of blocks.
Secondly, we propose \AlgName{Online Recursive Multi-Section}, a streaming algorithm for the process mapping problem that uses recursive multi-sections on-the-fly.
Our approach is the first in the literature designed specifically for process mapping, and its versatility allows it to also solve the graph partitioning problem.
Experimental results show that our recursive multisection on-the-fly algorithm produces significantly better solutions in terms of total communication cost while being up to multiple orders of magnitude faster than the state-of-the-art.
Thirdly, we introduce \AlgName{FREIGHT}, a streaming algorithm for hypergraph partitioning that extends the graph-based algorithm \AlgName{Fennel}.
Our proposed implementation for \AlgName{FREIGHT} is highly efficient, leading to significantly faster computation times and better cut-net and connectivity measures than previous state-of-the-art algorithms.
Lastly, we present an experimental comparison of our three algorithms for the \hbox{graph partitioning problem}.

\paragraph*{References.}
This chapter is based on~\cite{HeiStream}~and~\cite{StreamMultiSection}, which is joint work with Christian Schulz, and on~\cite{freight_paper}, which is joint work with Kamal Eyubov and Christian Schulz.
The experimental comparison of our three streaming algorithms is~new.

\section{Buffered Streaming Graph Partitioning}
\label{chap:heistream_Buffered Streaming Graph Partitioning}
\label{sec:heistream_Buffered Streaming Graph Partitioning}

In this section, we start to fill the gap currently observed between in-memory and streaming algorithms for graph partitioning algorithms.
We propose \AlgName{HeiStream}, an algorithm that can compute significantly better partitions of huge graphs than the currently available streaming algorithms while using a single machine without a lot of memory.
We adopt the buffered streaming model which allows a buffer of nodes to be received and stored before making assignment decisions.
Our algorithm is carefully engineered to produce partitions of improved quality by using a sophisticated multilevel scheme on a compressed model of the buffer \emph{and} the already assigned nodes. 
Our multilevel algorithm optimizes for the same objective as the previous state-of-the-art \AlgName{Fennel}. 
However, due to the multilevel scheme used on the compressed model, our local search algorithms have a global view on the optimization problem and hence compute better solutions overall. 
Lastly, using the multilevel scheme reduces the time complexity from $O(nk+m)$ of \AlgName{Fennel} to $O(n+m)$, where $k$ is the number of blocks a graph has to be partitioned in.
To this end, experiments indicate that our algorithm can partition huge networks on machines with small memory while computing \emph{better} solutions than the previous state-of-the-art in the streaming setting. 
At the same time our algorithm is \emph{faster} than the previous state-of-the-art for larger values of blocks~$k$.

\subsection{HeiStream}
\label{sec:heistream_HeiStream}

\subsubsection{Overall Structure}
\label{subsec:heistream_Overall Structure}
We now explain the overall structure of \AlgName{HeiStream}.
We slide through the streamed graph $G$ by repeating the following successive operations until all the nodes of~$G$ are assigned to blocks.
First, we load a batch containing $\delta$ nodes alongside with their adjacency lists.
Second, we build a model $\mathcal{B}$ to be partitioned. This model represents the already partitioned nodes as well as the nodes of the current batch.
Third, we partition $\mathcal{B}$ with a multilevel partitioning algorithm to optimize for the \AlgName{Fennel} objective function.
And finally, we permanently assign the nodes from the current batch to blocks.
Algorithm~\ref{alg:heistream_overall_algorithm} summarizes the general structure of \AlgName{HeiStream} and Figure~\ref{fig:heistream_comprehensive_scheme} shows the detailed structure of \AlgName{HeiStream}.

\begin{figure}[t]
	\centering
	\includegraphics[width=0.9\textwidth]{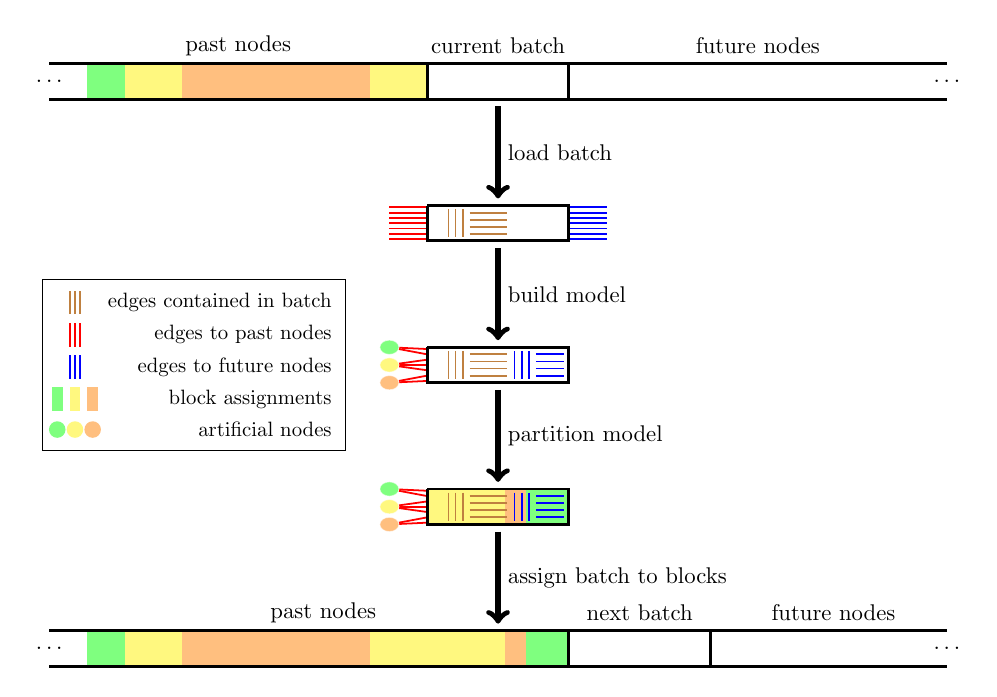}
	\caption{Detailed structure of \AlgName{HeiStream}. The algorithm starts by loading a batch of nodes alongside with their edges. Next, it builds a meaningful model based on the loaded nodes and edges. This model is then partitioned using a multilevel algorithm. Based on the partition of the model, the nodes of the current batch are permanently assigned to blocks. This whole process is repeated for the next batch until the whole graph is partitioned.}
	\label{fig:heistream_comprehensive_scheme}
\end{figure}

\begin{algorithm}[t] %
	\begin{algorithmic} %
		\While{$G$ has not been completely streamed}
			\State Load batch of nodes
			\State Build model $\mathcal{B}$
			\State Run multilevel partitioning on model $\mathcal{B}$
			\State Assign nodes of batch to permanent blocks
		\EndWhile
	\end{algorithmic}
	\caption{Structure of \AlgName{HeiStream}} %
	\label{alg:heistream_overall_algorithm} %
\end{algorithm}

\subsubsection{Model Construction}
\label{subsec:heistream_Graph Construction}
We build two different models, which yield a running time-quality trade-off. 
We start by describing the basic model and then extend this later.

When a batch is loaded, we build the \emph{basic model} $\mathcal{B}$ as follows.
We initialize $\mathcal{B}$ as the subgraph of $G$ induced by the nodes of the current batch.
If the current batch is not the first one, we add $k$ artificial nodes to the model.
These represent the $k$ preliminary blocks in their current state, i.e., filled with the nodes from the previous batches, which were already assigned.
The weight of an artificial node $i$ is set to the weight of block $V_i$.
A node of the current batch is connected to an artificial node $i$ if it has a neighbor from the previous batch that has been assigned to block $V_i$.
If this creates parallel edges, we replace them by a single edge with its weight set to the sum of the weight of the parallel edges. 
Note that the basic model does ignore edges towards nodes that will be streamed in future batches, \ie batches that have not been streamed yet.

Our \emph{extended model} incorporates edges towards nodes from future, not yet loaded, batches -- if the stream contains such edges.
We call edges towards nodes of future batches ghost edges and the corresponding endpoint in the future batch ghost node.
Ghost nodes and ghost edges provide partial information about parts of $G$ that have not yet been streamed.
Hence, representing them in the model $\mathcal{B}$ enhances the partitioning process.
Note though that simply inserting all ghost nodes and edges can overload memory in case there is an excessive number of them.
Thus our approach consists of randomly contracting the ghost nodes with one of their neighboring nodes from the current batch. 
Note that this contraction increments the weight of a node within our model and ensures that if there are more than one node from the current batch connected to the same future node, then there will be edges between those nodes in our model. 
Also note that the contraction ensures that the number of nodes in all models throughout the batched streaming process is constant. 
This prevents memory from being overloaded and makes it unnecessary to reallocate memory for $\mathcal{B}$ between successive batches.
In order to give a lower priority to ghost edges in comparison to the other edges, we divide the weight of each ghost edge by 2 in our~model.
The construction of the extended model is conceptually illustrated in Figure~\ref{fig:heistream_comprehensive_scheme}.

\subsubsection{Multilevel Weighted Fennel}
\label{subsec:heistream_Multilevel Fennel}
Our approach to partition the model $\mathcal{B}$ is a multilevel version of the algorithm \AlgName{Fennel}.
Recall that the multilevel scheme consists of three successive phases: coarsening, initial partitioning, and uncoarsening, as depicted in Figure~\ref{fig:sgc_MSGC}. %

It should be noted that the artificial nodes present in our model may become very heavy.
As an acknowledgement of their representation of permanent block assignments, any contraction or block change of artifical nodes is prohibited.
Consequently, these nodes require special handling in our multilevel algorithm.
Furthermore, note that the original formulation of \AlgName{Fennel} only works for unweighted graphs~\cite{tsourakakis2014fennel}.
However, our model $\mathcal{B}$ has nodes and edges that are weighted -- as a result of connections to artificial nodes and future nodes that may be contracted into the model, as well as due to the generation of a series of weighted graphs through the multilevel scheme. 
Hence, the \AlgName{Fennel} algorithm requires some adaptation to be utilized in \AlgName{HeiStream}.
We introduce a generalization of the \AlgName{Fennel} algorithm for weighted graphs that can be directly employed in a multilevel algorithm.
Within this section, we will present this generalized \AlgName{Fennel} objective and provide a detailed explanation of our multilevel algorithm for partitioning the model.

\paragraph*{Generalized Fennel.}
As a necessary feature, our generalization of the \AlgName{Fennel} gain function must ensure that the gain of a node at a coarser level equals the sum of the gains of the nodes that it represents at the finer levels. 
This way the algorithm gets a global view onto the optimization problem on the coarser levels and a very local view on the finer levels. 
Moreover, it is ensured that on each level of the hierarchy the algorithm works towards optimizing the given objective.

Our generalization of the gain function of \AlgName{Fennel} is as follows. %
Let $u$ be the node that should be assigned to a block. 
Our generalized \AlgName{Fennel} assigns $u$ to a block $i$ that maximizes Equation~(\ref{eq:gen_fennel}) such that $f(c(V_i)) =  \alpha * \gamma * c(V_i)^{\gamma-1}$.
Note that this is a direct generalization of the unweighted case. First, if the graph does not have edge weights, then the first part of the equation becomes $|V_i \cap N(u)|$ which is the first part of the \AlgName{Fennel} objective. Second, if the graph also does not have node weights, then the second part of the equation is the same as the second part of the equation in the \AlgName{Fennel} objective.
Moreover, observe that the penalty term $f(c(V_i))$ in our objective is multiplied by $c(u)$.
This is done to have the property stated above and formalized in Theorem~\ref{theo:generalized_fennel}.
Finally, we keep the original value of the parameters $\alpha$ and $\gamma$ in order to keep consistency. 

\begin{equation}
	\sum_{ v \in V_i \cap N(u)}{ \omega(u,v)} - c(u) f(c(V_i))
	\label{eq:gen_fennel}
\end{equation}

\begin{theorem}
	If a set of nodes $S \subseteq V$ is contracted into a node $w$, the generalized \AlgName{Fennel} gain function of $w$ is equal to the sum of the generalized \AlgName{Fennel} gain functions of the nodes in $S$.
	\label{theo:generalized_fennel}
\end{theorem}

\begin{proof}
	On the one hand, the generalized \AlgName{Fennel} gain of assigning a node $w$ to a block~$i$ is computed as:
	\begin{equation}
		\sum\limits_{ v \in V_i \cap N(w)}{ \omega(w,v)} - c(w) f(c(V_i))
		\label{eq:gain_w}
	\end{equation}
	On the other hand, the sum of the generalized \AlgName{Fennel} gains of assigning all nodes in $S$ to a block~$i$ consists of:
	\begin{equation}
		\sum\limits_{u\in S}\sum\limits_{ v \in V_i \cap N(u)}{ \omega(u,v)} - \sum\limits_{u\in S}c(u) f(c(V_i))
		\label{eq:gain_S}
	\end{equation}
	Since the factor $f(c(V_i))$ is identical in~(\ref{eq:gain_w})~and~(\ref{eq:gain_S}), the two of them are equivalent if equalities~(\ref{eq:equality1})~and~(\ref{eq:equality2}) hold: 
	\begin{equation}
		\sum\limits_{ v \in V_i \cap N(w)}{ \omega(w,v)} := \sum\limits_{u\in S}\sum\limits_{ v \in V_i \cap N(u)}{ \omega(u,v)}  
		\label{eq:equality1}
	\end{equation}
	\begin{equation}
		c(w) := \sum\limits_{u\in S}c(u)
		\label{eq:equality2}
	\end{equation}
	But equalities~(\ref{eq:equality1})~and~(\ref{eq:equality2}) are trivially true as a property of the contraction process.
\end{proof}

\paragraph*{Multilevel Fennel.}
In principle, our model $\mathcal{B}$ could be partitioned by any multilevel partitioning algorithm which allows fixed nodes. 
Nevertheless, our preliminary tests using an adaptation of \AlgName{KaHIP}~\cite{kaffpa} generated poor edge-cut.
This is the case because internal memory algorithms such as \AlgName{KaHIP} and \AlgName{Metis}~\cite{karypis1998fast} are designed to minimize edge-cut as much as possible while ensuring a low balancing constraint.
This leads to two possibilities, both of which are bad for the overall edge-cut within the buffered streaming model:
(i)~using a large balancing constraint for the first buffers, the internal memory partitioner will tend to assign all their nodes to a single block,
(ii)~using a low balancing constraint for each batch, the overall partitioning problem will be aggressively over-restricted.
Hence, we design a multilevel algorithm based on the Fennel function, which is a function that successively deals with the trade-off between edge-cut and imbalance associated with graph streaming.

We now explain our multilevel algorithm to partition the model $\mathcal{B}$.
\label{subsec:heistream_Contraction}
Our \emph{coarsening} phase is based on an adapted version of the size-constraint label propagation approach \cite{pcomplexnetworksviacluster}. To be self-contained, we shortly outline the coarsening approach and then show how to modify it to be able to handle artificial nodes. To compute a graph hierarchy, the algorithm computes a size-constrained clustering on each level and contract that to obtain the next level.
The clustering is contracted by replacing each of its clusters by a single node (as exemplified in Figure~\ref{fig:sgc_clustercontraction}), and the process is repeated recursively until the graph becomes small enough.
This hierarchy is then used by the algorithm.
Due to the way we define contraction, it ensures that a partition of a coarse graph corresponds to a partition of all the finer graphs in the hierarchy with the same edge-cut and balance. 
Note that cluster contraction is an aggressive coarsening strategy. 
In contrast to matching-based approaches, it enables us to drastically shrink the size of irregular networks.
The intuition behind this technique is that a clustering of the graph (one hopes) contains many edges running inside the clusters and only a few edges running between clusters, which is favorable for the~edge~cut~objective.

The algorithm to compute clusters is based on \emph{label propagation}~\cite{labelpropagationclustering} and avoids large clusters by using a \emph{size constraint}, as described in~\cite{pcomplexnetworksviacluster}. 
For a graph with $n$ nodes and $m$ edges, one round of size-constrained label propagation can be implemented to run in $O(n+m)$ time.
Initially, each node is in its own cluster/block, \ie the initial block ID of a node is set to its node ID.
The algorithm then works in rounds. 
In each round, all the nodes of the graph are traversed. 
When a node $v$ is visited, it is \emph{moved} to the block that has the strongest connection to $v$, \ie it is moved to the cluster $V_i$ that maximizes $\omega(\{(v, u) \mid u \in N(v) \cap V_i \})$. 
Ties are broken randomly. 
We perform at most $L$ rounds, where $L$ is a tuning~parameter.

In \AlgName{HeiStream}, we have to ensure that two artificial nodes are not contracted together since each of them should remain in its previously assigned block. 
We achieve this by ignoring artificial nodes and artificial edges during the label propagation, \ie artificial nodes cannot not change their label and nodes from the batch can not change their label to become a label of an artificial node.  As a consequence, artificial nodes are not contracted during coarsening.
Overall, we repeat the process of computing a size-constrained clustering and contracting it, recursively. 
As soon as the graph is small enough, \ie it has fewer nodes than an $O(\max({|\mathcal{B}|/k,k}))$ threshold, it is initially partitioned by an initial partitioning~algorithm. 
More precisely, we use the threshold $\max({\frac{|\mathcal{B}|}{2xk},xk})$, in which $x$ is a tuning parameter.
Note that, for large enough buffer sizes, this threshold will~be~$O(|\mathcal{B}|/k)$.

\label{subsec:heistream_Initial Partition}

\label{subsec:heistream_Uncoarsening}

When the coarsening phase ends, we run an \emph{initial partitioning algorithm} to compute an initial $k$-partition for the coarsest version of $\mathcal{B}$.
That means that all nodes other than the artificial nodes, which are already assigned, will be assigned to blocks. 
To assign the nodes, we run our generalized \AlgName{Fennel} algorithm with explicit balancing constraint $L_{\max}$, \ie the weight of no block will exceed $L_{\max}$.
To be precise, a node $u$ will be assigned to a block $i$ that maximizes $\sum_{ v \in V_i \cap N(u)}{ \omega(u,v)} - c(u) f(c(V_i))$, such that $f(c(V_i)) =  \alpha * \gamma * c(V_i)^{\gamma-1}$ \emph{and} $c(V_i \cup u) \leq L_{\max}$. 
Note that the algorithm at this point considers all possible blocks $i \in \{1, \ldots, k\}$ and hence has complexity proportional to~$k$. 
However, as the coarsest graph has $O(|\mathcal{B}|/k)$ nodes, overall the initial partitioning needs time which is linear in the size of the~input~model.
When initial partitioning is done, we transfer the current solution to the next finer level by assigning each node of the finer level to the block of its coarse representative. 
At each level of the hierarchy, we apply a \emph{local search} algorithm.
Our local search algorithm is the same size-constraint label propagation algorithm we used in the contraction phase but with a different objective function.
Namely, when visiting a non-artificial node, we remove it from its current block and then we assign it to the neighboring block which maximizes the generalized \AlgName{Fennel} gain function defined in Equation~(\ref{eq:gen_fennel}).
Note that, in contrast to the initial partitioning, only blocks of adjacent nodes are considered here. 
Hence, one round of the algorithm can still be implemented to run in linear  time in the size the current level.
As in the coarsening phase, artificial nodes cannot be moved between blocks.
Differently though, we do not exclude the artificial nodes from the label propagation here.
This is the case because the artificial nodes and their edges are used to compute the generalized \AlgName{Fennel} gain function of the other nodes. 
As in the initial partitioning, we use the explicit size constraint $L_{\max}$~of~$G$. As a side note, we also tried to use high-quality offline algorithms as initial partitioning algorithms, however, in preliminary experiments this results in very unbalanced blocks (even with adaptively configured balance constraints) and overall in reduced quality throughout the process. Hence, we did not consider this~further.

Assuming geometrically shrinking graphs throughout the hierarchy and assuming that the buffer size $\delta$ is larger than the number of blocks $k$, then the overall running time to partition a batch is linear in the size of the batch.
This is due to the fact that the overall running time of coarsening and local search sums up to be linear in the size of the batch, while the overall running time of the initial partitioning depends linearly on the size of the input model. 
Summing this up over all batches yields overall linear running time $O(n+m)$.

\subsubsection{Restreaming}
\label{subsec:heistream_Restreaming}
We now extend \AlgName{HeiStream} to operate in a restreaming setting.
During restreaming, the overall structure of the algorithm is roughly the same.
Nevertheless we need to implement some adaptations which we explain in this section.
The first adaptations concern model construction.
Recall that the nodes from the current batch are already assigned to blocks during the previous pass of the input.
We explicitly assign these nodes to their respective blocks in $\mathcal{B}$.
Furthermore, ghost nodes and edges are not needed to construct $\mathcal{B}$.
This is the case since all nodes from future batches are already known and assigned to blocks, \ie these nodes will be represented in the artificial nodes.
More precisely, we adapt the artificial nodes to represent the nodes from all batches except the current one.
Since a partition of the graph is already given,  we do not allow the contraction of cut edges during restreaming in the coarsening phase of our multilevel scheme.
That means that clusters are only allowed to grow inside blocks.
As a consequence, we can directly use the partition computed in the previous pass as initial partitioning for $\mathcal{B}$ so we do not need to run an initial~partitioning~algorithm.

\subsubsection{Implementation Details}
\label{subsec:heistream_Implementation Details}
Our implementation of $\mathcal{B}$ is based on an adjacency array and consecutive node IDs.
We reserve the first $\delta$ IDs for the nodes from the current batch, which keep their global order.
This means that, when we process the $i^{th}$ batch, nodes IDs can be easily converted from our model $\mathcal{B}$ to $G$ and the other way around by respectively summing or subtracting $(i-1)*\delta$ on their ID.
Similarly, we reserve the last $k$ IDs of $\mathcal{B}$ for the artificial nodes and keep their relative order for all batches.
Note that this configuration separates \emph{mutable} nodes (nodes from current batch) and \emph{immutable} nodes (artificial nodes).
This allows us to efficiently control which nodes are allowed to move during coarsening, initial partitioning, and local search.
We keep an array of size $n$ store the permanent block assignment of the nodes of $G$.
To improve running time, we use approximate computation of powering in our \AlgName{Fennel} function.

\subsection{Experimental Evaluation}
\label{sec:heistream_Experimental Evaluation}

\paragraph*{Setup.}  
We performed the implementation of \AlgName{HeiStream} and competing algorithms inside the \AlgName{KaHIP} framework (using C++) and compiled them using gcc 9.3 with full optimization turned on (-O3 flag). 
Most of our experiments were run on a single core of Machine A.
The only exceptions are the experiments with huge graphs, which were run on a single core of Machine B. 
When using machine A, we stream the input directly from the internal memory, and when using machine $B$, that only has 16 GB of main memory, we stream the input from the~hard~disk.

\paragraph*{Baselines.}
\label{apdx:heistream_Baselines}
We compare \AlgName{HeiStream} against various state-of-the-art algorithms for streaming graph partitioning.
Since no official versions of the one-pass streaming and restreaming algorithms are available in public repositories, we implemented them in our framework.
Our implementations of these algorithms reproduce the results presented in the respective papers and are optimized for running time as much as possible. 
To this end, we implemented \AlgName{Hashing}, \AlgName{LDG}, \AlgName{Fennel}, and \AlgName{ReFennel}.
\AlgName{Multilevel~LDG}~\cite{jafari2021fast} is also not publicly available. 
We sent a message to the authors requesting an executable version of their algorithm for our tests but we have not receive any response.
Hence, we compare \AlgName{HeiStream} against \AlgName{Multilevel~LDG} based on the results explicitly reported in~\cite{jafari2021fast}.
We have used two machines: Machine~A and Machine~B.

\paragraph*{Instances.}
\label{apdx:heistream_instances}
In this section, we use three disjoint sets of graphs.
Basic properties of the graphs under consideration can be found in Table~\ref{tab:heistream_graphs}.
The \emph{tuning} set is used for the parameter study experiments, the \emph{test} set is used for the comparisons against the state-of-the-art, and the set of \emph{huge graphs} is used for special larger~scale~tests. 
In any case, when streaming the graphs we use the natural given order of the nodes.
We use $k \in \{2,3,\ldots,128\}$ for most experiments. 
We allow an imbalance of $3\%$ for all experiments (and all algorithms). All partitions computed by all algorithms have been balanced.

\paragraph*{Methodology.}  
Depending on the focus of the experiment, we measure running time and/or edge-cut.
Unless explicitly mentioned otherwise, we average all results of each algorithm grouped by $k$.

\subsubsection{Parameter Study}
\label{subsec:heistream_Algorithm Configuration}

\begin{figure}[p]
	\captionsetup[subfigure]{justification=centering}
	\centering
		\vspace{-.5cm}
	\begin{subfigure}[t]{0.495\textwidth}
		\includegraphics[angle=-0, width=0.9\textwidth]{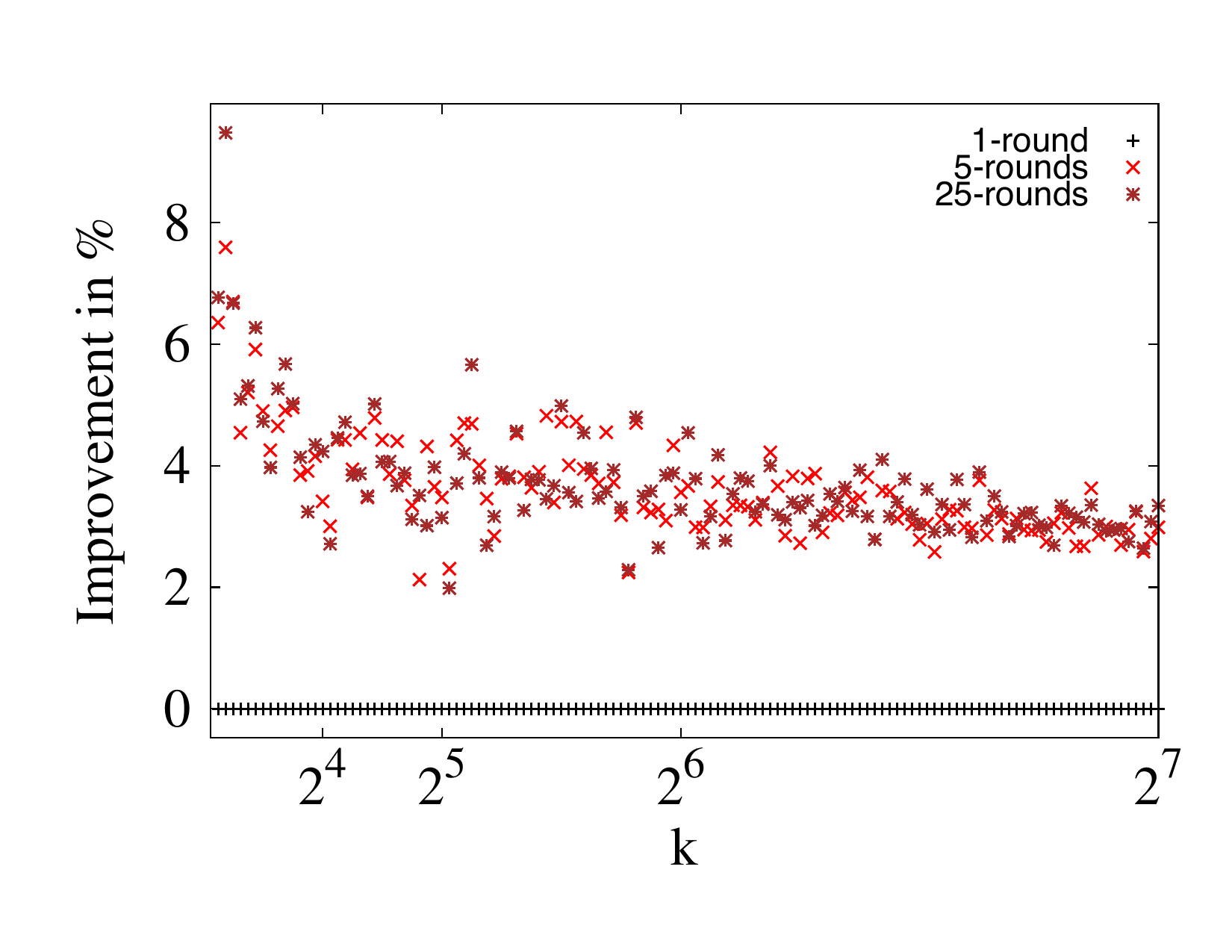}
		\caption{Quality improvement plot for label propagation rounds during uncoarsening.}
		\label{fig:heistream_res_rep_label}
	\end{subfigure}
	\begin{subfigure}[t]{0.495\textwidth}
		\includegraphics[angle=-0, width=0.9\textwidth]{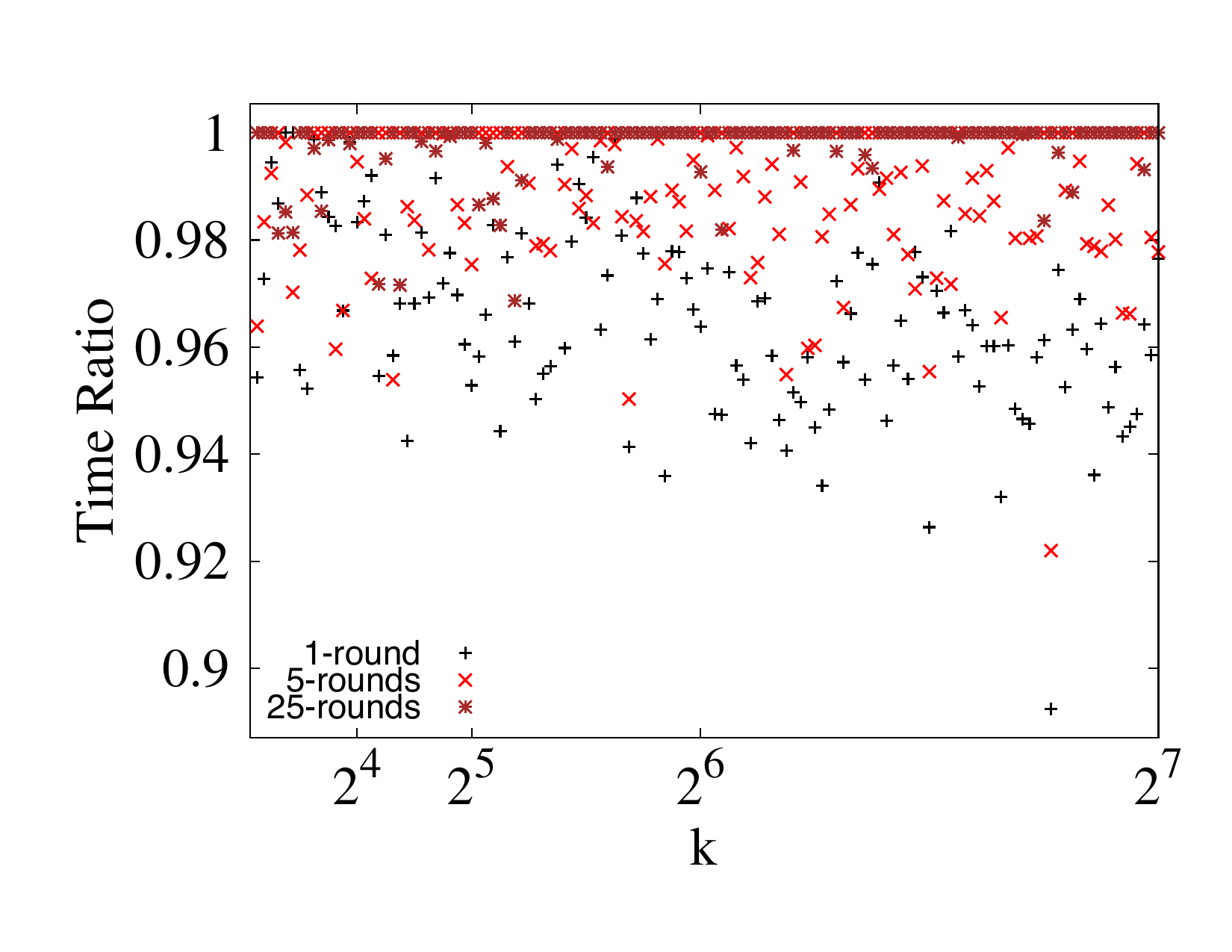}
		\caption{Running time ratio plot for label propagation rounds during uncoarsening.}
		\label{fig:heistream_tim_rep_label}
	\end{subfigure}
	\begin{subfigure}[t]{0.495\textwidth}
		\includegraphics[angle=-0, width=0.9\textwidth]{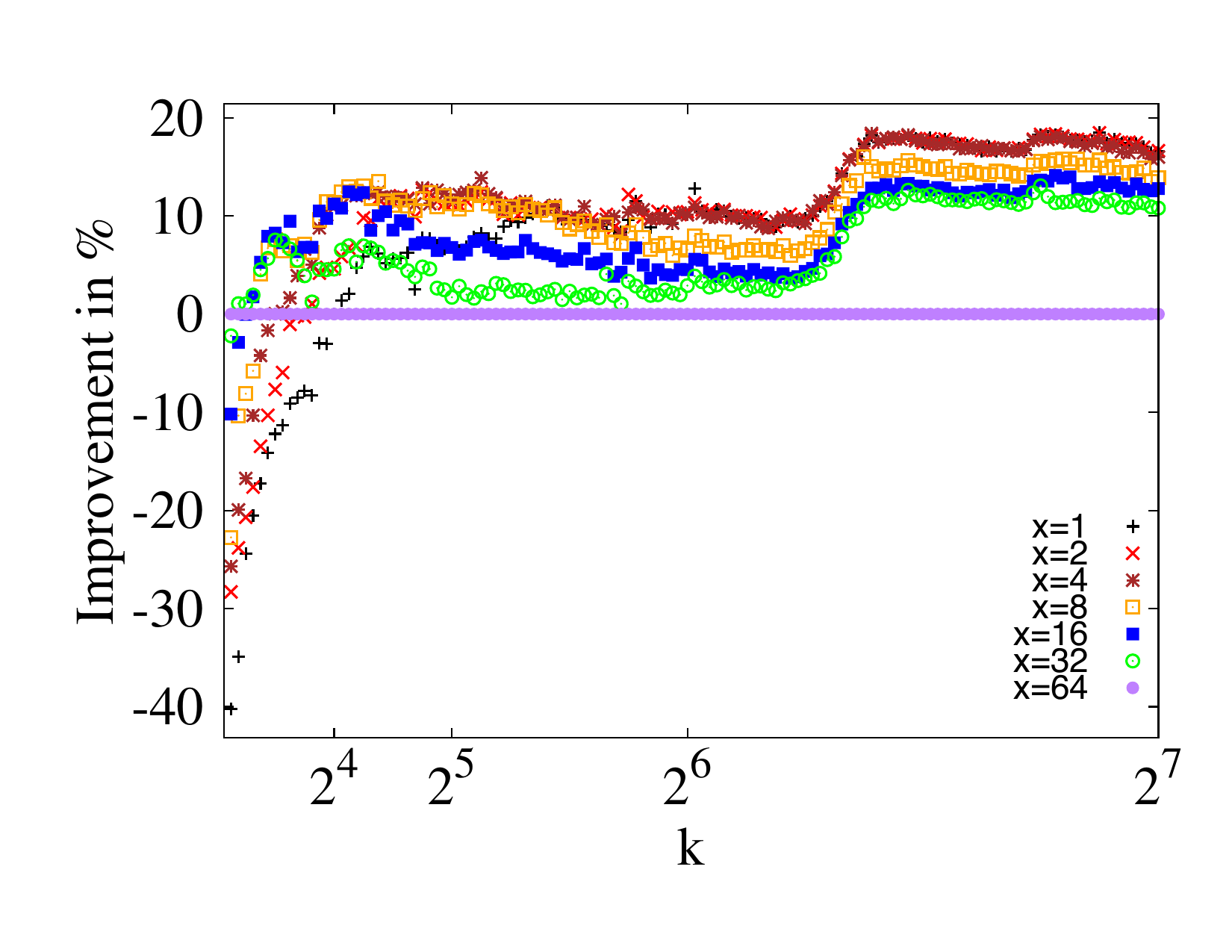}
		\caption{Quality improvement plot for parameter $x$ from expression $\max({|\mathcal{B}|/2xk,xk})$.}
		\label{fig:heistream_StopRule}
	\end{subfigure}
	\begin{subfigure}[t]{0.495\textwidth}
		\includegraphics[angle=-0, width=0.9\textwidth]{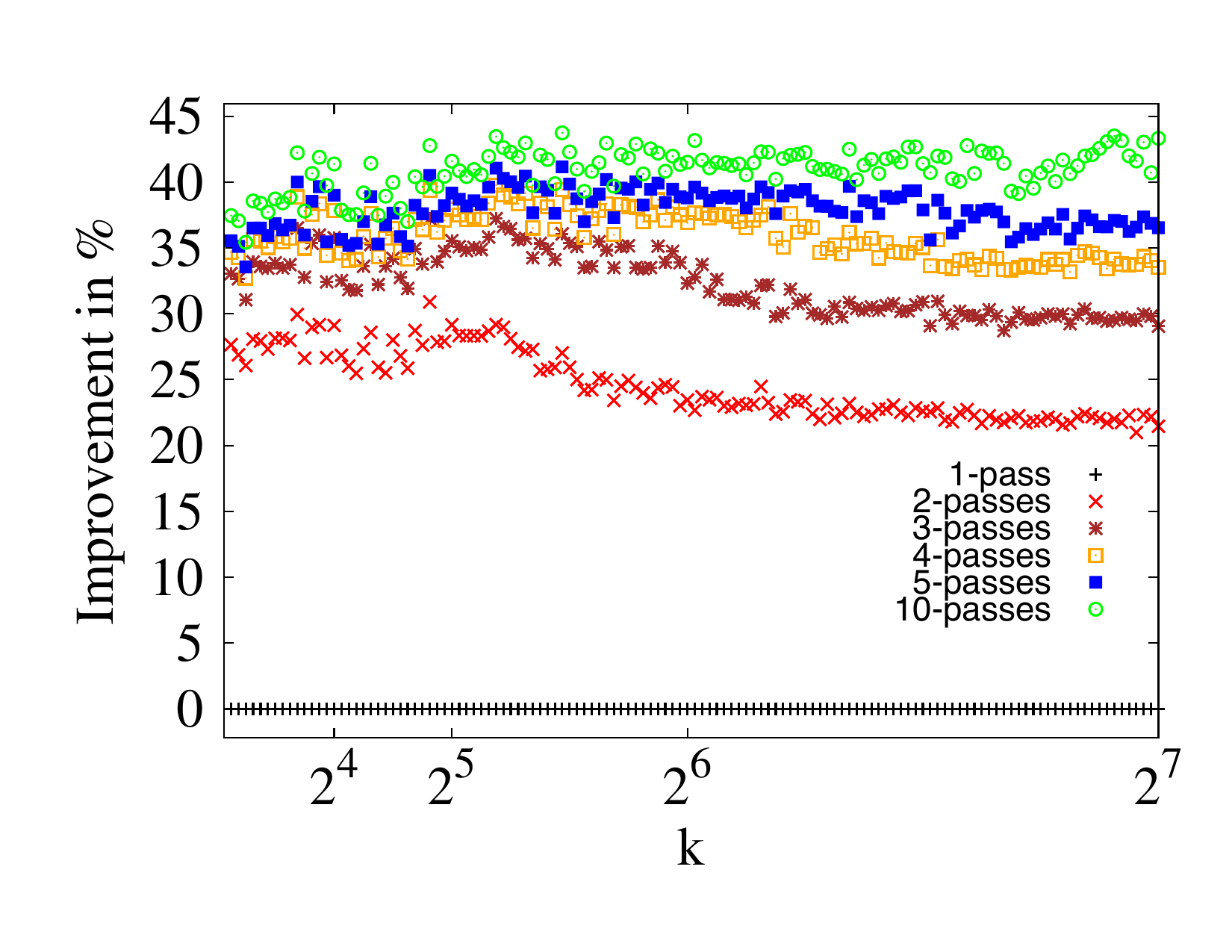}
		\caption{Quality improvement plot for restreaming.}
		\label{fig:heistream_res_restream}
	\end{subfigure}
	\begin{subfigure}[t]{0.495\textwidth}
		\includegraphics[angle=-0, width=0.9\textwidth]{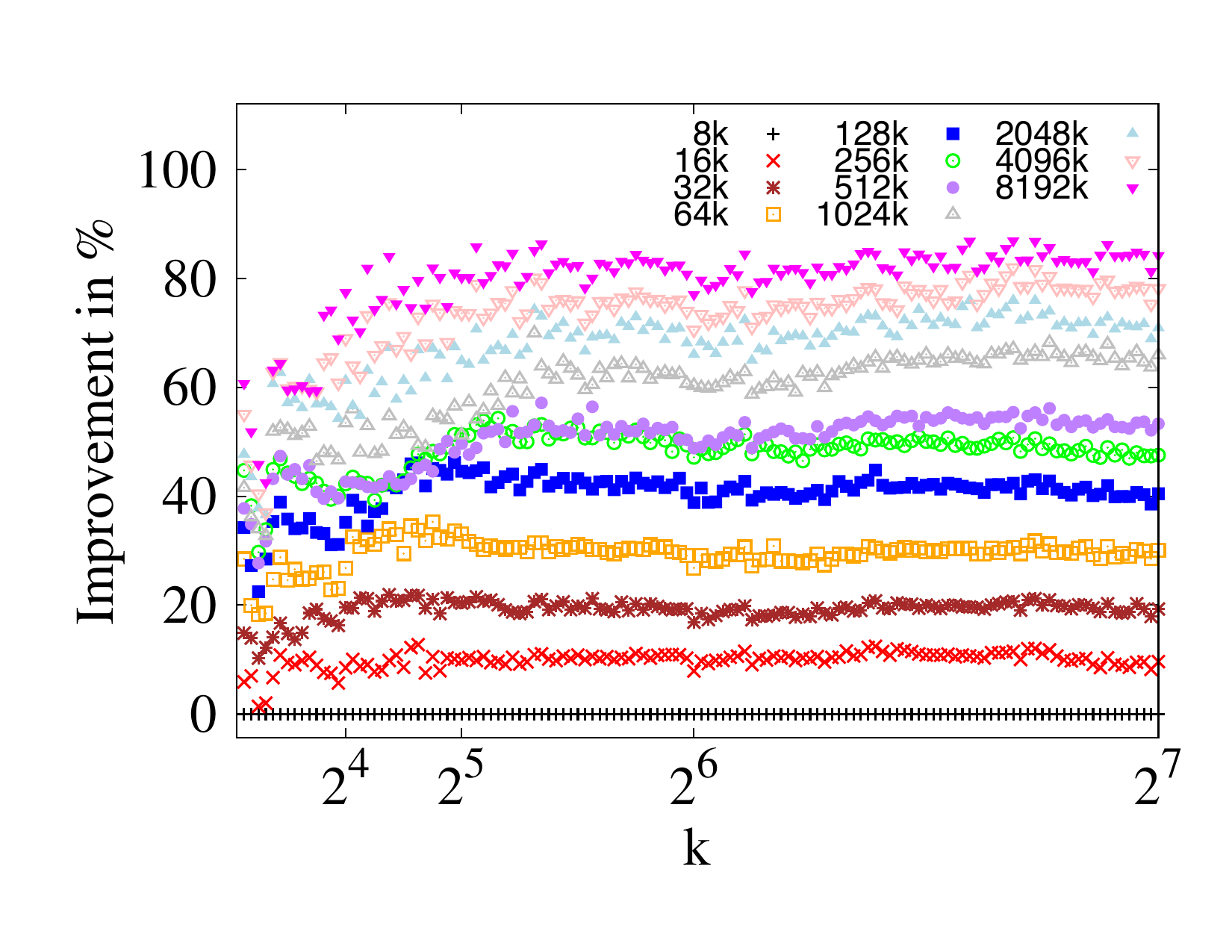}
		\caption{Quality improvement plot for buffer size.}
		\label{fig:heistream_res_batch_noFennel}
	\end{subfigure}
	\begin{subfigure}[t]{0.495\textwidth}
		\includegraphics[angle=-0, width=0.9\textwidth]{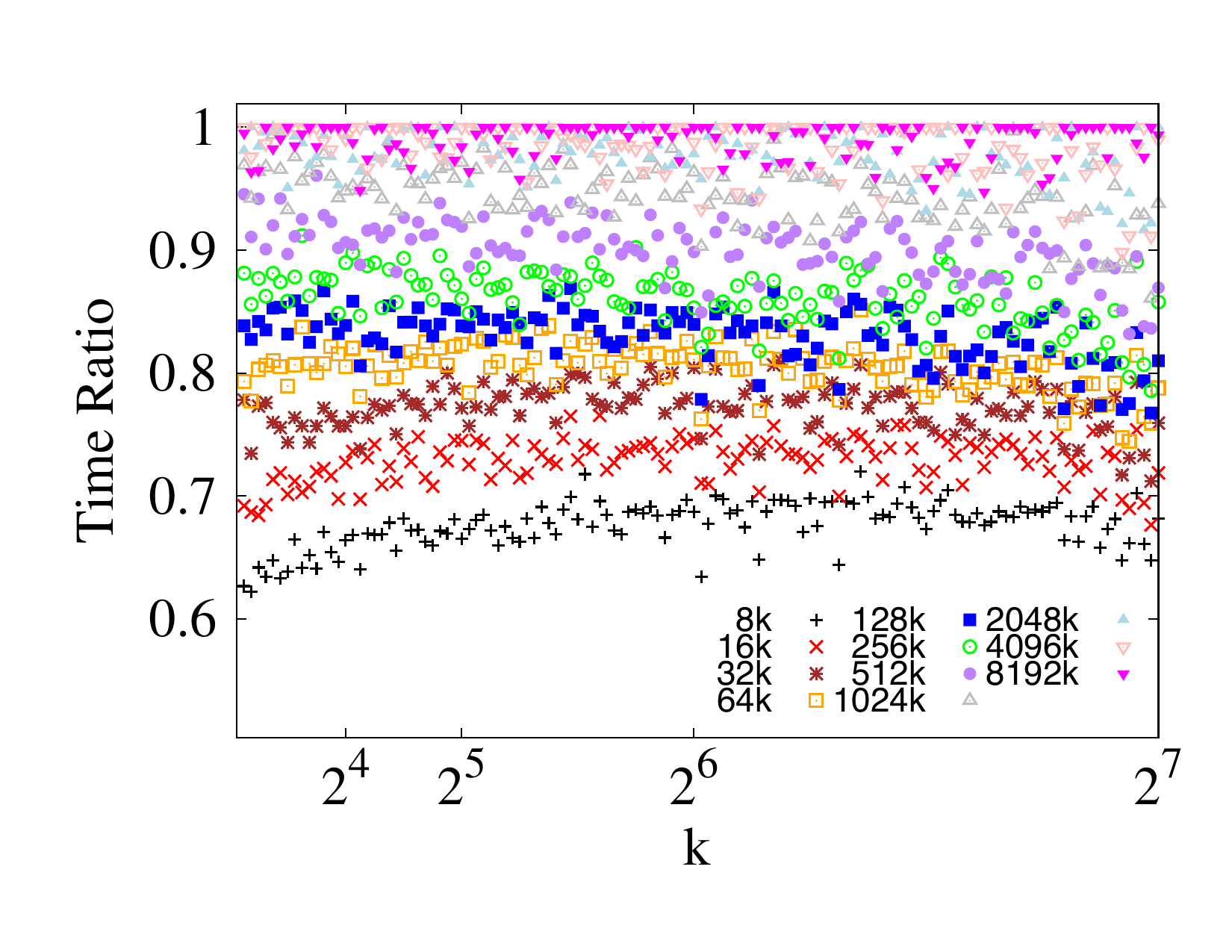}
		\caption{Running time ratio plot for buffer size.}
		\label{fig:heistream_tim_batch_noFennel}
	\end{subfigure}

	\caption{Results for tuning and exploration experiments. Higher is better for quality improvement plots. Lower is better for running time ratio plots.}
	\label{fig:heistream_tuning_plots2}
\end{figure}

\begin{figure*}[t]
	\captionsetup[subfigure]{justification=centering}
	\centering
	
	\begin{subfigure}[t]{0.495\textwidth}
		\centering
		\includegraphics[angle=-0, width=\imgScaleFactor\textwidth]{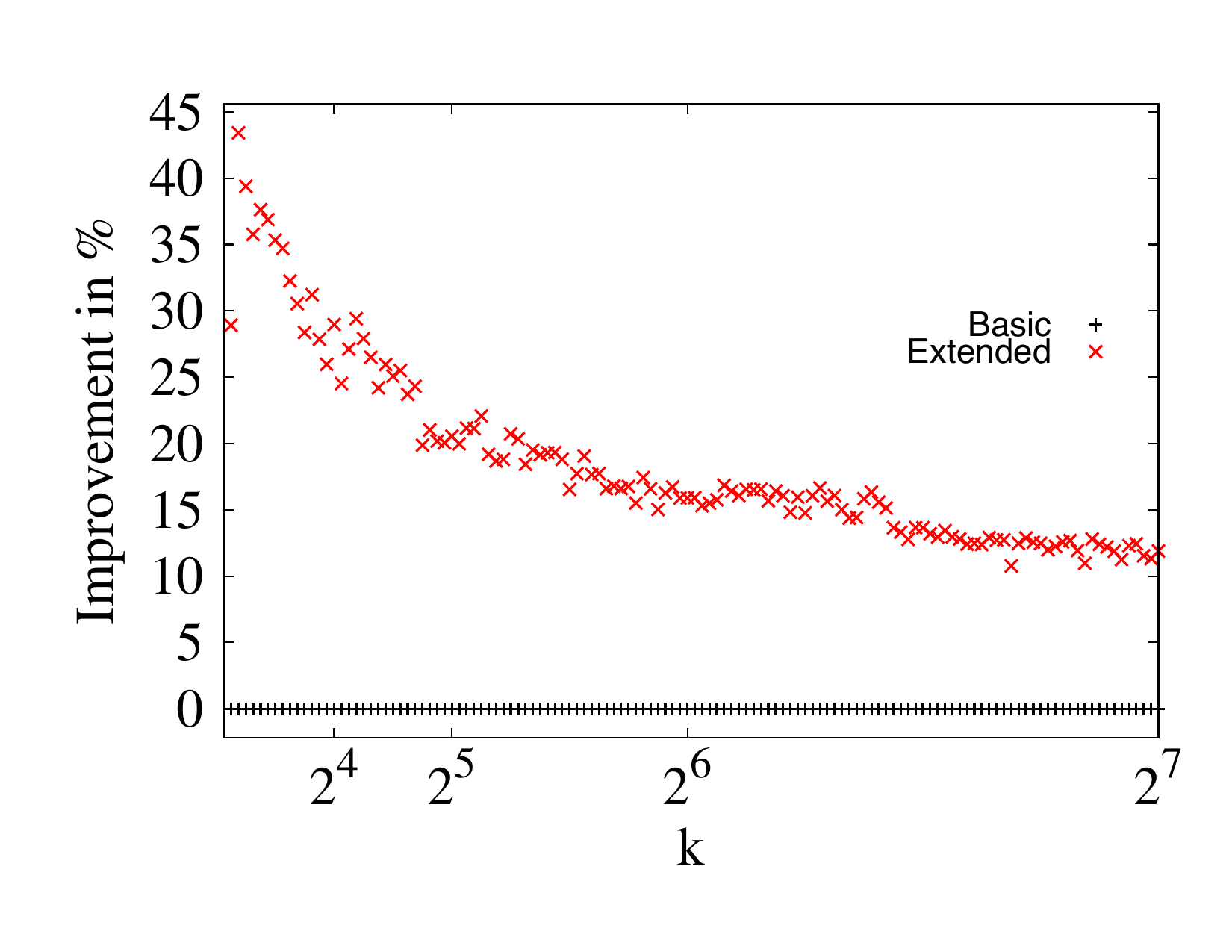}
		\vspace*{\capPosition}
		\caption{Quality improvement plot for model construction.}
		\label{fig:heistream_res_ghost}
	\end{subfigure}
	\begin{subfigure}[t]{0.495\textwidth}
		\centering
		\includegraphics[angle=-0, width=\imgScaleFactor\textwidth]{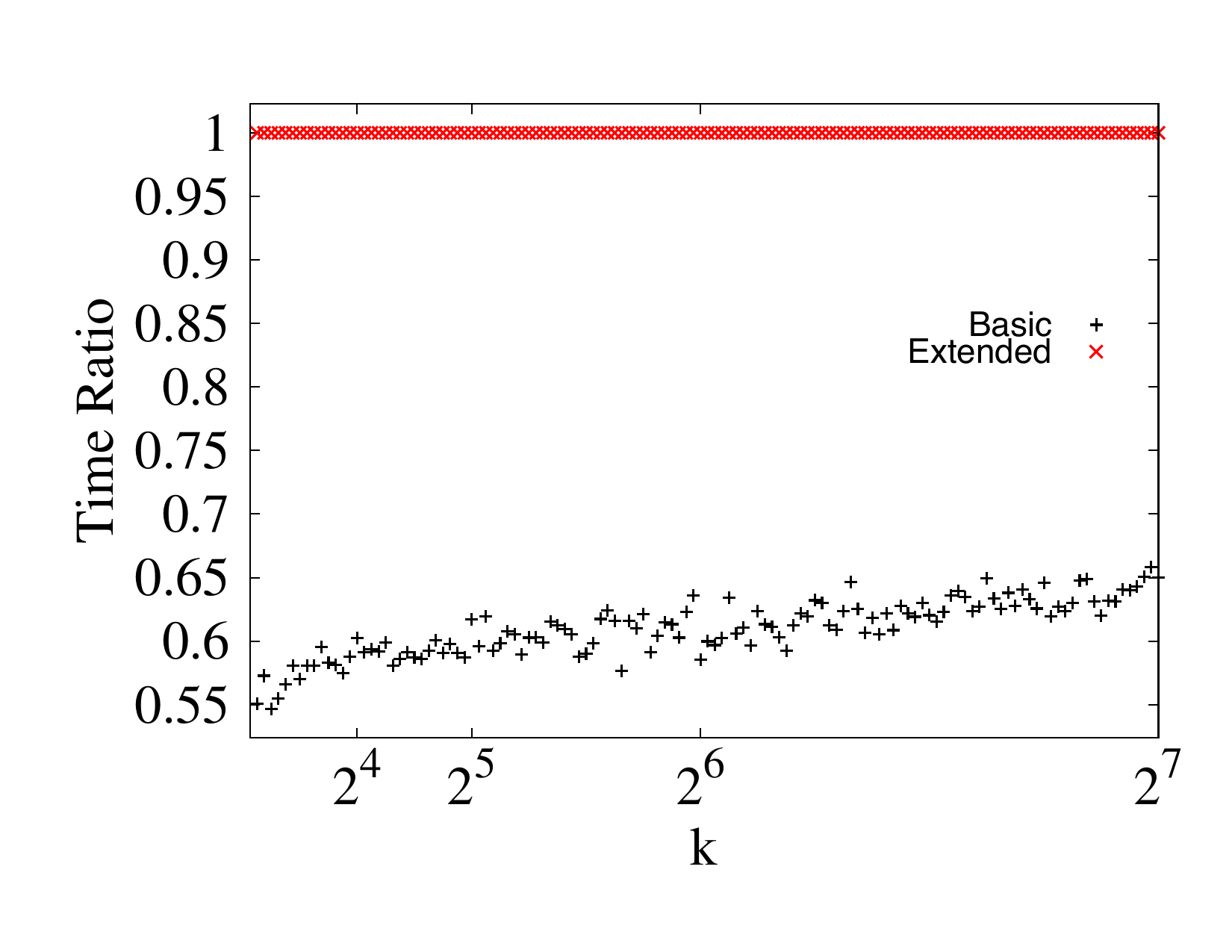}
		\vspace*{\capPosition}
		\caption{Running time ratio plot for model construction.}
		\label{fig:heistream_tim_ghost}
	\end{subfigure}

	\vspace*{.35cm}
	\caption{Results for tuning and exploration experiments. Higher is better for quality improvement plots. Lower is better for running time ratio plots.}
	\label{fig:heistream_tuning_plots3}

	\vspace*{-1cm}
\end{figure*}

We now present experiments to tune \AlgName{HeiStream} and explore its behavior.
We do this on the tuning set of our instance set.
In our strategy, each experiment focuses on a single parameter of the algorithm while all the other ones are kept invariable.
We start with a baseline configuration consisting of the following parameters:  $5$ rounds in the coarsening label propagation, $1$ round in the local search label propagation, and $x=64$ in the expression of the coarsest model size.
After each tuning experiment, we update the baseline to integrate the best found parameter.
Unless explicitly mentioned otherwise, we run the experiments of this section over all tuning graphs from Table~\ref{tab:heistream_graphs} based on the \emph{extended} model construction, \ie including ghost nodes and edges, for a buffer size of $\numprint{32768}$.

\paragraph*{Tuning.}
We begin by evaluating how the number of label propagation rounds during local search affects running time and solution quality.
In particular, we run configurations of \AlgName{HeiStream} with $1$, $5$, and $25$ rounds and report results in Figures~\ref{fig:heistream_res_rep_label}~and~\ref{fig:heistream_tim_rep_label}.
Observe that the results of the baseline have considerably lower solution quality than the other ones overall.
On the other hand, the results of the configurations with $5$ and $25$ rounds differ slightly to each other.
On average, they respectively improve solution quality $3.6\%$ and $3.7\%$ over the baseline.
Regarding running time, they respectively increase $2.2\%$ and $3.4\%$ on average over the baseline.
Since the variation of quality for these two configurations is not significant, we decided to integrate the fastest one among them in the algorithm, namely the $5$-round configuration.

Next we look at the parameter $x$ associated with the expression $\max\big({|\mathcal{B}|/2xk,xk}\big)$, which determines the size of the coarsest model.
We run experiments for $x=2^i$, with $i \in \{0,\ldots,6\}$, and report results in Figure~\ref{fig:heistream_StopRule}.
We omit running time charts for this experiment since the tested configurations present comparable behavior in this regard.
Figure~\ref{fig:heistream_StopRule} shows that the baseline presents the overall worst solution quality while $x=2$ and $x=4$ present the overall best solution quality.
Averaging over all instances, these two latter configurations produce results respectively $10.3\%$ and $11.2\%$ better than the baseline.
In light of that, we decide in favor of $x=4$ to compose \AlgName{HeiStream}.

\paragraph*{Exploration.}
We start the exploration of open parameters by investigating how the buffer size affects solution quality and running time.
We use as baseline a buffer of $\numprint{8192}$ nodes and successively double its capacity until any graph from the tuning set in Table~\ref{tab:heistream_graphs} fits in a single buffer.
We plot our results in Figures~\ref{fig:heistream_res_batch_noFennel}~and~\ref{fig:heistream_tim_batch_noFennel}.
Note that solution quality and running time increase regularly as the buffer size becomes larger.
This behavior occurs because larger buffers enable more comprehensive and complex graph structures to be exploited by our multilevel algorithm.
As a consequence, there is a trade-off between solution quality and resource consumption.
In other words, we can improve partitioning quality at the cost of considerable extra memory and slightly more running time.
Otherwise, we can save memory as much as possible and get a faster partitioning process at the cost of lowering solution quality.
In practice, it means that \AlgName{HeiStream} can be adjusted to produce partitions as refined as possible with the resources available in a specific system. 
For the extreme case of a single-node buffer, \AlgName{HeiStream} behaves exactly as \AlgName{Fennel}, while it behaves as an internal memory partitioning algorithm for the opposite extreme~case.

Next, we compare the effect of using the \emph{extended} model, which incorporates ghost nodes, over the \emph{basic} model, which ignores ghost nodes.
Figures~\ref{fig:heistream_res_ghost}~and~\ref{fig:heistream_tim_ghost} displays the results.
The results show that the extended model provides improved quality over the basic model, with an $18.3\%$ improvement on average.
This happens because the presence of ghost nodes and edges expands the perspective of the partitioning algorithm to future batches.
This has a similar effect to increasing the size of the buffer, but at no considerable extra memory cost.
Regarding running time, the results show that the extended model is consistently slower than the basic model for all values of $k$.
Averaging over all instances, the extended model costs $63.9\%$ more running time than the basic model.
This increase in running time is explained by the higher number of edges to be processed when ghost nodes are incorporated in the model.
As a practical conclusion from the experiment, the extended model can be used for better overall partitions with no significant extra memory but at the cost of extra running time.
Otherwise, the basic model can be used for a consistently faster execution at the cost of a lower solution quality.

Finally, we test to what extent solution quality can be improved by restreaming \AlgName{HeiStream} multiple times.
We investigate this by restreaming each input graph $10$ times.
We collect results after each pass and plot in Figure~\ref{fig:heistream_res_restream}.
The first restream generates a considerable quality jump, with an improvement over the baseline of $24.6\%$ on average.
Each following pass has a positive impact on solution quality, which converges to be a $40.9\%$-improvement on average over the baseline after the last pass.
On the other hand, the running time has a roughly linear increase for each pass over the graph.
In practice, this adds another degree of freedom to configure \AlgName{HeiStream} for the needs of real systems.

\subsubsection{State-of-the-Art}
\label{subsec:heistream_Scenarios}
\label{subsec:heistream_State-of-the-Art}

\begin{table}[p]
	\centering
	\scriptsize
	\setlength{\tabcolsep}{2.pt}
	\begin{tabular}{l@{\hskip 30pt}rr@{\hskip 30pt}rr@{\hskip 30pt}rrrr}
		\toprule	
		
		\multicolumn{1}{l}{}                        & \multicolumn{8}{c}{Cut Edges (\%)}                                 \\ 
		\multicolumn{1}{l}{\multirow{-2}{*}{ Graph}} & \AlgName{HS}(Int.) & \AlgName{MLDG}(Int.) & $2$-Re\AlgName{HS}(32k) & $2$-\AlgName{ReFennel} & \AlgName{HS}(32k) & \AlgName{Fennel} & \AlgName{LDG}                       & \AlgName{Hashing} 
		\\ 
                \midrule
		Dubcova1          & \textbf{13.68} & 14.26          & \textbf{12.92} &         29.19  & \textbf{13.68} & 33.99  &         33.96  & 95.62   \\ 
		hcircuit          & \textbf{ 2.73} & 17.75          & \textbf{ 2.04} &         21.86  &  \textbf{2.53} & 28.97  &         28.97  & 90.75   \\ 
		coAuthorsDBLP     & \textbf{15.99} & 24.82          & \textbf{16.90} &         24.28  & \textbf{17.80} & 27.12  &         27.12  & 94.80   \\ 
		Web-NotreDame     & \textbf{ 5.85} & 11.01          & \textbf{ 6.15} &         12.58  &  \textbf{9.20} & 19.52  &         19.56  & 95.97   \\
		Dblp-2010         & \textbf{11.31} & 18.52          & \textbf{12.20} &         22.93  & \textbf{13.42} & 28.82  &         28.80  & 92.49   \\ 
		ML\_Laplace       & \textbf{ 7.93} & 13.44          & \textbf{ 7.62} &          7.82  &          8.36  & 7.92   &  \textbf{5.77} & 96.37   \\ 
		coPapersCiteseer  & \textbf{ 8.23} & 11.22          & \textbf{ 8.35} &          9.63  & \textbf{10.29} & 12.88  &         12.27  & 96.52   \\ 
		coPapersDBLP      & \textbf{14.51} & 19.29          & \textbf{15.12} &         16.47  & \textbf{18.33} & 20.65  &         20.22  & 96.39   \\ 
		Amazon-2008       & \textbf{10.09} & 19.01          & \textbf{12.77} &         28.92  & \textbf{15.56} & 37.07  &         37.07  & 94.68   \\
		eu-2005           & \textbf{11.14} & 14.57          & \textbf{14.82} &         25.53  & \textbf{18.64} & 35.88  &         31.96  & 96.44   \\
		Web-Google        & \textbf{ 1.62} &  9.66          & \textbf{ 3.96} &         18.04  &  \textbf{9.48} & 30.64  &         30.64  & 96.87   \\
		ca-hollywood-2009 &         35.34  & \textbf{32.51} & \textbf{37.86} &         41.34  & \textbf{42.36} & 44.54  &         45.25  & 96.62   \\
		Flan\_1565        & \textbf{ 7.69} &  9.36          & \textbf{ 7.30} &         10.26  &          8.12  & 10.59  &  \textbf{6.70} & 96.61   \\ 
		Ljournal-2008     & \textbf{29.12} & 34.76          & \textbf{33.58} &         43.23  & \textbf{38.58} & 51.43  &         51.36  & 96.07   \\ 
		HV15R             &         14.09  & \textbf{11.48} &         15.38  & \textbf{15.05} &         17.48  & 16.39  & \textbf{15.49} & 96.84   \\ 
		Bump\_2911        & \textbf{ 8.66} & 11.73          & \textbf{ 7.65} &          8.61  &  \textbf{8.23} & 8.65   &         8.30   & 96.19   \\ 
		FullChip          & \textbf{38.20} & 48.16          & \textbf{42.24} &         57.39  & \textbf{45.71} & 61.93  &         64.23  & 95.06   \\ 
		patents           & \textbf{15.57} & 29.12          & \textbf{41.75} &         52.60  & \textbf{60.56} & 70.98  &         70.98  & 96.88   \\ 
		Cit-Patents       & \textbf{15.75} & 28.65          & \textbf{39.67} &         51.62  & \textbf{60.76} & 72.16  &         72.16  & 96.88   \\ 
		Soc-LiveJournal1  & \textbf{29.72} & 35.69          & \textbf{29.39} &         34.00  & \textbf{35.62} & 39.03  &         45.62  & 96.66   \\ 
		circuit5M         &         40.02  & \textbf{34.60} & \textbf{39.20} &         75.45  & \textbf{41.00} & 78.42  &         78.47  & 96.87   \\  
                \midrule
		del21             &  \textbf{1.38} & -              & \textbf{ 5.41} &         33.52  &  \textbf{8.53} & 40.21  &         40.21  & 93.39   \\ 
		rgg21             &  \textbf{1.53} & -              & \textbf{ 1.34} &          4.09  &  \textbf{1.52} & 5.02   &         4.88   & 96.89   \\ 
		soc-orkut-dir     & \textbf{37.86} & -              & \textbf{43.19} &         47.22  & \textbf{54.76} & 55.85  &         60.27  & 96.85   \\
		italy-osm         &  \textbf{0.13} & -              & \textbf{ 1.19} &          4.65  &  \textbf{1.34} & 4.80   &         4.80   & 78.11   \\ 
		great-britain-osm &  \textbf{0.16} & -              & \textbf{ 1.43} &          7.18  &  \textbf{1.63} & 7.34   &         7.34   & 79.94   \\ 
                \bottomrule
	\end{tabular}
	\caption{Edge-cut results against competitors for $k=32$. Internal memory algorithms are are on the 2 left columns and streaming algorithms are on the 5 right columns. We refer to setups of \AlgName{HeiStream} with specific buffer sizes as \AlgName{HS}($X$k), in which each buffer contains $X\times 1024$ nodes. \AlgName{HS}(Int.) uses a buffer size of $n$. We bold the best result for each graph for internal memory approaches and streaming approaches. The results for \AlgName{Multilevel~LDG} for the five bottom graphs are missing, as the graphs where not part of their benchmark set. Lower is better. }
	\label{tab:heistream_multiLDG}
\end{table}

\begin{figure}[t]
	\centering
	\includegraphics[width=\textwidth]{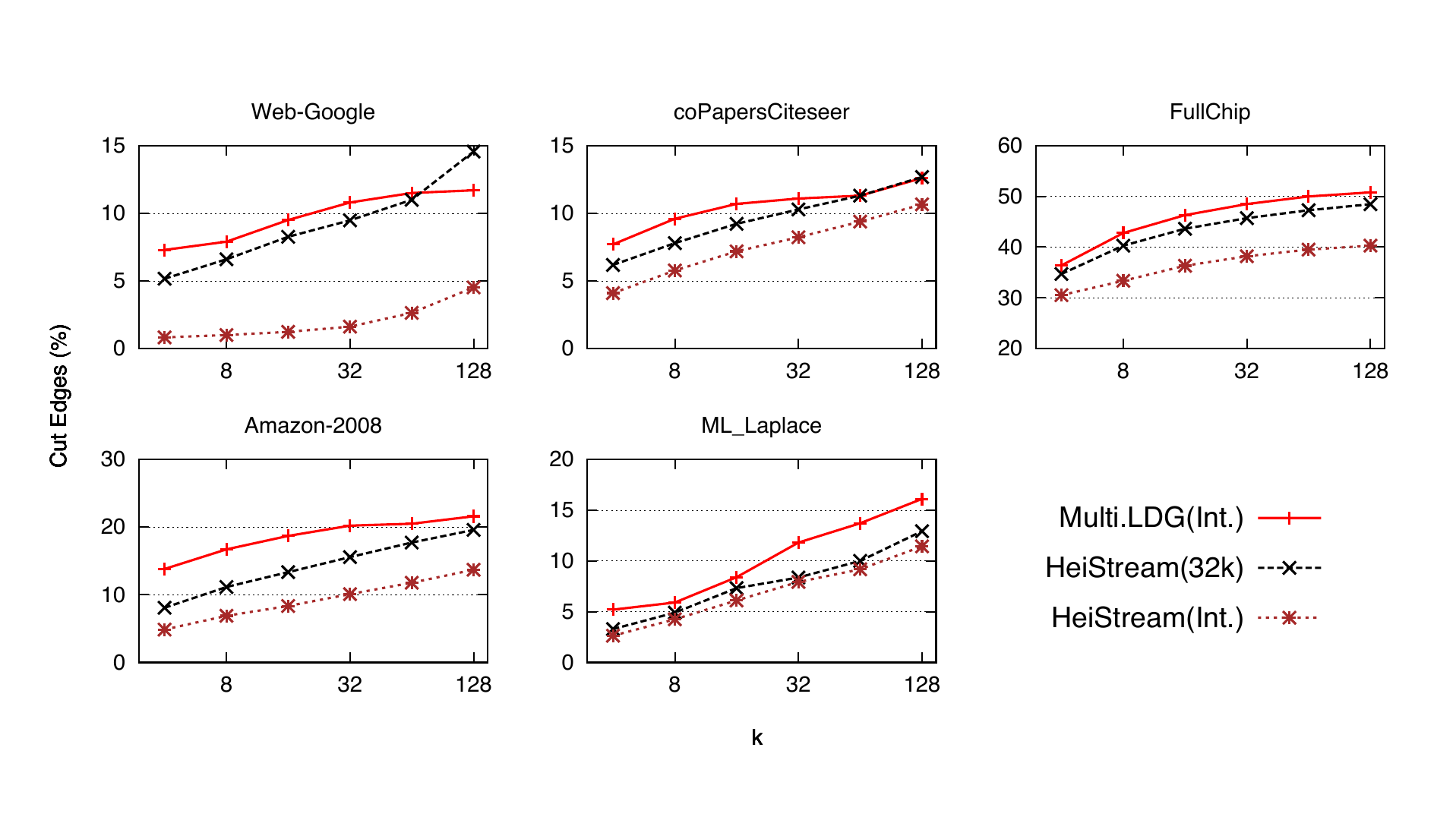}
	\caption{Comparison of \AlgName{HeiStream} against \AlgName{Multilevel~LDG} for buffer containing the whole graph.}
	\label{fig:heistream_multiLDG}
\end{figure}

\begin{figure}[htb]
	\vspace*{-.5cm}
	\captionsetup[subfigure]{justification=centering}
	\centering
	\begin{subfigure}[t]{0.495\textwidth}
		\centering
		\includegraphics[angle=-0, width=\imgScaleFactor\textwidth]{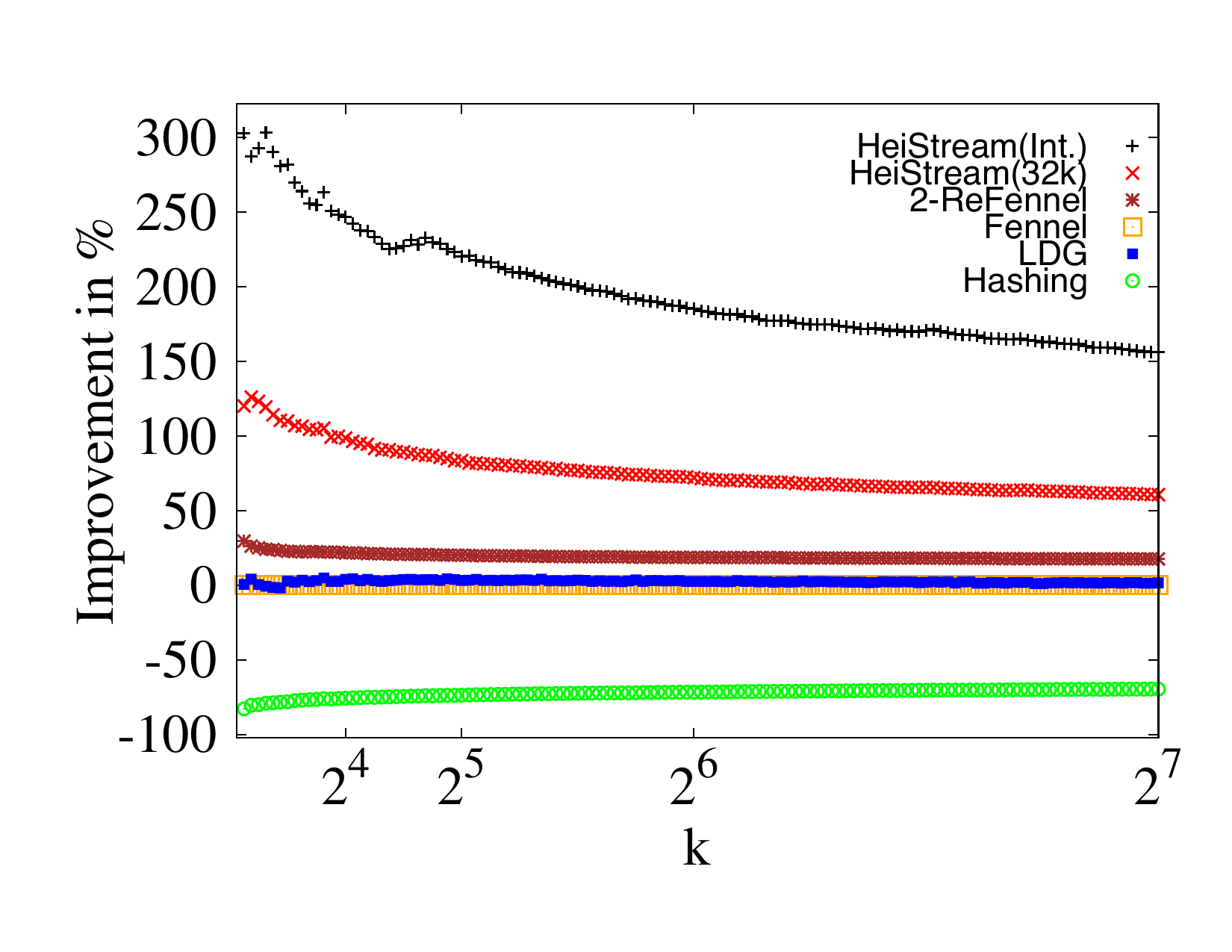}
		\caption{Quality improvement plot over \AlgName{Fennel}.}
		\label{fig:heistream_res_InitFennelOpt4}
	\end{subfigure}
	\begin{subfigure}[t]{0.495\textwidth}
		\centering
		\includegraphics[angle=-0, width=\imgScaleFactor\textwidth]{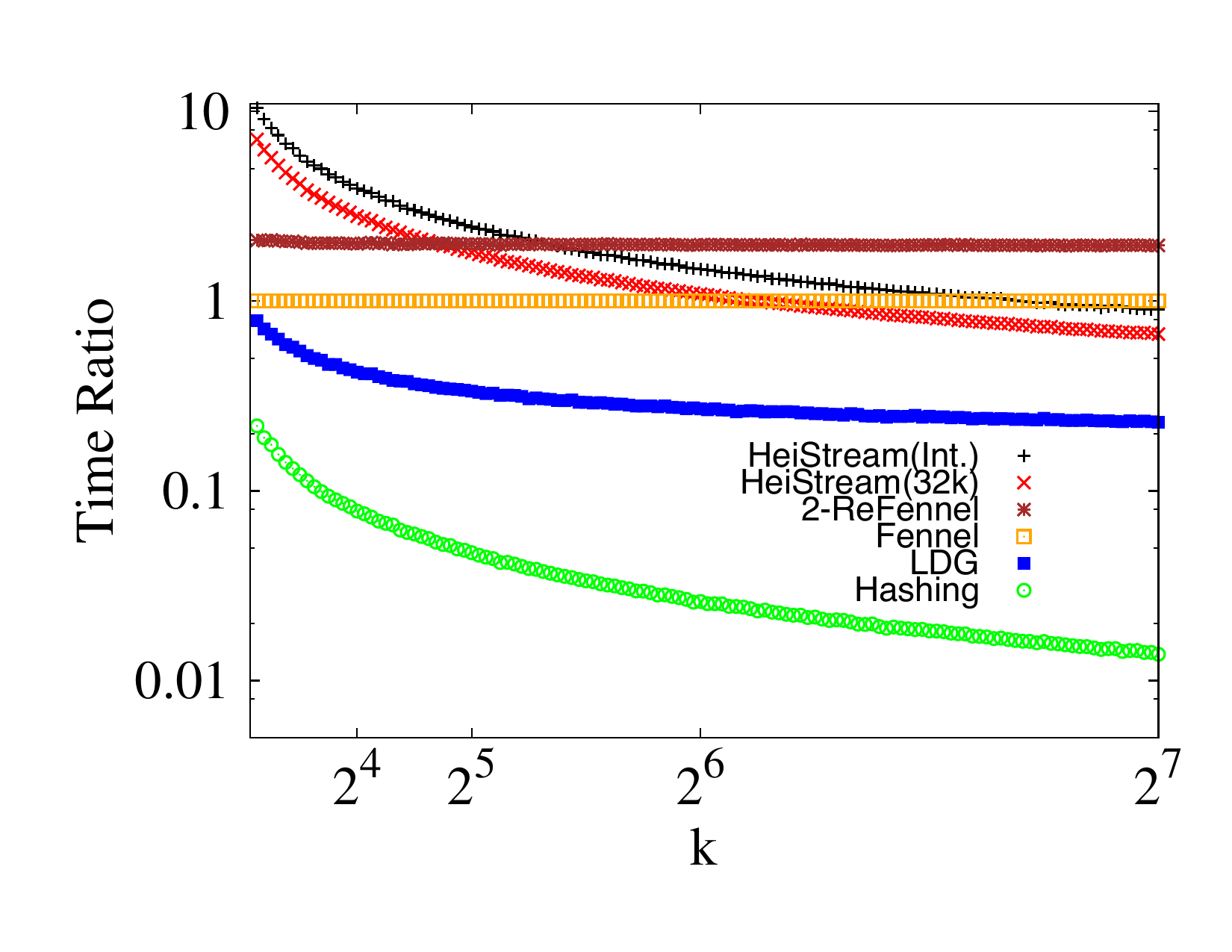}
		\caption{Relative running time plot over \AlgName{Fennel}.}
		\label{fig:heistream_tim_InitFennelOpt4}
	\end{subfigure}
	
	\begin{subfigure}[t]{0.495\textwidth}
		\includegraphics[angle=-0, width=\imgScaleFactor\textwidth]{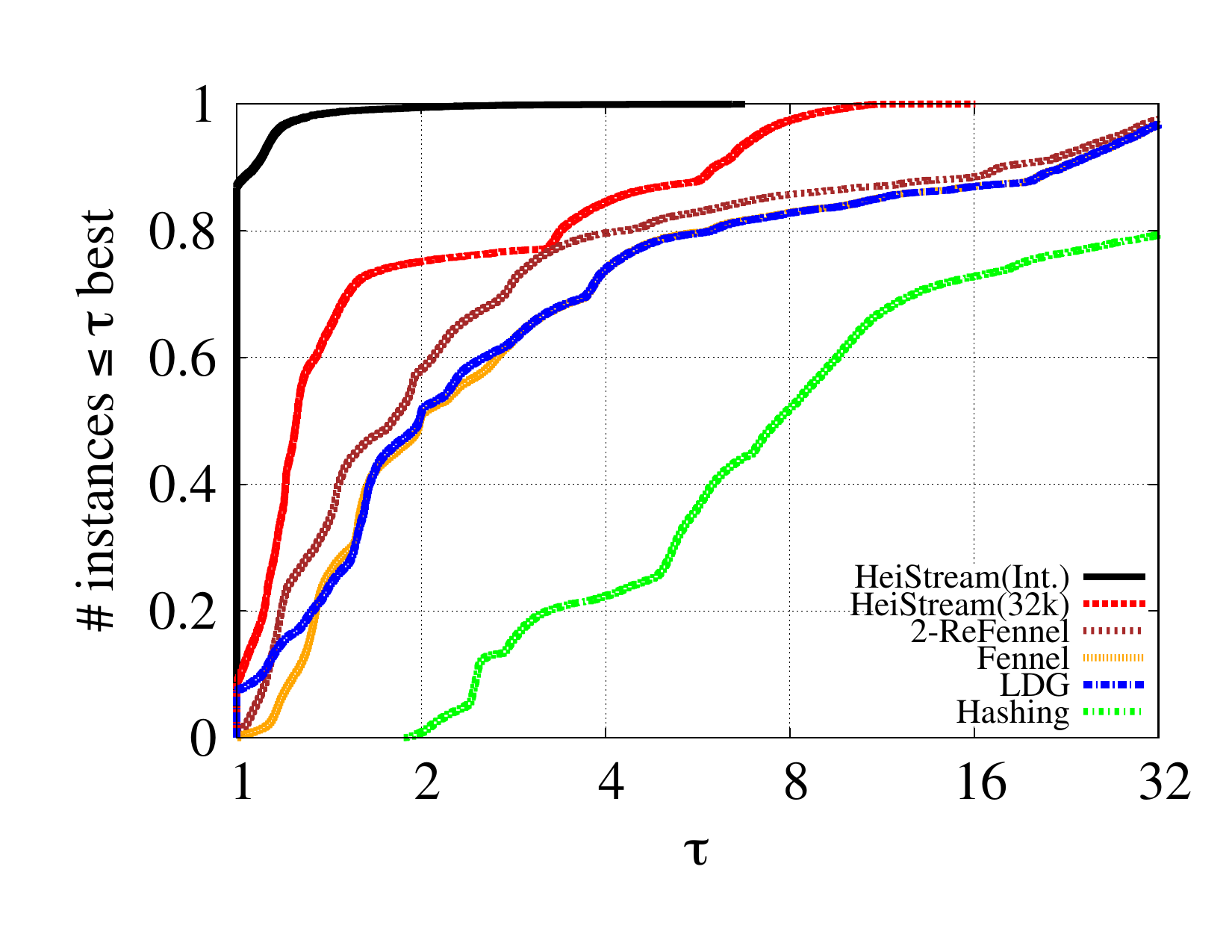}
		\caption{Quality performance profile.}
		\label{fig:heistream_pp_res_InitFennelOpt4}
	\end{subfigure}
	\begin{subfigure}[t]{0.495\textwidth}
		\includegraphics[angle=-0, width=\imgScaleFactor\textwidth]{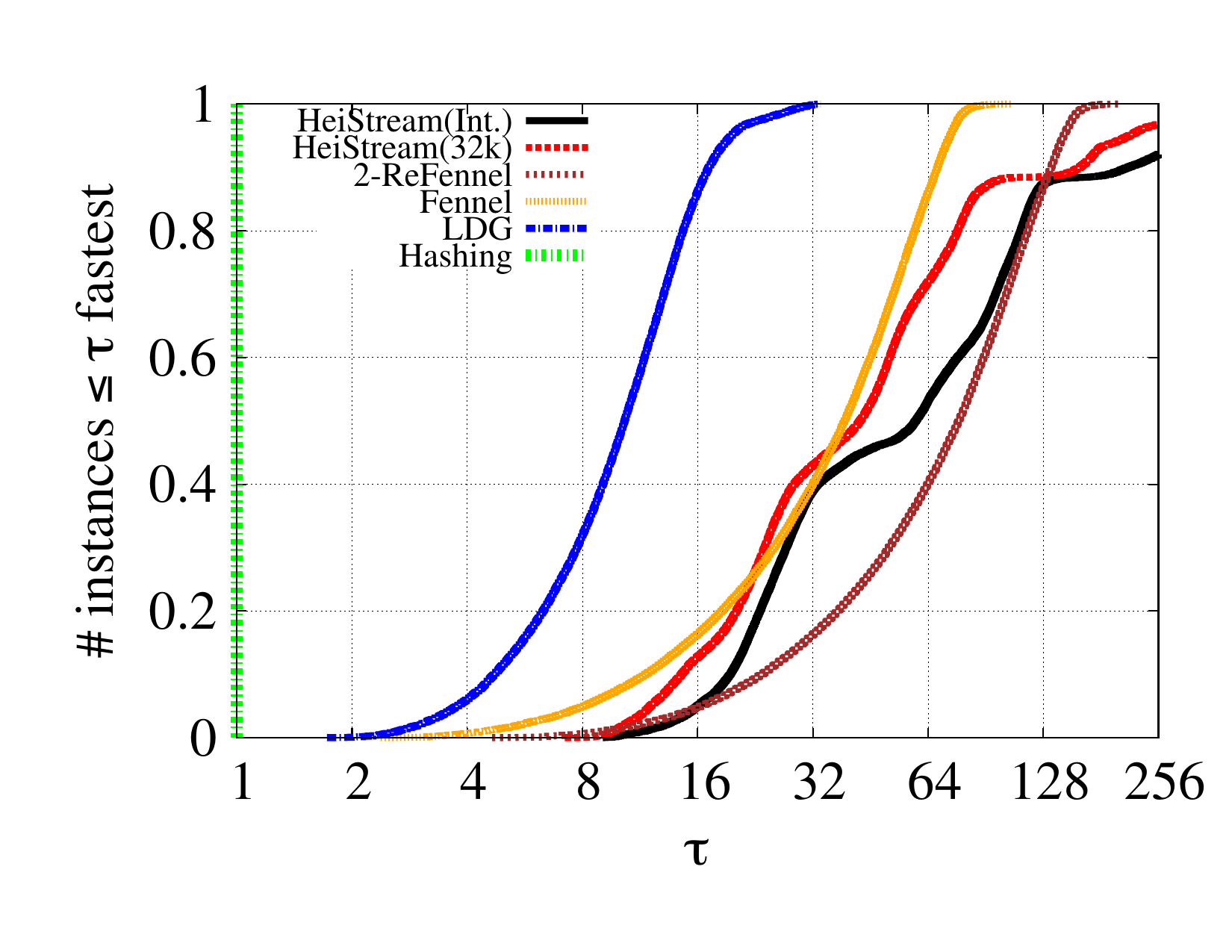}
		\caption{Running time performance profile.}
		\label{fig:heistream_pp_tim_InitFennelOpt4}
	\end{subfigure}
	\caption{Results for comparison against state-of-the-art one-pass (re)streaming algorithms using performance profiles. Higher is better for quality improvement plots. }
	\vspace*{-.5cm}
	\label{fig:heistream_one-pass2}
\end{figure}

\begin{figure*}[htb]
	\vspace*{-.5cm}
	\centering
	\begin{subfigure}{\scaleFactor\textwidth}
		\centering
		\includegraphics[width=\imgScaleFactor\textwidth]{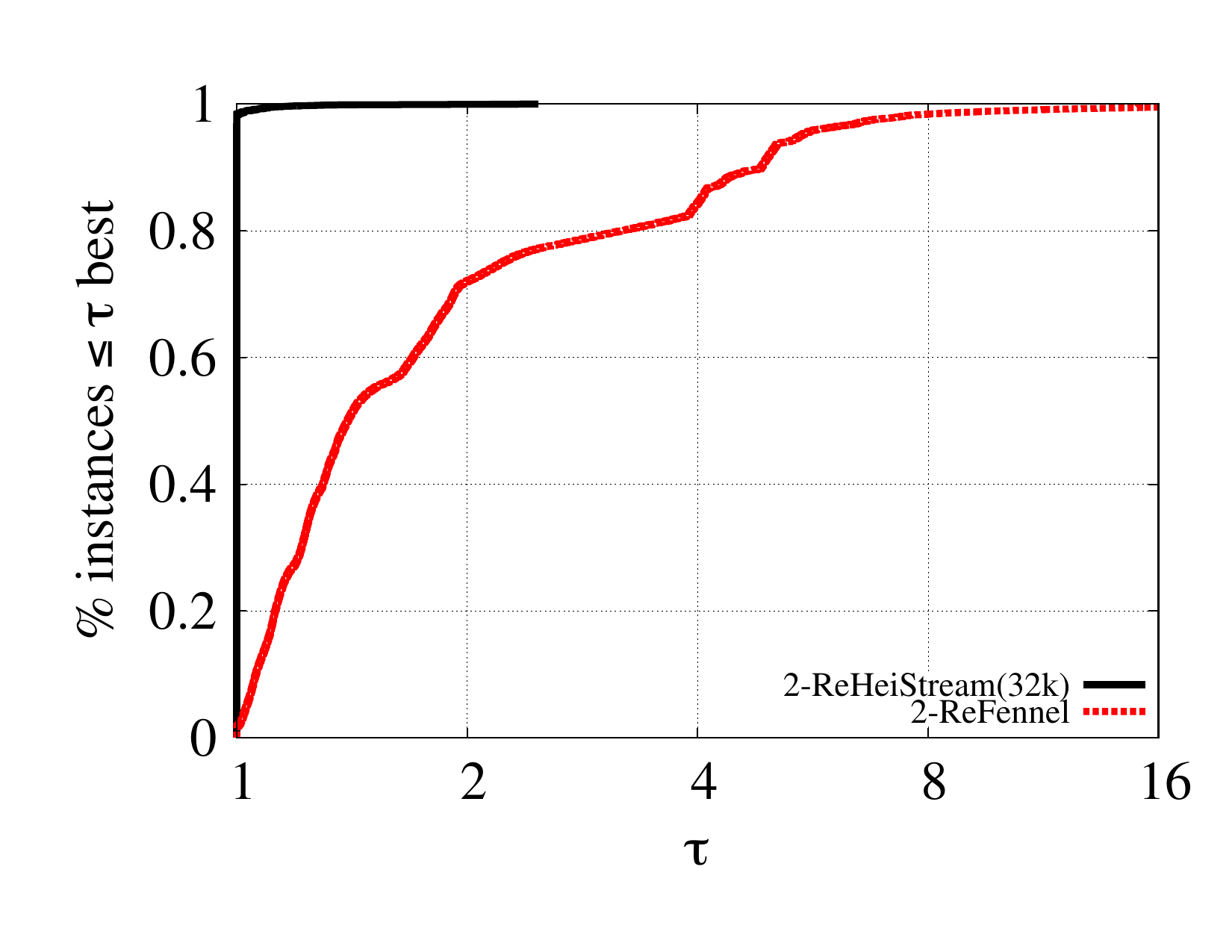}
		\vspace*{\capPosition}
	\end{subfigure}
	\caption{Quality performance profile of $2$-Re\AlgName{HeiStream}(32k) against $2$-\AlgName{ReFennel}.}
	\vspace*{-.5cm}
	\label{fig:heistream_restream}
\end{figure*}

In this section, we show experiments in which we compare \AlgName{HeiStream} against the current state-of-the-art algorithms.
Unless mentioned otherwise, these experiments involve all the graphs from the Test Set group in Table~\ref{tab:heistream_graphs} and focus on two particular configurations of \AlgName{HeiStream}, which we refer to as \AlgName{HeiStream}(32k) and \AlgName{HeiStream}(Int.).
The first configuration is based on batches of size \numprint{32768}, while the second one has enough batch capacity to operate as an internal memory algorithm -- both configurations perform a single pass over the input using the extended~model.
We also present results for the $2$-pass restreaming version of \AlgName{HeiStream}(32k), which we refer to as $2$-Re\AlgName{HeiStream}(32k).

Internal memory algorithms such as \AlgName{Metis}~\cite{karypis1998fast} and \AlgName{KaHIP}~\cite{kaffpa} are beyond the scope of this thesis since it is common knowledge that internal memory algorithms are better than streaming algorithms regarding partition quality if the instances fits in the memory of a machine.
For the sake of reproducibility, we ran \AlgName{Metis} and the fast social version of \AlgName{KaHIP} over our Test Set group of instances for $k \in \{2,4,8,16,32,64,128\}$.
On average, \AlgName{Fennel} cuts a factor $7.5$ more edges than \AlgName{Metis} and the fast social version of \AlgName{KaHIP}, while \AlgName{HeiStream}(Int.) cuts a factor $2.2$ more edges than both.
Furthermore, \AlgName{Fennel} is respectively a factor $2.2$ and a factor $6.0$ faster than \AlgName{Metis} and the fast social version of \AlgName{KaHIP}, while \AlgName{HeiStream}(Int.) is $56.3\%$ slower than \AlgName{Metis} and $72.9\%$ faster than the fast social version~of~\AlgName{KaHIP}.

\paragraph*{Results.}
We now present a detailed comparison of \AlgName{HeiStream} (HS) against the state-of-the-art.
In the results, we refer to the internal memory version of \AlgName{Multilevel~LDG} as \AlgName{MLDG}(Int.).
Moreover, we refer to the restreaming version of \AlgName{HeiStream} and \AlgName{Fennel} that passes over the graph $2$ times as $2$-ReHS and $2$-\AlgName{ReFennel} respectively. %
First, we focus on $k=32$ and later on choose a much wider range for the number of blocks.
Table~\ref{tab:heistream_multiLDG} shows the percentage of edges cut in the partitions generated by each algorithm for the graphs in the Test Set for $k=32$.
\AlgName{HeiStream}(Int.), $2$-Re\AlgName{HeiStream}(32k), and \AlgName{HeiStream}(32k) outperform all the other competitors for the majority of instances.
First, all of them outperform \AlgName{Hashing} for all graphs and \AlgName{LDG} for $23$ out of the $26$ graphs.
Next, they also outperform \AlgName{Fennel} in $24$ instances.
The algorithm $2$-\AlgName{ReFennel} is outperformed by \AlgName{HeiStream}(Int.), $2$-Re\AlgName{HeiStream}(32k), and \AlgName{HeiStream}(32k) in $24$, $25$, and $17$ instances, respectively.
Considering only the 21 instances for which there are results reported for \AlgName{Multi.LDG}(Int.) in literature, the algorithms \AlgName{HeiStream}(Int.), $2$-Re\AlgName{HeiStream}(32k), and \AlgName{HeiStream}(32k) compute better partitions for $18$, $15$,~and~$14$ instances respectively.

For a closer comparison against \AlgName{Multi.LDG}(Int.), we present Figure~\ref{fig:heistream_multiLDG}.
We plot edge-cut for \AlgName{Multi.LDG}(Int.) based on results graphically reported in~\cite{jafari2021fast}.
In this figure, we show of \AlgName{HeiStream}(Int.) and \AlgName{HeiStream}(32k) for $5$ particular graphs with $k=4,8,16,32,64,128$.
For all these instances, \AlgName{HeiStream}(Int.) outperforms \AlgName{Multi.LDG}(Int.) by a considerable margin.
Note that \AlgName{HeiStream}(32k) outperforms the \emph{internal memory} version of \AlgName{Multilevel~LDG} for the majority of instances.
We omit additional comparisons against buffered versions of \AlgName{Multilevel~LDG}, since they provide lower quality than the internal memory version.
We ran wider experiments over our whole Test Set for $127$ different values of $k$.
Figures~\ref{fig:heistream_res_InitFennelOpt4}~and~\ref{fig:heistream_tim_InitFennelOpt4} show a quality improvement plot over \AlgName{Fennel} and a relative running time plot, respectively.
Figures~\ref{fig:heistream_pp_res_InitFennelOpt4}~and~\ref{fig:heistream_pp_tim_InitFennelOpt4} show performance profiles for solution quality and running time, respectively.
Observe that \AlgName{HeiStream}(Int.) produces solutions with highest quality overall.
In particular, it produces partitions with smallest edge-cut for almost $86.9\%$ of the instances and improves solution quality over \AlgName{Fennel} $195.0\%$ on average. 
We now provide some results in which we exclude \AlgName{HeiStream}(Int.), since it has access to the whole graph.
The best algorithm is \AlgName{HeiStream}(32k), which produces the best solution quality for $63.3\%$ of the instances and improves solution quality over \AlgName{Fennel}  $75.9\%$ on average.
It is followed by $2$-\AlgName{ReFennel}, which is the best tested algorithm from the previous state-of-the-art. 
In particular, it computes the best partition for $26.8\%$ of the instances and improves on average $19.2\%$ over \AlgName{Fennel}.
\AlgName{LDG} and \AlgName{Fennel} come next.
Particularly, \AlgName{LDG} finds the best partition for $9.8\%$ of the instances and improves on average $2.4\%$ over \AlgName{Fennel}. 
\AlgName{Fennel} does not find the best partition for any instance.
Finally, \AlgName{Hashing} produces the worst solutions, with $72.5\%$ worse quality than \AlgName{Fennel} on average.
Regarding running time, \AlgName{Hashing} is the fastest one for all instances, which is expected since it is the only one with time complexity $O(n)$.
The second fastest one is \AlgName{LDG}, whose running time is a factor $9.6$ higher than \AlgName{Hashing} on average.
\AlgName{Fennel}, \AlgName{HeiStream}(32k) and \AlgName{HeiStream}(Int.) come next with factors $32.4$, $41.4$ and $56.2$ slower than \AlgName{Hashing} on average, respectively.
\AlgName{ReFennel} is the slowest algorithm of the test, being a factor $64.61$ slower than \AlgName{Hashing} on average.
In a direct comparison, \AlgName{HeiStream}(32k) and \AlgName{HeiStream}(Int.) are respectively $27.7\%$ and $73.2\%$ slower than \AlgName{Fennel} on average.
Note that both configurations of \AlgName{HeiStream} are faster than \AlgName{Fennel} for larger values of~$k$, which is consistent with the fact the running time of \AlgName{HeiStream} is $O(n+m)$ while the running time of \AlgName{Fennel}~is~$O(nk+m)$.

\paragraph*{Restreaming.}
Now we directly evaluate the restreaming version of \AlgName{HeiStream}.
As shown in Table~\ref{tab:heistream_multiLDG}, $2$-Re\AlgName{HeiStream}(32k) computes solutions with smaller edge-cut than $2$-\AlgName{ReFennel} for almost all our Test Set when $k=32$.
Figure~\ref{fig:heistream_restream} shows that this is the case for $k\in\{2, \ldots, 128\}$.
In particular, $2$-Re\AlgName{HeiStream}(32k) produces the best solution for $98.3\%$ of the instances and improves solution quality over $2$-\AlgName{ReFennel} by $79.6\%$ on average.
The average improvement of $2$-Re\AlgName{HeiStream}(32k) over $2$-\AlgName{ReFennel} is comparable to the average improvement of \AlgName{HeiStream}(32k) over \AlgName{Fennel}.
Nevertheless, the percentage of instances for which $2$-Re\AlgName{HeiStream}(32k) outperforms $2$-\AlgName{ReFennel} is almost $100\%$, which means that $2$-Re\AlgName{HeiStream}(32k) almost dominates $2$-\AlgName{ReFennel} with respect to edge-cut. 
Note that this is considerably higher than the percentage of instances for which \AlgName{HeiStream}(32k) outperforms \AlgName{Fennel} ($63.3\%$).

\begin{figure}[t]
	\centering
	\includegraphics[angle=-0, width=1.0\textwidth]{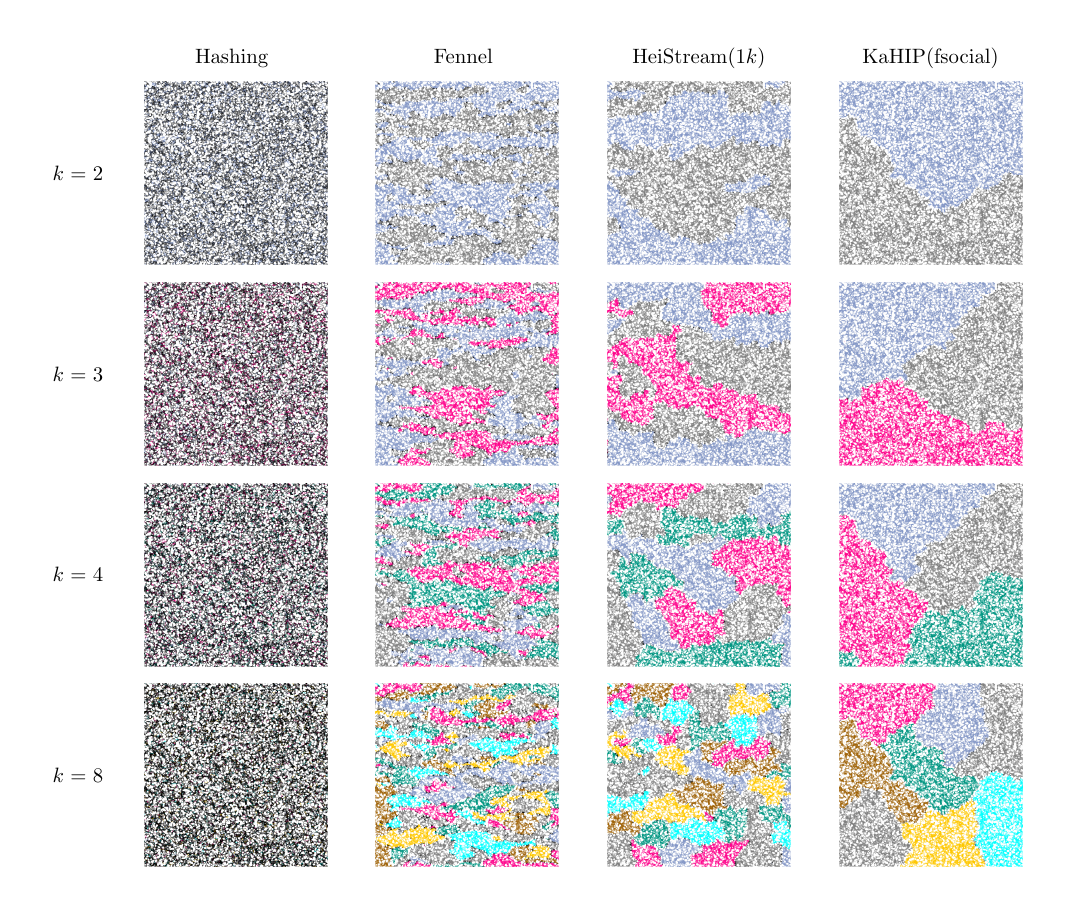}
	\caption{Visualization of partitions generated by the algorithms \AlgName{Hashing}, \AlgName{Fennel}, \AlgName{HeiStream($1k$)} and the fast social version of \AlgName{KaHIP} for the graph rgg15, which has \numprint{32768} nodes and \numprint{160240} edges.}
	\label{fig:heistream_visualization}
\end{figure}

\paragraph*{Visualization.}
As shown, the edge-cut of partitions produced by \AlgName{HeiStream} is on average lower than the edge-cut of partitions produced by its competitor streaming algorithms.
We shortly look at some visualizations in order to concretely understand why this happens.
In Figure~\ref{fig:heistream_visualization}, we show a visual comparison of some partition layouts generated by \AlgName{Hashing}, \AlgName{Fennel}, \AlgName{HeiStream} and the fast social version of \AlgName{KaHIP} for the graph rgg15.
Since this graph has only \numprint{32768} nodes, we use a buffer size of \numprint{1024} nodes for \AlgName{HeiStream} in order to partition the graph over multiple successive batches.
There is a leap of partitioning quality from \AlgName{Hashing} to \AlgName{Fennel}, \ie well-delimited clusters associated to a same block can be identified in the partitions generated by \AlgName{Fennel} but not in partitions generated by \AlgName{Hashing}.
Similarly there is a leap of partitioning quality from \AlgName{Fennel} to \AlgName{KaHIP}, \ie a block generated by \AlgName{Fennel} consists of multiple small clusters that are not mutually connected while a block generated by \AlgName{KaHIP} usually consists of a single connected cluster.
Note that the partitions produced by \AlgName{HeiStream} have intermediary characteristics between the partitions generated by \AlgName{Fennel} and \AlgName{KaHIP}.
More specifically, a block generated by \AlgName{HeiStream} consists of fewer and larger clusters than a block generated by \AlgName{Fennel} but not as few and as large clusters as those generated by \AlgName{KaHIP}.
This behavior is a direct consequence of the more or less global view provided by the distinct computational models used by these three algorithms.

\subsubsection{Huge Graphs}
\label{sec:heistream_huge_graphs}

\begin{table}[p]
	\centering
	\scriptsize
	\setlength{\tabcolsep}{2.pt}
	\begin{tabular}{lr@{\hskip 30pt}rrr@{\hskip 30pt}rr@{\hskip 30pt}rr@{\hskip 30pt}rr}
		\toprule	
		\multicolumn{1}{l}{\multirow{2}{*}{Graph}} & \multirow{2}{*}{k} & \multicolumn{3}{l}{\AlgName{HeiStream}(Xk)}                                               & \multicolumn{2}{l}{\AlgName{Fennel}}                               & \multicolumn{2}{l}{\AlgName{LDG}}                                  & \multicolumn{2}{l}{\AlgName{Hashing}}                              \\ 
		\multicolumn{1}{l}{}                       &                    & \multicolumn{1}{l}{X} & \multicolumn{1}{l}{CE(\%)} & \multicolumn{1}{l}{RT(s)} & \multicolumn{1}{l}{CE(\%)} & \multicolumn{1}{l}{RT(s)} & \multicolumn{1}{l}{CE(\%)} & \multicolumn{1}{l}{RT(s)} & \multicolumn{1}{l}{CE(\%)} & \multicolumn{1}{l}{RT(s)} \\ 
		
                \midrule
		\multirow{6}{*}{uk-2005}                   %
		& 8                  & 1024                  & \textbf{4.03}              & 290.23                      & 19.93                      & 37.19                       & 19.97                      & 19.36                       & 73.70                      & 3.28                         \\
		& 16                 & 1024                  & \textbf{6.01}              & 300.04                      & 22.76                      & 58.05                      & 22.72                      & 22.58                       & 78.86                      & 3.38                         \\
		& 32                 & 1024                  & \textbf{7.65}              & 310.72                      & 25.24                      & 98.78                       & 25.19                      & 29.19                       & 81.39                      & 3.30                         \\
		& 64                 & 1024                  & \textbf{8.99}              & 322.14                      & 26.88                      & 183.15                      & 26.81                      & 42.90                       & 82.60                      & 3.28                         \\
		& 128                & 1024                  & \textbf{9.94}              & 346.73                      & 27.89                      & 342.74                      & 27.76                      & 61.87                       & 83.18                      & 3.27                         \\	
		& 256                & 1024                  & \textbf{10.68}              & 386.64                      & 28.78          &  666.22           &            28.65           &   109.20                      &  83.46                     &   3.31                       \\ 
                \midrule
		\multirow{6}{*}{twitter7}                  %
		& 8                  & 512                   & \textbf{41.64}             & 1727.13                      &      45.18                 & 184.17                       & 56.11                      & 180.85                       & 71.66                      & 3.46                         \\
		& 16                 & 512                   & \textbf{47.04}             & 1774.92                      & 53.49                      & 213.17                       & 61.73                      & 186.27                       & 76.78                      & 3.57                        \\
		& 32                 & 512                   & \textbf{52.59}             & 1884.16                      & 58.15                      & 244.18                       & 66.84                      & 184.90                       & 79.34                      & 3.49                        \\
		& 64                 & 512                   & \textbf{57.53}             & 1988.11                      & 62.95                      & 330.46                      & 68.68                      & 197.86                       & 80.62                      & 3.50                        \\
		& 128                & 512                   & \textbf{61.87}             & 2113.34                      & 66.68                      & 504.53                      & 69.94                      & 219.65                       & 81.26                      & 3.93                         \\
		& 256                & 512                   & \textbf{65.47}             & 2357.92                      & 78.32                      & 846.20                     & 71.26                      & 280.99                    & 81.57                     & 3.51                         \\ 
                \midrule
		\multirow{6}{*}{sk-2005}                   %
		& 8                  & 1024                  & \textbf{3.23}              & 634.79                      & 21.95                      & 55.39                       & 21.13                      & 30.98                       & 81.50                      & 4.17                        \\
		& 16                 & 1024                  & \textbf{4.11}              & 648.48                      & 26.36                      & 82.41                       & 25.33                      & 35.44                       & 87.26                      & 4.20                        \\
		& 32                 & 1024                  & \textbf{5.32}              & 667.84                      & 29.59                      & 137.42                      & 27.97                      & 43.76                       & 90.11                      & 4.23                        \\
		& 64                 & 1024                  & \textbf{7.55}              & 695.60                      & 32.52                      & 238.54                      & 30.18                      & 59.19                       & 91.50                      & 4.21                        \\
		& 128                & 1024                  & \textbf{8.95}              & 733.05                      & 35.87                      & 449.19                      & 32.44                      & 91.36                       & 92.19                      & 4.20                        \\ 
		& 256                & 1024                  & \textbf{12.02}              & 798.73                     & 40.06                     & 857.76                    &  35.69                    &        150.64               & 92.55                     & 4.26                     \\ 
                \midrule
		\multirow{6}{*}{soc-friendster}           %
		& 8                  & 1024                  & \textbf{27.36}             & 4099.35                     & 30.57                      & 405.68                       & 45.60                      & 381.78                       & 87.53                      & 5.45                        \\
		& 16                 & 1024                  & \textbf{34.50}             & 4202.04                     & 45.74                      & 440.11                       & 58.98                      & 361.54                       & 93.77                      & 5.62                        \\
		& 32                 & 1024                  & \textbf{39.52}             & 4345.96                     & 54.87                      & 503.90                      & 61.00                      & 408.56                       & 96.89                      & 5.49                        \\
		& 64                 & 1024                  & \textbf{46.35}             & 4546.98                     & 59.27                      & 649.34                      & 64.02                      & 422.87                       & 98.45                      & 5.44                        \\
		& 128                & 1024                  & \textbf{52.41}             & 4796.56                     & 60.82                      & 888.14                      & 68.17                      & 475.16                       & 99.22                      & 5.72                        \\ 
		& 256                & 1024                  & \textbf{57.79}             & 5323.08                     & 64.25                      & 1426.16            & 71.90                & 523.66               &        99.61               &  5.53                       \\ 
                \midrule
		\multirow{6}{*}{er-fact1.5s26}           %
		& 8                  & 1024                  & \textbf{73.27}             & 2216.99                      & 73.44                      & 259.98                       & 73.44                      & 208.81                       & 87.50                      & 5.57                        \\
		& 16                 & 1024                  & \textbf{80.18}             & 2292.12                      & 80.40                      & 288.90                       & 80.40                      & 226.01                       & 93.75                      & 5.57                        \\
		& 32                 & 1024                  & \textbf{84.36}             & 2400.35                     & 84.63                      & 357.21                      & 84.63                      & 255.60                       & 96.87                      & 5.60                        \\
		& 64                 & 1024                  & \textbf{86.99}             & 2534.09                     & 87.31                      & 506.23                      & 87.31                      & 270.58                       & 98.44                      & 5.58                        \\
		& 128                & 1024                  & \textbf{88.72}             & 2725.81                     & 89.10                      & 769.61                      & 89.10                      & 407.59                       & 99.22                      & 5.57                        \\ 
		& 256                & 1024                  & \textbf{89.99}             & 2913.95                     & 90.45                      & 1329.57                 & 90.45                      &      408.35                 & 99.61                      & 5.65                        \\ 
                \midrule
		\multirow{6}{*}{RHG1}                     %
		& 8                  & 1024                  & \textbf{0.04}              & 380.04                      & 2.02                       & 86.95                       & 2.02                       & 44.42                       & 91.91                      & 8.37                        \\
		& 16                 & 1024                  & \textbf{0.06}              & 391.63                      & 2.12                       & 143.47                      & 2.12                       & 52.11                       & 97.39                      & 8.57                        \\
		& 32                 & 1024                  & \textbf{0.09}              & 406.56                      & 2.16                       & 252.15                      & 2.16                       & 65.65                       & 99.19                      & 8.38                        \\
		& 64                 & 1024                  & \textbf{0.15}              & 435.71                      & 2.17                       & 450.73                      & 2.17                       & 95.22                       & 99.74                      & 8.33                        \\
		& 128                & 1024                  & \textbf{0.22}              & 482.06                      & 2.18                       & 877.69                      & 2.18                       & 147.22                      & 99.90                      & 8.31                        \\
		& 256                & 1024                  & \textbf{0.34}              & 569.77                     &   2.19                    &   1708.33                 &   2.18                     &   273.76                   &   99.95                    &   8.45                   \\ 
                \midrule
		\multirow{6}{*}{RHG2}                      %
		& 8                  & 1024                  & 0.09                       & 621.56                      & 0.05                       & 103.83                       & \textbf{0.04}              & 56.23                       & 89.71                      & 8.31                        \\
		& 16                 & 1024                  & 0.13                       & 632.61                      & 0.07                       & 153.29                      & \textbf{0.04}              &    60.14                         & 96.15                      & 8.57                        \\
		& 32                 & 1024                  & 0.19                       & 648.68                      & 0.12                       & 262.91                      & \textbf{0.05}              & 77.34                       & 98.73                      & 9.85                        \\
		& 64                 & 1024                  & 0.29                       & 674.36                      & 0.18                       & 468.04                      & \textbf{0.05}              & 108.39                       & 99.57                      & 8.32                        \\
		& 128                & 1024                  & 0.44                       & 727.66                      & 0.28                       & 872.68                      & \textbf{0.07}              & 157.75                      & 99.85                      & 8.31                        \\
		& 256                & 1024                  & 0.68                       & 816.60                      & 0.44                       & 1686.18                   & \textbf{0.09}              & 278.51                     & 99.92                      & 8.47                      \\ 
                \midrule
		\multirow{6}{*}{uk-2007-05}                %
		& 8                  & 1024                  & \textbf{0.54}              & 1024.26                      & 25.23                      & 107.46                       & 25.21                      & 58.19                       & 87.91                      & 8.80                        \\
		& 16                 & 1024                  & \textbf{0.60}              & 1045.36                      & 28.02                      & 166.06                      & 28.19                      & 70.72                       & 94.08                      & 8.79                        \\
		& 32                 & 1024                  & \textbf{0.70}              & 1058.73                      & 29.40                      & 278.38                       & 29.32                      & 85.42                       & 97.12                      & 8.99                        \\
		& 64                 & 1024                  & \textbf{0.92}              & 1099.64                      & 29.94                      & 517.20                      & 29.85                      & 115.18                       & 98.61                      & 9.32                        \\
		& 128                & 1024                  & \textbf{1.31}              & 1163.45                      & 30.73                      & 935.98                      & 30.18                      & 175.01                      & 99.33                      & 8.79                        \\
		& 256                & 1024                  & \textbf{1.95}              & 1280.48                      & 31.65                      & 1808.52                    & 30.70                      & 324.71                    & 99.68                      & 8.92                      \\ 
                \bottomrule
	\end{tabular}
	\vspace*{-.25cm}
	\caption{Experiments with huge graphs. CE and RT denote cut edges and running time. 
	}
	\label{tab:heistream_hugeResults}
\end{table}

We now switch to the main use case of streaming algorithms: computing high-quality partitions for huge graphs on small machines. 
The experiments in this section are based on the \emph{huge graphs} listed in Table~\ref{tab:heistream_graphs} and are run on the relatively small Machine~B.
Namely, we ran experiments for $k=\{8,16,32,64,128,256\}$ %
and we did not repeat each test multiple times with different seeds as in previous experiments. We also ran \AlgName{Metis} and \AlgName{KaHIP} on those graphs on this machine, but they fail on all instances since they require more memory than the machine has.
For all instances, \AlgName{HeiStream} performs a single pass over the input based on the extended model construction.
We refer to setups of \AlgName{HeiStream} with specific buffer sizes as \AlgName{HeiStream}($X$k), in which a buffer contains $X\times 1024$ nodes. %
In~Table~\ref{tab:heistream_hugeResults}, we report detailed per-instance results with large buffer sizes able to run on Machine~B.
We exclude from Table~\ref{tab:heistream_hugeResults} the IO~delay to load the input graph from the disk, since it depends on the disk and is roughly the same independently of the used partitioning algorithm.
For completeness, we report this delay (in seconds) for the huge graphs listed in Table~\ref{tab:heistream_graphs} following their respective order: 131.3, 203.2, 313.2, 294.0, 164.9, 186.1, 340.7, 551.5.

The results show that \AlgName{HeiStream} outperforms all the competitor algorithms regarding solution quality for most instances.
Notably, \AlgName{HeiStream} computes partitions with considerably lower edge-cut in comparison to the one-pass algorithms for 4 out of the tested graphs: uk-2005, sk-2005, uk-2007-05, and RHG1.
For the social networks soc-friendster and twitter7, \AlgName{HeiStream} is the best for all instances, %
but the improvement over \AlgName{Fennel} and \AlgName{LDG} is not so large as in the other instances.
One outlier can be seen on the network RHG2.
While \AlgName{HeiStream} produces fairly small edge-cut values, which are all below $0.7\%$, \AlgName{Fennel} does outperform it and \AlgName{LDG} improves solution quality even further on this instance.
Furthermore, note that the running time of \AlgName{Fennel} increases with increasing $k$ to the point in which it becomes higher than the running time of \AlgName{HeiStream} for 5 out of the 8 huge graphs tested.

\vfill

\paragraph*{Memory Consumption.} We now shortly review the amount of memory needed by the streaming algorithms under consideration. First of all note that the memory of \AlgName{HeiStream} depends on the size of the buffer that is used. If the buffer only contains one node, then the memory requirements match those of \AlgName{Fennel} and \AlgName{LDG}. Here, we measure the memory consumption of \AlgName{HeiStream} for various buffer sizes and compare it to \AlgName{Fennel} and \AlgName{LDG}.  To do that, we measured the memory consumption of \AlgName{HeiStream}(1024k), \AlgName{HeiStream}(32k), \AlgName{Fennel} and \AlgName{LDG} on the three largest graphs (RHG1, RHG2, and uk-2007-05). On average, \AlgName{HeiStream}(1024k) consumes respectively $2.5$GB, $4.1$GB, and $9.9$GB  of memory to partition these graphs. while \AlgName{HeiStream}(32k) consumes $472$MB, $521$MB, and $3.7$GB, and \AlgName{Fennel} and \AlgName{LDG} use $399$MB, $401$MB, and $445$MB. Note that the increased amount of used memory of \AlgName{HeiStream}(1024k) is expected, since we use a fairly large buffer. Given the size of the graphs, we believe those required amounts of memory are more than reasonable.

\section{Streaming Process Mapping}
\label{sec:recmultisec_Streaming Process Mapping}
\label{chap:recmultisec_Streaming Process Mapping}

In this section, we develop \AlgName{Online Recursive Multi-Section}, a shared-memory parallel streaming algorithm for process mapping.
Our algorithm performs recursive multi-sections on the fly, which allows us to perform multi-sectioning along a hierarchical topology using only one pass over the communication graph.
To the best of our knowledge, ours is the first streaming algorithm designed specifically for the process mapping problem.
On average, our algorithm computes $41\%$ better process mappings and is 55 times faster than \AlgName{Fennel} which ignores the given hierarchy.
In scalability tests, our algorithm is only $3$ times slower than \AlgName{Hashing} when running on 32 threads.
If no hierarchy is specified as an input, our approach can also be used as a tool to solve the standard graph partitioning problem.
Our approach has a significantly lower running time complexity compared to state-of-the-art non-buffered one-pass partitioning algorithms.
Our experiments show that on average, our algorithm is $134.4$ times faster than \AlgName{Fennel} at the cost of $5\%$ more cut edges.

\subsection{Online Recursive Multi-Section}
\label{sec:recmultisec_Online Recursive Multi-Section}
\label{sec:recmultisec_OnlineRecursiveMulti-Section}

\subsubsection{Overall Scheme}
\label{subsec:recmultisec_Overall Scheme Map}

A successful offline algorithm to partition and map the nodes of a communication graph onto PEs is the recursive multi-section~\cite{schulz2017better,GlobalMultisection}.
This approach specializes the partitioning process for the case in which the communication cost between two processes (nodes) highly depends on the hierarchy level shared by the PEs (blocks) in which they are allocated.
Recall that a hierarchical topology can represented by a string $\mathcal{S}=a_1: a_2: ...:a_\ell$ which in the process mapping application means that each processor has $a_1$ cores, each node has $a_2$ processors, each rack has $a_3$ nodes, and so forth.
The offline recursive multi-section works as follows.
First, the whole graph is partitioned in $a_\ell$ blocks. 
Then, the subgraph induced by each block of an $a_i$-partition is recursively partitioned in $a_{i-1}$ sub-blocks until the whole graph is partitioned in $k=\prod_{i=1}^{\ell}{a_i}$ blocks.
This recursive approach have been developed for process mapping and exploits the fact that the communication between PEs is cheaper through lower layers of the communication hierarchy.
It creates a hierarchy of partitioning subproblems that directly reflects the hierarchical topology of the system, which yields an improved process mapping in practice~\cite{GlobalMultisection}.

Intuitively, for any given hierarchy $S$ (independent of the application), recursive multi-section can be implemented in the streaming model using $\ell$ successive passes of any one-pass partitioning algorithm over the input graph, i.e., by restreaming the graph $\ell$ times and subpartitioning the corresponding blocks in each pass.
In our algorithm, which we call \AlgName{Online Recursive Multi-Section}, we compress all steps performed during the $\ell$ passes into a single pass, as follows.
After a node is loaded, assign it to one of the $a_\ell$ blocks $\{V^\ell_1,\ldots,V^\ell_{a_\ell}\}$ in layer~$\ell$.
Then, for each layer $i<\ell$, assign the node to one of the $a_{i}$ sub-blocks of the block chosen in the previous step.
After going through all layers, the node is directly assigned to a final block, which makes this approach feasible for online execution.
Algorithm~\ref{alg:recmultisec_One-pass Process Mapping} summarizes the structure of our \AlgName{Online Recursive Multi-Section} and Figure~\ref{fig:recmultisec_hierarchy_problems} exemplifies it. Here, score depends on the algorithm logic used, e.g. \AlgName{Fennel} or \AlgName{LDG}, as well as other parameters that are specific to our multi-section algorithm. We give more details in Section \ref{subsec:recmultisec_subproblems}.
Note that it produces exactly the same result as the version with $\ell$ passes over the input since it does not violate any~data~dependency.

\begin{algorithm}[t] %
	\begin{algorithmic}[1] %
		\For{$u \in V(G)$}
			\State Load $N(u)$
			\State $X \leftarrow \{V^\ell_1,\ldots,V^\ell_{a_\ell}\}$, $V^* \leftarrow \emptyset$
			\For{$i \in \{\ell,\ldots,1\}$}
				\State \textbf{if} $i \neq \ell$ \textbf{then} $X \leftarrow $ sub-blocks of $V^*$
				\State $V^* \leftarrow$ $\argmax\limits_{W\in X}\big\{score(W)\big\}$
				\State Assign node to sub-block $V^*$
			\EndFor
		\EndFor
	\end{algorithmic}
	\caption{\AlgName{Online Recursive Multi-Section}} %
	\label{alg:recmultisec_One-pass Process Mapping} %
\end{algorithm}

\begin{figure}[t!]
	\centering
	\includegraphics[angle=0, width=0.6\linewidth]{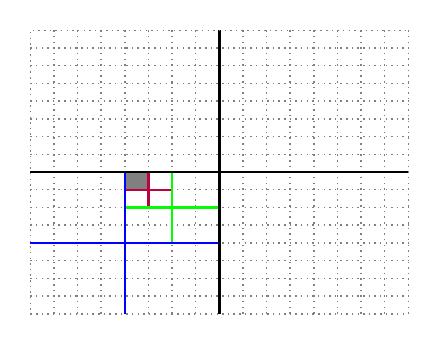}
	\caption{Assigning a node with \AlgName{Online Recursive Multi-Section} over a grid composed of ${256}$ blocks that are contained within a hierarchy $S=4:4:4:4$. In the first ${4}$-partition (black), the node is assigned to the lower left block. In the following ${4}$-partitions, it is assigned to the upper right sub-block (blue), its upper left sub-block (green), and finally its upper left sub-block (purple), which is a block from the original partitioning problem.}
	\label{fig:recmultisec_hierarchy_problems}
\end{figure}

\AlgName{Online Recursive Multi-Section} can be applied to compute any hierarchical partitioning and is not limited to the process mapping application. 
However, in case of the process mapping application, as the offline recursive multi-section our streaming algorithm exploits the inherent structure of the problem in two ways:
(i) its layout reproduces the hierarchical communication topology;
(ii) the top-down order in which the nodes are assigned to sub-blocks of previously assigned blocks reflects the order in which the communication costs decrease in a communication topology.
In other words, this top-down order reflects the need for primarily avoiding cut edges among modules of higher layers in the communication hierarchy.
We give the time and space complexity of our~algorithm.

\begin{lemma}
	\label{theo:recmultisec_num_blocks_msec}
	\AlgName{Online Recursive Multi-Section} needs $O(k)$ space to store block weights.
\end{lemma}
\begin{proof}
	By definition, the multi-section consists of $\ell$~layers such that a layer $i \in \{1,\ldots,\ell\}$ contains exactly $\prod_{r=i}^{\ell}{a_r}$~blocks whose weight we need to keep track of.
	As we define $a_r \geq 2 \ \forall r$, we can write $\prod_{r=i}^{\ell}{a_r} \leq (1/{2^{i-1}})\prod_{r=1}^{\ell}{a_r} = (1/{2^{i-1}})k$.
	Hence, the total number of block weights that we need to keep track of is $\sum_{i=1}^{\ell}{\prod_{r=i}^{\ell}{a_r}} \leq \sum_{i=1}^{\ell}{(1/{2^{i-1}})k} \leq 2k$.
\end{proof}

\begin{theorem}
	\label{theo:recmultisec_mem_usage_msec}
	\AlgName{Online Recursive Multi-Section} coupled with \AlgName{Fennel} or \AlgName{LDG} needs $O(n+k)$ memory.
\end{theorem}
\begin{proof}
	Due to the hierarchical structure of multi-section, \AlgName{Fennel} and \AlgName{LDG} may keep track of a single block assignment per node, which is enough to infer all its superblocks.
	Hence, the space complexity $O(n+k)$ directly follows from Lemma \ref{theo:recmultisec_num_blocks_msec}.
\end{proof}
\begin{theorem}
	\label{theo:recmultisec_running_time_msec}
	\AlgName{Online Recursive Multi-Section} coupled with \AlgName{Fennel} or \AlgName{LDG} has time complexity $O(m\ell + n\sum_{i=1}^{\ell}{a_i})$.
\end{theorem}
\begin{proof}
	\AlgName{Online Recursive Multi-Section} assigns each node~$u$ over $\ell$ layers.
	Using \AlgName{Fennel} or \AlgName{LDG}, the running time to assign~$u$ in a given layer~$i$ is $O(|N(u)|+a_i)$.
	Accounting for all layers and nodes, this sums up to $O(m\ell + n\sum_{i=1}^{\ell}{a_i})$.
\end{proof}
\begin{corollary}
	\label{theo:recmultisec_expected_running_time_msec}
	\AlgName{Online Recursive Multi-Section} coupled with \AlgName{Fennel} or \AlgName{LDG} has time complexity $O\big((m + n)\log{k}\big)$ if $a_i = b$, $\forall i \in \{1,\ldots,\ell\}$ for any constant $b \geq 2$.
\end{corollary}
\begin{proof}
	Based on the assumption, we derive the claimed bound from Theorem~\ref{theo:recmultisec_running_time_msec} by proving $\ell=O(\log{k})$ and $\sum_{i=1}^{\ell}{a_i} = O(\log{k})$.
	The first part trivially holds since $k=\prod_{i=1}^{\ell}{a_i}=b^\ell \Rightarrow \ell = \log_{b}{k}$.
	To prove the second part, notice that $\sum_{i=1}^{\ell}{a_i} = b\ell$.
	Since $\ell = \log_{b}{k}$, it follows that $k=\prod_{i=1}^{\ell}{a_i} = b\log_{b}{k} = O(\log{k})$.
\end{proof}

\subsubsection{Partitioning Subproblems}
\label{subsec:recmultisec_Process Mapping}
\label{subsec:recmultisec_subproblems}

Our \AlgName{Online Recursive Multi-Section} algorithm implies a hierarchy of one-pass partitioning subproblems.
As self-contained, one-pass partitioning problems, these subproblems can be solved using any one-pass partitioning algorithm in literature.
The focus of this section is to analyze the subproblems in relation to the parameters of a hierarchical partitioning problem.

Consider all the partitioning subproblems contained in some layer $i$ of our algorithm.
The initial point of consideration is the uniformity of these subproblems, which implies that they are supplied with an induced subgraph containing a similar quantity of nodes and edges, and subsequently partition it among $a_i$~blocks.
More specifically, a subproblem in layer $i$ partitions among $k_i = a_{i}$ blocks and receives as input a graph containing roughly $n_i = n/{\prod}_{r=i+1}^{\ell}{a_r}$ nodes and $m_i = m/{\prod}_{r=i+1}^{\ell}{a_r}$ edges.
Consequently, the size constraint $L_i$ for a block from the subproblems in layer $i$ is determined by the equation $L_i = \lceil(1+\epsilon) n_i/k_i \rceil \simeq L_{\max} \prod_{r=1}^{i-1}{a_r}$, which simply denotes the total capacity of all blocks from the original problem that are encompassed within it.
The variations observed in the dimensions subproblems at different layers within our algorithm have further implications based on the choice of streaming partitioning algorithm employed to address them, as we show next.

\paragraph*{Fennel Mapping.}
Using \AlgName{Fennel} within \AlgName{Online Recursive Multi-Section} requires attention to its constant $\alpha$.
Recall that it is defined as $\alpha = \sqrt{k} m/ n^{3/2}$ for partitioning the whole graph into $k$ blocks with vanilla \AlgName{Fennel}.
Using this value of $\alpha$ for all partitioning subproblems contained in our algorithm is not a natural choice since we intend to apply \AlgName{Fennel} independently for each subproblem.
Independently applying the definition of \AlgName{Fennel} for each partitioning subproblem contained in our algorithm implies~$\ell$ different parameters $\alpha_i$, $i \in \{1, \ldots, \ell\}$, for all multi-section layers.
We derive the value of~$\alpha_i$ by applying the \AlgName{Fennel} definition $\alpha_i = \sqrt{k_i} \frac{m_i}{ n_i^{3/2}}$ and substituting the values of $k_i$, $m_i$, and $n_i$ which we have already discussed.
It follows  that 
$\alpha_i = \frac{\alpha }{ \sqrt{\prod_{r=1}^{i-1}{a_r}}}$.

\paragraph*{LDG Mapping.}
Combining \AlgName{LDG} with \AlgName{Online Recursive Multi-Section} is straightforward, since it directly uses the remaining capacity of each block as a multiplicative~penalty.
Hence, we can configure \AlgName{LDG} for a subproblem in layer $i$ by simply computing this penalty based on the block capacity $L_i$, whose value we have already discussed.

\subsubsection{General Partitioning}
\label{subsec:recmultisec_General Partitioning}

In this section, we show how to partition a streamed graph into an arbitrary number of blocks using \AlgName{Online Recursive Multi-Section} when no explicit hierarchy is given. 
We do this by creating an artificial hierarchy.

The \emph{recursive bisection} is a successful offline approach to partition graphs into an arbitrary number $k$ of blocks~\cite{more_recent_advances_hgp}. %
If~$k$ is a power of $2$, the algorithm works as a recursive multi-section with $\log_2{k}$ layers of $2$-way partitioning subproblems. 
Otherwise, it is irregular and cannot be represented by a string~$S$.
In analogy to the offline algorithm, we define an online recursive bisection to partition a graph on the fly when no hierarchy is given.
Recall that the whole hierarchy of blocks and sub-blocks has to be kept in memory throughout the execution of \AlgName{Online Recursive Multi-Section}, hence the same requirement applies here.
We build this hierarchy, which we call \emph{multi-section tree}, as a preliminary step for the streaming partitioning process.
In Algorithm~\ref{alg:recmultisec_Recursively Define Subproblems}, we define the procedure $\textsc{BuildHierarchy}$ which recursively builds this multi-section tree for any value~of~$k$.
This procedure receives as input a parent block $P$ for the multi-section tree as well as the endpoints $k_L$ and $k_R$ of the range of blocks to be covered by the multi-section tree.
In line~2, it terminates the recursion when $P$ is a leaf of the multi-section tree, which is true when $k_L=k_R$.
Otherwise, it creates two sub-blocks for $P$ and inserts them as sons of $P$ in the multi-section tree (line~3). 
Then, it splits the range $\{k_L,\ldots,k_R\}$ in roughly equal parts and performs 2 recursive calls to itself. 

\begin{algorithm}[b!] %
	\raggedright
	\textbf{Input} 
	$P$: Parent block in hierarchy, 
	$k_L$: Left endpoint of blocks covered by hierarchy.
	$k_R$: Right endpoint of blocks covered by hierarchy \\
	\begin{algorithmic}[1] %
		\Procedure{BuildHierarchy}{$P, k_L, k_R$}
		
		\If{$k_L = k_R$}
		\Return \EndIf
		\State $P_L,P_R \leftarrow$ Create sub-blocks for $P$
		\State $\textsc{BuildHierarchy}\Big(P_L, k_L, \big\lfloor \frac{k_L+k_R}{2} \big\rfloor\Big)$
		\State $\textsc{BuildHierarchy}\Big(P_R, \big\lfloor \frac{k_L+k_R}{2} \big\rfloor+1, k_R\Big)$
		\EndProcedure
	\end{algorithmic}
	\caption{Create Blocks for Multi-Section Tree} 
	\label{alg:recmultisec_Recursively Define Subproblems} %
\end{algorithm}
 
We further generalize the recursive bisection to \emph{recursive $b$-section} for a \emph{base}~$b$.
Given a base $b \geq 2$, a recursive $b$-section is a recursive multi-section associated with a multi-section tree in which blocks have up to $b$ sub-blocks.
Algorithm~\ref{alg:recmultisec_Recursively Define Subproblems} can be adapted to deal with~$b$-section by creating $\min\{b,k_R-k_L+1\}$ sub-blocks in line~3 and, afterwards, making the same number of recursive calls with proper parameters.
We create the multi-section tree by calling the command $\textsc{BuildHierarchy}(\emptyset,1,k)$ at the cost of $O(k)$.
Given a multi-section tree, we solve it by using Algorithm~\ref{alg:recmultisec_One-pass Process Mapping}.
Analogously to Theorem~\ref{theo:recmultisec_mem_usage_msec}, it is possible to prove that the online recursive $b$-section respectively stores $O(k)$ blocks and needs $O(n+k)$ memory when coupled with \AlgName{Fennel} or \AlgName{LDG}.
Theorem~\ref{theo:recmultisec_expected_running_time_bisec} provides a running time bound.

\begin{theorem}
	\label{theo:recmultisec_expected_running_time_bisec}
	Online recursive $b$-section coupled with \AlgName{Fennel} or \AlgName{LDG} has time complexity $O\big( (m+nb) \log_b{k}\big)$.
\end{theorem}

\begin{proof}
	The number of layers in the multi-section tree is up to $\lceil\log_b{k}\rceil$.
	In other words, each node should be assigned through up to $1+\log_b{k}$ layers.
	Since all subproblems partition among up to $b$~blocks, then the running time to assign a node $u$ over a layer is $|N(u)|+b$.
	Accounting for all layers and nodes, this sums up to $(2m + nb)(\log_b{k}+1) = O\big( (m+nb) \log_b{k}\big)$.
\end{proof}

\paragraph*{Heterogeneous Partitioning.}
\label{subsec:recmultisec_Heterogeneous Partitioning}

When $k$ is not a power of $b$, the recursive $b$-section hierarchy may contain some partitioning subproblems with heterogeneous blocks.
We deal with this by computing the size constraint of each block in the multi-section tree individually.
For simplicity, we explain how to do this when $b=2$ (recursive bisection), but this can be easily extended to an arbitrary~$b$.
For example when $k=5$, the two blocks in the first 2-way partitioning subproblem respectively cover $2$ and $3$ of the blocks from the original $5$-way partitioning.
Hence, these two blocks shall respectively have capacities $2L_{\max}$ and $3L_{\max}$, where $L_{\max}$ is the size constraint of a block in the original $k$-way partitioning.
Putting it in general terms, each block from the multi-section tree created in line~$3$ of Algorithm~\ref{alg:recmultisec_Recursively Define Subproblems} covers $t = k_R - k_L + 1$ blocks of the original $k$-way partitioning.
For simplicity, we use $t$ to refer to this number covered by a given block, and we use $t_1$ and $t_2$ to refer to the numbers covered by the two blocks of a partitioning subproblem.
The size constraint of a block is~$t\times L_{\max}$.

When a subproblem has blocks with heterogeneous size constraints, the used partitioning algorithm has to cope with it.
We adapt \emph{\AlgName{Fennel}} to address this issue by increasing (decreasing) the constant $\alpha$ used to compute the score of a specific block when its size constraint is lower (higher) than the other block from the same subproblem.
Recall that $\alpha$ depends on the numbers of nodes, edges and blocks for a specific subproblem.
A subproblem receives as input an induced subgraph with roughly $\frac{t_1+t_2}{k}$ of the nodes and edges from the original $k$-way partitioning.
We redefine the number of blocks as $\frac{t_1+t_2}{t_1}$ for the first block and $\frac{t_1+t_2}{t_2}$ for the second block of a subproblem.
This value equals $2$ for both blocks when $t_1=t_2$. 
Nevertheless, if $t_1 \neq t_2$, this value is larger (smaller) than $2$ for the block with smaller (larger) size constraint.
Summing up, the value of $\alpha$ for a given block will be $\sqrt{t}$ times smaller than the value $\alpha$ from the original $k$-way partitioning problem.
Consequently, the imbalance penalty function for \AlgName{Fennel} will be more pronounced for blocks with lower capacities, thereby addressing the issue of heterogeneous balancing..
For \emph{\AlgName{LDG}}, a natural adaptation for heterogeneous blocks arises from its very definition, since it directly uses the remaining capacity of each block as a multiplicative~penalty.

\ifFull
\paragraph*{Hybrid  Meta-Partitioning.}
It is also possible to solve distinct subproblems with different partitioning algorithms.
	This possibility opens a door to a trade-off when we mix a high-quality algorithm such as \AlgName{Fennel} with a speedy algorithm such as \AlgName{Hashing}.
	In particular, we can use \AlgName{Fennel} to solve top-layer subproblems (whose communication is more expensive) and \AlgName{Hashing} to solve bottom-layer subproblems (whose communication is cheaper).
	If we do this by solving the whole $h$ ($1 \leq h <\ell$) bottom layers of subproblems with \AlgName{Hashing}, we come to an overall time complexity $O(m + nh + n\sum_{i=\ell}^{h+1}{a_i})$.
	Observe that this hybridization is faster than coupling our scheme with \AlgName{Fennel} only and slower than coupling it with \AlgName{Hashing} only.
\fi{}

\subsubsection{Shared-Memory Parallelization}
\label{subsec:recmultisec_Shared-Memory Parallelization}

Since \AlgName{Online Recursive Multi-Section} is a node-centric algorithm, we can parallelize it by independently splitting the nodes of the graph among threads.
More specifically, it can be achieved with OpenMP by parallelizing the \emph{for} loop in line 1 of Algorithm~\ref{alg:recmultisec_One-pass Process Mapping}.
This parallelization requires the nodes from the input graph to be concurrently loaded by distinct threads alongside with their neighborhoods, which is a reasonable assumption in many practical environments. 
Regarding data consistency, the only source of concern are the block weights, whose values can be concurrently read and incremented by multiple threads.
This is important because an inconsistency could compromise the load balance between blocks.
We ensure writing consistency by making the incrementation an atomic operation.
Potentially, a block can still be overloaded if multiple threads decide to assign a node to it at the same time.
Since this is very unlikely, we do not use any synchronization to keep it from~happening.

\subsection{Experimental Evaluation}
\label{sec:recmultisec_Experimental Evaluation}

\paragraph*{Setup.} 
We performed our implementations using C++ and compiled them using gcc 9.3 with full optimization turned on (-O3 flag). 
The majority of our experiments were executed on a single core of Machine~C, with the exception of the scalability experiments, which employed up to 32 threads using hyperthreading.
Unless otherwise mentioned we stream the input directly from the internal memory to obtain clear running time comparisons.
However, note that the streaming algorithma could also be run streaming the graph from hard~disk.

\paragraph*{Baselines.} 
We identify \AlgName{Fennel} and \AlgName{LDG} as the state-of-the-art of non-buffered one-pass stream partitioning algorithms which aim at minimizing edge-cut.
Since \AlgName{Fennel} generates better solutions on average than \AlgName{LDG}~\cite{tsourakakis2014fennel}, we focus our experiments on \AlgName{Fennel} without loss of generality.
We also include \AlgName{Hashing} as a competitor, since it is the fastest algorithm for streaming partitioning.
To the best of our knowledge, there are no streaming partitioning algorithms specifically designed for process mapping.
Since no official versions of \AlgName{Fennel}, \AlgName{LDG}, and \AlgName{Hashing} are available in public repositories, we implemented them in our framework.
For comparison purposes, we ran experiments with internal-memory tools, i.e. we compare against the fastest version of the integrated multi-level algorithm proposed in~Chapter~\ref{chap:intmap_Multilevel Process Mapping} of this dissertation, which we refer to as~\emph{\AlgName{IntMap}}.
In addition, we compare our results against \AlgName{KaMinPar}~\cite{gottesburen2021deep}, a very fast internal-memory algorithm that is orders of magnitudes faster than \AlgName{mt-Metis} in terms of running time and produces comparable cuts while also enforcing balance (in contrast to \AlgName{mt-Metis}). In particular, the purpose of running \AlgName{IntMap} and \AlgName{KaMinPar} is to provide a reference of streaming algorithms in comparison to internal memory algorithms. 
We set a timeout of $30$ minutes for an algorithm to partition a graph.
Only \AlgName{IntMap} exceeded this time limit for some instances.
Hence, we exclude this algorithm from the plots.
Our implementations of these algorithms reproduce the results presented in the respective papers and are optimized for running time as much as possible.

\paragraph*{Instances.}
In this section, we use two disjoint sets of graphs:
the \emph{tuning} set is used for the parameter study experiments, while the \emph{test} set is used for the comparisons against the state-of-the-art. 
Basic properties of the graphs under consideration can be found in Table~\ref{tab:heistream_graphs}.
In any case, when streaming the graphs we use the natural given order of the nodes.
Unless otherwise mentioned, we use the following configurations for process mapping experiments: $D=1:10:100$, $\mathcal{S}=4:16:r$, with $r \in \{1,2,3,\ldots,128\}$. Hence, $k=64 r$.
This is the same configuration used in other studies~\cite{schulz2017better,GlobalMultisection}.
Analogously, we use $k = 64s$, $s \in \{1,2,3,\ldots,128\}$ for graph partitioning experiments unless mentioned otherwise.
We allow a fixed imbalance of $3\%$ for all experiments (and all algorithms) since this is a frequently used value in the partitioning literature. 
All partitions computed by all algorithms were~balanced.

\paragraph*{Methodology.} 
We perform two types of experiments: experiments for the process mapping objective (with given hierarchies as specified below) and standard graph partitioning in which we create implicit hierarchies as described above.
Depending on the focus of the experiment, we measure running time, edge-cut, and/or the process mapping communication cost defined in Equation~(\ref{eq:process_mapping_obj}).

{}

{}

{}

\begin{figure}[p]
	\captionsetup[subfigure]{justification=centering}
	\centering
	\begin{subfigure}[]{\scaleFactorSmall\textwidth}
		\centering
		\includegraphics[angle=-0, width=\imgScaleFactorSmall\textwidth]{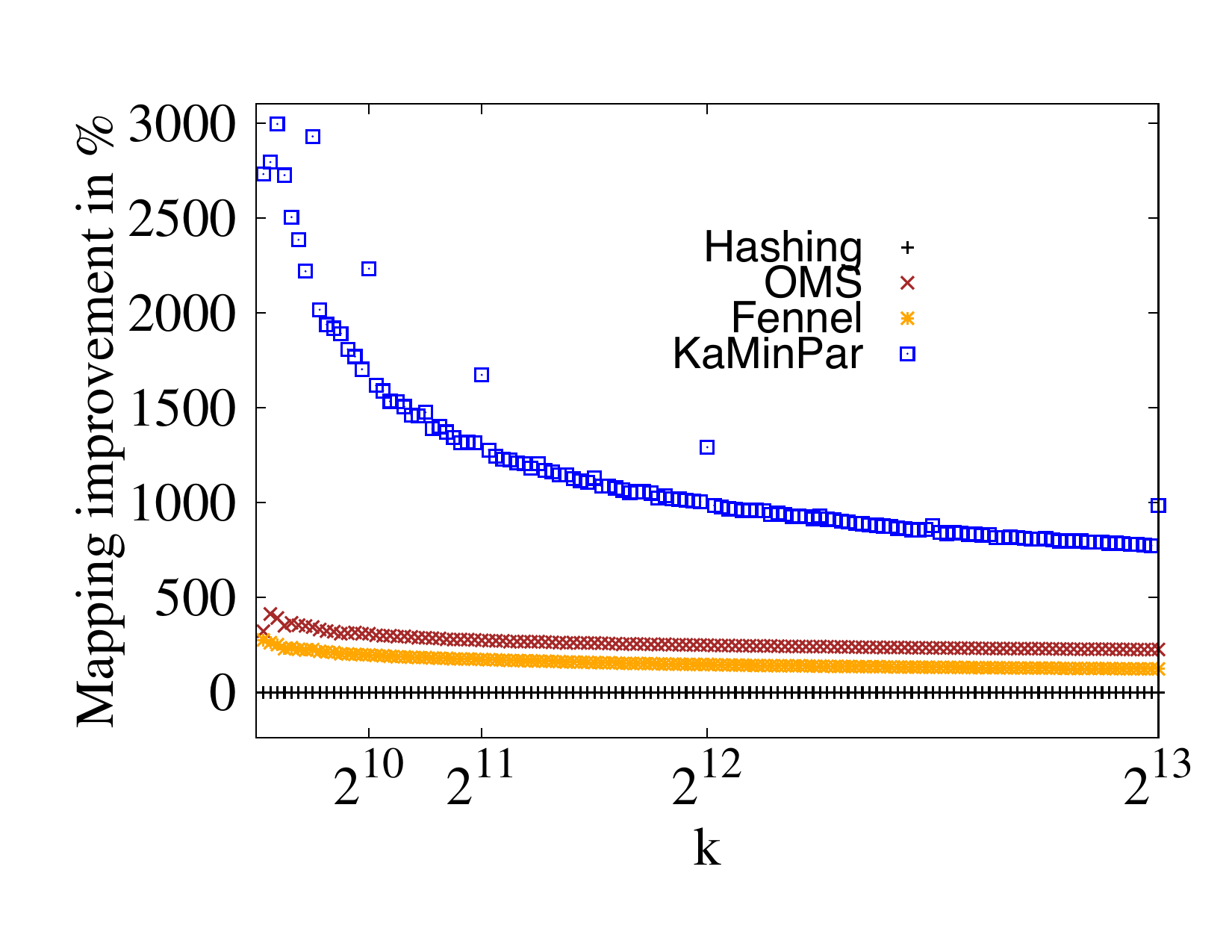}
		\vspace*{\capPositionSmall}
		\caption{Mapping improvement over \AlgName{Hashing}.}
		\label{fig:recmultisec_stateoftheartMap_res}
	\end{subfigure}\hspace{2mm}%
	\begin{subfigure}[]{\scaleFactorSmall\textwidth}
		\centering
		\includegraphics[angle=-0, width=\imgScaleFactorSmall\textwidth]{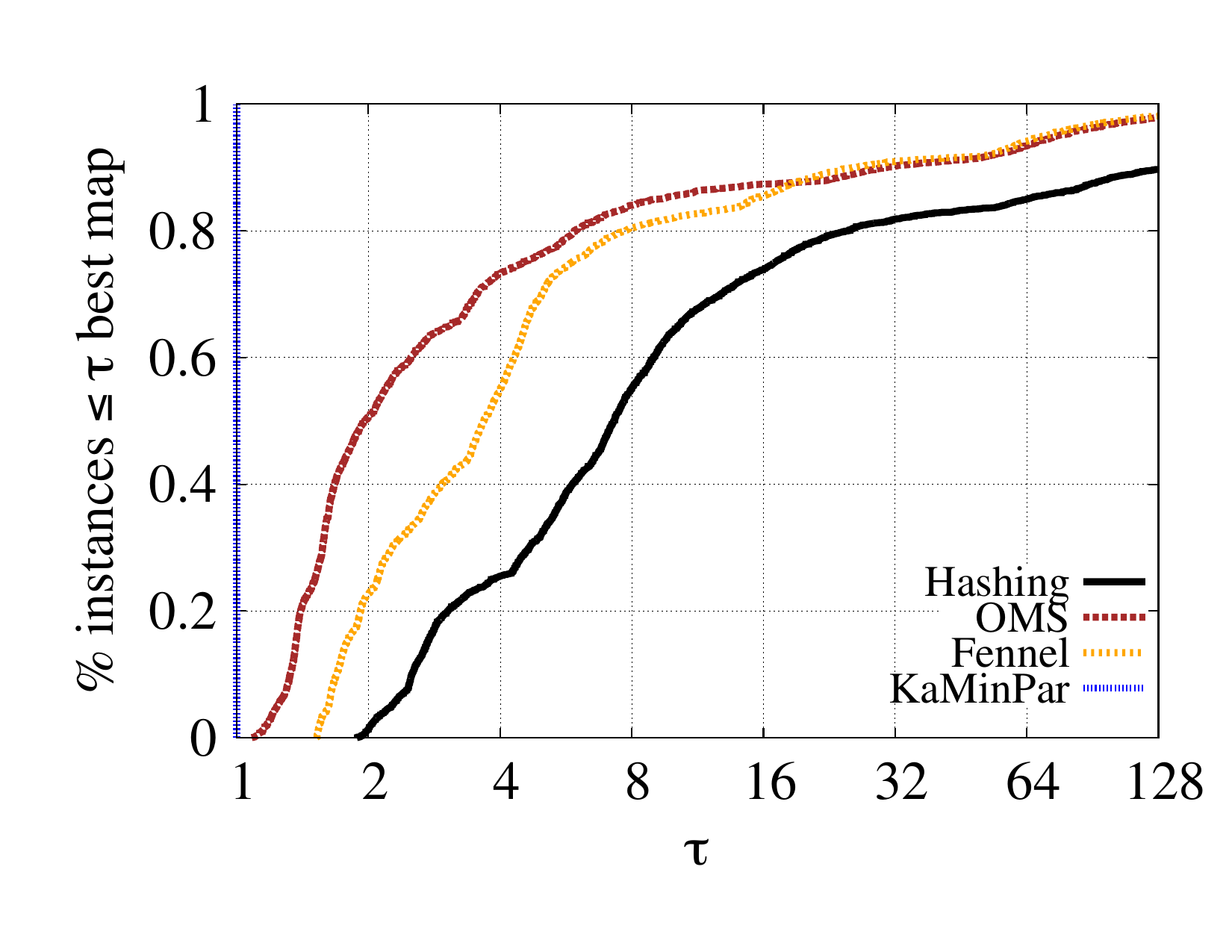}
		\vspace*{\capPositionSmall}
		\caption{Mapping performance profile.}
		\label{fig:recmultisec_stateoftheartMap_respp}
	\end{subfigure}
	\vspace*{\afterCapSmall}
	\begin{subfigure}[]{\scaleFactorSmall\textwidth}
		\centering
		\includegraphics[angle=-0, width=\imgScaleFactorSmall\textwidth]{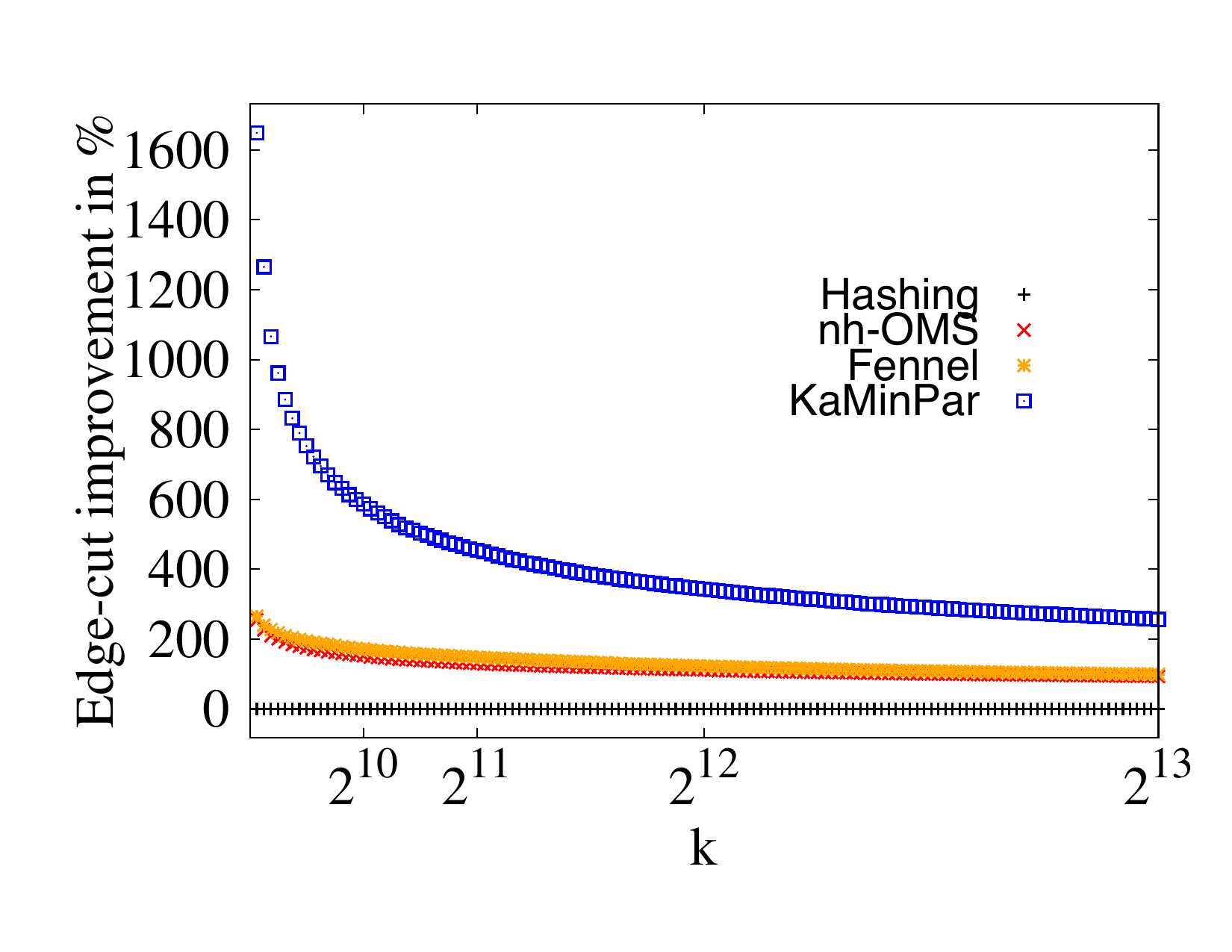}
		\vspace*{\capPositionSmall}
		\caption{Edge-cut improvement over \AlgName{Hashing}.}
		\label{fig:recmultisec_stateoftheartPar_res}
	\end{subfigure}\hspace{2mm}%
	\begin{subfigure}[]{\scaleFactorSmall\textwidth}
		\centering
		\includegraphics[angle=-0, width=\imgScaleFactorSmall\textwidth]{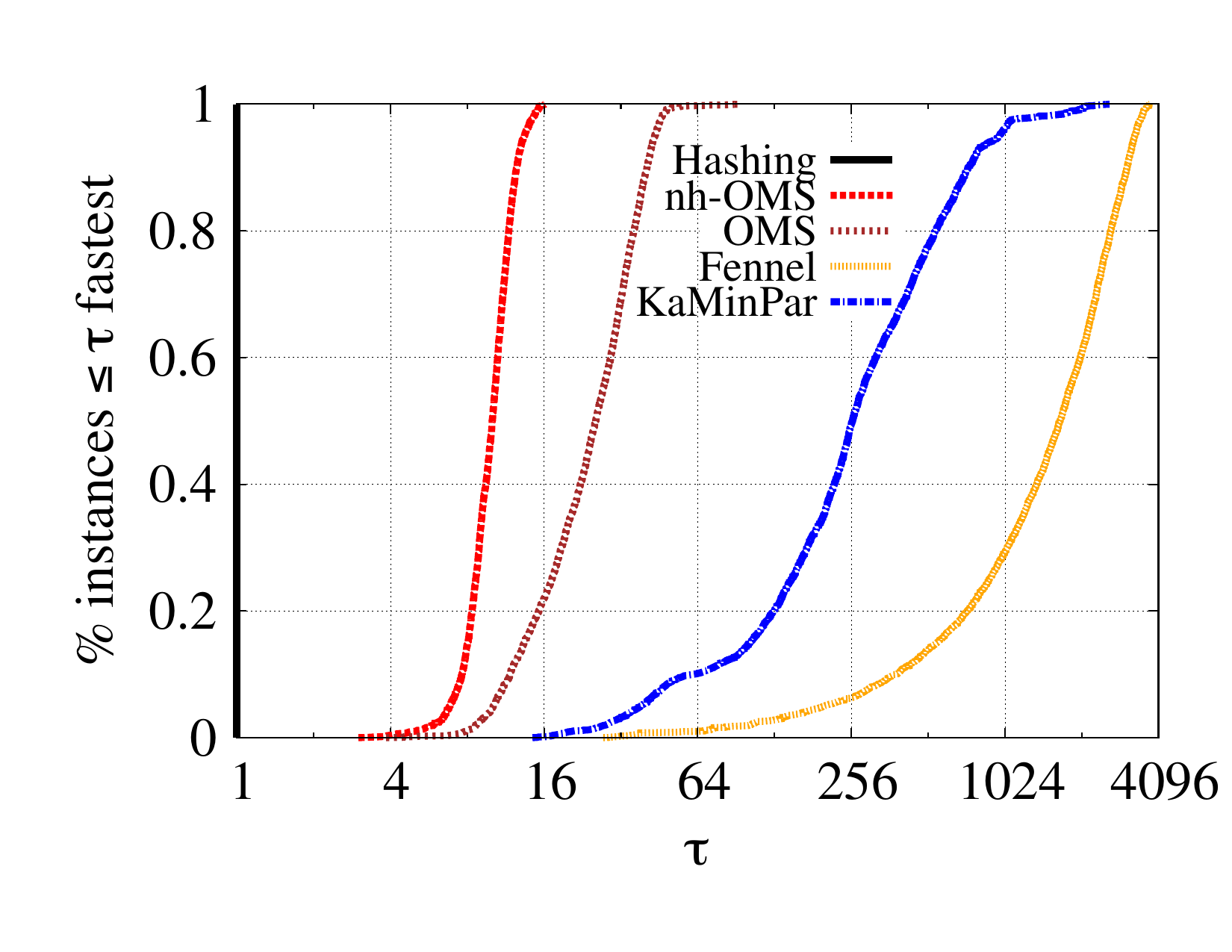}
		\vspace*{\capPositionSmall}
		\caption{Running time performance profile.}
		\label{fig:recmultisec_stateoftheartPar_timpp}
	\end{subfigure}\hspace{2mm}%
	\begin{subfigure}[]{\scaleFactorSmall\textwidth}
		\centering
		\includegraphics[angle=-0, width=\imgScaleFactorSmall\textwidth]{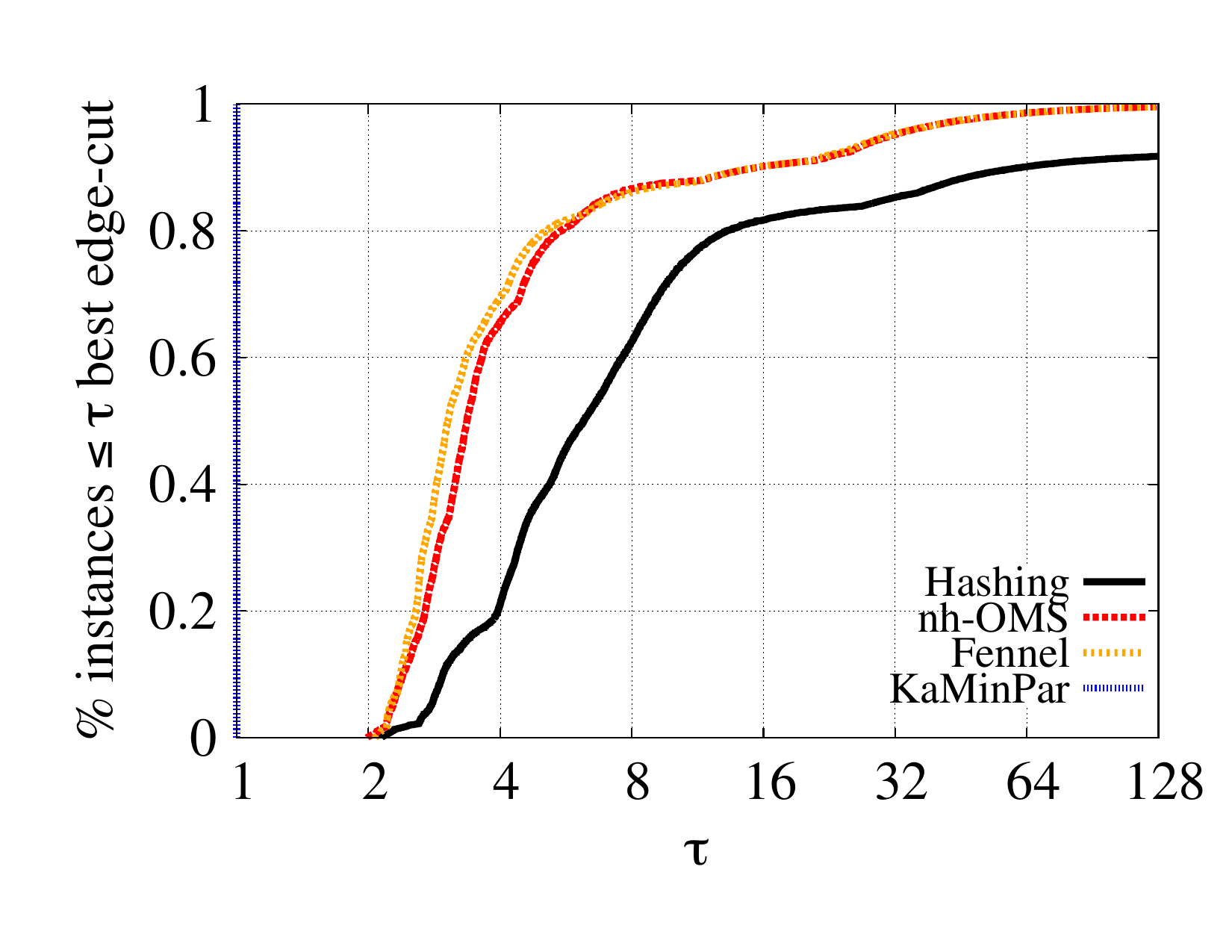}
		\vspace*{\capPositionSmall}
		\caption{Edge-cut performance profile.}
		\label{fig:recmultisec_stateoftheartPar_GPpp}
	\end{subfigure}\hspace{2mm}%
	\begin{subfigure}[]{\scaleFactorSmall\textwidth}
		\centering
		\includegraphics[angle=-0, width=\imgScaleFactorSmall\textwidth]{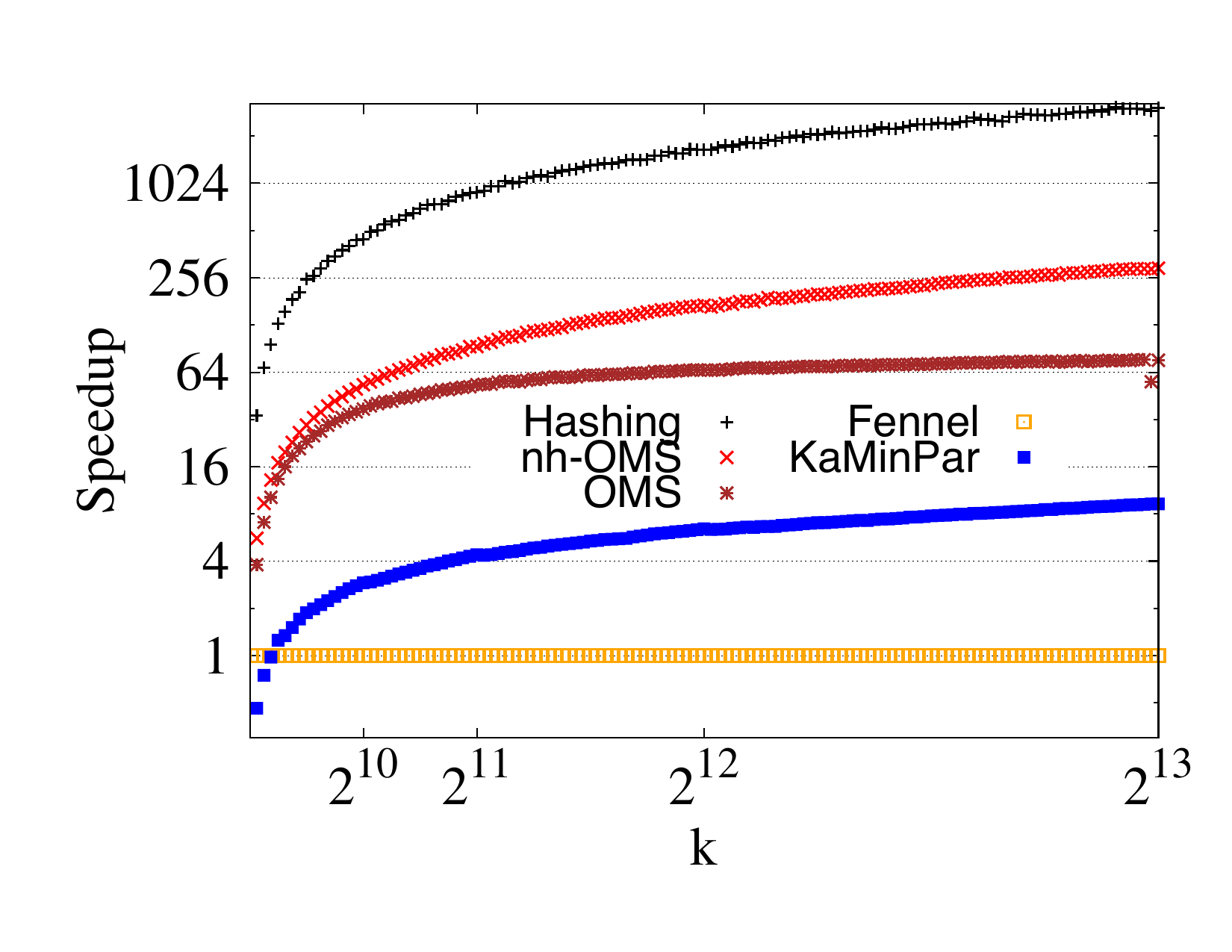}
		\vspace*{\capPositionSmall}
		\caption{Total speedup over \AlgName{Fennel}.}
		\label{fig:recmultisec_stateoftheartPar_tim}
	\end{subfigure}

	\vspace*{.55cm}
	\caption{Comparison against the state-of-the-art. Higher is better.}
	\label{fig:recmultisec_state-of_the_art}
	
\end{figure}

\begin{figure}[p]
	\captionsetup[subfigure]{justification=centering}
	\centering
		\begin{subfigure}[t]{\scaleFactorSmall\textwidth}
			\centering
			\includegraphics[angle=-0, width=\imgScaleFactorSmall\textwidth]{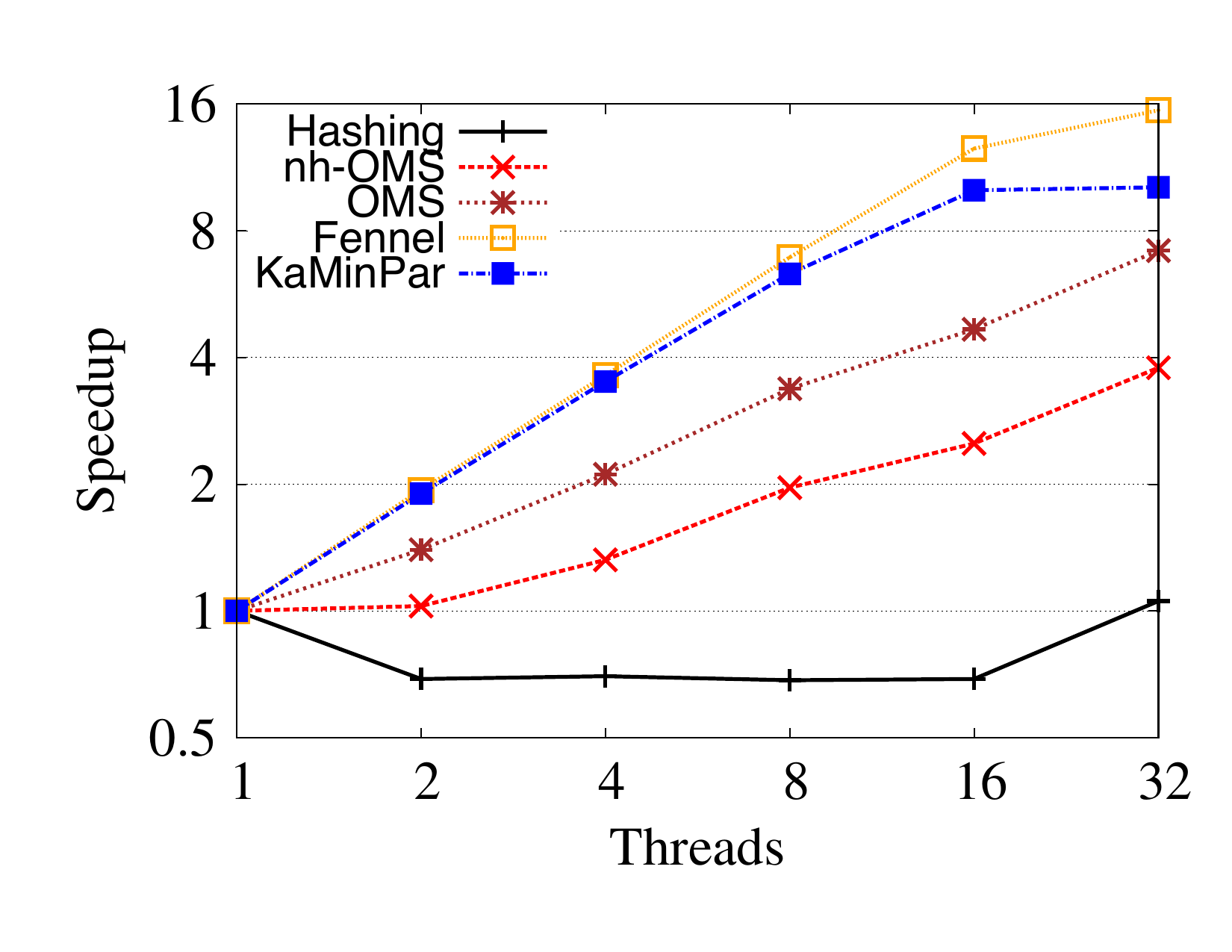}
			\vspace*{\capPositionSmall}
			\caption{Speedup versus number of used threads for graph soc-orkut-dir.}
			\label{fig:recmultisec_soc-orkut-dir_speedup}
		\end{subfigure}\hspace{2mm}%
		\begin{subfigure}[t]{\scaleFactorSmall\textwidth}
			\centering
			\includegraphics[angle=-0, width=\imgScaleFactorSmall\textwidth]{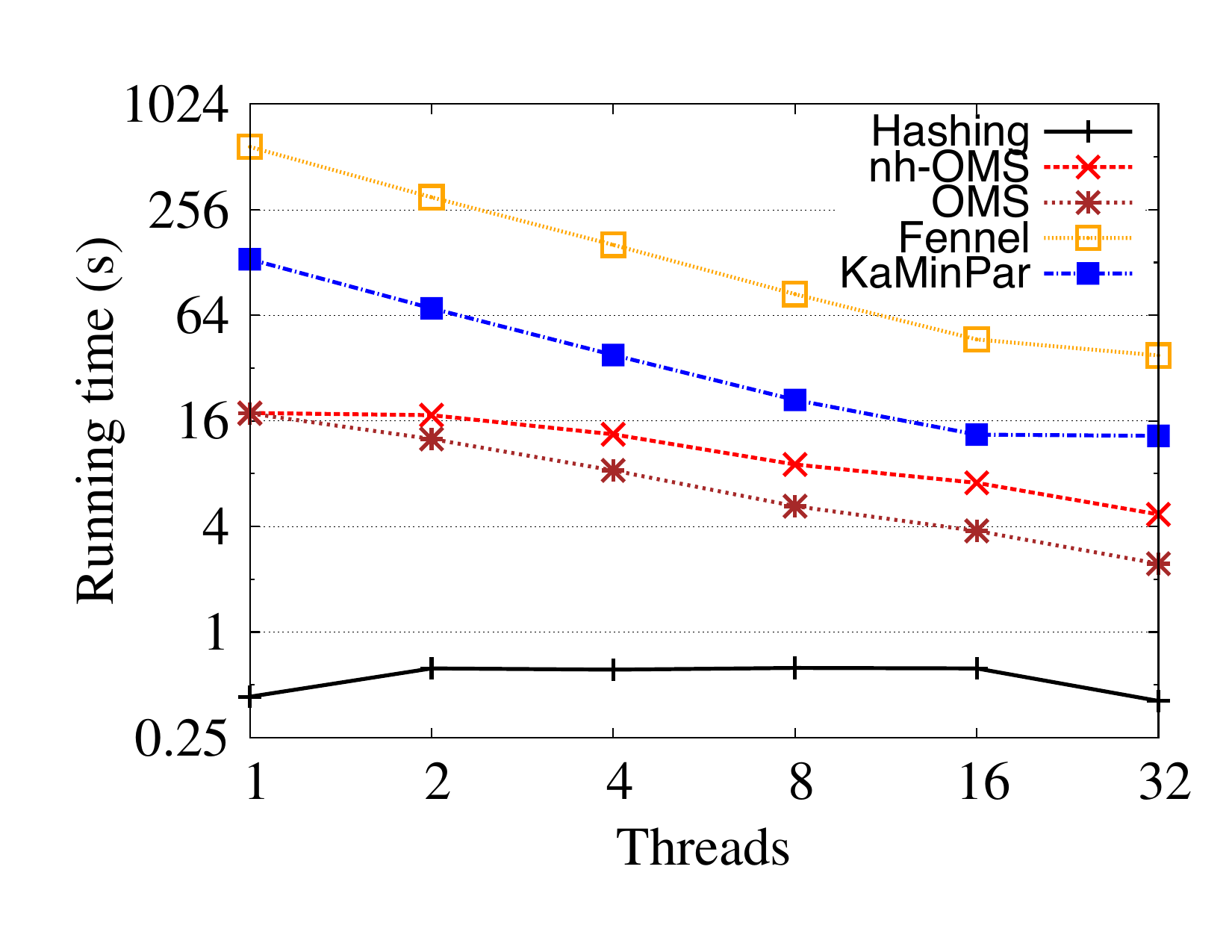}
			\vspace*{\capPositionSmall}
			\caption{Running time versus number of used threads for graph soc-orkut-dir.}
			\label{fig:recmultisec_soc-orkut-dir_times}
		\end{subfigure}\hspace{2mm}%
		\begin{subfigure}[t]{\scaleFactorSmall\textwidth}
			\centering
			\includegraphics[angle=-0, width=\imgScaleFactorSmall\textwidth]{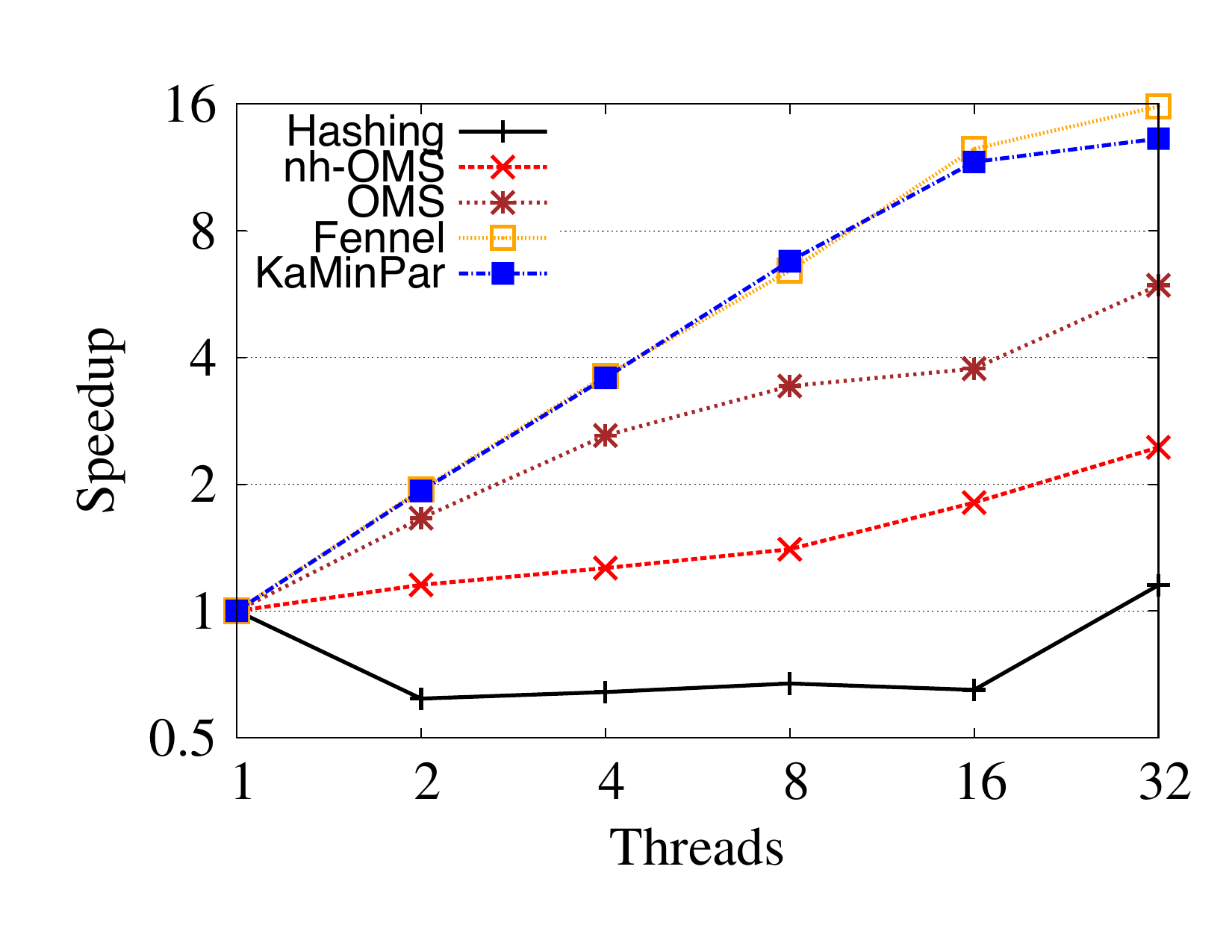}
			\vspace*{\capPositionSmall}
			\caption{Speedup versus number of used threads for graph HV15R.}
			\label{fig:recmultisec_HV15R_speedup}
		\end{subfigure}
		\vspace*{\afterCapSmall}
		\begin{subfigure}[t]{\scaleFactorSmall\textwidth}
			\centering
			\includegraphics[angle=-0, width=\imgScaleFactorSmall\textwidth]{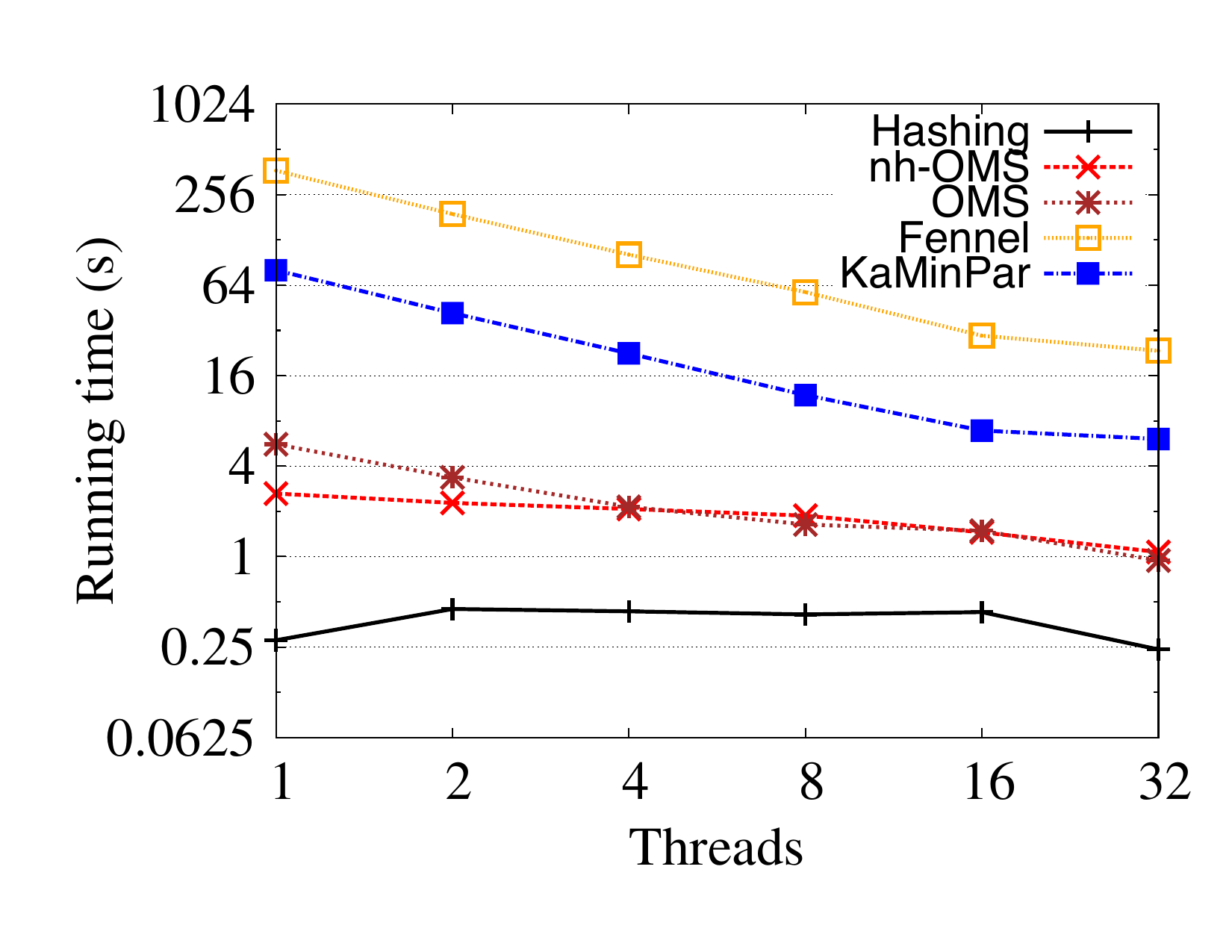}
			\vspace*{\capPositionSmall}
			\caption{Running time versus number of used threads for graph HV15R.}
			\label{fig:recmultisec_HV15R_times}
		\end{subfigure}\hspace{2mm}%
		\begin{subfigure}[t]{\scaleFactorSmall\textwidth}
			\centering
			\includegraphics[angle=-0, width=\imgScaleFactorSmall\textwidth]{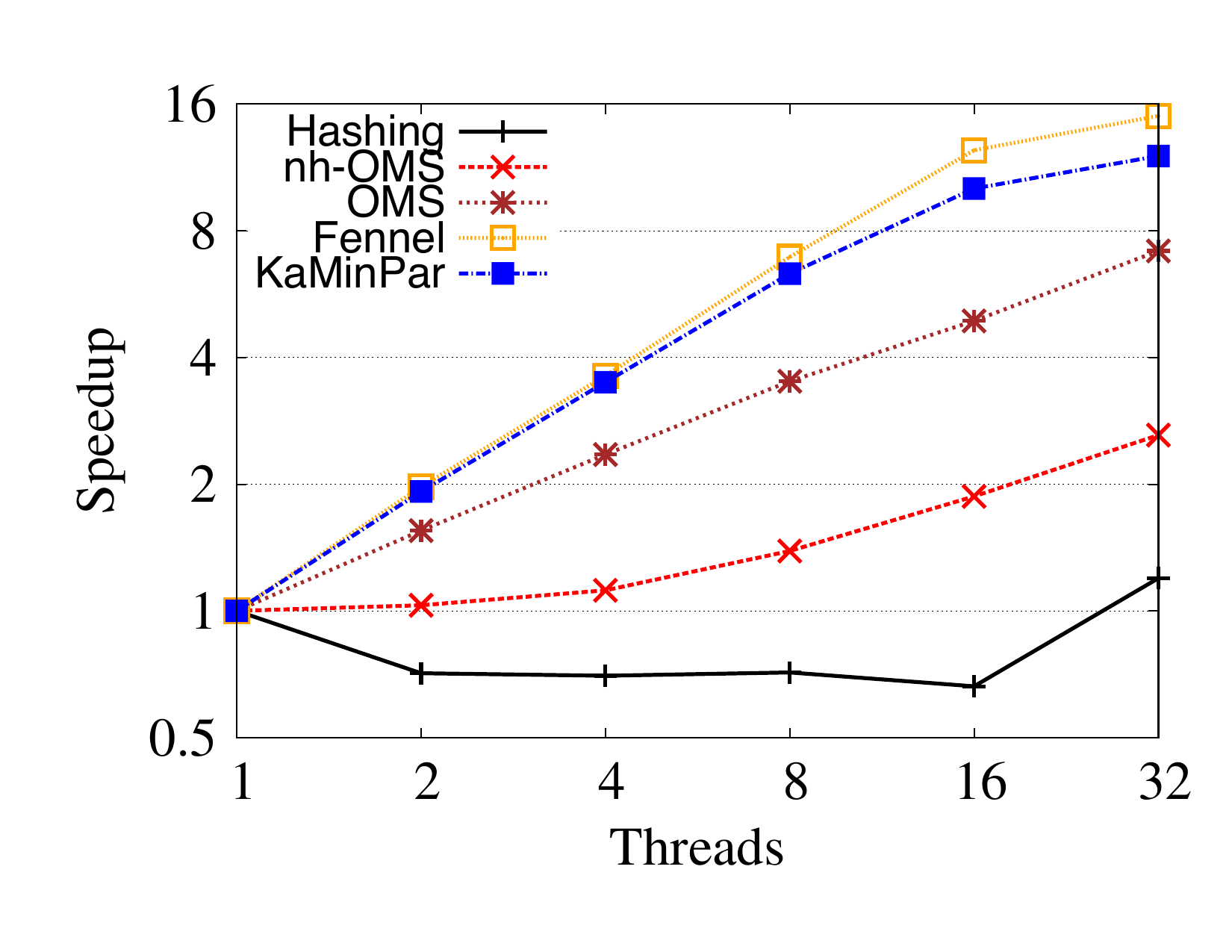}
			\vspace*{\capPositionSmall}
			\caption{Speedup versus number of used threads for graph soc-LiveJournal1.}
			\label{fig:recmultisec_soc-LiveJournal1_speedup}
		\end{subfigure}\hspace{2mm}%
		\begin{subfigure}[t]{\scaleFactorSmall\textwidth}
			\centering
			\includegraphics[angle=-0, width=\imgScaleFactorSmall\textwidth]{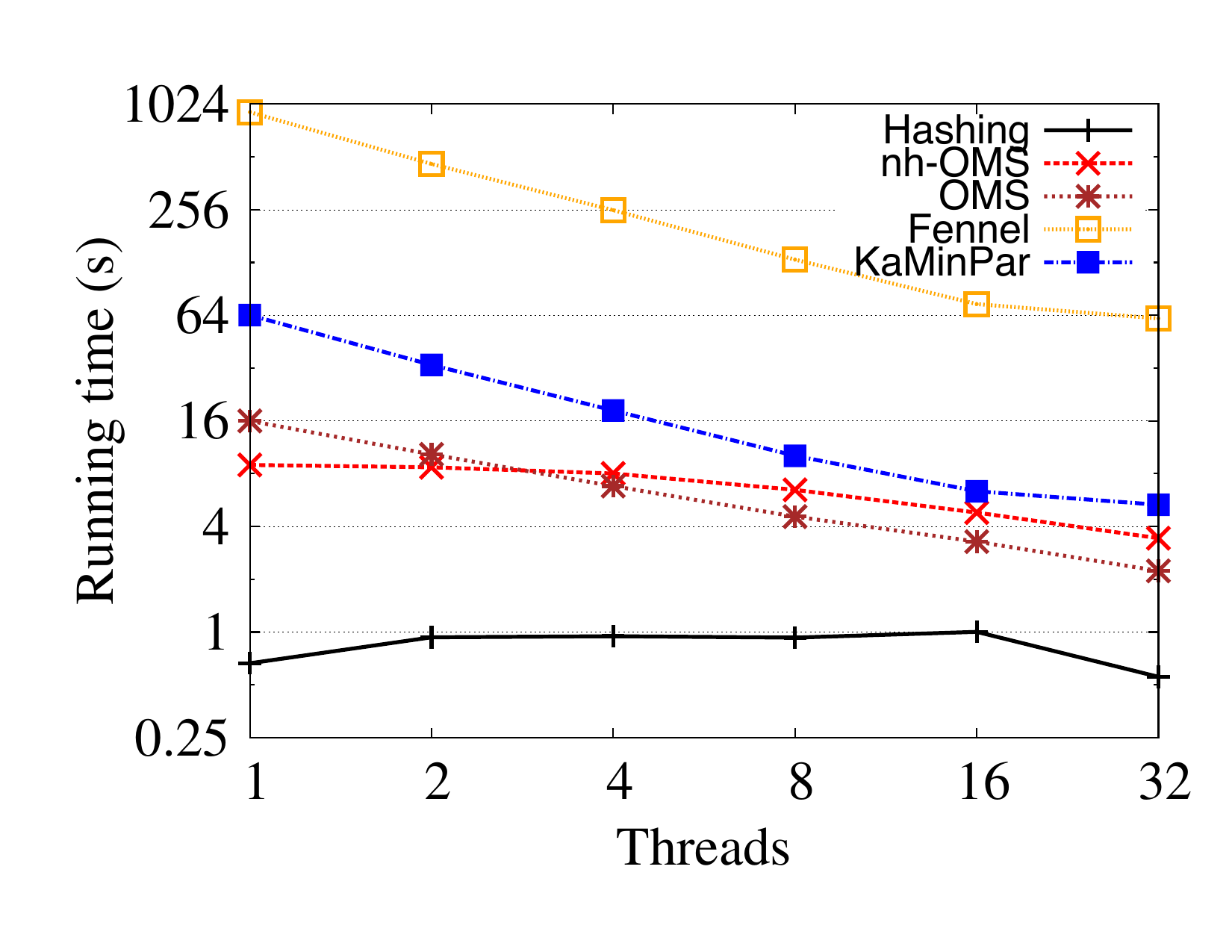}
			\vspace*{\capPositionSmall}
			\caption{Running time versus number of used threads for graph soc-LiveJournal1.}
			\label{fig:recmultisec_soc-LiveJournal1_times}
		\end{subfigure}

	\vspace*{.45cm}
	\caption{Speedup and time comparison for $k=8192$. Higher is better for speedup. Lower is better for time.}
	\label{fig:recmultisec_state-of_the_art_parallel}
	
	\vspace*{-.5cm}
\end{figure}

\subsubsection{Parameter Study}
\label{subsec:Parameter Study}
We performed extensive tuning experiments using the graphs from the tuning set in Table~\ref{tab:heistream_graphs}. 
We briefly summarize the main results here. %
\AlgName{Online Recursive Multi-Section} produces on average $3.89\%$ better mapping and $0.19\%$ better edge-cut when coupled with \AlgName{Fennel} than when coupled with \AlgName{LDG}.
Hence, we use \AlgName{Fennel} as our scoring function.
Computing adapted values of $\alpha$ for each partitioning subproblem is superior than using the default value of $\alpha$ of the original $k$-way partitioning. 
Particularly, it is on average $3.1\%$ faster while producing $9.7\%$ better mapping and cutting roughly the same number of edges. 
Hence, our algorithm uses \emph{adapted} $\alpha$ values.
When no communication hierarchy is given, using the \emph{base} $b=4$ to build the multi-section tree is the fastest configuration overall.  
Using $b=4$, our algorithm is $16.7\%$ faster and cuts $3.2\%$ fewer edges than using $b=2$.
Hence, our algorithm uses \emph{base} $b=4$.
From now on, we refer to our \AlgName{Online Recursive Multi-Section} algorithm as \emph{\AlgName{OMS}} when a communication hierarchy is given and \emph{\AlgName{nh-OMS}} otherwise.

\subsubsection{State-of-the-Art}
\label{subsec:recmultisec_State-of-the-Art}

In this section, we show experiments in which we compare \AlgName{Online Recursive Multi-Section} against the current state-of-the-art.

\paragraph*{Solution Quality (Process Mapping).}
We start by looking at the mapping quality produced by \AlgName{OMS}. 
In Figure~\ref{fig:recmultisec_stateoftheartMap_res}, we plot the average mapping improvement over \AlgName{Hashing}. 
\AlgName{KaMinPar} produces the best mapping overall, with an average improvement of $1117\%$ over \AlgName{Hashing}.
Among the instances \AlgName{IntMap} could solve, it improves on average $7.6\%$ over \AlgName{KaMinPar}. 
\AlgName{IntMap} produces the best overall mapping in $67\%$ of the cases it could solve. 
Note that this is in line with previous works in the area of graph partitioning, i.e. streaming algorithms typically compute worse solutions than internal memory algorithms that always have access to the whole graph. 
\AlgName{OMS} has an average improvement of $257.8\%$ over \AlgName{Hashing}, while \AlgName{Fennel} improves $153\%$ on average over \AlgName{Hashing}.
In a direct comparison, \AlgName{OMS} produces on average $41\%$ better mappings than \AlgName{Fennel}.
In Figure~\ref{fig:recmultisec_stateoftheartMap_respp}, we plot the mapping performance profile.
In the plot, \AlgName{KaMinPar} produces the best overall mapping for all instances.
We conclude that \AlgName{OMS} produces the best mapping among the streaming~competitors.

\paragraph*{Solution Quality (Edge-Cut).}
Next we look at the edge-cut of \AlgName{nh-OMS}.
In Figure~\ref{fig:recmultisec_stateoftheartPar_res}, we plot the edge-cut improvement over \AlgName{Hashing}.
\AlgName{KaMinPar} produces the best overall edge-cut, with an average improvement of $3024\%$ over \AlgName{Hashing}.
\AlgName{IntMap} cuts $20\%$ more edges on average than \AlgName{KaMinPar} for the instances it solved. 
Among the streaming algorithms, \AlgName{Fennel} and \AlgName{nh-OMS} produce improvements of respectively $130.5\%$ and $118.2\%$ on average over \AlgName{Hashing}.
In a direct comparison, \AlgName{nh-OMS} cuts on average $5\%$ more edges than \AlgName{Fennel}.
In Figure~\ref{fig:recmultisec_stateoftheartPar_GPpp}, we plot the edge-cut performance profile.
\AlgName{KaMinPar} produces the smallest edge-cut for all instances.
Among the streaming algorithms, \AlgName{Fennel} is slightly better than \AlgName{nh-OMS} and both are distinctly better than \AlgName{Hashing}.

\paragraph*{Running Time.} 
We now investigate the running time of \AlgName{OMS} and \AlgName{nh-OMS}.
In Figure~\ref{fig:recmultisec_stateoftheartPar_tim}, we plot the speedup over \AlgName{Fennel}.
On average, \AlgName{Hashing} is~$1301$~times faster than \AlgName{Fennel}, while \AlgName{nh-OMS} and \AlgName{OMS} are respectively $133$~and~$55.4$ times faster than \AlgName{Fennel}.
In a direct comparison, \AlgName{Hashing} is on average $9.7$ times faster than \AlgName{nh-OMS} and $23.4$ times faster than \AlgName{OMS}.
\AlgName{KaMinPar} comes next with an average speedup of $5.3$ over \AlgName{Fennel}.
In a direct comparison, \AlgName{nh-OMS} and \AlgName{OMS} are respectively $25.1$ and $10.5$ times faster than \AlgName{KaMinPar}.
\AlgName{KaMinPar} is on average $2.5$ times faster than \AlgName{IntMap} for the instances \AlgName{IntMap} could solve. 
In Figure~\ref{fig:recmultisec_stateoftheartPar_timpp}, we plot the running time performance profile.
Note that the running time of \AlgName{nh-OMS} is at most 16 times slower than \AlgName{Hashing} for $100\%$ of the experiments, which is in accordance with Theorem~\ref{theo:recmultisec_expected_running_time_bisec}.
As the third fastest algorithm, \AlgName{OMS} is considerably faster than all the other competitors, including \AlgName{Fennel}.

\paragraph*{Memory Requirements.}
We now look at the memory requirements of the different algorithms. We measured this on three graphs of our collection. In this case, we run stream graphs directly from disk for the streaming algorithms.
Besides being the fastest algorithm, \AlgName{Hashing} needs the least memory overall.
For soc-orkut-dir, HV15R, and soc-LiveJournal1, it respectively needs $17$MB, $13$MB, and $24$MB.
\AlgName{OMS}, \AlgName{nh-OMS}, and \AlgName{Fennel} have comparable consumption, all of which using $19$MB, $14$MB, and $25$MB for the mentioned graphs, respectively.
Finally, \AlgName{KaMinPar} respectively uses $4.1$GB, $4.1$GB, and $1.8$GB, while \AlgName{IntMap} respectively uses $34$GB, $12$GB, and $10$GB.

\begin{table}[t]
	\scriptsize
	\centering
	\setlength{\tabcolsep}{4pt}
	\begin{tabular}{l@{\hskip 30pt}rr@{\hskip 30pt}rr@{\hskip 30pt}rr@{\hskip 30pt}rr@{\hskip 30pt}rr}
		\toprule	
		\multirow{2}{*}{Threads} & \multicolumn{2}{l}{\AlgName{Hashing}} & \multicolumn{2}{l}{\AlgName{nh-OMS}} & \multicolumn{2}{l}{\AlgName{OMS}} & \multicolumn{2}{l}{\AlgName{Fennel}} & \multicolumn{2}{l}{\AlgName{KaMinPar}} \\
		
		& RT            & SU          & RT           & SU          & RT          & SU        & RT           & SU          & RT            & SU           \\
                \midrule
		1                        & 0.49          & 1.0         & 4.8          & 1.0         & 10.7        & 1.0       & 673.6        & 1.0         & 36.9          & 1.0          \\
		2                        & 0.72          & 0.7         & 4.6          & 1.1         & 6.4         & 1.7       & 346.3        & 1.9         & 19.3          & 1.9          \\
		4                        & 0.70          & 0.7         & 3.6          & 1.3         & 3.9         & 2.7       & 184.6        & 3.6         & 10.5          & 3.5          \\
		8                        & 0.72          & 0.7         & 3.0          & 1.6         & 2.6         & 4.1       & 96.3         & 7.0         & 5.8           & 6.4          \\
		16                       & 0.75          & 0.7         & 2.5          & 1.9         & 2.3         & 4.7       & 54.0           & 12.5        & 3.5           & 10.5         \\
		32                       & 0.46          & 1.1         & 1.7          & 2.8         & 1.3         & 8.2       & 44.2         & 15.2        & 3.1           & 11.9 \\
                \bottomrule
		
	\end{tabular}
	\caption{Average time in seconds~(RT) and average speedup~(SU) for $k=8192$.}
	\label{tab:recmultisec_scalability}
	
\end{table}

\subsubsection{Scalability}
\label{subsec:recmultisec_Scalability}

Now we evaluate the scalability of \AlgName{Online Recursive Multi-Section}.
As in Section~\ref{subsec:recmultisec_State-of-the-Art}, we refer to our algorithm as~\emph{\AlgName{OMS}} when a process mapping communication hierarchy is given and~\emph{\AlgName{nh-OMS}} otherwise.
As competitors, we include \AlgName{KaMinPar}, \AlgName{Fennel}, and \AlgName{Hashing}.
For a fair comparison against \AlgName{Fennel} and \AlgName{Hashing}, we implemented them with the same parallelization scheme of our algorithm, \ie a node-centric parallelization.
We do not include \AlgName{IntMap} in these experiments since it cannot run in parallel.
For these experiments, we selected the 12 graphs from our collection which have at least two million nodes and partitioned them into $k=8192$ blocks using all algorithms.

In Figure~\ref{fig:recmultisec_state-of_the_art_parallel}, we plot speedup and running time versus number of threads for the graphs soc-orkut-dir, HV15R, and soc-LiveJournal1.
In Table~\ref{tab:recmultisec_scalability}, we plot the average running time in seconds and speedup over all graphs as a function of the number of threads.
For all graphs, \AlgName{Hashing} presents the worst scalability, with speedups smaller than 1.
Although \AlgName{Hashing} is theoretically an embarrassingly parallel algorithm, it has two limitations:
(i) it is extremely fast, hence the overhead of the parallelization is too large in comparison to the overall running time and
(ii) it neither reuses data nor accesses sequential positions in memory, so there are almost no cache hits.
On the other hand, \AlgName{Fennel} presents the best scalability.
Differently than \AlgName{Hashing}, it is rather slow but reuses data, \eg the assignments of previous nodes to blocks, and goes through the array of blocks in order to compute their score. 
Following \AlgName{Fennel}, \AlgName{KaMinPar} has the second best scalability which roughly reproduces the behavior reported in~\cite{gottesburen2021deep}.
The algorithms \AlgName{OMS} and \AlgName{nh-OMS} have an intermediary scalability between \AlgName{KaMinPar} and \AlgName{Hashing}.
This is explained by their characteristics, which are intermediary between \AlgName{Fennel} and \AlgName{Hashing}.
Note that \AlgName{OMS} is more scalable than \AlgName{nh-OMS}.
This happens because \AlgName{OMS} goes through and scores several blocks in one of the partitioning subproblems contained in the multi-section hierarchy, which favors cache hits, whereas \AlgName{nh-OMS} partitioning subproblems have at most $4$ blocks.
For 32 threads, the average running time of \AlgName{OMS} is within a factor 3 of the running time of~\AlgName{Hashing}.

\section{Streaming Hypergraph Partitioning}
\label{sec:freight_Streaming Hypergraph Partitioning}
\label{chap:freight_Streaming Hypergraph Partitioning}
\label{chap:freight_FREIGHT: Fast Streaming Hypergraph Partitioning}

In this section, we propose \AlgName{FREIGHT}: a Fast stREamInG Hypergraph parTitioning algorithm that can optimize for cut-net and connectivity. 
By using an efficient data structure, we make the overall running time of \AlgName{FREIGHT} linearly dependent on the pin-count of the hypergraph and the memory consumption linearly dependent on the numbers of nets and blocks.
Our proposed algorithm demonstrates remarkable efficiency, with a running time comparable to the \AlgName{Hashing}  algorithm and a maximum discrepancy of only four in three quarters of the instancesr 
Our study establishes the superiority of \AlgName{FREIGHT} over all current (buffered) streaming algorithms and even the in-memory algorithm \AlgName{HYPE}, in both cut-net and connectivity measures. 
This shows the potential of our algorithm as a valuable tool for partitioning hypergraphs in the context of large and constantly changing \hbox{data processing environments}.

\subsection{FREIGHT}
\label{sec:freight_FREIGHT: Fast Streaming Hypergraph Partitioning}

\subsubsection{Mathematical Definition}
\label{subsec:freight_Mathematical Definition}
\label{subsubsec:freight_Mathematical Definition}

In this section, we provide a mathematical definition for \AlgName{FREIGHT} by expanding the idea of \AlgName{Fennel} to the domain of hypergraphs.
Recall that, assuming the nodes of a graph being streamed one-by-one, the \AlgName{Fennel} algorithm assigns an incoming node~$v$ to a block $V_d$ \hbox{where $d$ is computed as follows}:

\begin{equation}
d = \argmax\limits_{i,~|V_i| < L_{\max}}\big\{|V_i\cap N(v)|-\alpha * \gamma * |V_i|^{\gamma-1}\big\}
\label{eq:freight_Fennel}
\end{equation}

The term $-\alpha * \gamma * |V_i|^{\gamma-1}$, which penalizes block imbalance in \AlgName{Fennel}, is directly used in \AlgName{FREIGHT} without modification and with the same meaning.
The term $|V_i\cap N(v)|$, which minimizes edge-cut in \AlgName{Fennel}, needs to be adapted in~\AlgName{FREIGHT} to minimize the intended metric, i.e., either cut-net or connectivity.
Before explaining how this is adapted, recall that, in contrast to graph partitioning, in hypergraph partitioning the incident nets~$I(v)$ of an incoming node~$v$ might contain nets that are already cut, i.e., with pins assigned to multiple blocks.
The version of \AlgName{FREIGHT} designed to optimize for \emph{connectivity} accounts for already cut nets by keeping track of the block~$d_e$ to which the most recently streamed pin of each net~$e$ has been assigned.
More formally, the connectivity version of \AlgName{FREIGHT} assigns an incoming node~$v$ of a hypergraph to a block $V_d$ with $d$ given by Equation~(\ref{eq:freight_FREIGHT}), where \hbox{$I^i_{obj}(v) = I^i_{con}(v) =$} \hbox{$\{ e \in I(v): d_e = i\}$}.
On the other hand, the version of \AlgName{FREIGHT} designed to optimize for \emph{cut-net} ignores already cut nets, since their contribution to the overall cut-net of the hypergraph $k$-partition is fixed and cannot be changed anymore.
More formally, the cut-net version of \AlgName{FREIGHT} assigns an incoming node~$v$ of a hypergraph to a block $V_d$ with $d$ given by Equation~(\ref{eq:freight_FREIGHT}), where \hbox{$I^i_{obj}(v) = I^i_{cut}(v) =$} \hbox{$I^i_{con}(v) \setminus E'$}~and~$E'$ is the \hbox{set of already cut nets}.
\begin{equation}
d = \argmax\limits_{i,~|V_i| < L_{\max}}\big\{|I^i_{obj}(v)|-\alpha * \gamma * |V_i|^{\gamma-1}\big\}
\label{eq:freight_FREIGHT}
\end{equation}

Both configurations of \AlgName{FREIGHT} interpolate two objectives: favoring blocks with many incident (uncut) nets and penalizing blocks with large cardinality.
We briefly highlight that \AlgName{FREIGHT} can be adapted for weighted hypergraphs.
In particular, when dealing with weighted nets, the \hbox{term~$|I^i_{obj}(v)|$} is substituted \hbox{by~$\omega(I^i_{obj}(v))$}.
Likewise when dealing with weighted nodes, the term $-\alpha * \gamma * |V_i|^{\gamma-1}$ is substituted by $-c(v) * \alpha * \gamma * c(V_i)^{\gamma-1}$, where the weight~$c(v)$~of~$v$ is \hbox{used as a multiplicative factor in the penalty term}.

\begin{figure}[t]
	\centering
	\includegraphics[width=0.8\linewidth]{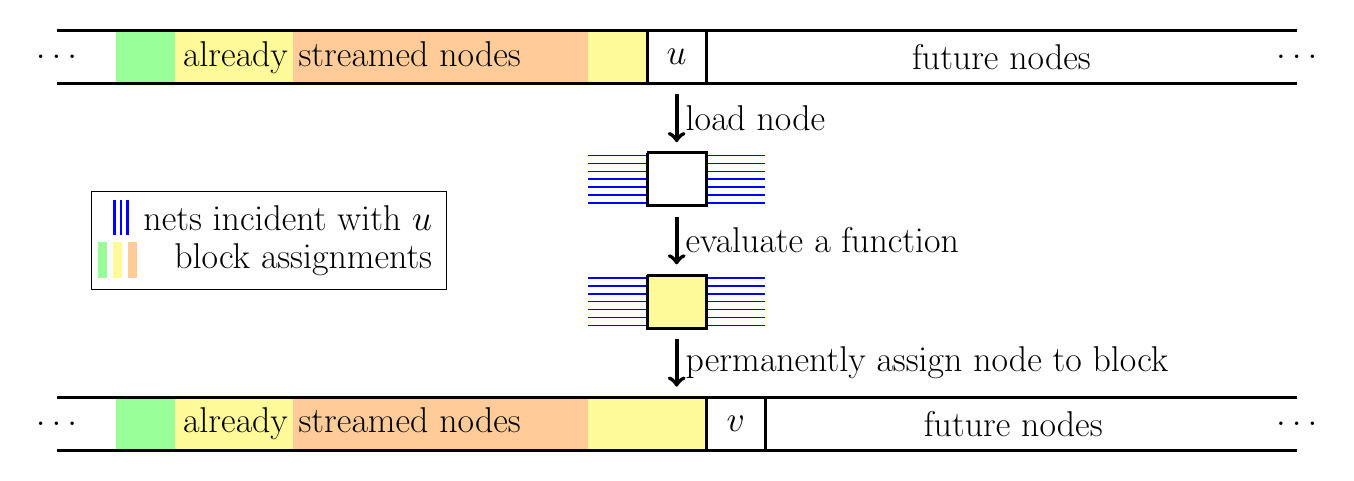}
	\caption{Typical layout of streaming algorithm for hypergraph partitioning.}
	\label{fig:freight_StreamingModel}
\end{figure}

\subsubsection{Streaming Hypergraphs}
\label{subsec:freight_Models for Streaming Hypergraphs}
\label{subsubsec:freight_Models for Streaming Hypergraphs}

In this section, we present and discuss the streaming model used by \AlgName{FREIGHT}.
Recall in the streaming model for graphs, nodes  are loaded one at a time alongside with their adjacency lists.
Thus, just streaming the graph (without doing additional compuations, implies a time cost $O(m+n)$.
In our model, the nodes of a hypergraph are loaded one at a time alongside with their incident nets, as illustrated in Figure~\ref{fig:freight_StreamingModel}.
Our streaming model implies a time cost $O(\sum_{e \in E}{|e|} + n)$ just to stream the hypergraph, where $O(\sum_{e \in E}{|e|})$ is the cost to stream each net $e$ exactly $|e|$ times.
\AlgName{FREIGHT} uses $O(m+k)$ memory, with $O(m)$ being used to keep track, for each net~$e$, of its cut/uncut status as well as the block~$d_e$ to which its most recently streamed pin was assigned.
This net-tracking information, which substitutes the need to keep track of node assignments, is necessary for executing \AlgName{FREIGHT}. 
Although \AlgName{FREIGHT} consumes more memory than required by graph-based streaming algorithms which often use $O(n+k)$ memory, it is still far better than the $O(mk)$ worst-case memory required by the state-of-the-art algorithms for streaming hypergraph partitioning~\cite{alistarh2015streaming,tacsyaran2021streaming}, all of which are also based on a computational model that implies a time cost \hbox{$O(\sum_{e \in E}{|e|} + n)$} \hbox{just to stream the hypergraph}.

\subsubsection{Efficient Implementation}
\label{subsec:freight_Efficient Implementation}
\label{subsubsec:freight_Efficient Implementation}

In this section, we describe an efficient implementation for \AlgName{FREIGHT}.
Recall that, for every node $v$ that is loaded, \AlgName{FREIGHT} uses Equation~(\ref{eq:freight_FREIGHT}) to find the block with the highest score among up to~$k$~options.
A simple method to accomplish this task consists of explicitly evaluating the score for each block and identifying the one with the highest score.
This results in a total of $O(nk)$ evaluations, leading to an overall complexity of \hbox{$O(\sum_{e \in E}{|e|}+nk)$}.
We propose an implementation that is significantly more efficient \hbox{than this approach}.

For each loaded node~$v$, our implementation separates the blocks $V_i$ for which \hbox{$|V_i|<L_{\max}$} into two disjoint sets, $S_1$ and $S_2$.
In particular, the set~$S_1$ is composed of the blocks $V_i$ where \hbox{$|I^i_{obj}(v)|>0$}, while the set~$S_2$ comprises the remaining blocks, i.e., blocks $V_i$ for which \hbox{$|I^i_{obj}(v)|=0$}.
Using the sets provided, we break down Equation~(\ref{eq:freight_FREIGHT}) into Equation~(\ref{eq:freight_FREIGHT_E_1}) and Equation~(\ref{eq:freight_FREIGHT_E_2}), which are solved separately. 
The resulting solutions are compared based on their \AlgName{FREIGHT} scores to ultimately find the solution for Equation~(\ref{eq:freight_FREIGHT}).
The overall process is illustrated \hbox{in~Figure~\ref{fig:freight_DecomposeEquation}}.
\begin{equation}
d = \argmax\limits_{i \in S_1}\big\{|I^i_{obj}(v)|-\alpha * \gamma * |V_i|^{\gamma-1}\big\}
\label{eq:freight_FREIGHT_E_1}
\end{equation}
\begin{equation}
d = \argmax\limits_{i \in S_2}\big\{|I^i_{obj}(v)|-\alpha * \gamma * |V_i|^{\gamma-1}\big\} = \argmin\limits_{i \in S_2}|V_i|
\label{eq:freight_FREIGHT_E_2}
\end{equation}

\begin{figure}[t]
	\centering
	\includegraphics[width=0.9\linewidth]{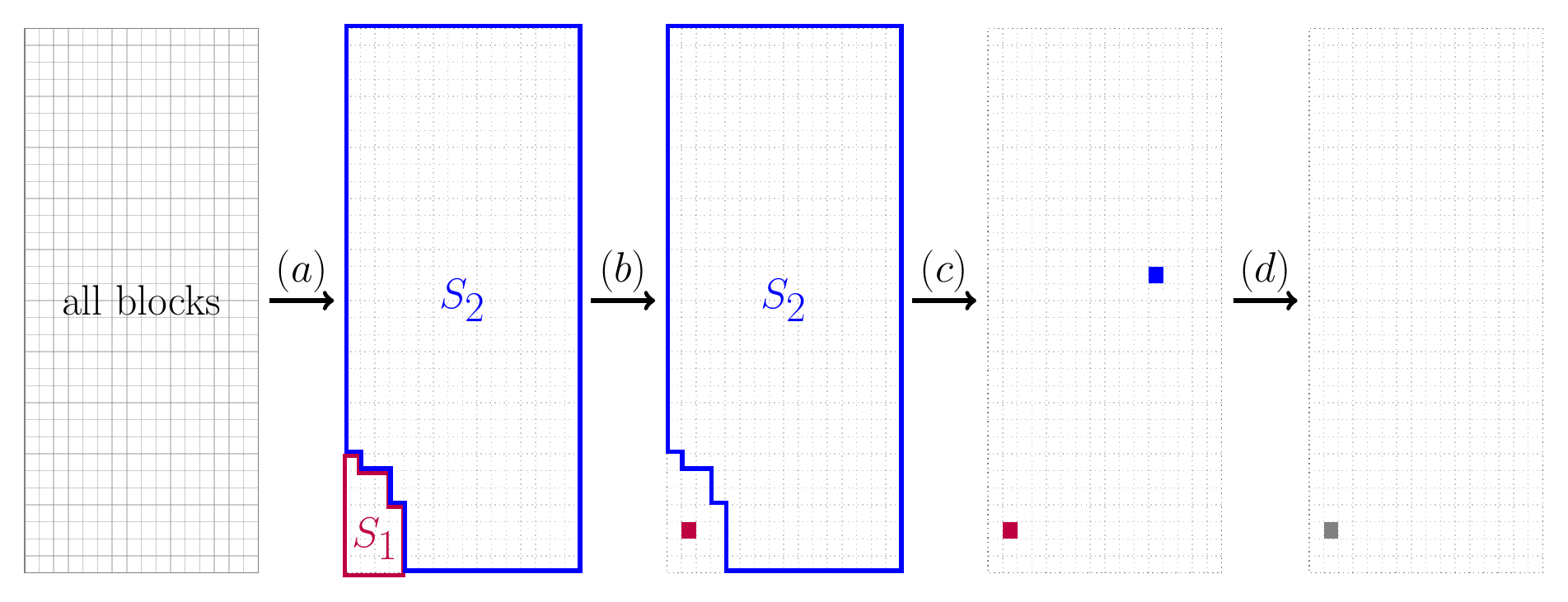}
	\caption{Illustration of the process to solve Equation~(\ref{eq:freight_FREIGHT}) for an incoming node~$u$ with \hbox{$k=512$}~blocks. 
		(a)~The~$k$~blocks are decomposed into~$S_1$~and~$S_2$, with~\hbox{$|S_1| = O(|I(u)|)$}. 
		(b)~Equation~(\ref{eq:freight_FREIGHT_E_1}) is explicitly solved at cost~\hbox{$O(|I(u)|)$}.  
		(c)~Equation~(\ref{eq:freight_FREIGHT_E_2}) is implicitly solved at cost~$O(1)$.
		(d)~Both solutions are then evaluated using their \AlgName{FREIGHT} scores to determine \hbox{the final solution for Equation~(\ref{eq:freight_FREIGHT})}.}
	\label{fig:freight_DecomposeEquation}
\end{figure}

Now we explain how we solve Equation~(\ref{eq:freight_FREIGHT_E_1})~and~Equation~(\ref{eq:freight_FREIGHT_E_2}).
To solve Equation~(\ref{eq:freight_FREIGHT_E_1}), we use the theoretical complexity outlined in Theorem~\ref{theo:freight_E_1} and solve it explicitly. 
In contrast, Equation~(\ref{eq:freight_FREIGHT_E_2}) is implicitly solved by identifying the block with minimal cardinality.
We use an efficient data structure to keep all blocks sorted by cardinality throughout the entire execution, which enables us to solve \hbox{Equation~(\ref{eq:freight_FREIGHT_E_2}) in constant time}.

\begin{theorem}
	Equation~(\ref{eq:freight_FREIGHT_E_1}) can be solved in time~$O(|I(v)|)$.
	\label{theo:freight_E_1}
\end{theorem}

\begin{proof}
	The terms \hbox{$|I^i_{obj}(v)|$} in Equation~(\ref{eq:freight_FREIGHT_E_1}) can be computed by iterating through the nets of~$v$ at a cost of \hbox{$O(|I(v)|)$} and determining their status as cut, unassigned, or assigned to a block.
	The calculation of the factors \hbox{$-\alpha * \gamma * |V_i|^{\gamma-1}$} in Equation~(\ref{eq:freight_FREIGHT_E_1}) can be done in time \hbox{$O(|S_1|) = O(|I(v)|)$}, \hbox{thus completing the proof}.
\end{proof}

Now we explain our data structure to keep the blocks sorted by cardinality during the whole algorithm execution.
The data structure is implemented with two arrays~$A$~and~$B$, both with $k$ elements, and a list~$L$.
The array~$A$ stores all $k$~blocks always in ascending order. 
The array~$B$ maps the index~$i$ of a block~$V_i$ to its position in~$A$.
Each element in the list~$L$ represents a bucket.
Each bucket is associated with a unique block cardinality and contains the leftmost and the rightmost positions~$\ell$~and~$r$ of the range of blocks in~$A$ which currently have this cardinality.
Reciprocally, each block in~$A$ has a pointer to the unique bucket in~$L$ corresponding to its cardinality.
To begin the algorithm, $L$ is set up with a single bucket for cardinality~$0$ which covers the $k$ positions of~$A$, i.e., its paramenters $\ell$~and~$r$ are $1$~and~$k$, respectively. 
The blocks in~$A$ are sorted in any order initially, however, as each block starts with a cardinality of $0$, they will be \hbox{ordered by their cardinalities}.

\begin{figure}[t]
	\centering
	\includegraphics[width=0.7\linewidth]{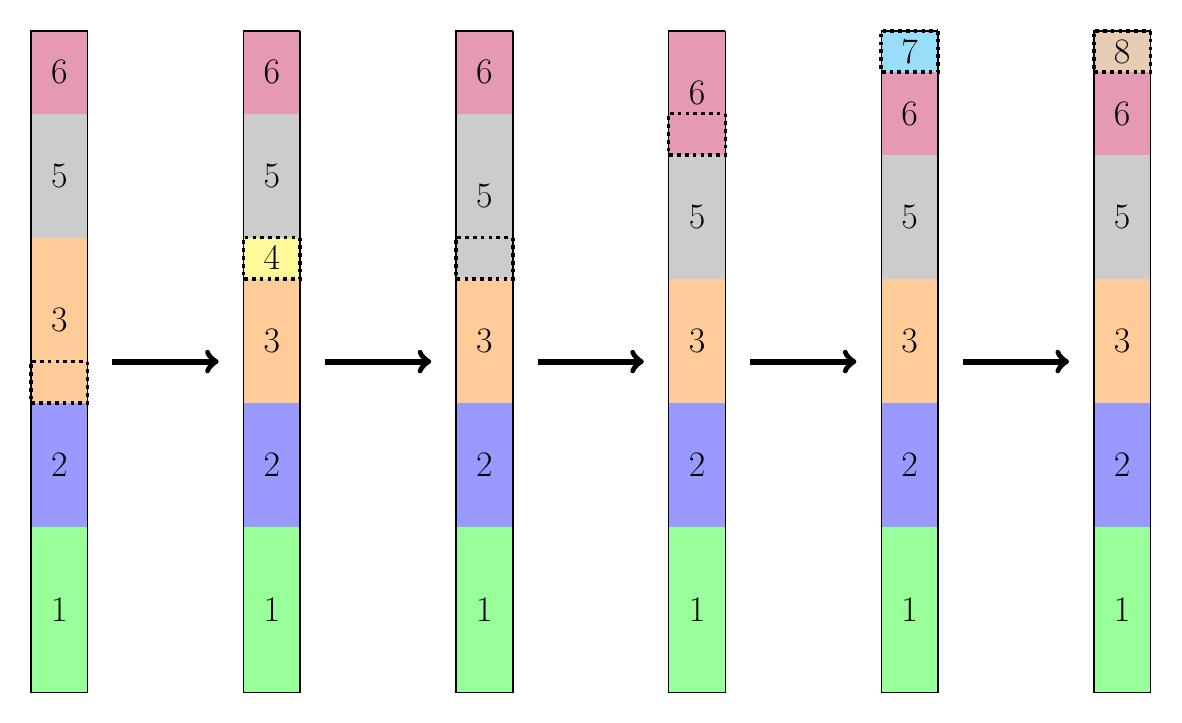}
	\caption{Illustration of our data structure used to keep the blocks sorted by cardinality throughout the execution of \AlgName{FREIGHT}. The array~$A$ is represented as a vertical rectangle. Each region of~$A$ is covered by a unique bucket, which is represented by a unique color filling the corresponding region in~$A$. The cardinality associated with each bucket is written in the middle of the region of~$A$ covered by it. Here we represent the behavior of the data structure when assigning nodes \hbox{to the block surrounded by a dotted rectangle five times consecutively}.}
	\label{fig:freight_SortingArray}
\end{figure}

When a node is assigned to a block~$V_d$, we update our data structure as detailed in Algorithm~\ref{alg:freight_IncrementCardinalityBlock} and exemplified in Figure~\ref{fig:freight_SortingArray}.
We describe Algorithm~\ref{alg:freight_IncrementCardinalityBlock} in detail now.
In line~1, we find the position~$p$ of~$V_d$ in~$A$ and find the bucket~$C$ associated with it.
In line~2, we exchange the content of two positions in~$A$: the position where~$V_d$ is located and the position identified by the variable $r$~in~$C$, which marks the rightmost block in~$A$ covered by~$C$.
This variable $r$ is afterwards decremented in line~3 since $V_d$ is now not covered anymore by the bucket $C$.
In lines~4~and~5, we check if the new (increased) cardinality of~$V_d$ matches the cardinality of the block located right after it in~$A$.
If so, we associate $V_d$ to the same bucket as it and decrement this bucket's leftmost position~$\ell$ in line~6;
Otherwise, we push a new bucket to~$L$ and match it to $V_d$ adequately in lines~8~and~9.
Finally, in line~10, we delete~$C$ in case its range $[\ell,r]$ is empty.
Figure~\ref{fig:freight_SortingArray} shows our data structure through five consecutive executions of Algorithm~\ref{alg:freight_IncrementCardinalityBlock}.
Theorem~\ref{theo:freight_sorting_correctness} proves the correctness of our data structure. 
Theorem~\ref{theo:freight_E_2} shows that, using our proposed data structure, we need time~$O(1)$ to either solve Equation~(\ref{eq:freight_FREIGHT_E_2}) or prove that the solution for Equation~(\ref{eq:freight_FREIGHT_E_1}) solves Equation~(\ref{eq:freight_FREIGHT}).
Note that our data structure can only handle unweighted vertices.
In case of weighted vertices, a bucket queue can be used instead of our data structure, resulting in the same overall complexity and requiring $O(k+L_{\max})$ memory, while our data structure only requires $O(k)$ memory.
The overall complexity of \AlgName{FREIGHT}, which directly follows from Theorem~\ref{theo:freight_E_1} and Theorem~\ref{theo:freight_E_2}, \hbox{is expressed in Corollary~\ref{cor:freight_complexity_without_sampling}}.

\begin{algorithm}[] %
	\begin{algorithmic}[1] %
		\State $p \gets B_d$; $C \gets A_p.bucket$;
		\State $q \gets C.r$; $c \gets A_q.id$; $Swap(A_p, A_q)$; $Swap(B_c, B_d)$; 
		\State $C.r \gets C.r - 1$;
		\State $C^\prime \gets A_{q+1}.bucket$;
		\If{$C.cardinality + 1 = C^\prime.cardinality$}
		\State $A_q.bucket \gets C^\prime$; $C^\prime.\ell \gets C^\prime.\ell - 1$;
		\Else
		\State $C^{\prime\prime} \gets NewBucket()$; $A_q.bucket \gets C^{\prime\prime}$; $L \gets L \cup \{C^{\prime\prime}\}$; 
		\State $C^{\prime\prime}.cardinality \gets C.cardinality + 1$; $C^{\prime\prime}.\ell \gets q$; $C^{\prime\prime}.r \gets q$;
		\EndIf
		\State \textbf{if} $C.r = C.\ell$ \textbf{then} $L \gets L \setminus \{C\}$;
	\end{algorithmic}
	\caption{Increment cardinality of block $V_d$ in the proposed data structure} %
	\label{alg:freight_IncrementCardinalityBlock} %
\end{algorithm}

% \vfill \pagebreak
\begin{theorem}
	Our proposed data structure keeps the blocks within array A consistently sorted in ascending order of cardinality.
	\label{theo:freight_sorting_correctness}
\end{theorem}

\begin{proof}
	We inductively prove two claims at the same time: 
	(a)~the variables~$\ell$~and~$r$ contained in each bucket from~$L$ respectively store the leftmost and the rightmost positions of the unique range of blocks in~$A$ which currently have this cardinality;
	(b)~the array~$A$ contains the blocks sorted in ascending order of cardinality.
	Both claims are trivially true at the beginning, since all blocks have cardinality 0 and~$L$ is initialized with a single bucket with $\ell=1$ and $r=k$.
	Now assuming that (a)~and~(b) are true at some point, we show that they keep being true after Algorithm~\ref{alg:freight_IncrementCardinalityBlock} is executed.
	Note that line~2 performs the only position exchange in~$A$ throughout the whole algorithm.
	As~(a) is assumed, it is the case that~$V_d$ swaps positions with the rightmost block in $A$ containing the same cardinality of~$V_d$.
	Since the cardinality of~$V_d$ will be incremented by one and all blocks have integer cardinalities, this concludes the proof of~(b).
	To prove that~(a) remains true, note that the only buckets in~$L$ that are modified are~$C$ (line~3),~$C^\prime$ (line~6), and~$C^{\prime\prime}$ (line~9).
	Claim~(a) remains true for~$C$ because $V_d$, whose cardinality will be incremented, is the only block removed from its range.
	Claim~(a) remains true for~$C^\prime$ because line~6 is only executed if the new cardinality of~$V_d$ equals the cardinality of~$C^\prime$, whose current range starts right after the new position of~$V_d$ in~$A$.
	Bucket~$C^{\prime\prime}$ is only created if the new cardinality of~$V_d$ is respectively larger and smaller than the cardinalities of~$C$~and~$C^\prime$.
	Since~(b) is true, then this condition only happens if there is no block in~$A$ with the same cardinality as the new cardinality of~$V_d$.
	Hence, claim~(a) remains true for~$C^{\prime\prime}$, which is created \hbox{covering only the position of~$V_d$ in~$A$}.
\end{proof}

\begin{theorem}
	By utilizing our proposed data structure, solving Equation~(\ref{eq:freight_FREIGHT_E_2}) or demonstrating that any solution for Equation~(\ref{eq:freight_FREIGHT_E_1}) is also a solution for Equation~(\ref{eq:freight_FREIGHT}) can be accomplished~in~$O(1)$~time.
	\label{theo:freight_E_2}
\end{theorem}

\begin{proof}
	Algorithm~\ref{alg:freight_IncrementCardinalityBlock} contains no loops and each command in it has a complexity of $O(1)$, thus the total cost of the algorithm is $O(1)$.
	Our data structure executes Algorithm~\ref{alg:freight_IncrementCardinalityBlock} once for each assigned node, hence it costs $O(1)$ per node.
	Say we are evaluating an incoming node~$v$.
	According to Theorem~\ref{theo:freight_sorting_correctness}, the block~$V_d$ with minimum cardinality is stored in the first position of the array~$A$, hence it can be accessed in time $O(1)$.
	In~case~\hbox{$V_d \in S_2$}, then~$d$ is a solution for Equation~(\ref{eq:freight_FREIGHT_E_2}).
	On the other hand, if $V_d$ is in $S_1$, the \AlgName{FREIGHT} score of $V_d$ will be larger than the \AlgName{FREIGHT} score of the solution for Equation~(\ref{eq:freight_FREIGHT_E_2}) by at least \hbox{$|I^{d}_{obj}(v)| > 0$}.
	In this case, it follows that \hbox{any solution for Equation~(\ref{eq:freight_FREIGHT_E_1}) solves Equation~(\ref{eq:freight_FREIGHT})}.
\end{proof}

\begin{corollary}
	The overall complexity of \AlgName{FREIGHT} is \hbox{$O\big(\sum_{e \in E}{|e|} + n\big)$}.
	\label{cor:freight_complexity_without_sampling}
\end{corollary}

\subsection{Experimental Evaluation}
\label{sec:freight_Experimental Evaluation}

\paragraph*{Setup.} 
We performed our implementations in C++ and compiled them using gcc 11.2 with full optimization turned on (-O3 flag). 
Unless mentioned otherwise, all experiments are performed on a single core of Machine~C. 
Unless otherwise mentioned we stream (hyper)graphs directly from the internal memory to obtain clear running time comparisons. 
However, note that \AlgName{FREIGHT} as well as most of the other used algorithms can also be run streaming \hbox{the hypergraphs from hard~disk}.

\paragraph*{Baselines.}
We compare \AlgName{FREIGHT} against various state-of-the-art algorithms.
In this section we will list these algorithms and explain our criteria for algorithm selection.
We have implemented \AlgName{Hashing}  in C++, since it is a simple algorithm.
It basically consists of hashing the IDs of incoming nodes into $\{1,\ldots,k\}$.
The remaining algorithms were obtained either from official repositories or privately from the authors, with the exception of  \AlgName{Min-Max}, for which there is no official implementation available.
Here, we use the \AlgName{Min-Max} implementations by Ta{\c s}yaran~et~al.~\cite{tacsyaran2021streaming}. 
\hbox{All algorithms were compiled with gcc 11.2}.

We run \AlgName{Hashing}, \AlgName{Min-Max}~\cite{alistarh2015streaming} and its improved versions proposed in~~\cite{tacsyaran2021streaming}: %
\AlgName{Min-Max-BF}, \AlgName{Min-Max-N2P}, \AlgName{Min-Max-L$\ell$}, \AlgName{Min-Max-MH}, \AlgName{REF}, \AlgName{REF\_RLX}, and \AlgName{REF\_RLX\_SV}.
(see Section~\ref{subsec:Streaming Hypergraph Partitioning} for details on the different \AlgName{Min-Max} versions), \AlgName{HYPE}~\cite{HYPE2018}, and \AlgName{PaToH}  v3.3~\cite{ccatalyurek2011patoh}.
\AlgName{Hashing}  is relevant because it is the simplest and fastest streaming algorithm, which gives us a lower bound for partitioning time.
\AlgName{Min-Max} is a current state-of-the-art for streaming hypergraph partitioning in terms of cut-net and connectivity.
The improved and buffered versions of \AlgName{Min-Max} proposed in~\cite{tacsyaran2021streaming} are relevant because some of them are orders of magnitude faster than \AlgName{Min-Max} while others produce improved partitions in comparison to it.
\AlgName{HYPE}   and \AlgName{PaToH}  are in-memory algorithms for hypergraph partitioning, hence they are not suitable for the streaming setting.
However, we compare against them because \AlgName{HYPE}   is among the fastest in-memory algorithms while \AlgName{PaToH}  is very fast and also computes partitions with very good cut-net and connectivity.
Note that KaHyPar~\cite{schlag2016k} is the leading tool with respect to solution quality, however it is also much slower than \AlgName{PaToH}.

\paragraph*{Instances.}
In our experiments we use 310~hypergraphs which are described in more detail in Section~\ref{subsec:Hypergraphs}.
Prior to each experiment, we converted all hypergraphs to the appropriate streaming formats required by each algorithm.
We removed parallel and empty hyperedges and self loops, and assigned unitary weight to all nodes and hyperedges.
In all experiments with streaming algorithms, we stream the hypergraphs with the natural given order of the nodes.
We use a number of blocks $k \in \{512,1024,1536,2048,2560\}$ 
unless mentioned otherwise.
We allow a fixed imbalance of $3\%$ for all experiments (and all algorithms) since this is a frequently used value in the partitioning literature. 
All algorithms always generated balanced partitions, except for \AlgName{HYPE}  
which generated highly unbalanced partitions in around \hbox{$5\%$ of its experiments}.

\paragraph*{Methodology.}
Depending on the focus of the experiment, we measure running time, cut-net, and-or connectivity.
We perform 5 repetitions per algorithm and instance using random seeds for non-deterministic algorithms.

\begin{figure*}[p!]
	\captionsetup[subfigure]{justification=centering}
	\centering
	\begin{subfigure}[]{\scaleFactorSmall\textwidth}
		\centering
		\includegraphics[angle=-0, width=\imgScaleFactorSmall\textwidth]{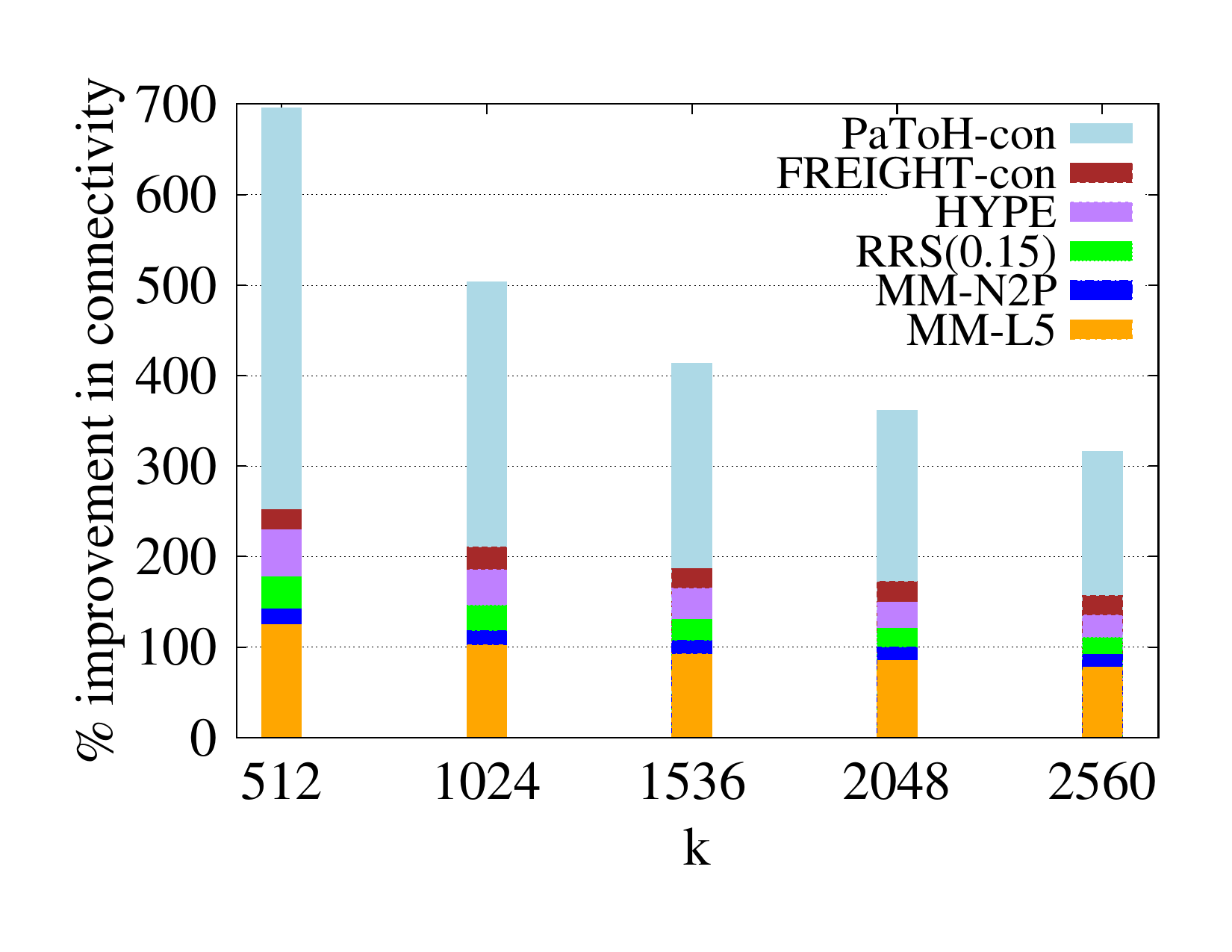}
		\vspace*{\capPositionSmall}
		\caption{Connectivity improvement over \AlgName{Hashing}.}
		\label{fig:freight_connectivity_impr}
	\end{subfigure}%
	\begin{subfigure}[]{\scaleFactorSmall\textwidth}
		\centering
		\includegraphics[angle=-0, width=\imgScaleFactorSmall\textwidth]{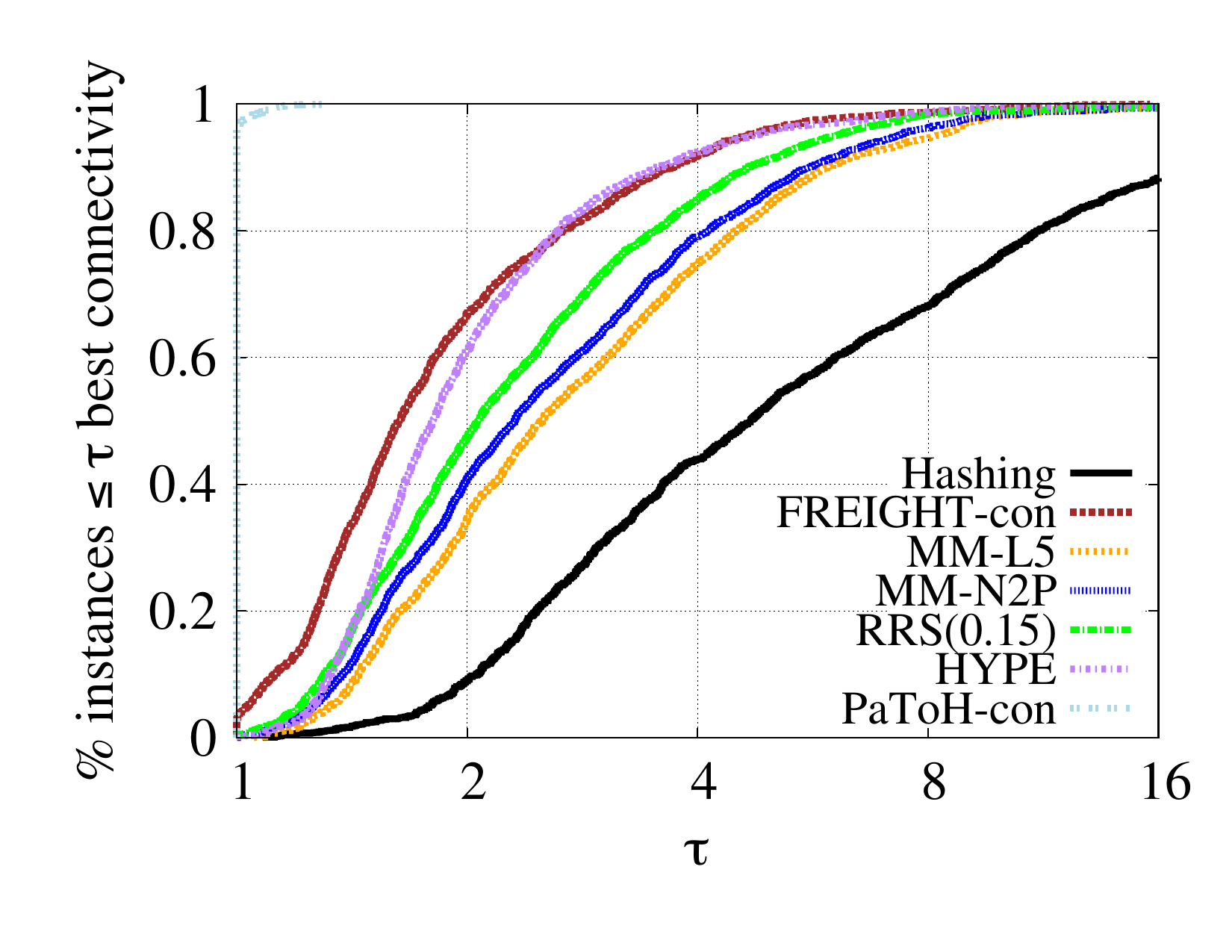}
		\vspace*{\capPositionSmall}
		\caption{Connectivity performance profiles.}
		\label{fig:freight_connectivity_pp}
	\end{subfigure}%

	\begin{subfigure}[]{\scaleFactorSmall\textwidth}
		\centering
		\includegraphics[angle=-0, width=\imgScaleFactorSmall\textwidth]{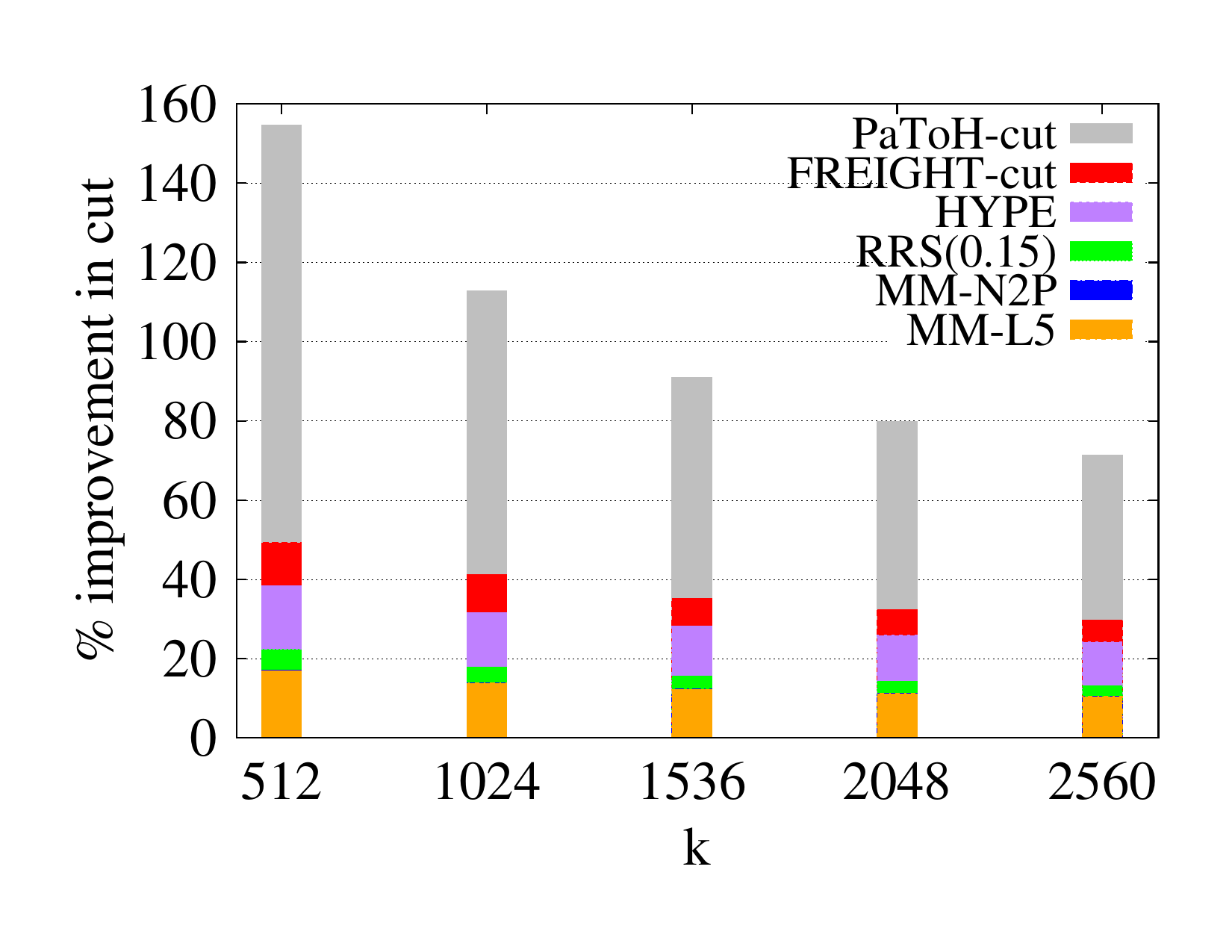}
		\vspace*{\capPositionSmall}
		\caption{Cut-net improvement over \AlgName{Hashing}.}
		\label{fig:freight_cut_impr}
	\end{subfigure}%
	\begin{subfigure}[]{\scaleFactorSmall\textwidth}
		\centering
		\includegraphics[angle=-0, width=\imgScaleFactorSmall\textwidth]{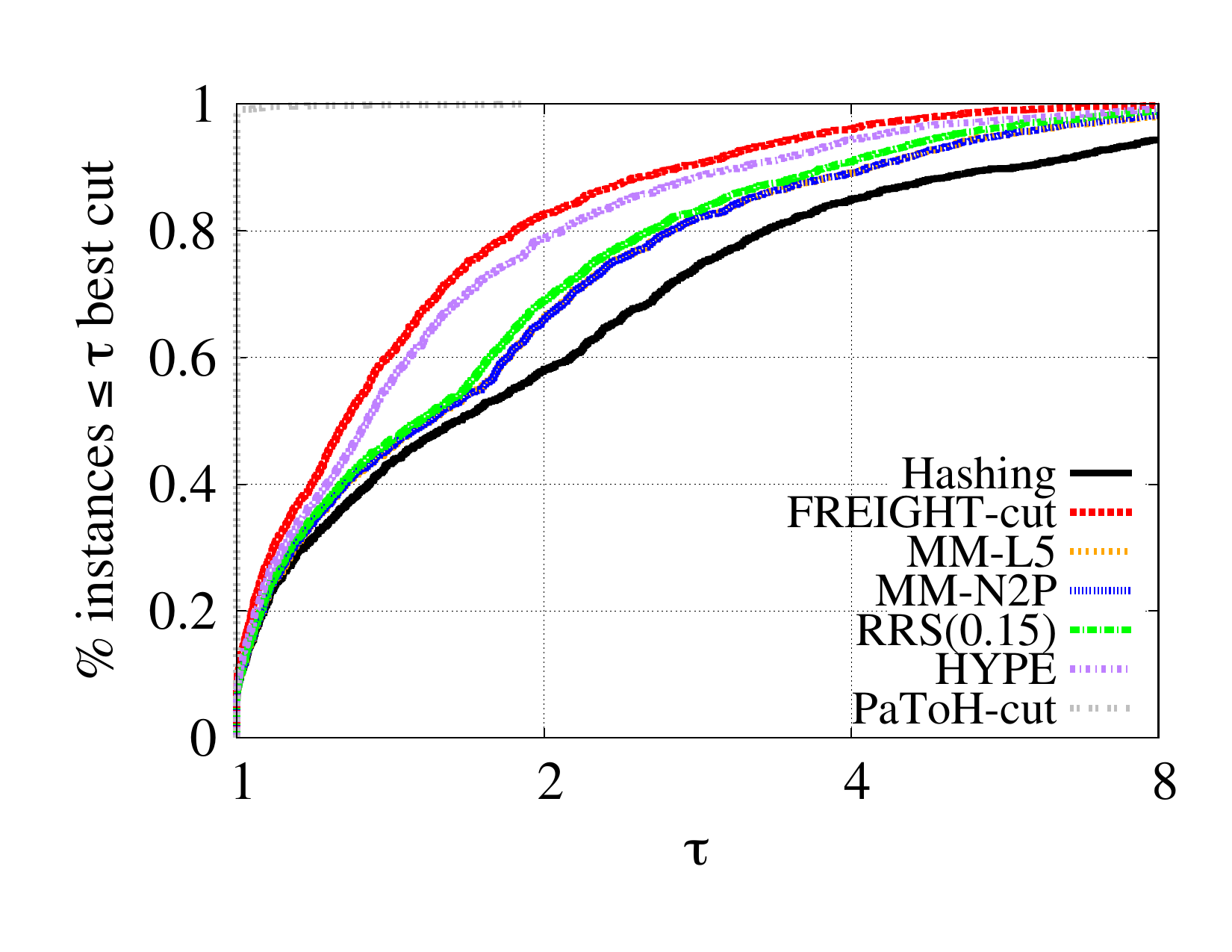}
		\vspace*{\capPositionSmall}
		\caption{Cut-net performance profiles.}
		\label{fig:freight_cut_pp}
	\end{subfigure}

	\begin{subfigure}[]{\scaleFactorSmall\textwidth}
		\centering
		\includegraphics[angle=-0, width=\imgScaleFactorSmall\textwidth]{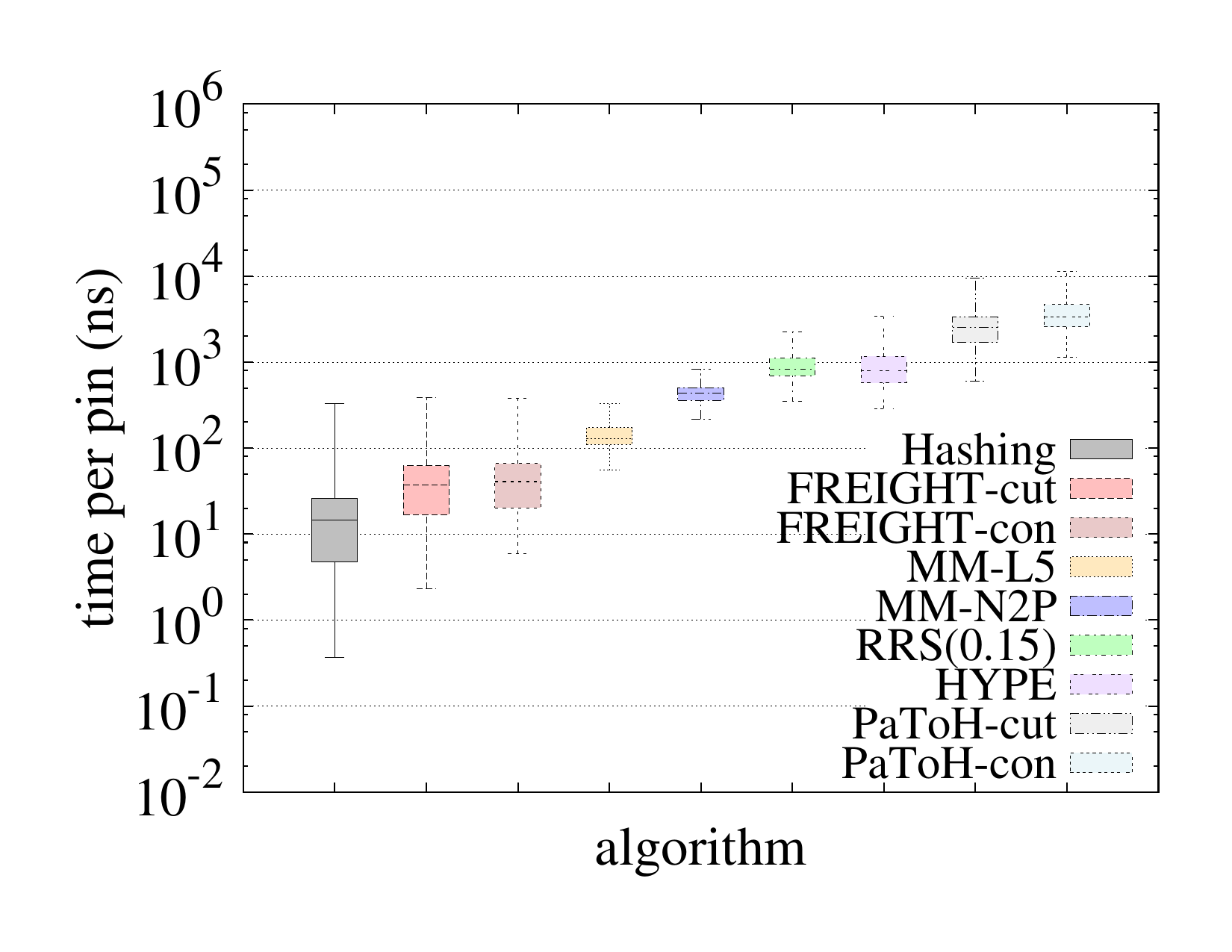}
		\vspace*{\capPositionSmall}
		\caption{Running time boxplots.}
		\label{fig:freight_time_bp}
	\end{subfigure}
	\begin{subfigure}[]{\scaleFactorSmall\textwidth}
		\centering
		\includegraphics[angle=-0, width=\imgScaleFactorSmall\textwidth]{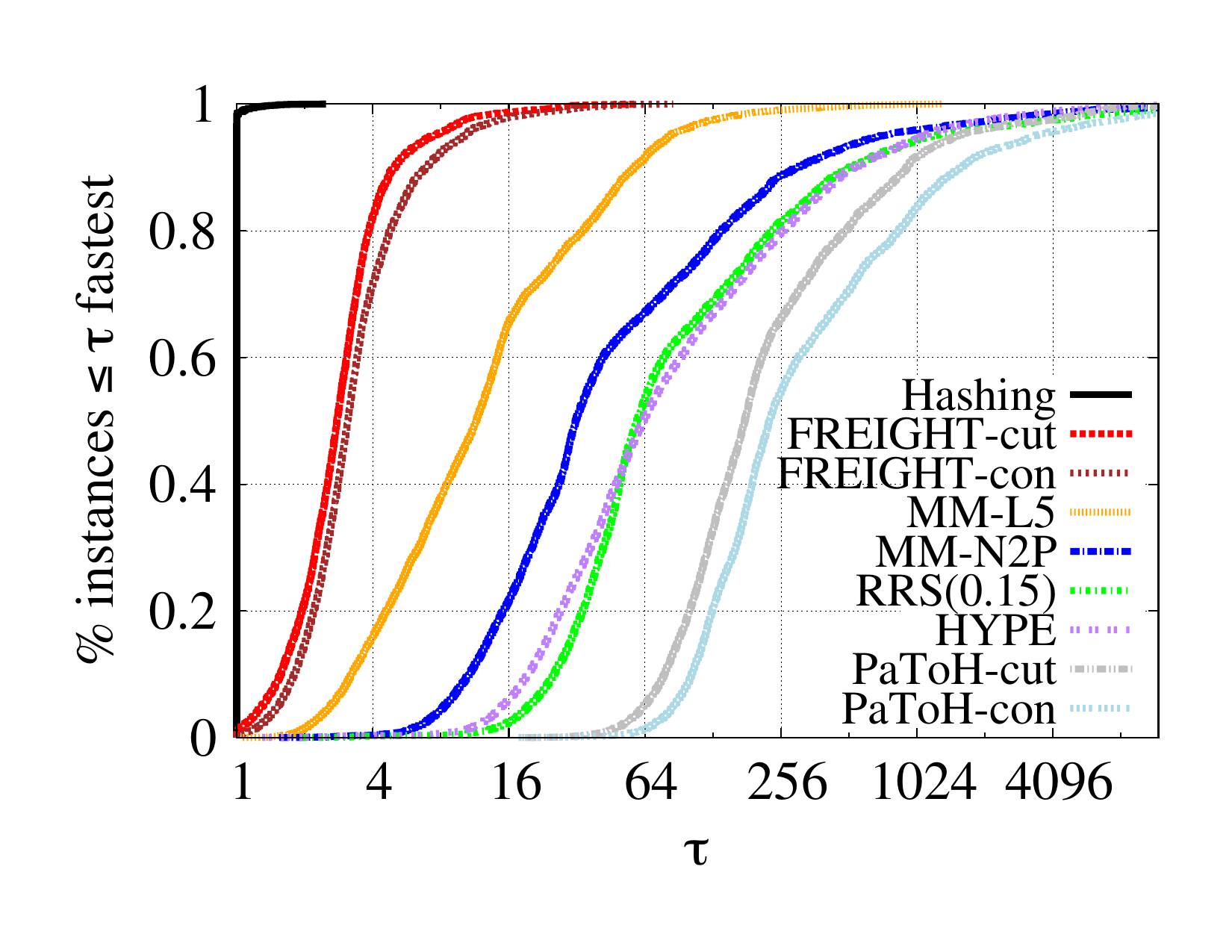}
		\vspace*{\capPositionSmall}
		\caption{Running time performance profiles.}
		\label{fig:freight_time_pp}
	\end{subfigure}%
	
	\vspace*{0.5cm}
	
	\caption{Comparison against the state-of-the-art streaming algorithms for hypergraph partitioning. We show performance profiles, improvement plots over \AlgName{Hashing}, and boxplots. Note that \AlgName{PaToH-con}, \AlgName{PaToH-cut}, and \AlgName{Hashing}  align almost perfectly with the y-axis in Figures~\ref{fig:freight_connectivity_pp},~\ref{fig:freight_cut_pp}, and~\ref{fig:freight_time_pp}, respectively. Also the curves and bars of \AlgName{MM-N2P} and \AlgName{MM-L5} roughly overlap \hbox{with one another in Figure~\ref{fig:freight_cut_pp} and Figure~\ref{fig:freight_cut_impr}}.}
	\label{fig:freight_state-of_the_art}
	
\end{figure*}

\subsubsection{State-of-the-Art}
\label{subsec:freight_Results}
\label{subsec:freight_State-of-the-Art}
\label{subsubsec:freight_Results}
\label{subsubsec:freight_State-of-the-Art}

In this section, we show experiments in which we compare \AlgName{FREIGHT} against the current state-of-the-art of streaming hypergraph partitioning.	
As already mentioned, we also use two internal-memory algorithms~\cite{HYPE2018,ccatalyurek2011patoh} as more general baselines for comparison. 
We focus our experimental evaluation on the comparison of solution quality and running time.
Observe that \AlgName{PaToH}  and \AlgName{FREIGHT} have distinct versions designed to optimize for each quality metric (i.e., connectivity and cut-net).
For a meaningful comparison, we only take into account the relevant version when dealing with each quality metric, however, both versions are still considered for running time comparisons.
To differentiate between the versions, suffixes \AlgName{-con} and \AlgName{-cut} are added to represent the connectivity-optimized and cut-net \hbox{versions respectively}.

For clarity, we refrain from discussing state-of-the-art streaming algorithms that are \emph{dominated} by another algorithm.
We define a dominated algorithm as one that has worse running time compared to another without offering a superior solution quality in return, or vice-versa.
In particular, we leave out \AlgName{Min-Max} and \AlgName{Min-Max-BF} since they are dominated by \AlgName{Min-Max-N2P}, which is referred to as \AlgName{MM-N2P} hereafter. %
Similarly, we omit \AlgName{Min-Max-MH} 
because it is dominated by \AlgName{Hashing}.
We use a buffer size of $15\%$ for testing the buffered algorithms \AlgName{REF}, \AlgName{REF\_RLX}, and \AlgName{REF\_RLX\_SV}, following the best results outlined in~\cite{tacsyaran2021streaming}. %
We omit the first two of them since they are dominated by the latter one, which is referred to as \AlgName{RRS(0.15)}   from now on.
Since \AlgName{Min-Max-L$\ell$} is not dominated by any other algorithm, we exhibit its results with $\ell=5$, as seen in the best results in~\cite{tacsyaran2021streaming}, and \hbox{we refer to it as \AlgName{MM-L5} from this point}.

\paragraph*{Connectivity.}
We start by looking at the connectivity metric. 
In Figure~\ref{fig:freight_connectivity_impr}, we plot the average connectivity improvement over \AlgName{Hashing}  for each value of~$k$. 
\AlgName{PaToH-con}  produces the best connectivity on average, yielding an average improvement of $443\%$ when compared to \AlgName{Hashing}.
This is in line with previous works in the area of (hyper)graph partitioning, i.e. streaming algorithms typically compute worse solutions than internal memory algorithms, which have access to the whole graph. 
\AlgName{FREIGHT-con} is found to be the second best algorithm in terms of connectivity, outperforming both the internal memory algorithm \AlgName{HYPE}   and the buffered streaming algorithm \AlgName{RRS(0.15)}.
On average, these three algorithms improve $194\%$, $171\%$, and $136\%$ over \AlgName{Hashing}, respectively.
Finally, \AlgName{MM-N2P} and \AlgName{MM-L5} compute solutions which improve $111\%$ and $96\%$ over \AlgName{Hashing}  on average, respectively.
In direct comparison, \AlgName{FREIGHT-con} shows average connectivity improvements of $8\%$, $24\%$, $39\%$, and $50\%$ over \AlgName{HYPE}, \AlgName{RRS(0.15)}, \AlgName{MM-N2P}, and \AlgName{MM-L5}, respectively.
Note that each algorithm retains its relative ranking in terms of average connectivity \hbox{over all values of~$k$}.

In Figure~\ref{fig:freight_connectivity_pp}, we plot connectivity performance profiles across all experiments.
Note that \AlgName{PaToH-con}  produces the best overall connectivity for $96.4\%$ of the instances, while \AlgName{FREIGHT-con} produces the best connectivity for $3.1\%$ of the instances and no other algorithm computes the best connectivity for more than $0.35\%$ of the instances.
The connectivity produced by  \AlgName{FREIGHT-con}, \AlgName{HYPE}, \AlgName{RRS(0.15)}, \AlgName{MM-N2P}, \AlgName{MM-L5}, and \AlgName{Hashing}  are within a factor~$2$ of the best found connectivity for $67\%$, $61\%$, $47\%$, $41\%$, $34\%$, and $9\%$ of the instances, respectively.
In summary, \AlgName{FREIGHT-con} produces the best connectivity among (buffered) streaming competitors, outperforming even \hbox{in-memory algorithm \AlgName{HYPE}}.

\paragraph*{Cut-Net.}
Next we examine at the cut-net metric.
In Figure~\ref{fig:freight_cut_impr}, we plot the cut-net improvement over \AlgName{Hashing}.
\AlgName{PaToH-cut}  produces the best overall cut-net, with an average improvement of $100\%$ compared to \AlgName{Hashing}.
\AlgName{FREIGHT-cut} is found to be the second best algorithm with respect to cut-net, superior to internal-memory algorithm \AlgName{HYPE}   and buffered streaming algorithm \AlgName{RRS(0.15)}.
These three algorithms improve connectivity over \AlgName{Hashing}  by $37\%$, $30\%$, and $17\%$ respectively.
Finally, both \AlgName{MM-N2P} and \AlgName{MM-L5} improve connectivity by $13\%$ on average over \AlgName{Hashing}.
In direct comparison, \AlgName{FREIGHT-cut} shows average connectivity improvements of $6\%$, $18\%$, $22\%$, and $22\%$ over \AlgName{HYPE}, \AlgName{RRS(0.15)}, \AlgName{MM-N2P}, and \AlgName{MM-L5}, respectively.
Each algorithm preserves its relative ranking in average cut-net \hbox{across all values of~$k$}.

% \vfill\pagebreak

In Figure~\ref{fig:freight_cut_pp}, we plot cut-net performance profiles across all experiments.
In the plot, \AlgName{PaToH-cut}  produces the best overall connectivity for $98.0\%$ of the instances, while \hbox{\AlgName{FREIGHT-cut}} and \AlgName{HYPE}   produce the best cut-net for $6.8\%$ and $5.2\%$ of the instances and all other streaming algorithms (\AlgName{RRS(0.15)}, \AlgName{MM-N2P}, \AlgName{MM-L5}, and \AlgName{Hashing}) produce the best cut-net for $4.8\%$ of the instances.
The cut-net results produced by  \AlgName{FREIGHT-cut}, \AlgName{HYPE}, \AlgName{RRS(0.15)}, \AlgName{MM-N2P}, \AlgName{MM-L5}, and \AlgName{Hashing}  are within a factor~$2$ of the best found cut-net for $83\%$, $79\%$, $69\%$, $66\%$, $66\%$, and $58\%$ of the instances, respectively.
This shows that \AlgName{FREIGHT-cut} produces the best cut-net among all (buffered) streaming competitors and \hbox{even beats the in-memory algorithm \AlgName{HYPE}}.

\paragraph*{Running Time.}
Now we compare the algorithms' runtime.
Boxes and whiskers in Figure~\ref{fig:freight_time_bp} display the distribution of the running time per pin, measured in nanoseconds, for all instances.
\AlgName{Hashing}, \AlgName{FREIGHT-cut}, and \AlgName{FREIGHT-con} are the three fastest algorithms, with median runtimes per pin of \numprint{15}ns, \numprint{38}ns, and \numprint{41}ns, respectively.
\AlgName{MM-L5}, \AlgName{MM-N2P}, \AlgName{HYPE}, and \AlgName{RRS(0.15)}   follow with median runtimes per pin of \numprint{130}ns, \numprint{437}ns, \numprint{792}ns, and \numprint{833}ns, respectively.
Lastly, the algorithms with the highest median runtime per pin are \AlgName{PaToH-cut}  and \AlgName{PaToH-con}, with \numprint{2516}ns and \numprint{3333}ns respectively.
The measured runtime per pin for both \AlgName{HYPE}   and \AlgName{PaToH}  align with \hbox{values reported in prior research~\cite{schlag2020high}}.

In Figure~\ref{fig:freight_time_pp}, we show running time performance profiles.
\AlgName{Hashing}  is the fastest algorithm for $98.3\%$ of the instances, while \AlgName{FREIGHT-cut} is the fastest one for $1.2\%$ of the instances and no other algorithm is the fastest one for more than $0.4\%$ of the instances. 
The running time of \AlgName{FREIGHT-cut} and \AlgName{FREIGHT-con} is within a factor 4 of that of \AlgName{Hashing}  for $82\%$ and $72\%$ of instances, respectively. 
In contrast, for only $16\%$ of instances does this occur for \AlgName{MM-L5}, and for less than $0.4\%$ of instances for all other algorithms.
The close running times of \AlgName{FREIGHT} to \AlgName{Hashing}  are surprising given \AlgName{FREIGHT}'s superior solution quality compared to \AlgName{Hashing}  and \hbox{all other streaming algorithms and even \AlgName{HYPE}}.

\paragraph*{Further Comparisons.}
\label{subpar:freight_Further Comparison}
For graph partitioning \AlgName{FREIGHT} and \AlgName{Fennel} are mathematically equivalent. 
However, \AlgName{FREIGHT} exhibits a lower computational complexity of $O(m+n)$ compared to the standard implementation of \AlgName{Fennel}, which has a complexity of $O(m+nk)$ due to evaluating all blocks for each node. 
To optimize its performance for this use case, we have implemented an optimized version of \AlgName{FREIGHT} with a memory consumption of $O(n+k)$, matching that of \AlgName{Fennel}. 
In our experiments, we utilized the same graphs as in~\cite{StreamMultiSection} and tested with $k \in \{512,1024,1536,2048,2560\}$.
On average, \AlgName{FREIGHT} proves to be 109 times faster than the standard implementation of \AlgName{Fennel}. 
Moreover, the performance gap is found to increase as the value of $k$ grow, with \AlgName{FREIGHT} reaching up to 261 times \hbox{faster than \AlgName{Fennel} in some instances}.

% \vfill\pagebreak

\section{Experimental Comparison}
\label{sec:Experimental Comparison}

In this section, we present an unpublished experimental comparison of our three streaming algorithms.
We performed experiments for $k={8,16,32,64,128,256}$ on the \emph{huge graphs} listed in Table~\ref{tab:streaming_hugeResults} using Machine~B, which is a relatively modest machine.
As in Section~\ref{sec:heistream_huge_graphs}, we did not repeat each test multiple times with different seeds.
As mentioned in Section~\ref{sec:heistream_huge_graphs}, we also ran \AlgName{Metis} and \AlgName{KaHIP} on these graphs, but they failed on all instances as they require more memory than \hbox{the machine has}.

Like in Section~\ref{sec:heistream_huge_graphs}, we refer to setups of \AlgName{HeiStream} with specific buffer sizes as \AlgName{HeiStream}($X$k), where a buffer contains $X\times 1024$ nodes.
Apart from \AlgName{HeiStream}, we ran \AlgName{Online Recursive Multi-Section} (referred to as \AlgName{nh-OMS}) and \AlgName{FREIGHT}.
In particular, we use here the best parameters of \AlgName{nh-OMS} found in the tuning experiments in Section~\ref{subsec:Parameter Study}.
Table~\ref{tab:streaming_hugeResults} provides detailed per-instance results with \AlgName{HeiStream} using large buffer sizes capable of running on Machine~B.
We have excluded the IO delay of loading the input graph from the disk in Table~\ref{tab:streaming_hugeResults} as it depends on the disk and is roughly the same regardless of the used partitioning algorithm.
The delay (in seconds) is \hbox{reported in Section~\ref{sec:heistream_huge_graphs}}.

\begin{table}[p]
	\centering
	\scriptsize
	\setlength{\tabcolsep}{2.pt}
	\begin{tabular}{lr@{\hskip 30pt}rrr@{\hskip 30pt}rr@{\hskip 30pt}rr}
		\toprule	
		\multicolumn{1}{l}{\multirow{2}{*}{Graph}} & \multirow{2}{*}{k} & \multicolumn{3}{l}{\AlgName{HeiStream}(Xk)}                                               & \multicolumn{2}{l}{\AlgName{nh-OMS}}                               & \multicolumn{2}{l}{\AlgName{FREIGHT}}   \\ 
		\multicolumn{1}{l}{}                       &                    & \multicolumn{1}{l}{X} & \multicolumn{1}{l}{CE(\%)} & \multicolumn{1}{l}{RT(s)} & \multicolumn{1}{l}{CE(\%)} & \multicolumn{1}{l}{RT(s)} & \multicolumn{1}{l}{CE(\%)} & \multicolumn{1}{l}{RT(s)}  \\ 
		
                \midrule
		\multirow{6}{*}{uk-2005}                   %
		& 8                  & 1024                  & \textbf{4.03}              & 290.23                      & 22.08                      & 26.96                       & 23.87                      & 20.07                                    \\
		& 16                 & 1024                  & \textbf{6.01}              & 300.04                      & 25.40                      & 27.85                       & 26.79                      & 21.49                                    \\
		& 32                 & 1024                  & \textbf{7.65}              & 310.72                      & 28.34                      & 34.01                       & 30.15                      & 21.16                                    \\
		& 64                 & 1024                  & \textbf{8.99}              & 322.14                      & 30.58                      & 36.52                       & 32.10                      & 21.89                                    \\
		& 128                & 1024                  & \textbf{9.94}              & 346.73                      & 32.05                      & 41.89                       & 33.31                      & 22.74                                    \\	
		& 256                & 1024                  & \textbf{10.68}             & 386.64                      & 33.71                      & 44.50                       & 34.35                      & 23.51                                  \\ 
                \midrule
		\multirow{6}{*}{twitter7}                  %
		& 8                  & 512                   & \textbf{41.64}             & 1727.13                      & 60.14                      & 178.59                       & 55.13                      & 173.88                                    \\
		& 16                 & 512                   & \textbf{47.04}             & 1774.92                      & 72.74                      & 182.74                       & 61.73                      & 184.13                                   \\
		& 32                 & 512                   & \textbf{52.59}             & 1884.16                      & 81.78                      & 187.33                       & 65.07                      & 179.02                                   \\
		& 64                 & 512                   & \textbf{57.53}             & 1988.11                      & 87.64                      & 205.12                       & 76.78                      & 183.84                                   \\
		& 128                & 512                   & \textbf{61.87}             & 2113.34                      & 92.06                      & 228.06                       & 81.20                      & 185.98                                    \\
		& 256                & 512                   & \textbf{65.47}             & 2357.92                      & 94.28                      & 249.24                       & 83.39                      & 188.20                                   \\ 
                \midrule
		\multirow{6}{*}{sk-2005}                   %
		& 8                  & 1024                  & \textbf{3.23}              & 634.79                      & 20.86                      & 40.10                       & 23.89                      & 34.13                                   \\
		& 16                 & 1024                  & \textbf{4.11}              & 648.48                      & 25.69                      & 43.63                       & 28.87                      & 33.97                                   \\
		& 32                 & 1024                  & \textbf{5.32}              & 667.84                      & 28.71                      & 49.29                       & 31.90                      & 35.58                                   \\
		& 64                 & 1024                  & \textbf{7.55}              & 695.60                      & 32.59                      & 53.15                       & 35.20                      & 35.87                                   \\
		& 128                & 1024                  & \textbf{8.95}              & 733.05                      & 35.65                      & 61.62                       & 38.59                      & 37.56                                   \\ 
		& 256                & 1024                  & \textbf{12.02}              & 798.73                     & 39.64                      & 64.80                       & 43.15                      & 38.58                               \\ 
                \midrule
		\multirow{6}{*}{soc-friendster}           %
		& 8                  & 1024                  & \textbf{27.36}             & 4099.35                     & 37.68                      & 362.80                      & 35.54                      & 383.06                                   \\
		& 16                 & 1024                  & \textbf{34.50}             & 4202.04                     & 52.70                      & 393.57                      & 50.06                      & 383.60                                   \\
		& 32                 & 1024                  & \textbf{39.52}             & 4345.96                     & 64.99                      & 426.01                      & 54.70                      & 374.97                                   \\
		& 64                 & 1024                  & \textbf{46.35}             & 4546.98                     & 76.39                      & 427.77                      & 57.88                      & 388.27                                   \\
		& 128                & 1024                  & \textbf{52.41}             & 4796.56                     & 80.88                      & 481.40                      & 60.59                      & 394.40                                   \\ 
		& 256                & 1024                  & \textbf{57.79}             & 5323.08                     & 83.32                      & 477.39                      & 63.32                      & 408.62                                  \\ 
                \midrule
		\multirow{6}{*}{er-fact1.5s26}           %
		& 8                  & 1024                  & \textbf{73.27}             & 2216.99                     & 74.33                      & 233.06                      & 73.44                      & 231.99                                   \\
		& 16                 & 1024                  & \textbf{80.18}             & 2292.12                     & 81.59                      & 255.07                      & 80.40                      & 243.52                                   \\
		& 32                 & 1024                  & \textbf{84.36}             & 2400.35                     & 86.27                      & 293.15                      & 84.63                      & 255.88                                   \\
		& 64                 & 1024                  & \textbf{86.99}             & 2534.09                     & 89.05                      & 316.62                      & 87.31                      & 253.72                                   \\
		& 128                & 1024                  & \textbf{88.72}             & 2725.81                     & 90.82                      & 374.28                      & 89.10                      & 248.61                                   \\ 
		& 256                & 1024                  & \textbf{89.99}             & 2913.95                     & 91.80                      & 393.10                      & 90.45                      & 261.57                                   \\ 
                \midrule
		\multirow{6}{*}{RHG1}                     %
		& 8                  & 1024                  & \textbf{0.04}              & 380.04                      & 1.89                       & 58.12                       & 2.02                       & 44.17                                   \\
		& 16                 & 1024                  & \textbf{0.06}              & 391.63                      & 1.99                       & 60.36                       & 2.12                       & 44.83                                   \\
		& 32                 & 1024                  & \textbf{0.09}              & 406.56                      & 2.02                       & 70.22                       & 2.16                       & 46.87                                   \\
		& 64                 & 1024                  & \textbf{0.15}              & 435.71                      & 2.05                       & 75.38                       & 2.17                       & 45.09                                   \\
		& 128                & 1024                  & \textbf{0.22}              & 482.06                      & 2.07                       & 92.34                       & 2.18                       & 47.09                                   \\
		& 256                & 1024                  & \textbf{0.34}              & 569.77                      & 2.09                       & 92.16                       & 2.19                       & 46.92                              \\ 
                \midrule
		\multirow{6}{*}{RHG2}                      %
		& 8                  & 1024                  & 0.09                       & 621.56                      & \textbf{0.04}              & 73.03                       & 0.05                       & 57.02                                   \\
		& 16                 & 1024                  & 0.13                       & 632.61                      & \textbf{0.04}              & 72.78                       & 0.08                       & 57.87                                   \\
		& 32                 & 1024                  & 0.19                       & 648.68                      & \textbf{0.04}              & 84.19                       & 0.12                       & 57.34                                   \\
		& 64                 & 1024                  & 0.29                       & 674.36                      & \textbf{0.05}              & 92.05                       & 0.18                       & 57.78                                   \\
		& 128                & 1024                  & 0.44                       & 727.66                      & \textbf{0.06}              & 97.14                       & 0.27                       & 56.64                                   \\
		& 256                & 1024                  & 0.68                       & 816.60                      & \textbf{0.08}              & 105.69                      & 0.44                       & 61.23                                 \\ 
                \midrule
		\multirow{6}{*}{uk-2007-05}                %
		& 8                  & 1024                  & \textbf{0.54}              & 1024.26                      & 23.04                      & 76.25                       & 25.28                     & 62.97                                   \\
		& 16                 & 1024                  & \textbf{0.60}              & 1045.36                      & 26.69                      & 82.78                       & 28.18                     & 65.81                                   \\
		& 32                 & 1024                  & \textbf{0.70}              & 1058.73                      & 28.14                      & 94.47                       & 29.46                     & 66.08                                   \\
		& 64                 & 1024                  & \textbf{0.92}              & 1099.64                      & 29.74                      & 103.21                      & 29.92                     & 66.86                                   \\
		& 128                & 1024                  & \textbf{1.31}              & 1163.45                      & 30.80                      & 118.08                      & 30.63                     & 67.78                                   \\
		& 256                & 1024                  & \textbf{1.95}              & 1280.48                      & 32.10                      & 125.89                      & 31.54                     & 70.58                                 \\ 
                \bottomrule
	\end{tabular}
	\vspace*{-.25cm}
	\caption{Our algorithms compared. CE and RT denote cut edges and running time. 
	}
	\label{tab:streaming_hugeResults}
\end{table}

The results indicate that \AlgName{HeiStream} performs better than its competitors in terms of solution quality for most instances.
Notably, it produces partitions with significantly lower edge-cut compared to our other streaming algorithms for four tested graphs, namely uk-2005, sk-2005, uk-2007-05, and RHG1.
For social networks soc-friendster and twitter7, \AlgName{HeiStream} outperforms all other algorithms, but the improvement over \AlgName{nh-OMS} and \AlgName{FREIGHT} is not as substantial as in other cases.
However, there is an exception on the RHG2 network, where \AlgName{HeiStream} produces edge-cut values below $0.7\%$, but \AlgName{FREIGHT} performs better and \AlgName{nh-OMS} produces even lower edge-cut values.
Overall, \AlgName{FREIGHT} and \AlgName{nh-OMS} produce comparable quality results, with \AlgName{FREIGHT} being slightly better on three graphs and \AlgName{nh-OMS} being slightly better \hbox{on five graphs}.

In terms of running time, \AlgName{HeiStream} is considerably slower than \AlgName{nh-OMS} and \AlgName{FREIGHT}.
Although the runtime complexity of \AlgName{HeiStream} is $O(m+n)$, which is not worse than that of \AlgName{nh-OMS} and \AlgName{FREIGHT}, its constant factors are larger due to more complex computations.
It is also worth noting that the running time of \AlgName{nh-OMS} increases with increasing $k$, while the running time of \AlgName{FREIGHT} remains roughly constant for any $k$.
This is expected because the running time of \AlgName{nh-OMS} is $O((m+n)\log{k})$, while the running time of \AlgName{FREIGHT} for graphs is $O(m+n)$.
As a result, \AlgName{FREIGHT} becomes faster than \AlgName{nh-OMS} \hbox{as $k$ increases}.

\vfill

\section{Conclusion}
\label{sec:streaming_conclusion}

In this chapter, we proposed three streaming algorithms for (hyper)graph decomposition: \AlgName{HeiStream}, \AlgName{Online Recursive Multi-Section}, and \AlgName{FREIGHT}.
All were discussed in detail and were subjected to extensive experimental evaluation against the state-of-the-art.
Lastly, we presented an experimental comparison between them for the graph \hbox{partitioning problem}.

We proposed \AlgName{HeiStream}, a buffered streaming graph partitioning algorithm.
It combines the buffered streaming model with multilevel graph partitioning techniques and an extension of \AlgName{Fennel} to a multilevel algorithm.
Compared to the previous state of-the-art, \AlgName{HeiStream} computes significantly better solutions while at the same time being faster in many cases.
An important property of \AlgName{HeiStream} is that its running time does not depend on the number of blocks, while the previous state-of-the-art streaming partitioning algorithms have running time almost proportional to this number of blocks. 

We proposed \AlgName{Online Recursive Multi-Section}, a streaming algorithm to compute hierarchical partitionings of graphs. 
In terms of time complexity our algorithm outperforms previous state-of-the-art one-pass streaming algorithms for graph partitioning while also implicitly optimizing process mapping objectives.
To the best of our knowledge, this is the first streaming algorithm for the process mapping problem. %
We present extensive experimental results in which compare it against the previous state-of-the-art algorithms for streaming non-buffered one-pass graph partitioning, and show that we can speed it up even further with a multi-threaded parallelization.
Experiments show that our algorithm is up to two orders of magnitude faster than the previous state of-the-art while producing solutions with better communication cost and slightly worse edge-cut.
Moreover, our algorithm is only~three times slower than \AlgName{Hashing} when running on 32 threads while computing significantly better results.

We proposed \AlgName{FREIGHT}, a streaming algorithm for hypergraph partitioning. 
Our algorithm leverages an optimized data structure, resulting in linear running time with respect to pin-count and linear memory consumption in relation to the numbers of nets and blocks.
The results of our extensive experimentation demonstrate that the running time of \AlgName{FREIGHT} is competitive with the \AlgName{Hashing}  algorithm, with a maximum difference of a factor of four observed in three fourths of the instances.
Importantly, our findings indicate that \AlgName{FREIGHT} consistently outperforms all existing (buffered) streaming algorithms and even the in-memory algorithm \AlgName{HYPE}, with regards to both cut-net and connectivity measures. 
This underscores the significance of our proposed algorithm as a highly efficient and effective solution for hypergraph partitioning in the context of large-scale and \hbox{dynamic data processing.}

% \vfill\pagebreak

We presented an unpublished experimental comparison of our three streaming algorithms.
Although each of them is designed for a specific problem, they are all capable of solving the graph partitioning problem, so we compare them for this problem.
In our evaluation \AlgName{HeiStream} produces the lowest edge-cut while consuming the most runtime.
Our results show that, in terms of graph partitioning, \AlgName{FREIGHT} and \AlgName{Online Recursive Multi-Section} produce comparable solution quality but \AlgName{FREIGHT} becomes considerably faster than \AlgName{Online} \AlgName{Recursive} \AlgName{Multi-Section} for larger \hbox{numbers of blocks}.

In future work, we intend to parallelize \AlgName{HeiStream} and extend it to solve other (hyper)graph decomposition problems such as hypergraph partitioning and process mapping.
While \AlgName{Online Recursive Multi-Section} is already useful for a wide-range of applications that need (hierarchical) partitions very fast, in future work, we plan to parallelize the algorithm in the distributed memory model and want to port it to GPUs.
Moreover, we intend to parallelize \AlgName{FREIGHT} and try to make it even more effective by adapting its mathematical formulation to more accurately encode the connectivity \hbox{gains involved}.

\chapter{Local Algorithms}
\label{chap:Local Algorithms}

In this chapter, we propose our contributions in the field of local graph decomposition.
Our contributions comprise the detailed design of two algorithms, namely, \AlgName{LMCHGP} and \AlgName{SOCIAL}, for solving the local motif clustering problem.
\AlgName{LMCHGP} constructs a (hyper)graph model centered around the seed node in such a way that optimizing for conductivity in the model is equivalent to optimizing for motif conductance in the original graph.
The model is then partitioned utilizing high-quality (hyper)graph partitioning techniques available in the literature.
On the other hand, \AlgName{SOCIAL} builds the same hypergraph model around the seed node but converts it into a flow network that guarantees exact quality improvement.
The flow network is then recursively solved using a standard max-flow algorithm until a locally optimal cluster is obtained.
In our experiments, both \AlgName{LMCHGP} and \AlgName{SOCIAL} outperform the state-of-the-art by producing clusters with lower motif conductance value while being significantly faster, up to multiple orders of magnitude.

\paragraph*{References.}
This chapter is based on~\cite{LocMotifClusHyperGraphPartition}~and~\cite{LocMotifClusHyperGraphPartitionExtAbs}, which is joint work with Adil Chhabra and Christian Schulz, and on the technichal report~\cite{LocMotifClusMaxFlows}, which is also joint work with Adil Chhabra and Christian Schulz and is currently under submission.

\section{Local Motif Clustering}
\label{chap:lmcvhgp_Local Motif Clustering via (Hyper)Graph Partitioning}
\label{sec:lmcvhgp_Local Motif Clustering via (Hyper)Graph Partitioning}

Within this section, we propose two novel algorithms to solve the local motif clustering problem using sophisticated combinatorial algorithms. 

As a first contribution, we propose \AlgName{LMCHGP}~(Local Motif Clustering via (Hyper)Graph Partitioning).
This is an algorithm based on (hyper)graph partitionig which has two versions: one based on a graph model and the other one based on hypergraph model.
To begin, our algorithm constructs a (hyper)graph model that represents the distribution of motifs around the seed node on the original graph.
While the graph model accurately captures motifs of size at most three, the hypergraph model accommodates arbitrary motifs and is designed to minimize motif conductance in the original network.
We then partition the (hyper)graph model using a powerful multi-level hypergraph or graph partitioner to directly minimize the motif conductance of the corresponding partition in the original graph.

As a second contribution, we propose an algorithm named \AlgName{SOCIAL} (\emph{faSter mOtif Clustering vIa mAximum fLows}).
This algorithm optimizes for motif conductance by combining the strongly local hypergraph model proposed in Chapter~\ref{chap:lmcvhgp_Local Motif Clustering via (Hyper)Graph Partitioning} with an adapted version of the fast and effective algorithm \emph{max-flow quotient-cut improvement}(\AlgName{MQI})\cite{mqipaper2004}.
Using the same hypergraph model as our previous algorithm, we construct a flow model in which certain cuts correspond one-to-one with subsets of the initial cluster that include the seed node and have lower motif conductance than that of the whole cluster.
We then utilize a push-relabel algorithm to find such a cut and recursively repeat the process, or prove that the current cluster is optimal among all its sub-clusters containing the seed node.

We extensively evaluate our proposed algorithms \AlgName{LMCHGP} and \AlgName{SOCIAL} against the state-of-the-art.
In experiments involving triangle motifs, \AlgName{LMCHGP} and \AlgName{SOCIAL} compute clusters with a motif conductance value lower than the state-of-the-art, while also being up to multiple orders of magnitude faster.

\subsection{LMCHGP}
\label{sec:lmcvhgp_Local Motif Clustering via (Hyper)Graph Partitioning}
\label{sec:lmcvhgp_LMCHGP}

\subsubsection{Overall Strategy}
\label{subsec:lmcvhgp_Overall Strategy}

Given a graph $G=(V,E)$, a seed node $u$, and a motif $\mu$, our strategy for local clustering is based on four consecutive phases.
First, we select a set $S \subseteq V$ containing~$u$ and close-by nodes.
From now on, we refer to this set $S$ as a \emph{ball around} $u$.
Second, we enumerate the collection $M$ of occurrences of the motif $\mu$ which contain at least one node in~$S$.
Next, we build a graph or a hypergraph model $H_\mu$ depending on the configuration of the algorithm.
In particular, we design $H_\mu$ in such a way that the motif-conductance metric in $G$ can be computed directly in $H_\mu$.
Then, we partition this model into two blocks using a high-quality (hyper)graph partitioning algorithm.
The obtained partition of $H_\mu$ is directly translated back to $G$ as a local cluster around the seed node.
Figure~\ref{fig:lmcvhgp_overall_algorithm} provides a comprehensive illustration of the consecutive phases of our algorithm.
Note that (hyper)graph partitioning algorithms do not optimize for traditional clustering objectives such as conductance.
Instead, they aim at minimizing the edge-cut (resp. cut-net) value while respecting a hard balancing constraint.
To improve for the correct objective, we repeat the partitioning phase $\beta$ times with different imbalance constraints and pick the clustering with best motif conductance. 
Especially for the graph-based version of our model $H_\mu$, we subsequently run a special label propagation for each of these $\beta$ iterations in order to increase the chances of reaching a local minimum motif conductance.
Moreover, the first three phases of our strategy are repeated $\alpha$ times with different balls around the seed node in order to better explore the vicinity of the seed node in the original graph.
Our overall strategy including the mentioned repetitions is outlined in Algorithm~\ref{alg:lmcvhgp_overall_strategy}.

\begin{figure}[t]
	\centering
	\includegraphics[width=0.83\linewidth]{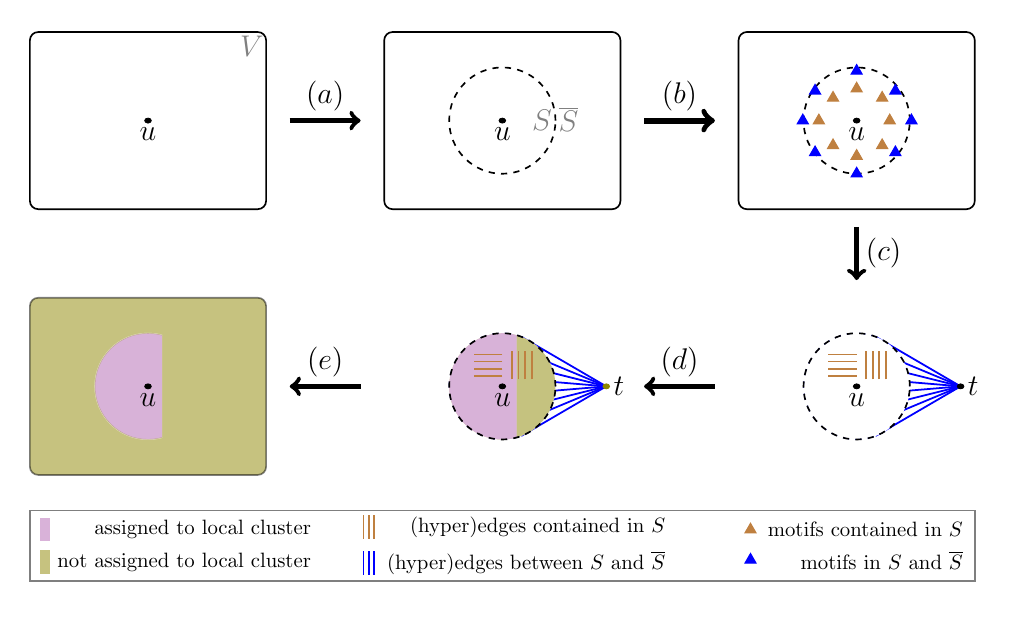}
	\caption{Illustration of the phases of our algorithm. (a)~Given a seed node $u$ and a graph $G$, a ball $S$ around $u$ is selected. (b)~Motif occurrences of $\mu$ with at least a node in $S$ are enumerated. (c)~The (hyper)graph model $H_\mu$ is built by converting motifs into (hyper)edges and contracting $\overline{S}$ into a node~$r$. (d)~The model is partitioned into two blocks using a multi-level (hyper)graph partitioner. (e)~The partition of $H_\mu$ is converted in a local cluster around the seed node in $G$.}
	\label{fig:lmcvhgp_overall_algorithm}
\end{figure}

\begin{algorithm}[t]
	\hspace*{-0.0cm} \textbf{Input} graph $G=(V,E)$; seed node $u \in V$; motif $\mu$ \\
	\hspace*{-.0cm} \textbf{Output} cluster $C^* \subseteq V$ 
	\begin{algorithmic}[1]  
		\State $C^* \leftarrow \emptyset$
		\For{$i=1,\ldots,\alpha$}
			\State Select ball $S$ around $u$
			\State $M \leftarrow$ Enumerate motifs in $S$
			\State Build (hyper)graph model $H_\mu$ based on $S$ and $M$%
		\For{$j=1,\ldots,\beta$}
		\State Partition model $H_\mu$ into $(C, \overline{C})$, where $u \in C$
		\If{$C^* = \emptyset \lor \phi_\mu(C) <  \phi_\mu(C^*)$}
		\State $C^* \leftarrow C$
		\EndIf
		\EndFor
		\EndFor
		\State Convert $C^*$ into a local motif cluster in $G$
	\end{algorithmic}
	\caption{Local Motif Clustering via (Hyper)Graph Partitioning}
	\label{alg:lmcvhgp_overall_strategy}
\end{algorithm}

\subsubsection{Ball around Seed Node}
\label{subsec:lmcvhgp_Ball around Seed}

Our approach to select $S$ is a fixed-depth breadth-first search (BFS) rooted on $u$.
More specifically, we compute the first $\ell$ layers of the BFS tree rooted on $u$, then we include all its nodes in~$S$.
For each of the $\alpha$ repetitions of our overall algorithm, we use different amounts $\ell$ of layers for a better algorithm exploration. 
Two exceptional cases are handled by our algorithm, namely a ball~$S$ that is either too small or disconnected from $\overline{S}$.
We avoid the first exceptional case by ensuring that~$S$ contains $100$ or more nodes in at least one repetition of our overall algorithm.
More specifically, in case this condition is not automatically met, then we accomplish it in the last repetition by growing additional layers in our partial BFS tree while it contains fewer than $100$ nodes.
The number $100$ is based on the findings of 
\citet{leskovec2009community}, which show that most well characterized communities from real-world graphs have a relatively small size, in the order of magnitude of $100$ nodes.
If the second exceptional case happens, it means that the whole BFS tree rooted on the seed node has at most $\ell$ layers.
In this case, we simply stop the algorithm and return the entire ball $S$, which corresponds to an optimal community with motif conductance $0$ provided that there is at least one motif in $S$.

The approach described above makes sure that there is a reasonable chance that a well characterized community containing $u$ is contained in $S$, since it has at least $100$ nodes~\cite{leskovec2009community} which are all very close to $u$.
This likelihood is further increased due to the multiple repetitions of our overall algorithm using balls $S$ of different sizes.
Our BFS approach to select $S$ can be executed in time linear on the subgraph induced in~$G$ by the closed neighborhood $N[S]$ of $S$.
After selecting~$S$, the further phases of our algorithm do not deal with the whole graph, but exclusively with~$S$, its edges, and its motif occurrences.
As a consequence, our algorithm operates on a much smaller problem dimension than the size of input graph~$G$, hence its running time corresponds to the same smaller problem dimension.
The number~$\alpha$ of repetitions as well as the amount $\ell$ of layers used in each repetition are tuning~parameters.

\subsubsection{Motif Enumeration}
\label{subsec:lmcvhgp_Motif Enumeration}

We now describe and discuss the motif-enumeration phase of our algorithm.
We optimally solve it for the triangle motif in time roughly linear on the size of the subgraph induced in~$G$ by the closed neighborhood $N[S]$ of $S$.
Moreover, we show that there are good heuristics approaches to enumerate higher-order motifs~efficiently.
The general problem of finding out if a given motif is a subgraph of some graph is NP-hard~\cite{read1977graph}, hence the enumeration of all such motifs is also NP-hard.
Nevertheless, some simpler motifs can be enumerated in polynomial time, which is the case for the triangle motif~\cite{ortmann2014triangle}.
Triangles, which can be defined as cycles or cliques of length three, have a wide variety of relevant applications on network analysis and clustering~\cite{holme2002growing,batagelj2007short,prat2012shaping}.
Without loss of generality, we specifically focus on the triangle motif within our algorithm. 
Nevertheless, note that many other small motifs can also be polynomially enumerated, such as small (directed and undirected) paths and cycles.
Moreover, our overall algorithm can also be adapted for more arbitrary motifs if we relax the optimality of the enumeration, which can be done using efficient heuristics such as the one proposed by 
\citet{kimmig2017shared}.
A simple and exact algorithm for triangle enumeration was proposed by 
\citet{chiba1985arboricity}.
Roughly speaking, this algorithm works by intersecting the neighborhoods of adjacent nodes.
For each node $v$, the algorithm starts by marking its neighbors with degree smaller than or equal to its own degree.
For each of these specific neighbors of $v$, it then scans its neighborhood and enumerates new triangles as soon as marked nodes are found.
The running time of this algorithm is $O(ma) = O(m^{\frac{3}{2}})$, where $a$ is the arboricity of the graph.
For the motif-enumeration phase of our algorithm, we apply the algorithm of 
\citet{chiba1985arboricity} only on the subgraph induced in $G$ by $N[S]$.
This is enough to find all triangles containing at least one node in $S$, as exemplified by transformation~(a) in figures~\ref{fig:lmcvhgp_model_construction}~and~\ref{fig:lmcvhgp_model_construction_graph}. 
Assuming a constant-bounded arboricity, the overall cost of our motif-enumeration phase for triangles is $O\big(|N[S] \times N[S])\cap E|\big)$. 

\subsubsection{Hypergraph Model}
\label{subsec:lmcvhgp_Hypergraph model}
In this section, we conceptually describe the hypergraph version of our model~$H_\mu$ and explain how to build it.
We show that it can be constructed in time linear on the amount of nodes in~$S$ and motifs in~$M$.
We also discuss advantages and limitations of our hypergraph model approach.

\begin{figure*}[t!]
	\centering
	\includegraphics[width=.7\linewidth]{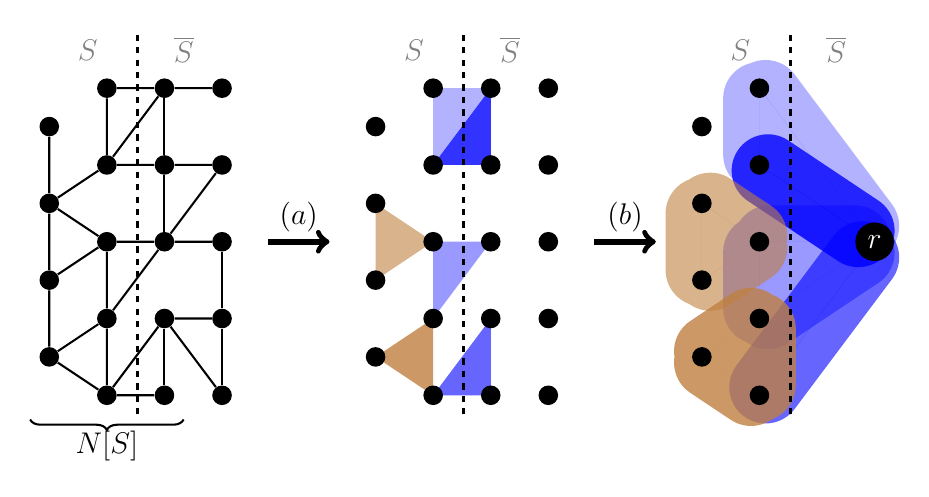}
	\caption{Example of motif-enumeration and model-construction phases of our algorithm for triangle motif and the hypergraph model. In the left, the nodes of $G$ are split into sets $S$ and $\overline{S}$. In the center, motif occurrences containing nodes in $S$ are enumerated. In the right,~$H_\mu$ is built by converting motifs into nets and contracting $\overline{S}$ into a node~$r$.}
	\label{fig:lmcvhgp_model_construction}
	\label{fig:social_model_construction}
\end{figure*}

Our hypergraph model~$H_\mu$ is built in two conceptual operations.
First, define a hypergraph containing $V$ as nodes and a set~$\mathcal{E}$ of nets such that, for each motif in~$M$, $\mathcal{E}$ has a net with pins equal to the endpoints of this motif. 
Then, we contract together all nodes in~$\overline{S}$ into a single node~$r$ and substitute parallel nets by a single net whose weight is equal to the summed weights of the removed parallel nets.
More formally, we define the hypergraph version of our model as $H_\mu = (S \cup \{r\},\mathcal{E})$ where the set~$\mathcal{E}$ of nets contains one net~$e$ associated with each motif occurrence $G'=(V',E') \in M$ such that $e = V'$ if $V'\subseteq S$, and $e = V' \cap S \cup \{r\}$ otherwise.
In the former case the net has weight~$1$, in the latter case the net has weight equal to the amount of motif occurrences in $M$ represented by it.
In a typical hypergraph contraction, the weight of $r$ would be $\mathfrc{c}(r) = c(\overline{S})$.
Nevertheless, we opt to make the nodes of $H_\mu$ unweighted, which is convenient for the purpose of our overall algorithm as will become clear later.
In practice, the hypergraph version of $H_\mu$ can be built by instantiating the nodes in $S \cup \{r\}$ and the nets in $\mathcal{E}$.
Assuming that the number of nodes in $\mu$ is a constant, our model is built in time $O(|S|+|M|)$ and uses memory $O(|S|+|M|)$.
The construction of~$H_\mu$ is illustrated in transformation~(c) of Figure~\ref{fig:lmcvhgp_overall_algorithm} and demonstrated for a particular example in transformation~(b)~of~Figure~\ref{fig:lmcvhgp_model_construction}.

Observing the relationship between $G$ and $H_\mu$, we can distinguish three groups of components.
The first group comprises nodes in $S$ and motifs with all endpoints in $S$, all of which are represented in $H_\mu$ without any contraction as nodes and nets.
The second group consists of nodes in $\overline{S}$ and motifs with all endpoints in $\overline{S}$, which are compactly represented in $H_\mu$ as the contracted node~$r$.
The third group comprises motifs with nodes in both $S$ and $\overline{S}$, all of which are abstractly represented in $H_\mu$ as nets containing individual pins in $S$ as well as the pin~$r$.
Summing up, our hypergraph model is a concise representation of the whole graph $G$ where relevant information for local motif clustering is emphasized in two perspectives:
Edges are omitted while motifs are made explicit and global information is abstracted while local information is preserved in detail.
Theorem~\ref{theo:lmcvhgp_cut_equivalence} shows that the cut-net of a partition of $H_\mu$ directly corresponds to the motif-cut of an equivalent partition in $G$ if our motif enumeration step is exact.
On top of that, Theorem~\ref{theo:lmcvhgp_conductance_equivalence} shows that the motif conductance of this equivalent partition of $G$ can be directly computed from $H_\mu$ assuming $d_\mu(S) \leq d_\mu(\overline{S})$.
Assuming $d_\mu(S) \leq d_\mu(\overline{S})$ is fair since $S$ is ideally much smaller than~$\overline{S}$. 
Enumerating the motifs in $\overline{S}$ is not reasonable for a local clustering algorithm, but we did verify that our assumption holds during all our experiments.

\begin{theorem}
	Every $k$-way partition $P$ of the hypergraph model $H_\mu$ corresponds to a distinct $k$-way partition $P^\prime$ of $G$, such that the cut-net of $P$ is equivalent to the motif-cut of $P^\prime$, assuming an exact motif enumeration step.
	\label{theo:lmcvhgp_cut_equivalence}
\end{theorem}

% \vfill\pagebreak

\begin{proof}
	For simplicity, we prove the claim assuming that parallel nets are not substituted by a single net whose weight is equal to their summed weights.
	This proof directly extends to our model since the contribution of a contracted cut net to the overall cut-net equals the contribution of the parallel nets represented by it.
	Due to the design of our hypergraph model $H_\mu$, there is a direct correspondence between its nodes and the nodes of $G$.
	Hence, any partition $P$ of $H_\mu$ corresponds to a partition $P^\prime$ of $G$ where corresponding nodes are simply assigned to the same blocks.
	Since $\overline{S}$ is represented by the single node~$r$ in $H_\mu$, no motif occurrence totally contained in $\overline{S}$ can be cut in $P^\prime$.
	All the remaining motif occurrences in $G$ can be potentially cut in $P^\prime$, but these motif occurrences are bijectively associated with the nets of $H_\mu$ with a direct correspondence between motif endpoints in $G$ and net pins in $H_\mu$.
	As a consequence, a motif occurrence of $G$ is cut in $P^\prime$ if, and only if, the corresponding net in $H_\mu$ is cut in~$P$.
\end{proof}

\begin{theorem}
	Given a $2$-way partition $P=(C,\overline{C})$ of our hypergraph model $H_\mu$ with~$r \in \overline{C}$, the motif conductance $\phi_\mu(C^\prime)$ of the corresponding $2$-way partition $P^\prime=(C^\prime,\overline{C^\prime})$ of $G$ is the ratio of the cut-net of~$P$ to $\mathfrc{d}_{\mathfrc{w}}(C)$, assuming an exact motif enumeration step and $d_\mu(S) \leq d_\mu(\overline{S})$.
	\label{theo:lmcvhgp_conductance_equivalence}
	\label{theo:social_conductance_equivalence}
\end{theorem}

\begin{proof}
	From Theorem~\ref{theo:lmcvhgp_cut_equivalence}, the motif-cut of $P^\prime$ can be substituted by the cut-net of $P$ in the numerator of the definition of $\phi_\mu(C^\prime)$.  
	To complete the proof, it suffices to show that the denominator of $\phi_\mu(C^\prime)$, namely $min(d_\mu(C^\prime),d_\mu(\overline{C^\prime}))$, is equal to $\mathfrc{d}_{\mathfrc{w}}(C)$.
	Due to the design of $H_\mu$, the values of $d_\mu(C^\prime)$ and $\mathfrc{d}_{\mathfrc{w}}(C)$ are identical.
	Our assumption~$r \in \overline{C}$ leads to $\overline{S} \subseteq \overline{C^\prime}$ and $C^\prime \subseteq S$, which respectively imply $d_\mu(\overline{S}) \leq d_\mu(\overline{C^\prime})$ and $d_\mu(C^\prime) \leq d_\mu(S)$. 
	Since $d_\mu(S) \leq d_\mu(\overline{S})$, hence $\mathfrc{d}_{\mathfrc{w}}(C) = d_\mu(C^\prime) \leq d_\mu(S) \leq d_\mu(\overline{S}) \leq d_\mu(\overline{C^\prime})$.
\end{proof}

\subsubsection{Graph Model}
\label{subsec:lmcvhgp_Graph model}

In this section, we describe the graph version of our model~$H_\mu$ and explain how to build it.
Similarly to the hypergraph version of this model, our graph model can be built in time linear on the amount of nodes in~$S$ and motifs in~$M$.
We discuss advantages and limitations of the graph model in comparison to the hypergraph model.

\begin{figure*}[t!]
	\centering
	\includegraphics[width=.7\linewidth]{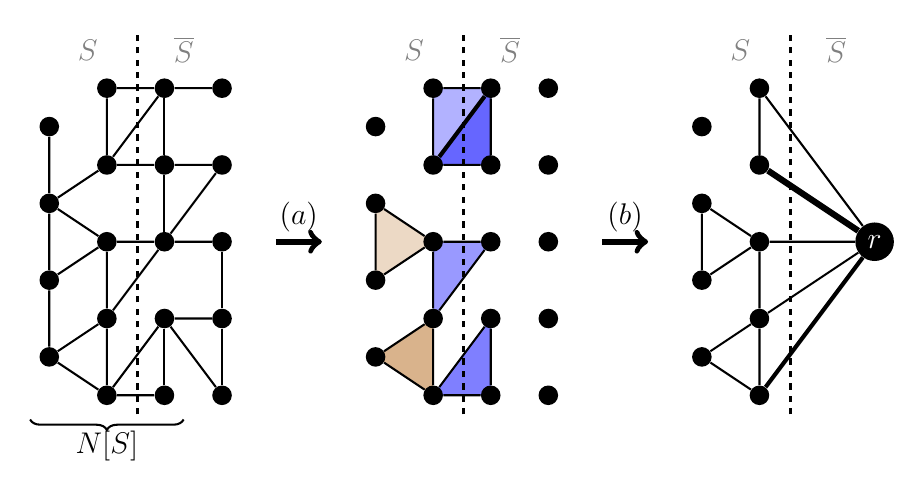}
	\caption{Example of motif-enumeration and model-construction phases of our algorithm for triangle motif and the graph model. In the left, $G$ is shown with its nodes split into the sets $S$ and $\overline{S}$. In the center, the motif occurrences containing at least a node in $S$ are enumerated and the weight of an edge equals the number triangles it touches. In the right, the model $H_\mu$ is built by contracting $\overline{S}$ into a single node~$r$. The weight of an edge is represented by its thickness.}
	\label{fig:lmcvhgp_model_construction_graph}
\end{figure*}

The graph version of our model~$H_\mu$ is built in two conceptual operations.
The first one consists of obtaining the weighted graph $W$ proposed by
\citet{benson2016higher} and used in the state-of-the-art algorithm MAPPR~\cite{yin2017local}.
The graph~$W$ contains $V$ as nodes and a set of edges such that two conditions are met: 
(i)~there is an edge between a pair of nodes if, and only if, both nodes belong at the same time to at least a motif in $G$; 
(ii)~the weight of an edge is equal to the number of motif occurrences containing both its endpoints.
The second operation consists of contracting all nodes in~$\overline{S}$ into a single node~$r$ and substitute parallel edges by a single edge whose weight is equal to the summed weights of the removed edges.
The construction of the graph version of~$H_\mu$ is illustrated in transformation~(c) of Figure~\ref{fig:lmcvhgp_overall_algorithm} and demonstrated for a particular example in transformation~(b) of Figure~\ref{fig:lmcvhgp_model_construction_graph}.
More formally, we define the model as $H_\mu=(S \cup \{r\},{E}_\mu)$ where~${E}_\mu$ contains an edge~$e$ for each pair of nodes sharing a motif $G'=(V',E') \in M$ provided that at least one of its endpoints is contained in $S$. 
The weight of $r$ is set to ${c}(r) = c(\overline{S})$.
Similarly to our approach with the hypergraph version of the model, we opt to make the nodes of the graph model unweighted in our experiments. 
In practice, the graph version of $H_\mu$ can be built by instantiating the nodes in $S \cup \{r\}$ and directly the computing the edges in ${E}_\mu$ and their weights.
Assuming that the number of nodes in $\mu$ is a constant, our graph model is built in time $O(|S|+|M|)$ and uses memory $O(|S|+|E_\mu|)$.
Especially for the triangle motif, the memory requirement of the graph model is $O(|S|+|N[S] \times N[S] \cap E|)$, which is linear on $n$ and $m$ in the worst~case.

We reproduce here Theorem~\ref{theo:lmcvhgp_conductance_equivalence_graph_complete} by \citet{yin2017local}, which shows that conductance in the weighted graph $W$ is equivalent to motif conductance in $G$ as long as the motif has at most $3$ nodes.
Based on this result, Theorem~\ref{theo:lmcvhgp_conductance_equivalence_graph} shows that we can compute motif conductance directly from our graph model $H_\mu$ if $d_\mu(S) \leq d_\mu(\overline{S})$.
As we mentioned, assuming $d_\mu(S) \leq d_\mu(\overline{S})$ is fair since $S$ is ideally much smaller than~$\overline{S}$. 
Recall that the hypergraph version of our model~$H_\mu$ is flexible enough to represent any motif as a net such that Theorems~\ref{theo:lmcvhgp_cut_equivalence}~and~\ref{theo:lmcvhgp_conductance_equivalence} continue valid.
Although the graph version of our model can technically represent any motif, Theorem~\ref{theo:lmcvhgp_conductance_equivalence_graph} is only valid for motifs with at most three nodes, while other models only allow a heuristic computation of the motif conductance~\cite{benson2016higher}.
Nevertheless, the biggest drawback of the hypergraph-based approach is the need for storing up to $|M|$ nets, which costs $O(n^{3})$ in the worst case.
In contrast, the memory needed to store our graph model is $O(n^2)$ in the worst case and $O(n+m)$ specifically for the triangle motif.

\begin{theorem}[Theorem 4.1 by \citet{yin2017local}]
	Given a $2$-way partition $P^{\prime\prime}=(C^{\prime\prime},\overline{C^{\prime\prime}})$ of the weighted graph $W$, the motif conductance $\phi_\mu(C^\prime)$ of the corresponding $2$-way partition $P^\prime=(C^\prime,\overline{C^\prime})$ in $G$ is equal to the conductance $\phi(C^{\prime\prime})$ of $C^{\prime\prime}$ in $W$, assuming a motif with at most three nodes.
	\label{theo:lmcvhgp_conductance_equivalence_graph_complete}
\end{theorem}

\begin{theorem}
	Given a $2$-way partition $P=(C,\overline{C})$ of the graph version of model $H_\mu$ with~$r \in \overline{C}$, the motif conductance $\phi_\mu(C^\prime)$ of the corresponding $2$-way partition $P^\prime=(C^\prime,\overline{C^\prime})$ of $G$ is the ratio of the edge-cut of~$P$ to ${d}_\omega(C)$, assuming an exact motif enumeration step, $d_\mu(S) \leq d_{\mu}(\overline{S})$, and a motif with at most three nodes.
	\label{theo:lmcvhgp_conductance_equivalence_graph}
\end{theorem}

\begin{proof}
	From Theorem~\ref{theo:lmcvhgp_conductance_equivalence_graph_complete}, the motif conductance of any $2$-way partition of $G$ is equal to the conductance of the equivalent partition in $W$ assuming a motif with at most $3$ nodes. 
	Since we assume $d_\mu(S) \leq d_\mu(\overline{S})$, hence the conductance in $W$ of any community $C^{\prime\prime} \in S$ is equal to its edge-cut divided by the volume of $C^{\prime\prime}$.
	From the construction of our graph model $H_\mu$, the assumed community~$C$ has an equivalent community~$C^{\prime\prime}$ in $W$ with same edge-cut and same volume, which completes the proof.
\end{proof}

\subsubsection{Partitioning}
\label{subsec:lmcvhgp_Partitioning}

In this section, we describe the (hyper)graph partitioning phase of our local motif clustering algorithm.
We present the used (hyper)graph partitioning algorithms and discuss how we enforce feasibility of the found solution and maximize its quality.
Moreover, we provide remarks about the running time of our partitioning phase.

The partitioning phase of our algorithm consists of a $2$-way partitioning of $H_\mu$. 
When using the hypergraph model, the partition is computed by the multi-level hypergraph partitioner KaHyPar~\cite{schlag2016k}. 
When using the graph model, the partition is computed by the multi-level graph partitioner KaHIP~\cite{kaHIPHomePage}.
These partitioners contain sophisticated algorithms to produce low-cut partitions of (hyper)graphs efficiently.
As already shown,
any $2$-way partition of $H_\mu$ automatically corresponds to a community in $G$.
Nevertheless, our aim is to obtain a \emph{consistent} partition of $H_\mu$, which we define as a partition where the seed node~$u$ and the contracted node~$r$ are in different blocks.
This consistent partition ultimately corresponds to a local community in~$G$ which contains the seed node~$u$ and is completely contained in the ball~$S$.
This \emph{consistency} criterion is important since the nodes in $\overline{S}$ have not been explored by our algorithm and are farther from $u$ than the nodes in~$S$.
Note that the $2$-way partition illustrated in Figure~\ref{fig:lmcvhgp_overall_algorithm} generates a consistent local community according to our definition of it.
While KaHyPar allows partitioning hypergraphs with fixed nodes, KaHIP does not offer such functionality.
Nevertheless, we ensure block feasibility for both versions of our algorithm by simply assigning the seed node to the block that does not contain $r$ after the partition is computed (before computing the motif conductance).
Although simplistic, this approach has not affected solution quality considerably for the hypergraph-based version of our algorithm in preliminary comparisons against a fixed nodes-based approach.

Besides corresponding to a consistent local clustering on the original graph, our solution should have as low a motif conductance as possible.
However, KaHyPar and KaHIP {are randomized algorithms and} do not directly optimize for this objective, but rather minimize the cut value while enforcing a hard balancing constraint.
To improve our results, we explore different combinations of edge-cut (resp. cut-net) and imbalance by repeating the partitioning procedure $\beta$~times with random balancing constraints for each built (hyper)graph model, where $\beta$ is a tuning parameter.
For each obtained partition, our algorithm computes the motif conductance of the corresponding local cluster as shown in Theorem~\ref{theo:lmcvhgp_conductance_equivalence_graph} (resp. Theorem~\ref{theo:lmcvhgp_conductance_equivalence}) and keeps the partition with the best motif~conductance.

KaHyPar and KaHIP run in time close to linear in practice.
Nevertheless, the partitioning phase can be the dominating operation of our overall local clustering algorithm. 
This is the case because KaHyPar and KaHIP use sophisticated algorithms and data structures in order to minimize the cut value, which increases constant factors in the algorithm running time complexity.
KaHyPar can be especially much slower than KaHIP since the number of nets in the hypergraph model can be considerably larger than the number of edges in the graph model.
{We can make the hypergraph version of our algorithm faster by using Mt-KaHyPar~\cite{gottesburen2021scalable} instead of KaHyPar.
	Mt-KaHyPar obtains significant speedups by using allowing shared-memory parallel execution.}{In both cases, running time can be improved further by using parallelized tools such as Mt-KaHyPar~\cite{gottesburen2021scalable} or \AlgName{KaMinPar}~\cite{DBLP:conf/esa/GottesburenH00S21}. However, parallelization is not the focus~of~this~thesis.}

\textbf{Local Search.}
\label{subsec:lmcvhgp_Local Search}
We implement a local search inspired by \emph{label propagation} \cite{labelpropagationclustering} for the graph model-based version of our algorithm. 
This local search is designed to optimize for the correct objective, i.e., motif conductance.
We apply it directly on each $2$-way partition generated by KaHIP in order to increase the chances of reaching a local minimum motif conductance.
The local search algorithm works in rounds.
In each round, it visits all nodes of $H_\mu$ in a random order, starting with the labels being the current assignment of nodes to blocks.
When a node $v$ is visited, it is moved to the opposite block if this movement causes a decrease in the motif conductance of the clustering.
Movements of nodes with zero-gain can occasionally occur with $50\%$ probability.
We ensure that the seed node $u$ and the contracted node $r$ continue in opposite blocks by simply skipping them.
We stop the local search when a local optimum is reached or after at most~$\ell$ rounds, where $\ell$ is a tuning parameter.

\subsection{SOCIAL}
\label{sec:social_Faster Local Motif Clustering via Maximum Flows}
\label{sec:social_SOCIAL}

\subsubsection{Overall Strategy}
\label{subsec:social_Overall Strategy}

Given a graph $G=(V,E)$, a seed node $u$, and a motif $\mu$, our strategy for local clustering is based on the following phases.
First, we select a set $S \subseteq V$ containing~$u$ and close-by nodes.
As in Section~\ref{subsec:lmcvhgp_Overall Strategy}, we refer to this set $S$ as a \emph{ball around}~$u$.
Second, we enumerate the collection $M$ of occurrences of the motif $\mu$ which contain at least one node in~$S$.
Third, we build a hypergraph model $H_\mu$ in such a way that the motif-conductance of any cluster~$C \subseteq S$ in~$G$ can be computed directly in~$H_\mu$.
Fourth, we set~$C_0=S$ as our initial cluster and use it to build our \AlgName{MQI}-based~\cite{mqipaper2004} flow model~$G_f$ from the hypergraph model~$H_\mu$.
Fifth, we use~$G_f$ to either find a new cluster~$C \subset C_0$ containing~$u$ with strictly smaller motif conductance than~$C_0$ or prove that such cluster does not exist.
While~$C \subset C_0$ is found, we take it as our new initial cluster, rebuild~$G_f$, and repeat the previous~phase.
When eventually no such strict sub-set is found, the best obtained cluster is directly translated back to $G$ as a local cluster around the seed node.
Figure~\ref{fig:social_overall_algorithm} provides a comprehensive illustration of the consecutive phases of \AlgName{SOCIAL}.
Note that there is no guarantee of finding the best overall cluster including~$u$ strictly contained in~$S$.
Instead, we find a succession of clusters with strictly decreasing cardinality and motif conductance until a local optimum~is~reached.
To better explore the vicinity of~$u$~in~$G$ and overcome the fact we only find clusters inside~$S$, we repeat the overall strategy~$\alpha$~times with distinct balls~$S$.
Our overall algorithm including the mentioned repetitions is outlined in Algorithm~\ref{alg:social_overall_strategy}.

\begin{figure*}[t]
	\centering
	\includegraphics[width=1.0\linewidth]{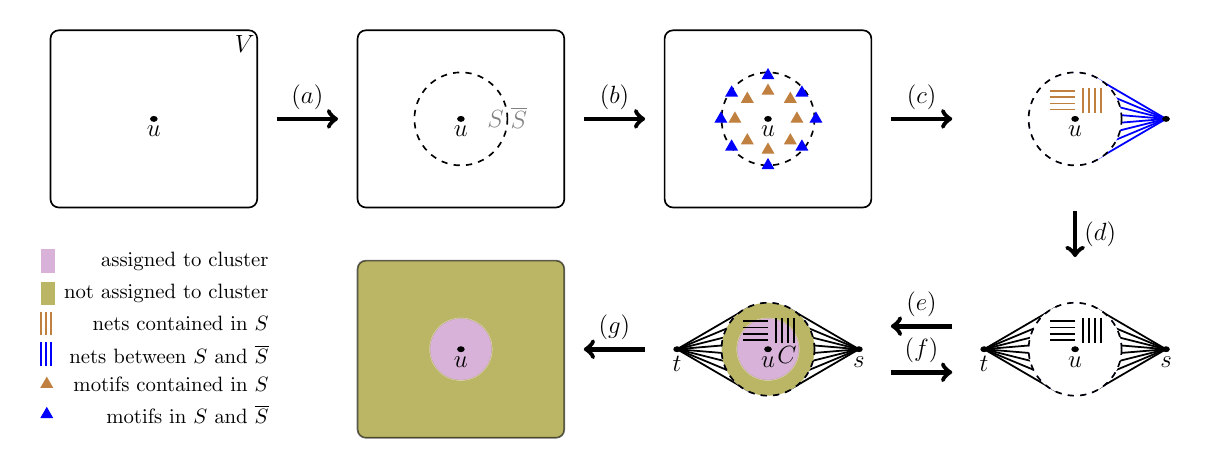}
	\caption{Illustration of the phases of \AlgName{SOCIAL}. (a)~Given a seed node $u$ and a graph $G$, a ball $S$ around $u$ is selected. (b)~Motif occurrences of $\mu$ with at least a node in $S$ are enumerated. (c)~The hypergraph model $H_\mu$ is built by converting motifs into nets and contracting $\overline{S}$ into a single node. The ball~$S$ is taken as the initial cluster~$C_0$. (d)~The flow model~$G_f$ is built based on~$C_0$~in~$H_\mu$. (e)~A cluster~$C \subseteq C_0$ containing~$u$ is found using maximum flows. (f)~While~$C \subset C_0$, the model~$G_f$ is rebuilt based on~$C$, which is taken as the initial cluster~$C_0$. (g)~When eventually~$C = S$,~$C$ is converted in a local cluster around the seed~node~in~$G$.}
	\label{fig:social_overall_algorithm}
\end{figure*}

\begin{algorithm}[t]
	\hspace*{-0cm} \textbf{Input} graph $G=(V,E)$; seed node $u \in V$; motif $\mu$ \\
	\hspace*{-0cm} \textbf{Output} cluster $C^* \subseteq V$ 
	\begin{algorithmic}[1]  
		\State $C^* \leftarrow \emptyset$
		\For{$i=1,\ldots,\alpha$}
			\State Select ball $S$ around $u$
			\State $M \leftarrow$ Enumerate motifs in $S$			
			\State Build hypergraph model $H_\mu$ based on~$S$~and~$M$%
			\State $C \leftarrow S$
			\State\Do
				\State \hskip1.5em $C_0 \leftarrow C$
				\State \hskip1.5em Build flow model~$G_f$ based on~$C_0$~in~$H_\mu$
				\State \hskip1.5em Solve~$G_f$ to obtain cluster~$C \subseteq C_0$ including~$u$ %
			\State \MyWhile{$C \subset C_0$}
			\If{$C^* = \emptyset \lor \phi_\mu(C) <  \phi_\mu(C^*)$} 
				\State $C^* \leftarrow C$
			\EndIf
		\EndFor
		\State Convert $C^*$ into a local motif cluster in $G$
	\end{algorithmic}
	\caption{Local Motif Clustering via Maximum Flows}
	\label{alg:social_overall_strategy}
\end{algorithm}

\subsubsection{Hypergraph Model}
\label{subsec:social_Hypergraph Model}
The hypergraph model employed in \AlgName{SOCIAL} is constructed in a manner that is identical to the hypergraph model utilized in \AlgName{LMCHGP}.
Specifically, the construction process entails the creation of a ball~$S$ centered around the seed node, as detailed in Section~\ref{subsec:lmcvhgp_Ball around Seed}.
Subsequently, the collection~$M$ of occurrences of the motif~$\mu$ that are incident to this ball are enumerated according to the procedure outlined in Section~\ref{subsec:lmcvhgp_Motif Enumeration}.
The hypergraph model~$H_\mu$ is then constructed, as illustrated in Figure~\ref{fig:social_model_construction} and described in Section~\ref{subsec:lmcvhgp_Hypergraph model}, to ensure the \hbox{validity of Theorem~\ref{theo:social_conductance_equivalence}}.

\ifFull

\begin{proof}
	From Theorem~\ref{theo:social_cut_equivalence}, the motif-cut of $P^\prime$ can be substituted by the cut-net of $P$ in the numerator of the definition of $\phi_\mu(C^\prime)$.  
	To complete the proof, it suffices to show that the denominator of $\phi_\mu(C^\prime)$, namely $min(d_\mu(C^\prime),d_\mu(\overline{C^\prime}))$, is equal to $\mathcal{d}_{\mathfrc{w}}(C)$.
	Due to the design of $H_\mu$, the values of $d_\mu(C^\prime)$ and $\mathcal{d}_{\mathfrc{w}}(C)$ are identical.
	Our assumption~$r \in \overline{C}$ leads to $\overline{S} \subseteq \overline{C^\prime}$ and $C^\prime \subseteq S$, which respectively imply $d_\mu(\overline{S}) \leq d_\mu(\overline{C^\prime})$ and $d_\mu(C^\prime) \leq d_\mu(S)$. 
	Since $d_\mu(S) \leq d_\mu(\overline{S})$, hence $\mathcal{d}_{\mathfrc{w}}(C) = d_\mu(C^\prime) \leq d_\mu(S) \leq d_\mu(\overline{S}) \leq d_\mu(\overline{C^\prime})$.
\end{proof}

\fi

\subsubsection{Flow Model}
\label{subsec:social_Flow Model}

\begin{figure*}[t]
	\centering
	\includegraphics[width=.9\linewidth]{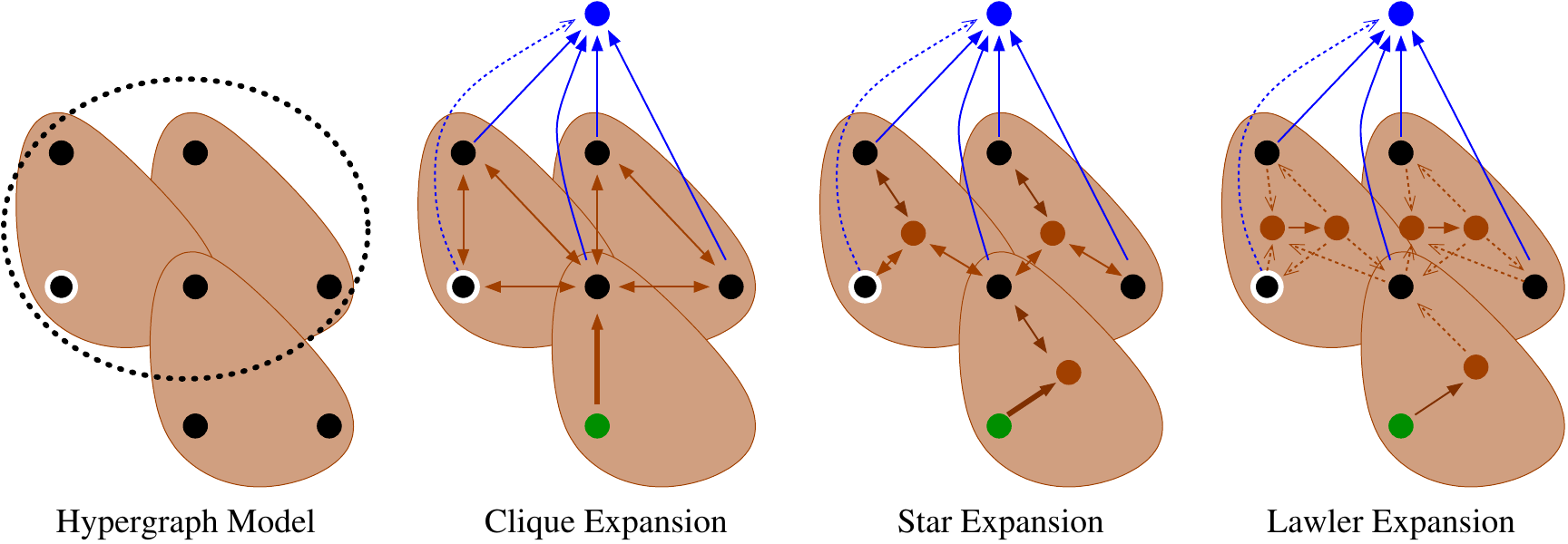}
	\caption{Flow model~$G_f$ given a hypergraph model~$H_\mu$ and an initial cluster~$C_0$.
		Nodes and nets of~$H_\mu$ are respectively represented by black circles and brown areas around them. 
		The seed node~$u$ is circled in white and the initial cluster~$C_0$ is surrounded by a dotted ellipse. 
		Auxiliary artificial nodes and edges used in each net-expansion are respectively represented by brown circles and arrows. 
		Bidirectional arrows represent pairs of edges in both directions.
		The seed node~$s$, the sink node~$t$, and the in-edges of~$t$ are respectively represented by a green circle, a blue circle, and blue arrows.
		Solid and dashed arrows respectively represent edges with finite and \hbox{infinite weight}.}
	\label{fig:social_flow_model}
\end{figure*}

In this section, we describe the process of constructing our \AlgName{MQI}-based flow model~$G_f$ using the hypergraph model~$H_\mu$ and an initial cluster~$C_0~\subseteq~S$ which contains the seed node~$u$. 
There are three possible implementations of~$G_f$ based on the three already explained techniques to represent hypergraphs using graphs, namely \emph{clique expansion}, \emph{star expansion}, and \emph{Lawler expansion} (see Figure~\ref{fig:social_hyperedge_expansion}).
We show a bijective correspondence between certain s-t cuts in $G_f$ and clusters~$C \subseteq C_0$~in~$G$ that include the seed node~$u$ and have motif conductance less~than~that~of~$C_0$.

We start by converting our hypergraph model~$H_\mu$ in a directed graph using the chosen net expansion technique.
Second, we find a corresponding cluster~$C_0^\prime$ for~$C_0$ in the created graph.
For the \emph{clique expansion}, $C_0^\prime~=~C_0$ since this transformation does not create artificial nodes.
For the \emph{star expansion},~$C_0^\prime$ consists of~$C_0$ and also the auxiliary artificial nodes connected to at least one node in~$C_0$.
For the \emph{Lawler expansion},~$C_0^\prime$ consists of~$C_0$, the auxiliary artificial nodes~$w_1$ having in-edges only from nodes in~$C_0$, and the auxiliary artificial nodes~$w_2$ having out-edges to at least one node in~$C_0$.
Third, we contract~$\overline{C_0^\prime}$ to a single \emph{source}~node~$s$ and then remove all its in-edges.
Fourth, we multiply the weight of all the remaining edges by $\mathfrc{d}_{\mathfrc{w}}(C_0)$, i.e., the weighted~volume of~$C_0$~in~$H_\mu$.
Fifth, we introduce a sink node~$t$ and include in-edges to it from each of the nodes~$v \in C_0\setminus\{u\}$, such that the weight of~$(v,t)$ is set to~$cut(C_{0})\mathfrc{d}_{\mathfrc{w}}(v)$, i.e., the cut-net of~$C_0$~in~$H_\mu$ multiplied by the weighted degree of~$v$~in~$H_\mu$.
Finally, we include an edge $(u,t)$ from the seed node to the sink and set its weight to infinity.
Our flow network model~$G_f$ is concluded by setting edge capacities to match edge weights.
Figure~\ref{fig:social_flow_model} shows the three possible configurations of our flow model~$G_f$ for a given hypergraph model~$H_f$ and an initial~cluster~$C_0$.

% \vfill\pagebreak

We now analyze the theoretical guarantees provided by the defined flow model~$G_f$.
Theorem~\ref{theo:social_flow_model_guarantee} shows that there is a set~${C\subset{C_0}}$ in~$G$ including the seed node~$u$ with motif conductance smaller than~that~of~$C_0$ if, and only if, the value of the maximum flow on $G_f$ is less than~$cut(C_{0})\mathfrc{d}_{\mathfrc{w}}(C_0)$, which is the weight of the trivial cut $(\{s\},V(G_f)\setminus\{s\})$.
In the proof, we show how such improved cluster~$C$ can be directly obtained from a maximum flow on~$G_f$.
Assumptions (a)~and~(b) in Theorem~\ref{theo:social_flow_model_guarantee} are the same used in Theorem~\ref{theo:social_conductance_equivalence}, which were previously shown to be reasonable in practice.
Note that the claim is only valid for motifs with three nodes for clique and star expansion models, while it is valid in general for the \hbox{Lawler expansion model}.

\begin{theorem}
	There is a set~${C\subset{C_0}}$ in~$G$ including the seed node~$u$ with motif conductance smaller than that of~$C_0$ if, and only if, the maximum flow on~$G_f$ is less than $cut(C_{0})\mathfrc{d}_{\mathfrc{w}}(C_0)$ under the following assumptions:
	\begin{enumerate}
		\item the motif enumeration phase is exact;
		\item $d_\mu(S) \leq d_\mu(\overline{S})$~in~$G$;
		\item in case~$G_f$ is based on clique expansion or star expansion, the motif~$\mu$ has three nodes;
	\end{enumerate}
	\label{theo:social_flow_model_guarantee}
\end{theorem}

% \vfill\pagebreak

\begin{proof}
	We start with the backward direction, i.e., if the maximum flow on~$G_f$ is less than $cut(C_{0})\mathfrc{d}_{\mathfrc{w}}(C_0)$, then there is a sub-set of~${C_0}$ in~$G$ including the seed node~$u$ with motif conductance smaller than that of~$C_0$.
	According to the Max-Flow Min-Cut Theorem~\cite{ford_fulkerson_1956}, the weight of the maximum s-t flow on a network equals the weight of its minimum s-t cut, hence it follows that there is an s-t cut~$(B_1,B_2)$~of~$G_f$ with weight smaller than $cut(C_{0})\mathfrc{d}_{\mathfrc{w}}(C_0)$.
	Without loss of generality, let~$s \in B_1$, $t \in B_2$, $C = C_0 \cap B_2$, and hence, by definition,~${C_0\setminus{C}=C_0\cap{B_1}}$.
	Necessarily~${u\in{C}}$, otherwise the edge $(u,t)$, which has infinite weight, would be cut.
	There are two kinds of edges from $B_1$~to~$B_2$.
	First, there are the edges~$(x,t)$, with~$x \in C_0 \setminus C$.
	By design, the total weight of these edges is given by~$cut(C_{0})\mathfrc{d}_{\mathfrc{w}}(C_0\setminus{C})$.
	The second kind consists of edges from~$B_1$~to~${B_2\setminus\{t\}}$.
	These edges vary based on the net-expansion technique used, but their total weight is~$cut(C)\mathfrc{d}_{\mathfrc{w}}{({C}_0)}$ by design under \hbox{the~given~assumptions}. 
	
	Now we show that the total weight of the edges from~$B_1$~to~${B_2\setminus\{t\}}$ is~$cut(C)\mathfrc{d}_{\mathfrc{w}}{({C}_0)}$ for the three net-expansion techniques.
	In the clique expansion under assumption~(c), each cut net~$e$ of~$C$ in~$H_\mu$ corresponds directly to two cut in-edges of~$B_2$ in~$G_f$.
	The weight of each of these edges is, by design, set to $\mathfrc{w}(e)\mathfrc{d}_{\mathfrc{w}}{({C}_0)} / 2$, so they add up to the specified total weight.
	In the star expansion under assumption~(c), for each cut net~$e$ of~$C$ in~$H_\mu$ there is a single cut in-edge of~$B_2$ in~$G_f$ that connects an auxiliary artificial node and a node from $C_0$. 
	The weight of this cut edge is set by design to $\mathfrc{w}(e)\mathfrc{d}_{\mathfrc{w}}{({C}_0)}$, so the total weight of these edges is as claimed.
	In the Lawler expansion, for each cut net~$e$ of~$C$ in~$H_\mu$ there is exactly one cut in-edge of~$B_2$ in~$G_f$ that connects two auxiliary artificial nodes, namely~$w_1 \in B_1$~and~$w_2 \in B_2$. 
	If this were not the case, there would be an edge from~$B_1$~to~$B_2$ with infinite weight. 
	The weight of this single cut edge is set by design to $\mathfrc{w}(e)\mathfrc{d}_{\mathfrc{w}}{({C}_0)}$, so the sum of the weights of the cut edges is as stated.
	By adding up the weights of the two kinds of edges that cross the cut~$(B_1,B_2)$ and verifying that their total weight is less than $cut(C_{0})\mathfrc{d}_{\mathfrc{w}}(C_0)$, we derive Equation~(\ref{eq:social_forward_direction}), which can further be simplified to Equation~(\ref{eq:social_forward_direction_final}).
	We conclude the proof of the backward direction by applying Theorem~\ref{theo:social_conductance_equivalence}, Equation~(\ref{eq:social_forward_direction_final}), and the assumptions~(a)~and~(b).
	
	\begin{equation}
		cut(C_{0})\mathfrc{d}_{\mathfrc{w}}(C_0\setminus{C})~+
		cut(C)\mathfrc{d}_{\mathfrc{w}}{({C}_0)}
		< cut(C_{0})\mathfrc{d}_{\mathfrc{w}}(C_0) %
		\label{eq:social_forward_direction}
	\end{equation}
	
	\begin{equation}
		\frac{cut(C)}{\mathfrc{d}_{\mathfrc{w}}(C)}
		<
		\frac{cut(C_0)}{\mathfrc{d}_{\mathfrc{w}}(C_0)}
		\label{eq:social_forward_direction_final}
	\end{equation}
	
	Now we prove the forward direction, i.e., given a set~$C~\subset~C_0$ including the seed node~$u$ in~$G$ with motif conductance smaller than~that~of~$C_0$, then the maximum flow on~$G_f$ is less than $cut(C_{0})\mathfrc{d}_{\mathfrc{w}}(C_0)$.
	With the assumptions~(a)~and~(b) in place, Equation~(\ref{eq:social_forward_direction_final}) holds true for the selected set~$C$.
	Since Equation~(\ref{eq:social_forward_direction_final}) can be rewritten as Equation~(\ref{eq:social_forward_direction}).
	To complete the proof, we will show that there exists an s-t~cut $(B_1,B_2)$ of $G_f$ such that $s \in B_1$, $t \in B_2$, $B_2 \cap C_0 = C$, and the total weight of the in-edges of $B_2$ is \hbox{$cut(C_{0})\mathfrc{d}{\mathfrc{w}}(C_0\setminus{C})+cut(C)\mathfrc{d}{\mathfrc{w}}{({C}_0)}$}.
	Using the Max-Flow Min-Cut Theorem~\cite{ford_fulkerson_1956} again, it follows that if there is an s-t cut with a weight of less than $cut(C_{0})\mathfrc{d}{\mathfrc{w}}(C_0)$, the maximum flow value on~$G_f$ must also be less than $cut(C_{0})\mathfrc{d}_{\mathfrc{w}}(C_0)$.
	Since $B_2 \cap C_0 = C$, it follows that $B_1 \cap C_0 = C_0 \setminus C$, hence there are $|C_0 \setminus C|$ cut edges of the form~$(x,t)$.
	By design, the total weight of these edges is $cut(C_{0})\mathfrc{d}{\mathfrc{w}}(C_0\setminus{C})$.
	
	Now we show that the weights of the remaining cut edges, i.e., edges from~$B_1$~to~${B_2\setminus\{t\}}$ add up to~$cut(C)\mathfrc{d}_{\mathfrc{w}}{({C}_0)}$ for each net-expansion technique. 
	In the clique expansion, we forcibly have $B_2=C \cup \{t\}$.
	Under assumption~(c), each cut net~$e$ of~$C$ in~$H_\mu$ corresponds to two cut in-edges of~$B_2$ in~$G_f$, both with weight set to $\mathfrc{w}(e)\mathfrc{d}_{\mathfrc{w}}{({C}_0)} / 2$ by design, hence they add up to the specified total weight.
	In the star expansion, under assumption~(c), we make $B_2 = C \cup A \cup \{t\}$, where~$A$ is the set of artificial nodes~$a$ with~$|N(a) \cap C| \geq |N(a)|/2$.
	For each cut net~$e$ of~$C$ in~$H_\mu$ there is a single cut in-edge of~$B_2$ in~$G_f$ that connects an auxiliary artificial node and a node from $C_0$. 
	The weight of this cut edge is set by design to $\mathfrc{w}(e)\mathfrc{d}_{\mathfrc{w}}{({C}_0)}$, so the total weight of these cut edges is as stated.
	In the Lawler expansion, we make~$B_2 = C \cup A \cup \{t\}$, where~$A$ consists of the auxiliary artificial nodes~$w_1$ having in-edges only from nodes in~$C$, and the auxiliary artificial nodes~$w_2$ having out-edges to at least one node in~$C$.
	Hence, for each cut net~$e$ of~$C$ in~$H_\mu$ there is exactly one cut in-edge of~$B_2$ in~$G_f$ that connects two auxiliary artificial nodes, namely~$w_1 \in B_1$~and~$w_2 \in B_2$. 
	The weight of this single cut edge is set by design to $\mathfrc{w}(e)\mathfrc{d}_{\mathfrc{w}}{({C}_0)}$, so~the~sum of the weights of the cut edges is as expected.
\end{proof}

\AlgName{SOCIAL} utilizes a push-relabel approach to iteratively search for a maximum s-t flow in the model~$G_f$.
If the found maximum flow is strictly smaller than $cut(C_{0})\mathfrc{d}_{\mathfrc{w}}(C_0)$, then we can directly find a minimum cut with the same weight as it and, consequently, a cluster~$C \subset C_0$ containing the seed node~$u$ that has a strictly smaller motif conductance value~$\phi_\mu(C)$ than that of~$C_0$~in~$G$. 
If such a cut is found, the algorithm repeats the process recursively setting the identified sub-cluster~$C$ as the new initial cluster, i.e., it constructs a new flow model based on~$H_\mu$ and the initial cluster and uses the push-relabel algorithm to continue searching for sub-clusters with even lower motif conductance values. 
If, on the other hand, the maximum flow is not strictly smaller than $cut(C_{0})\mathfrc{d}_{\mathfrc{w}}(C_0)$, it means that the current cluster~$C_0$ is optimal among all of its sub-clusters containing the seed node~$u$, and the algorithm terminates for the given ball~$S$.

\subsection{Experimental Evaluation}
\label{sec:social_Experimental Evaluation}

\paragraph*{Setup.} 
We implemented \AlgName{LMCHGP} and \AlgName{SOCIAL} in C++ within the KaHIP framework. 
In particular, \AlgName{LMCHGP} was implemented using the public libraries for KaHyPar~\cite{schlag2016k} and KaHIP~\cite{kaffpa}. 
We use the fastest configuration of these libraries throughout our experiments.
Note that using stronger configurations would likely yield better solutions at the cost of higher running time.
We compiled all algorithms using gcc 11.2 with full optimization turned on (-O3 flag).
Unless mentioned otherwise, all experiments are performed on a single core of Machine D.
For the reported experiments, we use a time limit of one hour for our overall algorithm.
This time limit is checked between repetitions of each of our algorithms, hence it can be violated in case a particular partitioning procedure takes too long.

\paragraph*{Baselines.} 
We experimentally compare our algorithms against the state-of-the-art competitors, namely \AlgName{MAPPR}~\cite{yin2017local}.
We also ran preliminary experiments with \AlgName{HOSPLOC}~\cite{zhou2021high}.
However, the algorithm is very slow even for small graphs and not scalable as their algorithm works using an adjacency matrix and hence needs $\Omega(n^2)$ space and time. 
Moverover,  experiments done in their paper are on graphs that are  multiple orders of magnitude smaller than the graphs used in our evaluation.  
Hence, we are not able to run the algorithm on the scale of the instances used in this thesis.  
We were not able to explicitly compare against~\AlgName{LCD-Motif}~\cite{zhang2019local} since their code is not available (neither public, nor privately\footnote{Personal communication with the authors}) and the data presented in the respective paper does not warrant explicit comparisons (e.g. seed nodes are typically not presented in the papers).
However, we try to make implicit comparisons \hbox{in Section~\ref{subsec:social_addcomp}}. %
We compare our results against the globally best cluster computed for each seed node by \AlgName{MAPPR} using its standard parameters ($\alpha=0.98$, $\epsilon=10^{-4}$).

\paragraph*{Instances.}
Basic properties of the graphs used in our experiments can be found in Table~\ref{tab:lmcvhgp_graphs}.
For our experiments, we split the graphs in two disjoint sets: a \emph{tuning} set for the parameter study experiments and a \emph{test} set for the comparisons against the state-of-the-art.
The graphs in the test set are exactly the graphs used in the MAPPR paper \cite{yin2017local}.
Prior to our experiments, we removed parallel edges, self-loops, and directions of edges and assigning unitary weight to all nodes and edges.
For each graph, we pick $50$ random seed nodes and use all of them as input for each algorithm.

\paragraph*{Methodology.} 
For the sake of simplicity, all our experiments are based on the triangle motif, i.e., the undirected clique of size three.
However, in general our method is also applicable to larger motifs.
Since this motif has three nodes, Theorems~\ref{theo:lmcvhgp_conductance_equivalence},~\ref{theo:lmcvhgp_conductance_equivalence_graph},~and~\ref{theo:social_flow_model_guarantee} are valid in general.
In other words, this setup ensures that the graph model and the hypergraph model of \AlgName{LMCHGP} are exact and guarantees that \AlgName{SOCIAL} yields exact improvements for all net expansion techniques. 
We ensure the integrity of our results by using the same motif-conductance evaluator function for all tested algorithms.

We measure running time, motif-conductance, and/or size of the computed cluster.
When averaging motif conductance over multiple instances, the final score is computed via arithmetic mean.
This is a necessary averaging strategy since motif conductance can be zero, which makes the geometric mean infeasible to compute.
When averaging running time or cluster size over multiple instances, we still use the geometric mean in order to give every instance the same influence on the \textit{final score}.

\begin{table}[t]
	\centering
	\setlength{\tabcolsep}{4pt}
	\scriptsize
	\begin{tabular}{l@{\hskip 30pt}rrr@{\hskip 30pt}rrr@{\hskip 30pt}rrr}
	\toprule	
        \multirow{2}{*}{Graph} & \multicolumn{3}{c}{\hspace*{-1cm}\AlgName{SOCIAL}}  & \multicolumn{3}{c}{\hspace*{-1cm}\AlgName{LMCHGP}}    & \multicolumn{3}{c}{\hspace*{-1cm}\AlgName{MAPPR}} \\
		& $\phi_\mu$     & $|C|$   & t(s)   & $\phi_\mu$     & $|C|$   & t(s)   & $\phi_\mu$  & $|C|$ & t(s)    \\ 
                \midrule

		com-amazon             & \textbf{0.031} & 76    & $<$0.01  & 0.037 & 64    & 0.22  & 0.153      & 58  & 2.68   \\
		com-dblp               & \textbf{0.090} & 58    & 0.02  & 0.115 & 56    & 0.38  & 0.289      & 35  & 3.04   \\
		com-youtube            & \textbf{0.125} & 1832  & 4.52  & 0.172 & 1443  & 7.93  & 0.910      & 2   & 10.44  \\
		com-livejournal        & \textbf{0.158} & 494   & 3.33  & 0.244 & 387   & 8.17  & 0.507      & 61  & 173.80 \\
		com-orkut              & 0.273 & 1041 & 256.21 & \textbf{0.150} & 13168 & 496.94 & 0.407      & 511 & 923.26 \\
		com-friendster         & 0.388      &  2060     & 1194.50       & \textbf{0.368}      &  10610     &   1339.99     &   0.741         &   121  &   16565.99      \\ 
                \midrule
		Overall                & \textbf{0.178} & 453   & 2.33  & 0.181 & 823   & 12.67  & 0.500      & 50  & 79.34  \\ 
                \bottomrule
	\end{tabular}
	\caption{Average comparison against state-of-the-art.}
	\label{tab:social_resultsoverall}
\end{table}

% \vfill\pagebreak

\subsubsection{Parameter Study}
\label{sec:lmcvhgp_Parameter Study}

We tune the paramenters of \AlgName{LMCHGP} using the graphs in the Tuning Set of Table~\ref{tab:lmcvhgp_graphs}, which are disjoint from the graphs used for the evaluation against state-of-the-art. 
In a comparison of the hypergraph-based version of our algorithm against its graph-based version, each approach produces the best motif conductance for around $50\%$ of the instances.
Nevertheless, the hypergraph-based version is $23$ times slower on average and uses up to $68.7$ times more memory since the hypergraph model stores a large number of nets.
Hence, we exclusively use the graph-based version of our algorithm for the remaining experiments.
Next, we compare the effect of using different values for the parameters $\beta$ and $\alpha$.
The results show the expected regular relationship between solution quality, running time, and these parameter: the larger $\alpha$ (resp. $\beta$), the smaller the motif conductance and the larger the running time.
Finally, we checked that the impact of including the label propagation local search in \AlgName{LMCHGP} is, on average, a $13\%$ decrease in the motif conductance at the cost of only $1.5\%$ more running time.
Summing up, we use the following parameters for \mbox{\AlgName{LMCHGP}}: graph model, label propagation, $\alpha=3$, $\ell \in \{1,2,3\}$, and $\beta=80$.
In view of that fact that the parameters $\alpha$ and $\ell$ utilized in \AlgName{SOCIAL} carry the same significance as their counterparts in \AlgName{LMCHGP}, we directly incorporate the finely-tuned values of these parameters from the latter in \AlgName{SOCIAL}.
In addition, the clique expansion method proves to be the most efficient for our three-node motif as it requires no auxiliary artificial nodes and involves the fewest number of auxiliary artificial edges compared to all other net expansion techniques.
In summary, we use the following parameters for \AlgName{SOCIAL}: clique expansion, \hbox{$\alpha=3$, and $\ell \in \{1,2,3\}$}.

\subsubsection{State-of-the-Art}
\label{subsec:social_State-of-the-Art}
\label{subsec:social_Results}

In this section we present and discuss our experiments against the state-of-the-art.
All the reported results are based on the Tet Set of Table~\ref{tab:lmcvhgp_graphs}. 
In the performance profile plots shown in Figures~\ref{fig:social_SIMPLEstateoftheart_graph_res_pp}~and~\ref{fig:social_SIMPLEstateoftheart_graph_tim_pp}, we compare \AlgName{MAPPR}~\cite{yin2017local}, \mbox{\AlgName{LMCHGP}},and \AlgName{SOCIAL}.
In Table~\ref{tab:social_resultsoverall}, we show average results for each graph.

As shown in Figure~\ref{fig:social_SIMPLEstateoftheart_graph_res_pp}, \AlgName{SOCIAL} obtains the best or equal motif conductance value for $62\%$ of the instances, while \mbox{\AlgName{LMCHGP}} and \mbox{\AlgName{MAPPR}} respectively obtain the best or equal motif conductance for $49\%$~and~$19\%$ of the instances.
This result can be explained with two observations.
First, \AlgName{SOCIAL}~and~\AlgName{LMCHGP} explore the solution space considerably better than \AlgName{MAPPR}, since both build the (hyper)graph model multiple times, while \AlgName{MAPPR} simply uses the \AlgName{APPR} algorithm.
Second, \AlgName{SOCIAL} is based on a flow approach which directly optimizes for motif conductance, whereas \mbox{\AlgName{LMCHGP}} is based on a graph partitioning algorithm which is repeated multiple times to compensate for its design to minimize the number of cut motifs rather than motif conductance.
In Table~\ref{tab:social_resultsoverall}, \AlgName{SOCIAL} outperforms \mbox{\AlgName{LMCHGP}} for 4 of the 6 graph and overall, and outperforms \mbox{\AlgName{MAPPR}} for all graphs and overall.
On the other hand,\AlgName{LMCHGP} outperforms \AlgName{SOCIAL} for 2 of the 6 graphs, and outperforms \AlgName{MAPPR} with respect to motif conductance and running time for all graphs and overall.
Overall, \AlgName{SOCIAL} computes clusters with motif conductance $0.178$ while \mbox{\AlgName{LMCHGP}} and \mbox{\AlgName{MAPPR}} compute clusters with motif conductance $0.181$~and~$0.500$,~respectively.

\begin{figure*}[htb]
	\captionsetup[subfigure]{justification=centering}
	\centering
	\begin{subfigure}{0.48\textwidth}
		\centering
		\includegraphics[width=\textwidth]{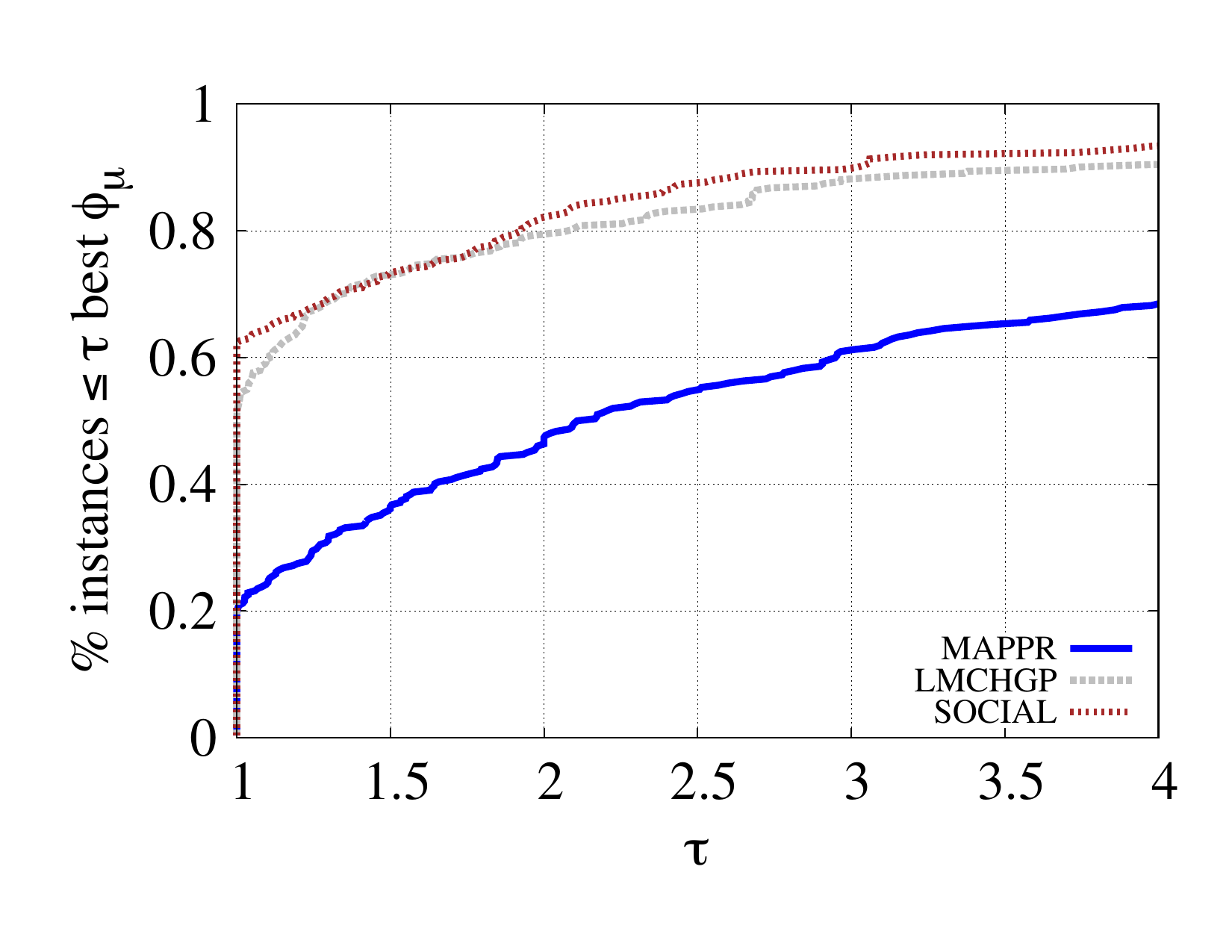}
		\caption{Motif conductance.}
		\label{fig:social_SIMPLEstateoftheart_graph_res_pp}
	\end{subfigure}
	\begin{subfigure}{0.48\textwidth}
		\centering
		\includegraphics[width=\textwidth]{chapters/social/plots/mqi_clique/tim/pp_chart.png}
		\caption{Running time.}
		\label{fig:social_SIMPLEstateoftheart_graph_tim_pp}
	\end{subfigure}
	\caption{Performance profiles comparing our algorithms against the state-of-the-art in terms of solution quality and \hbox{running time}.}
	\label{fig:social_SIMPLEstateoftheart_graph_pp}
\end{figure*}

As exhibited in Figure~\ref{fig:social_SIMPLEstateoftheart_graph_tim_pp}, \AlgName{SOCIAL} is the fastest one for $87\%$ of the instances, while \mbox{\AlgName{LMCHGP}} and \mbox{\AlgName{MAPPR}} are the fastest ones for $12\%$~and~$1\%$ of the instances, respectively.
Furthermore, the running time of \AlgName{SOCIAL} is within a factor $1.18$ of the running times of the fastest competitors for all instances.
\AlgName{SOCIAL} is respectively up to~\numprint{237}~and~\numprint{144063}~times faster than \mbox{\AlgName{LMCHGP}} and \mbox{\AlgName{MAPPR}} while being a factor~$5.4$~and~$34.1$ faster than them on average.
The reason for \AlgName{MAPPR} being considerably slower than the other algorithms is that it must enumerate motifs across the entire graph, while \AlgName{SOCIAL} and \AlgName{LMCHGP} only require enumeration of motifs in a ball around the seed node.
The reduced but still substantial difference in running time between \AlgName{SOCIAL} and \AlgName{LMCHGP} is a result of \AlgName{LMCHGP}'s repeated partitioning of each ball around the seed node, while \AlgName{SOCIAL} employs a flow model to greedily improve the motif conductance metric until a local optimum cluster is obtained.
Table~\ref{tab:social_resultsoverall} demonstrates that, on average, \AlgName{SOCIAL} has a lower running time than \AlgName{LMCHGP}, and \AlgName{LMCHGP} has a lower running time than \AlgName{MAPPR}, \hbox{for each graph as well as overall}.

For a more intuitive analysis of the quality of our results, Figure~\ref{fig:social_SIMPLEstateoftheart_graph_res_double} plots motif conductance versus cluster size for all communities computed by the three algorithms.
Observe that the communities found by \AlgName{SOCIAL} are densely localized in the lower left area of the chart, which is the region with smaller motif conductance and smaller cluster size. 
On the other hand, communities computed by \AlgName{MAPPR} are often in the upper area of the chart and communities computed by \mbox{\AlgName{LMCHGP}} are often in the right area~of~the~chart.

\begin{figure}[t]{}
	\centering
	\includegraphics[width=0.6\textwidth]{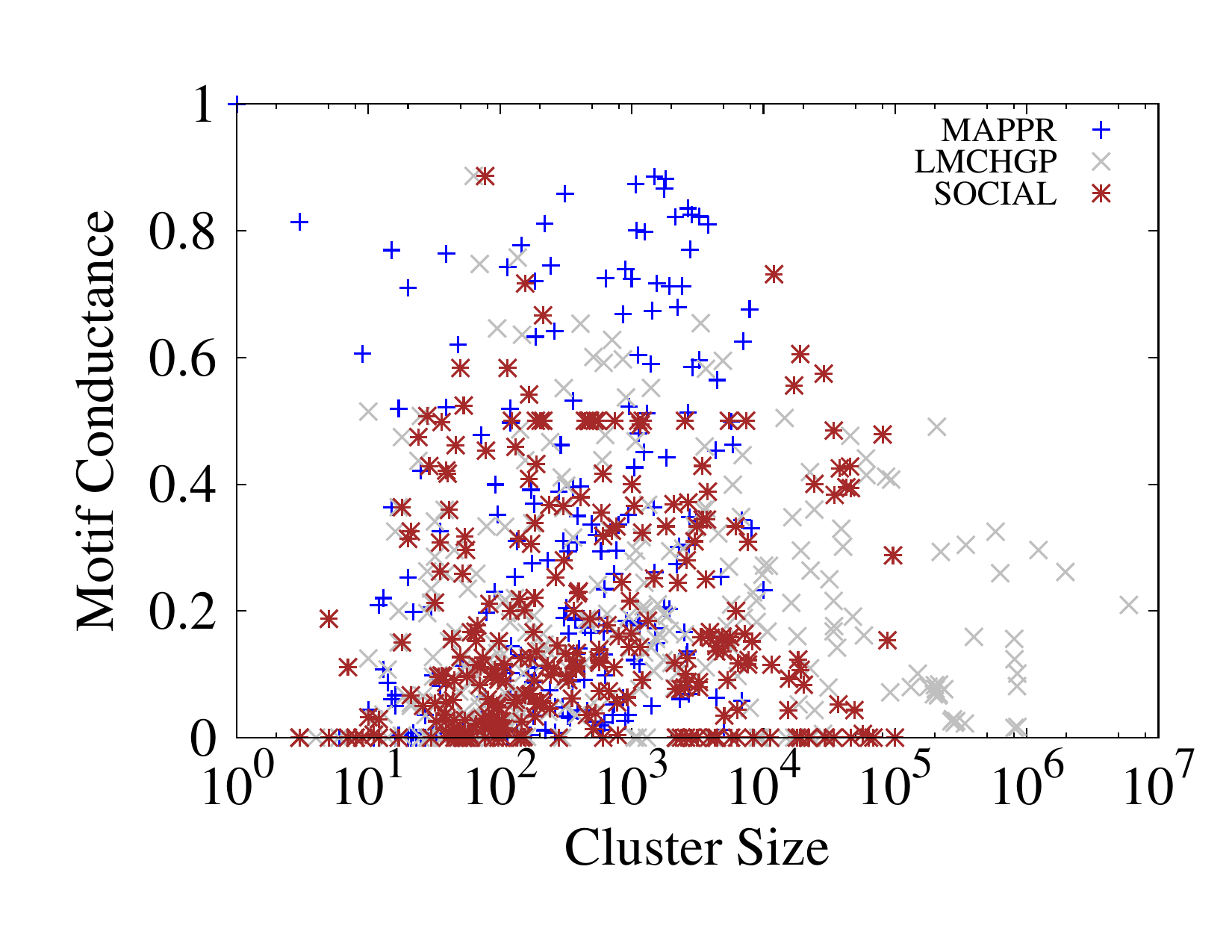}
	\caption{Motif conductance vs cluster size.}
	\label{fig:social_SIMPLEstateoftheart_graph_res_double}
\end{figure}

\subsubsection{Additional Comparisons}
\label{subsec:social_addcomp}
As mentioned above, we were not able to compare against \AlgName{LCD-Motif}~\cite{zhang2019local} explicitly since their code is not available (neither publicly, nor privately) and the data presented in the respective papers does not warrant explicit comparisions (e.g.,~seed nodes are typically not presented in papers, and in this case instances are directed rather than undirected).
Here, we make an attempt at implicit comparisons.
Zhang~\etal~\cite{zhang2019local} (Table 4 therein) compare motif conductance against \AlgName{MAPPR} on three directed instances (cit-HepPh, Slashdot, StanfordWeb) and report an geometric mean improvement of 54\% in motif conductance for the triangle motif. 
As \AlgName{LMCHGP} and \AlgName{SOCIAL} work for undirected instances, we have build undirected version of those graphs and run both as well as \AlgName{MAPPR} for the triangle motif.
Respectively, the geometric mean improvement \AlgName{LMCHGP} and \AlgName{SOCIAL} obtain over \AlgName{MAPPR} is 194\% and 223\%. %
Both improvements are significantly larger than the improvement of Zhang~\etal over \AlgName{MAPPR}.
Also note that in our experiments from Table~\ref{tab:social_resultsoverall}, the geometric mean improvement (using the average motif conductance values) of \AlgName{LMCHGP} and \AlgName{SOCIAL} over \AlgName{MAPPR} in motif conductance is \hbox{192\% and 219\%, respectively}.%

\section{Conclusion}
\label{sec:local_Conclusion}

In this chapter, we proposed three local algorithms for graph decomposition: \AlgName{LMCHGP} and \AlgName{SOCIAL}.
They were discussed in detail and were experimentally compared against one another and agains the state-of-the-art.
Both proposed algorithms solve the local motif clustering problem with sophisticated combinatorial techniques.
Our first contribution, \AlgName{LMCHGP} (Local Motif Clustering via (Hyper)Graph Partitioning), is an algorithm based on (hyper)graph partitioning.
We construct both a graph model and a hypergraph model to represent the motifs around the seed node, and then partition the (hyper)graph model using a multi-level partitioner to minimize the motif conductance in the original graph.
Our second contribution, \AlgName{SOCIAL} (faSter mOtif Clustering vIa mAximum fLows), combines the strongly local hypergraph model proposed in the previous algorithm with the \emph{max-flow quotient-cut improvement} algorithm (\AlgName{MQI}) to optimize for motif conductance.
We construct a flow model and utilize a push-relabel algorithm to find subsets of the initial cluster containing the seed node with lower motif conductance than the whole cluster.
In experiments involving triangle motifs, \AlgName{LMCHGP} produces communities with lower motif conductance compared to \AlgName{MAPPR}, while being up to multiple orders of magnitude faster than it.
On the other hand, \AlgName{SOCIAL} produces communities with lower motif conductance compared to \AlgName{MAPPR} and \AlgName{LMCHGP}, while being up to multiple orders of magnitude faster than both.
In future work, we intend to conduct experiments with larger motifs and use the Lawler-expansion version of our flow graph, since it is the only one whose quality guarantee holds true for larger motifs. 
Laslty, we intend to add parallelization to improve the speed on \hbox{large instances further}.

\chapter{Multilevel Algorithms}
\label{chap:Multilevel Algorithms}

In this chapter, we propose our contributions in the field of multilevel graph decomposition.
Firstly, we introduce an integrated multilevel algorithm for solving the process mapping problem, which is designed and implemented in multiple versions with different goals.
Experimental results demonstrate that the multiple versions of our algorithm outperform the state-of-the-art in terms of both solution quality and running time.
Secondly, we design a multilevel algorithm for solving the signed graph clustering problem and utilize it as a building block to construct a powerful distributed memetic algorithm.
Our experimental evaluation demonstrates that our memetic algorithm is capable of converging effectively and significantly outperforms the state-of-the-art with respect to edge-cut.

\paragraph*{References.}
This section is based on~\cite{HighQualityHierarchicalPM20}, which is joint work with Alexander van der Grinten, Henning Meyerhenke, Jesper Larsson Tr{\"{a}}ff, and Christian Schulz, on~\cite{DistrMemMLevelSigGraphClustering_short}, which is joint work with Felix Hausberger and Christian Schulz, and on the technical report~\cite{DistrMemMLevelSigGraphClustering}, which is also joint work with Felix Hausberger and Christian Schulz.

\section{Multilevel Process Mapping}
\label{sec:intmap_Multilevel Process Mapping}
\label{chap:intmap_Multilevel Process Mapping}
\label{chap:intmap_High-Quality Hierarchical Process Mapping}

In this section, we propose a high-quality algorithm for the process mapping problem.
Our algorithms is an integrated approach based on a multilevel scheme coupled with sophisticated initial partitioning and local search techniques. 
Additionally, we introduce faster techniques to keep precise topology information without the need to store distance matrices.
Our experiments indicate that our new integrated algorithm improves mapping quality over other state-of-the-art algorithms. 
For example, one configuration of our algorithm yields similar solution quality as the previous state-of-the-art in terms of mapping quality for large values of $k$ while being a factor~9.3 faster.  
Compared to the currently fastest iterated multilevel mapping algorithm \AlgName{Scotch}, we obtain 16\% better solutions while investing slightly more running time.
Most importantly, hierarchical multi-section algorithms that take the system hierarchy into account for model creation improve the results of the overall process mapping significantly.

\subsection{Integrated Mapping}
\label{s:intmap_main}
\label{sec:intmap_integrated_mapping_approach}
\label{sec:intmap_Integrated Process Mapping}

\subsubsection{Overall Algorithm}
\label{subsec:intmap_Overall Algorithm}

We engineered all the components of a multilevel algorithm to solve the process mapping problem in an \emph{integrated} way, as illustrated in Figure \ref{fig:intmap_multilevel}.
This algorithm includes coarsening-uncoarsening schemes, methods to construct and refine initial solutions, local search methods, and additional tools to explore trade-offs in memory \hbox{usage and performance}.

\begin{figure}[tb]
	\centering
	\includegraphics[width=0.65\columnwidth]{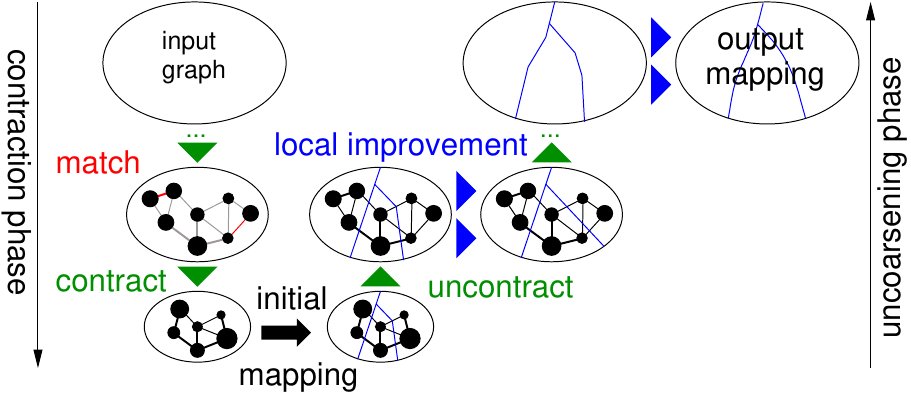}
	\caption{Multilevel scheme used to solve the process mapping problem (adapted from~\cite{kaffpa}).}
	\label{fig:intmap_multilevel}
\end{figure}

\subsubsection{Coarsening}
\label{subsec:intmap_Contraction}
We use a matching-based coarsening scheme.
The \emph{matching--based} coarsening is the most common choice in multilevel partitioning algorithms due to its simplicity, speed, and generality. 
It has two consecutive steps: An edge rating function and a matching algorithm. 
Based on local information, the \emph{edge rating function} scores each edge to estimate the benefit of contracting it.
We employ the same edge rating function $\text{exp*}(e)=  \omega(e)/(|\Gamma(u)|*|\Gamma(v)|)$ as used in~\cite{kaffpa}.
Then, the \emph{matching algorithm} obtains a maximal match to maximize the sum of the ratings of the contracted edges. 
As in \cite{kaffpa}, we computed matchings with the \emph{Global Paths Algorithm} \cite{kaffpa}, which is a $\frac{1}{2}$-approximate algorithm.

\subsubsection{Initial Solution Algorithms}
\label{subsec:intmap_Initial Solution}

We compute the initial mapping using a two-phase approach.
To solve the graph partitioning phase, we compare two multilevel recursive bisection algorithms:
(i)~\emph{standard bisection} setup, in which we perform a recursive bisection to obtain $k$ blocks;
(ii)~\emph{multi-section} setup, in which we perform recursive bisections throughout the hierarchical structure of PEs.
To construct a solution for one-to-one mapping, we apply two different construction methods: 
(i)~\emph{identity}, which assigns each block to the PE with the same ID to favor locality; 
(ii)~\emph{hierarchy top down}, which partitions the set of blocks throughout the hierarchical structure of PEs.
To refine the one-to-one mapping solution, we perform an $N^{10}_\mathcal{C}$ swap neighborhood local search.
Hence, the resulting map~$\Pi$ of nodes to PEs becomes our initial solution.

Our \emph{standard bisection} setup for initial partition  corresponds to the initial partition step in \AlgName{KaHIP}.
Moreover, it is a canonical choice to produce an initial partition using multilevel algorithms. 
On the other hand, the \emph{multi-section} setup draws inspiration from the scheme used in~\cite{GlobalMultisection}.
It is an attempt to specialize the initial partition for the particular case tackled in this paper: a regularly hierarchical distribution of PEs in which the communication cost between two processes (nodes) highly depends on the hierarchy level shared by their corresponding PEs (blocks).
Particularly, we apply a recursive partitioning scheme that splits all the nodes in $a_\ell$ blocks, then splits the nodes in each block in $a_{\ell-1}$ sub-blocks, then splits the nodes in each sub-blocks in $a_{\ell-2}$ sub-sub-blocks, and so forth.
Observing that the communication costs decrease as the communicating processes share lower hierarchy levels, the multi-section approach implies a hierarchy of sub-problems that directly reflects the problem cost hierarchy.

In both setups of the partitioning step, we recursively assign consecutive IDs to blocks throughout the process in order to maintain locality.
Moreover, the PEs belonging to each hierarchy module are labeled with consecutive IDs, which also promotes locality. 
Then, the \emph{identity} method is a fast way to construct a solution for one-to-one mapping which takes advantage of this locality: it assigns each block $V_i$ to the PE with the ID $i$.
Note, the \emph{standard bisection} setup conveniently combines with the identity mapping approach when $k$ is a power of $2$ since the recursive bisections will be automatically performed throughout the hierarchical topology.
For an analogous reason, the \emph{multi-section} setup is a good algorithm to create a coarse model to be mapped by the \emph{identity} mapping approach independently of~$k$.
The \emph{hierarchy top down}~\cite{schulz2017better} is a more general approach to construct solutions for one-to-one mapping when the PEs are hierarchically organized.
Its mechanism is similar to the idea of multi-section throughout the hierarchy.

\subsubsection{Uncoarsening}
\label{subsec:intmap_uncoarsening}

After obtaining an initial solution at the coarsest level, we apply a sequence of four local refinement methods to move nodes between blocks (which are already associated to unique PEs).
Then, we undo each of the contractions performed previously, from the coarsest graph until the original input graph.
After each uncoarsening step, we repeat our four local refinement methods.
The refinements run in a specific order based on their characteristics.
First, a \emph{quotient graph refinement} exhaustively tries to improve solution quality and eliminate imbalance by moving nodes between each pair of blocks connected by an edge in the quotient graph.
Second, a \emph{$k$-way Fiduccia-Mattheyses (FM) algorithm~\cite{fiduccia1982lth} refinement} greedily goes through the boundary nodes trying to relocate them with a more global perspective in order to improve the mapping.
Third, a \emph{label propagation refinement} randomly visits all nodes and moves each one to the most appropriate block while not increasing the objective.
Finally, a \emph{multi-try FM refinement} is exhaustively applied in rounds with random starting points throughout the graph in order to escape local optima as many times as possible.
Before explaining the local search algorithms, we introduce the notion of \emph{gain} for the process mapping problem.

\paragraph*{Gain.}

All our refinement methods are based on the concept of \emph{gain}.
We define $\Psi_b(v)$ as the \emph{partial} contribution of a node $v$ to the objective function $J(\mathcal{C},\mathcal{D}, \Pi)$ (Equation~(\ref{eq:process_mapping_obj})) in case $v$ is assigned to the PE $b$.
More precisely, $\Psi_b(v)$ represents the total cost of the communications involving~$v$ if $\Pi(v)=b$ and the neighbors of $v$ remain assigned to their current PEs.
The \emph{gain}~$g_b(v)$ represents the value that will be subtracted from $J(\mathcal{C},\mathcal{D}, \Pi)$ if a node $v$ is moved from its current PE $\Pi(v)$ to PE $b$. More precisely,
$    \Psi_b(v)  := \sum_{\set{v,u} \in I(v)} {\mathcal{C}_{v,u} \mathcal{D}_{b,\Pi(u)}}$ and
    $g_b(v) := \Psi_{\Pi(v)}(v) - \Psi_b(v)$.
Note that $g_{\Pi(v)}(v) \equiv 0$. Observe that a positive (resp., negative) gain indicates improvement (resp., worsening) of the solution.
Computing the gains of~$v$ to all blocks in $R(v)$ costs $O\big(|R(v)||I(v)|\big) = O\big(|I(v)|^2\big)$.
For comparison purposes, the computation of the same corresponding gains in the context of graph partitioning and edge-cut objective function costs $O\big(|I(v)|\big)$.

\paragraph*{Quotient Graph Refinement.}

We implemented an adapted version of the \emph{quotient graph refinement} \cite{kaffpa} to incorporate our definition of gains.
Within this refinement, we visit each pair of neighboring blocks in the quotient graph $\mathcal{Q}$ underlying the current $k$-partition. 
Then we apply an FM algorithm~\cite{fiduccia1982lth} to move nodes between the two currently visited blocks, keeping two respective gain--based priority queues of eligible nodes.
Each queue is randomly initialized with the boundary in its corresponding block.
After a node is moved (which can only happen once during an execution of the local search), its unmoved neighbors become eligible.

\paragraph*{K-Way FM Refinement.}
\label{subsec:intmap_K-Way_FM_Refinement}

Our $k$-way FM refinement was adapted from the implementation in~\cite{kaffpa}.
Unlike the quotient graph refinement, the $k$-way FM does not restrict the movement of a node to a certain pair of blocks, but performs global-aware movement choices.
Our implementation of $k$-way FM uses only one gain--based priority queue $P$, which is initialized with the \emph{complete} partition boundary in a random order.   
Then, the local search repeatedly looks for the highest-gain node $v$ and moves it to the best $c(v)$-underloaded neighboring block.
When a node is moved, we insert in $P$ all its neighbors that were not in $P$ and have not been moved yet.
The $k$-way local search stops if $P$ is empty (i.e., each node was moved once) or when a stopping criterion based on a random-walk model described in~\cite{kaffpa} applies.
To escape from local optima, this refinement allows some movements with negative gain or to blocks that are not  $c(v)$-underloaded.
Afterwards local search is rolled back to the lowest cut fulfilling the balance criterion that occurred.

\paragraph*{Label Propagation Refinement.}

We propose a local search inspired by \emph{label propagation}~\cite{labelpropagationclustering}. 
The algorithm works in rounds.
In each round, the algorithm visits all nodes in a random order, starting with the labels being the current assignment of nodes to blocks.
When a node $v$ is visited, it is moved to the $c(v)$-underloaded neighboring block  with highest positive gain.
We consider only $c(v)$-underloaded blocks since this ensures that the target block is not overloaded when the node is moved there.
Ties are broken randomly and a 0-gain neighboring block can be occasionally chosen with $50\%$ probability if there is no neighboring $c(v)$-underloaded block with positive gain.
We perform at most~$\ell$ rounds of the algorithm, where $\ell$ is a tuning parameter.

\paragraph*{Multi-Try FM Refinement.}
\label{subsec:intmap_Multi-Try_FM_Refinement}

We also adapted our gain concept to a localized variant of the $k$-way local search algorithm similar to that proposed in \cite{kaffpa} under the name of \emph{multi-try} FM.
Instead of being initialized with all boundary nodes, as in $k$-way FM, multi-try FM is repeatedly initialized with a single boundary node. %
This introduces a higher diversification to the search since it is not restricted to movements in boundary nodes with global largest gain.
As a result, this local search can escape local optima more easily than \emph{$k$-way} FM.

\subsubsection{Additional Techniques}
\label{subsec:intmap_additional_techniques}

\paragraph*{Implicit Distance Matrix.}

When the topology matrix $\mathcal{D}$ is stored in memory, access time to obtain the distance between a pair of PEs is $O(1)$, but this requires $O(k^2)$ space. 
From now on, we refer to the algorithm explicitly keeping $\mathcal{D}$ in memory as \emph{matrix--based} approach.
We implement three alternative approaches to save memory by exploiting the fact that our topology matrix is a hierarchy and the IDs of PEs in each of the hierarchy modules are sequential.
For simplification reasons, we call these approaches: (i) \emph{division--based}; (ii) \emph{stored division--based}; and (iii) \emph{binary notation--based}.

In the \emph{division--based} approach, we perform $O(\ell)$ successive integer divisions and comparisons in the ID of two PEs when we need to find out their distance. Here, $\ell$ is the number of levels in the system hierarchy. 
As a preprocessing step to be executed only once, we create a vector $h=\Big(k\Big/{\prod_{t=1}^{\ell} a_{t}}, \; k\Big/{\prod_{t=2}^{\ell} a_{t}},\; \ldots,\; k\Big/{a_{\ell}}\Big)$.
To find the distance between PEs $b$ and $b^\prime$ with $b \neq b^\prime$, we loop through the hierarchy layers from $i=\ell$ to $i=1$.
In each iteration, we perform the integer division of $b$ and $b^\prime$ by $h_{i}$.
Whenever the division results differ, then we break the loop and return $\mathcal{D}_{b,b^\prime}=d_{i}$.
This approach does not require any additional memory other than a vector with $O(\ell)$ integers and has time complexity $O(\ell)$.

The \emph{stored division--based} approach works in a similar way as the \emph{division--based} one.
The only difference is that we avoid repetitive integer divisions of IDs by elements of $h$ by storing the results of all possible divisions in a preprocessing step executed only once.
Although we still need $O(\ell)$ running time to perform comparisons in order to obtain the distance between a pair of PEs, the constant factors involved are much lower.
This improvement in running time comes at the cost of additional $O(k\ell)$ memory.

\begin{figure}[bt]
	\centering
	\includegraphics[width=0.7\textwidth]{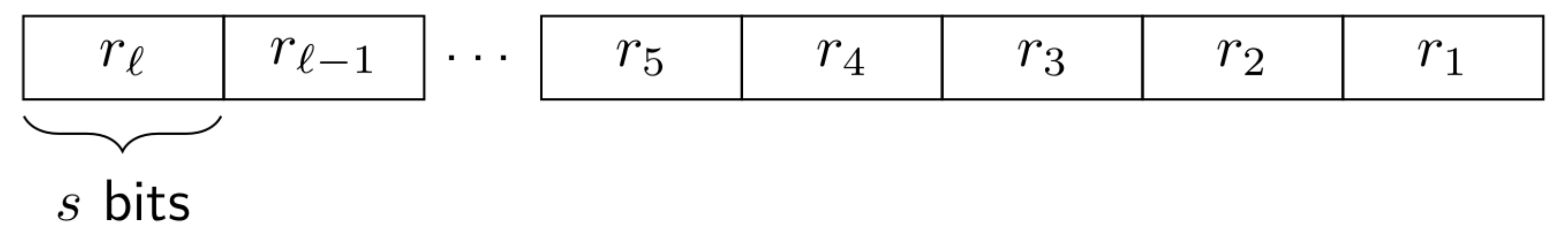}
	\caption{Section structure of the binary number used to represent PE~$b$.} 
	\label{fig:bina_repr}
\end{figure}

The \emph{binary notation--based} approach is a more compact way of decomposing the IDs of PEs.
Instead of storing $\ell$ numbers for each PE, we keep in memory a single binary number per PE.
This binary number $r$ consists of $\ell$ sections $r_i$, each containing $s$ bits, with $s$ given by Equation~(\ref{eq:bitsbinary}) (see Figure \ref{fig:bina_repr}).
To describe the construction of $r$ for a PE~$b$, let a variable~$t$ be initialized as $t=b$.
Then, we loop through the hierarchy layers, from $i=1$ to $i=\ell$.
In each iteration~$i$, $r_{i}$ receives the remainder of the division of $t$ by $a_{i}$ and, then, $t$ is updated to store the integer quotient of $t$ by $a_{i}$.
Afterwards, it is possible to precisely locate $b$ at the hierarchy by sweeping the sections of $r$ from $r_{\ell}$ to $r_1$.
In particular, $r_{\ell}$ specifies its data center, $r_{\ell-1}$ specifies its server among those belonging to its data center, and so forth.
Obtaining the distance between distinct PEs $b$ and $b^\prime$ is equivalent to finding which section $r_i$ contains the leftmost nonzero bit in the result of the bit-wise operation XOR($b$,$b^\prime$).
The running-time complexity of finding the section of the leftmost nonzero bit is $O(\log(\ell))$.
Furthermore, current processors often implement a \emph{count leading zeros} (CLZ) operation in hardware which allows the identification of the leftmost nonzero bit in $O(1)$ time, under the assumption that the size $\log r = O(\log k)$ of the
   binary numbers is smaller than the size of a machine word.

\tikzset{Nodo/.style={circle,fill=gray!25, minimum size=7mm}}

\begin{equation}
	s=\big\lceil \log_{2}\big(\max_{1\leq t \leq \ell}(a_t) \big) \big\rceil
	\label{eq:bitsbinary}
\end{equation}

% \vfill\pagebreak

\paragraph*{Delta-Gain Updates.}

Our local searches frequently need to compute \emph{gains} involved in the movement of nodes.
A \emph{base approach} to check these gains consists of computing them from scratch whenever they are needed, which can yield many gain recomputations.
For this reason, we implement a technique to save running time  called \emph{delta-gain updates} \cite{schlag2016k}.

In \emph{delta-gain updates}, we store a vector $R$ of length $|R(v)|=O(|\Gamma(v)|)$ for each node~$v$.
In this vector, we keep the gains $g_{b}(v)$ for all PEs $b$ containing neighbors of $v$.
Additionally, we store an $n$-sized vector $h$ to keep flags that indicate whether a node has up-to-date gains in memory.
Asymptotically speaking, these vectors represent $O(n+m)$ extra memory.
Each flag is initialized with an inactive seed and is considered active if its value equals a counter that is increased after each uncoarsening steps.
When we need to check a gain of some node $v$, we look at $h_v$ to verify if the gains of $v$ are up-to-date.
If they are not, we compute all gains $g_{b}(v)$ from scratch, which costs $O\big(|I(v)|^2\big)$, and activate $h_v$.
Otherwise, we just access the required gain from memory in $O(1)$ time.

If a node $v$ moves from its current PE to another one, we have to update all delta gains of $v$ and $u \in \Gamma(v)$ with $h_u$ being active.
Assume that $h_v$ and $h_u$ are active and $v$ moves from PE 1 to PE 2 during some local refinement.
After this movement, we should change the delta gains of $u$ and $v$ in memory.
For $v$, it suffices to subtract $g_2(v)$ from all other gains of $v$ and then set $g_2(v)$ to 0.
For $u$, it is slightly trickier, but we do not need to recalculate all its gains from scratch since their only source of change is the edge $e$ that connects $u$ and $v$.
Hence, we respectively subtract and add to $g_b(u)$ the corresponding contribution of $e$ before and after the movement of $v$.
We end up doing the update in time $O\big( |I(v)| + |I(v)| * \overline{|R(u)|}  \big)$, where $\overline{|R(u)|}$ is the average of $|R(u)|, \forall \{v,u\} \in I(v)$.

\subsection{Experimental Evaluation}
\label{sec:intmap_experiments}
\label{s:intmap_exp}

\paragraph*{Setup.} 
We performed our implementations using the KaHIP framework (using C++) and compiled them using gcc 8.3 with full optimization turned on (-O3 flag). 
All of our experiments were run on a single core of a  Machine~E. 

\paragraph*{Baselines.} 
For experiments based on the two-phase approach for tackling the process mapping problem, we solve the graph partitioning phase using \AlgName{KaHIP}~\cite{kaffpa}. 
We use its configurations \emph{fast}, \emph{eco} and \emph{strong} which are described in~\cite{kaffpa} -- we respectively refer to them as \emph{\AlgName{K(Fast)}}, \emph{\AlgName{K(Eco)}} and \emph{\AlgName{K(Strong)}}.
We select the most successful one-to-one mapping algorithms from \cite{schulz2017better} and also \AlgName{Scotch} for our comparison:
(i)~\emph{Top down} with $N_{\mathcal{C}}^{d}$ local search (TopDownN), which represent the state-of-the-art for one-to-one mapping when $k$ is not a power of 2;
(ii)~\emph{identity} mapping, which (when coupled with the \AlgName{KaHIP} multilevel partitioning algorithm) represents the state-of-the-art for the process mapping problem via two-phase approach when $k$ is a power of 2; 
(iii)~the algorithm of \emph{Müller-Merbach}~\cite{muller2013optimale} (Müller-Merbach), whose results are also used as a reference algorithm to calculate solution improvements in~\cite{schulz2017better}; and
(iv)~\emph{\AlgName{Scotch}}~\cite{Scotch}.
We run the two-phase approches TopDownN, \AlgName{Identity}, and Müller-Merbach coupled with \AlgName{K(Fast)}, \AlgName{K(Eco)} and \AlgName{K(Strong)} as a partitioning algorithm.
We also compare against global multi-section \cite{GlobalMultisection}, in which the graph partitioning step is already using  \emph{hierarchical multi-section} itself.
In particular, we use its configurations strong, econ, and fast, which we refer to ass ~\AlgName{Gmsec(Strong)}, \AlgName{Gmsec(Eco)} and \AlgName{Gmsec(Fast)}.
Recall that this algorithm is also non-integrated: it uses different quality configurations of \AlgName{KaHIP} to partition the graph, compute the coarser communication model, and then apply TopDownN to solve one-to-one mapping.
We also run \AlgName{Scotch}~\cite{Scotch} configured to only use recursive bipartitioning methods using the quality setting. 
\AlgName{Scotch} is among the algorithms with best running times in our experiments. 
Hence, we add an algorithm (ScotchTC) which reports the best solution out of multiple runs of \AlgName{Scotch} with different random seeds when given the same amount of time to compute a solution as our \emph{strong} configuration has used.
We contacted Christopher Walshaw, who informed us that \AlgName{Jostle}~\cite{walshaw2000mpm} is not available anymore.

\paragraph*{Instances.}
Our instances come from various sources. 
Basic properties of the graphs under consideration can be found in Table~\ref{tab:test_instances_walshaw}. 
To keep the evaluation simple, we use the following hierarchy configurations for all the experiments: $D=1:10:100$, $\mathcal{S}=4:16:r$, with $r \in \{1,2,3,\ldots,128\}$. Hence, $k=64 \cdot r$.

\paragraph*{Methodology.} 
Depending on the focus of the experiment, we measure running time and/or the process mapping communication cost defined in Equation~(\ref{eq:process_mapping_obj}).
Some of our plots are \emph{modified} performance profiles. 
These plots relate the running times of all algorithms to the slowest algorithm on a per-instance basis.
For each algorithm, these ratios are sorted in increasing order. 
The plots show $\big(\frac{\sigma_\text{A}}{\sigma_\text{slowest}}\big)$ on the y-axis. 
A point close to zero indicates that the algorithm was considerably faster than the slowest algorithm.

\subsubsection{Parameter Study}
\label{subsec:Parameter Study}

In this section, we present a sequence of experiments to test the performance of our algorithmic components regarding solution quality and running time.
Our general goal consists of individually evaluating the effectiveness and significance of each component.
Our specific goal consists of tuning three different configurations of the algorithm based on different principles:
(i)~a \emph{strong} configuration, mostly concerned with maximizing solution quality;
(ii)~a \emph{fast} configuration, mostly concerned with minimizing running time; and
(iii)~an \emph{eco} configuration, which seeks to balance running time and solution quality.

% \vfill\pagebreak

Our experimental strategy consists of defining a single \emph{focus} aspect of the algorithm for each experiment.
Then, this specific aspect is tested with different components or setups while other parameters  of the algorithm are kept constant.
Initially we use standard components. Then we use the best component found in an experiment the next section.

We begin by focusing on a representative component of the multilevel scheme: 
(i)~initial mapping;
(ii)~local search\ifFull; and (iii)~contraction scheme\fi{}.
Then, we evaluate algorithmic aspects which only affect running time and memory consumption: the distance matrix representation.
\ifFull We do not test different global search schemes and all tests are based on a single iteration of the multilevel scheme. \fi{}
The \emph{standard configuration} consists of the matching--based contraction, all local search methods, explicit storage of distance matrix, and no delta-gains updates.
All experiments in this section ran for the six \emph{tuning graphs} from Table \ref{tab:test_instances_walshaw}.

\vspace*{-.5cm}
\paragraph*{Initial Mapping.}

\begin{figure}[t!]
    \captionsetup[subfigure]{justification=centering}
	\centering
	\begin{subfigure}[t]{0.485\textwidth}
		\centering
		\includegraphics[width=\textwidth]{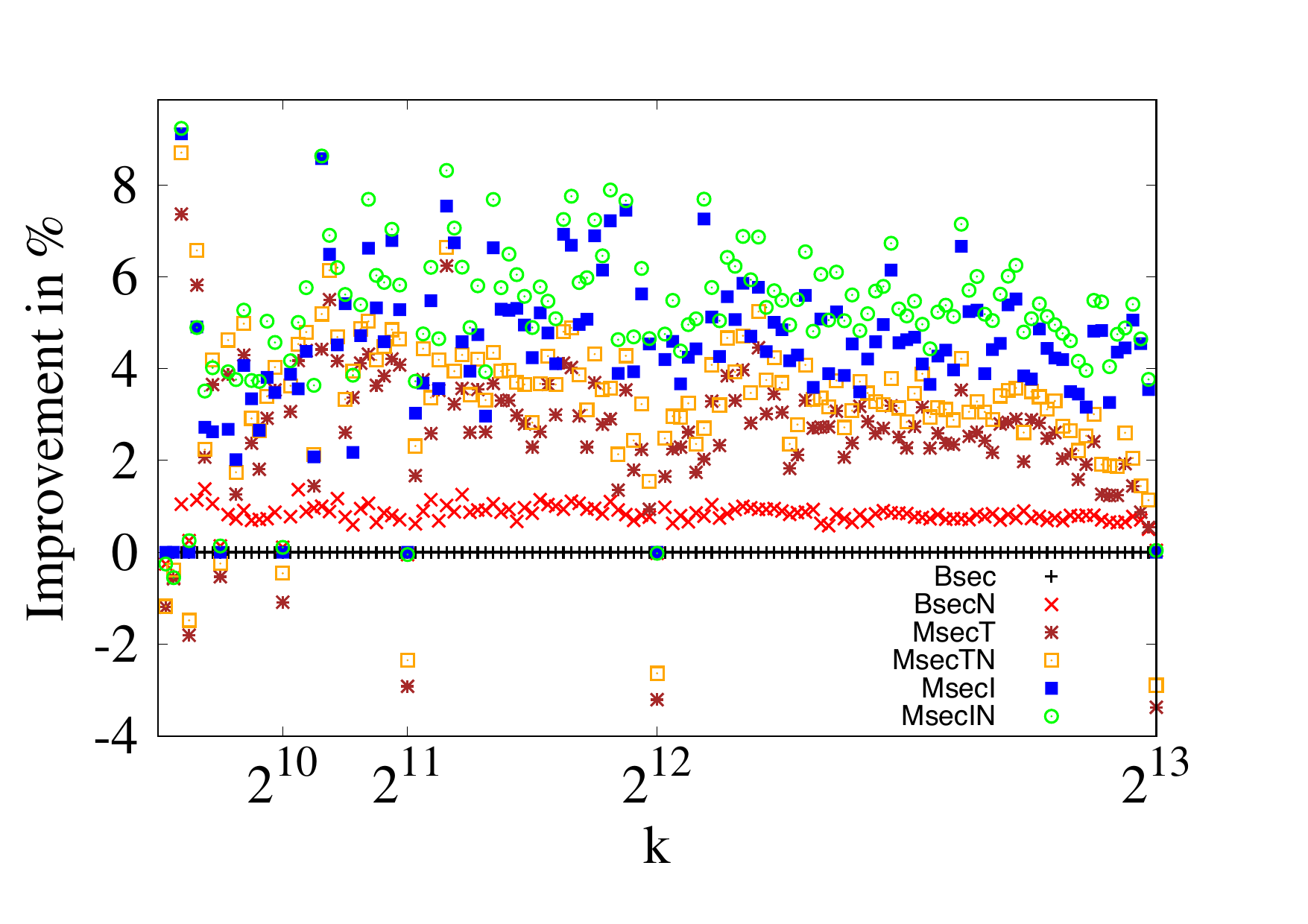}
		\caption{Improvements in objective function over Bsec. Higher is better.}
		\label{fig:InitSolRes}
	\end{subfigure}
	\begin{subfigure}[t]{0.485\textwidth}
		\centering
		\includegraphics[width=\textwidth]{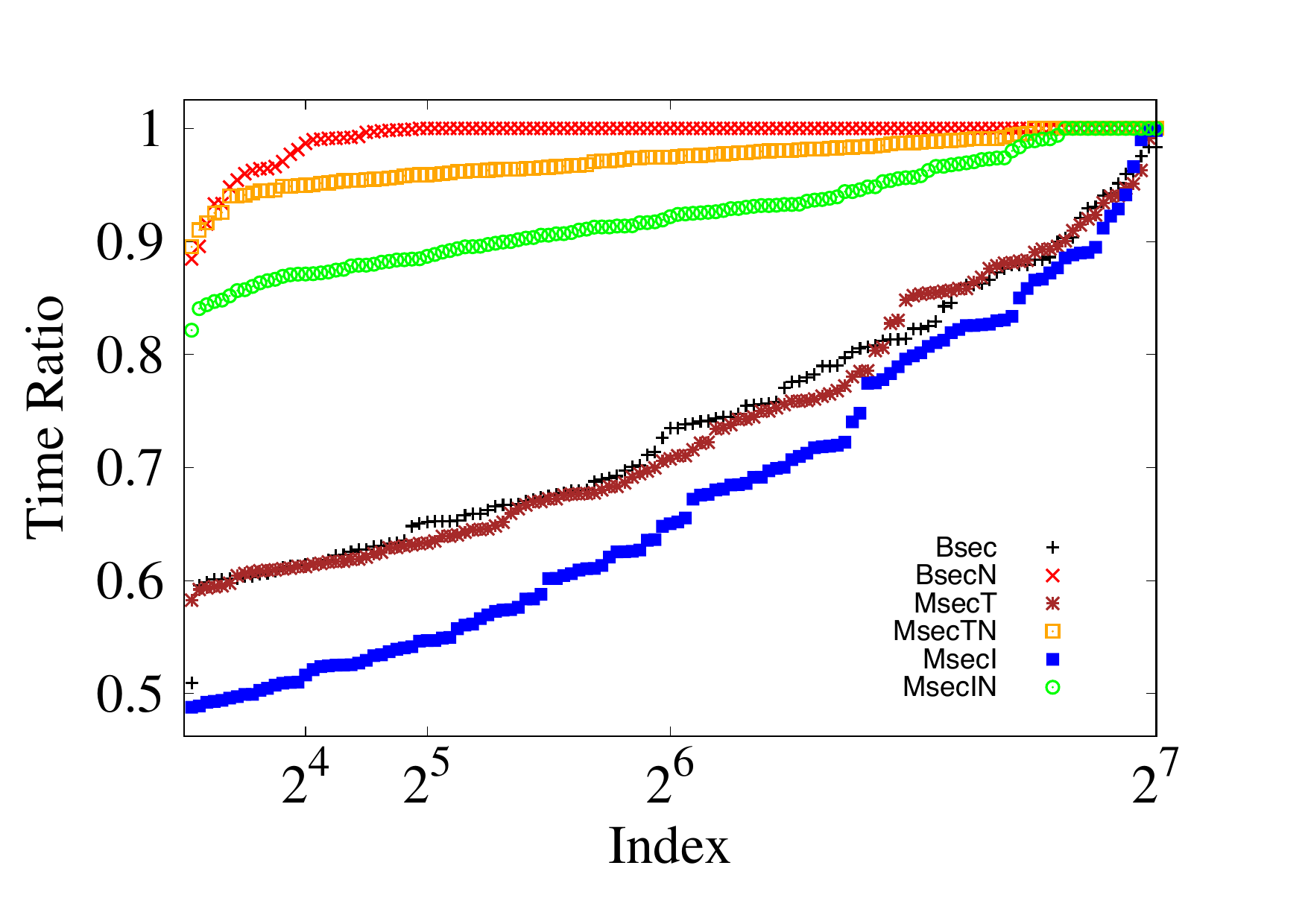}
		\caption{Performance profile for running time (ordered running time ratios). Lower is better.}
		\label{fig:InitSolTim}
	\end{subfigure}

	\caption{Comparing initial mapping algorithms from Table \ref{tab:res_init_sol}.}
	\label{fig:InitSol}
\end{figure}

For the computation of initial mappings, we consider the six configurations listed in Table \ref{tab:res_init_sol}.
Observe that Bsec and BsecN should apply either \emph{identity} or \emph{top down} depending on $k$.
This choice is based on results obtained in \cite{schulz2017better} comparing these two construction algorithms for one-to-one mapping.
Figure \ref{fig:InitSol} plots the results regarding solution quality and running time for our six configurations.

Looking at solution quality, the configurations using multisection dominate those using standard bisection except for instances having $k$ as a power of 2.
This exception was expected since the standard bisection naturally performs a multisection partition for these instances.
Among the configurations using multisection, \emph{identity} produces overall better solutions than \emph{hierarchy top down}, which is explained by the inherent locality of the multisection approach.
Finally, the $N^{10}_\mathcal{C}$ local search is the least significant factor for solution quality, although it slightly improves solution compared to the similar configurations that skip local search.

The $N^{10}_\mathcal{C}$ local search is the dominant factor regarding running time.
Observe that the configurations using \emph{identity} are always the fastest ones among those algorithms that either use $N^{10}_\mathcal{C}$ local search or among those that don't.
Hence, the construction algorithm for one-to-one mapping is the second most relevant factor for running time.
Finally, the partitioning algorithm has little influence over running time, which reflects the rather small time difference between each of the pairs \{BsecN, MsecTN\} and \{Bsec, MsecT\}.

Since MsecIN has the best overall solution quality results, it is the natural choice for \emph{strong}.
Notice that MsecI has the best overall running times, which makes it the perfect choice for \emph{fast}.
Nevertheless, it is also the second best regarding solution quality, which suffices to make it also the best choice for \emph{eco}.

\begin{table}[t]
	\scriptsize
	\centering
	\begin{tabular}{l@{\hskip 60pt}rr@{\hskip 60pt}rrr}
		\toprule	
		{Config.} & Std.Bisec. & Multisec. & Identity           & Top Down & $N^{10}_\mathcal{C}$   \\ 
                \midrule
		Bsec          & yes                & no           & if $k$ power of 2 & if $k$ not power of 2 & no  \\ 
		BsecN           & yes                & no           & if $k$ power of 2 & if $k$ not power of 2 & yes \\ 
		MsecT          & no                 & yes          & no                 & yes                    & no  \\ 
		MsecTN           & no                 & yes          & no                 & yes                    &  yes   \\ 
		MsecI          & no                 & yes          & yes                & no                     & no  \\ 
		MsecIN           & no                 & yes          & yes                & no                     &   yes  \\ 
                \bottomrule
	\end{tabular}
	\caption{Various configurations for the evaluation of different initial mapping algorithms.}
	\label{tab:res_init_sol}
\end{table}

\vspace*{-.5cm}
\paragraph*{Local Search.}

For local search experiments, we start looking at the \emph{fast} algorithm.
To obtain a fast algorithm, we restrict its number of local search methods to one.
Experiments with single local search algorithms do not yield much insight except that  label propagation with delta-gain updates yields a very good trade-off for running time and solution quality.

\begin{figure}
    \captionsetup[subfigure]{justification=centering}
	\centering
	\begin{subfigure}[t]{0.485\textwidth}
		\centering
		\includegraphics[width=\textwidth]{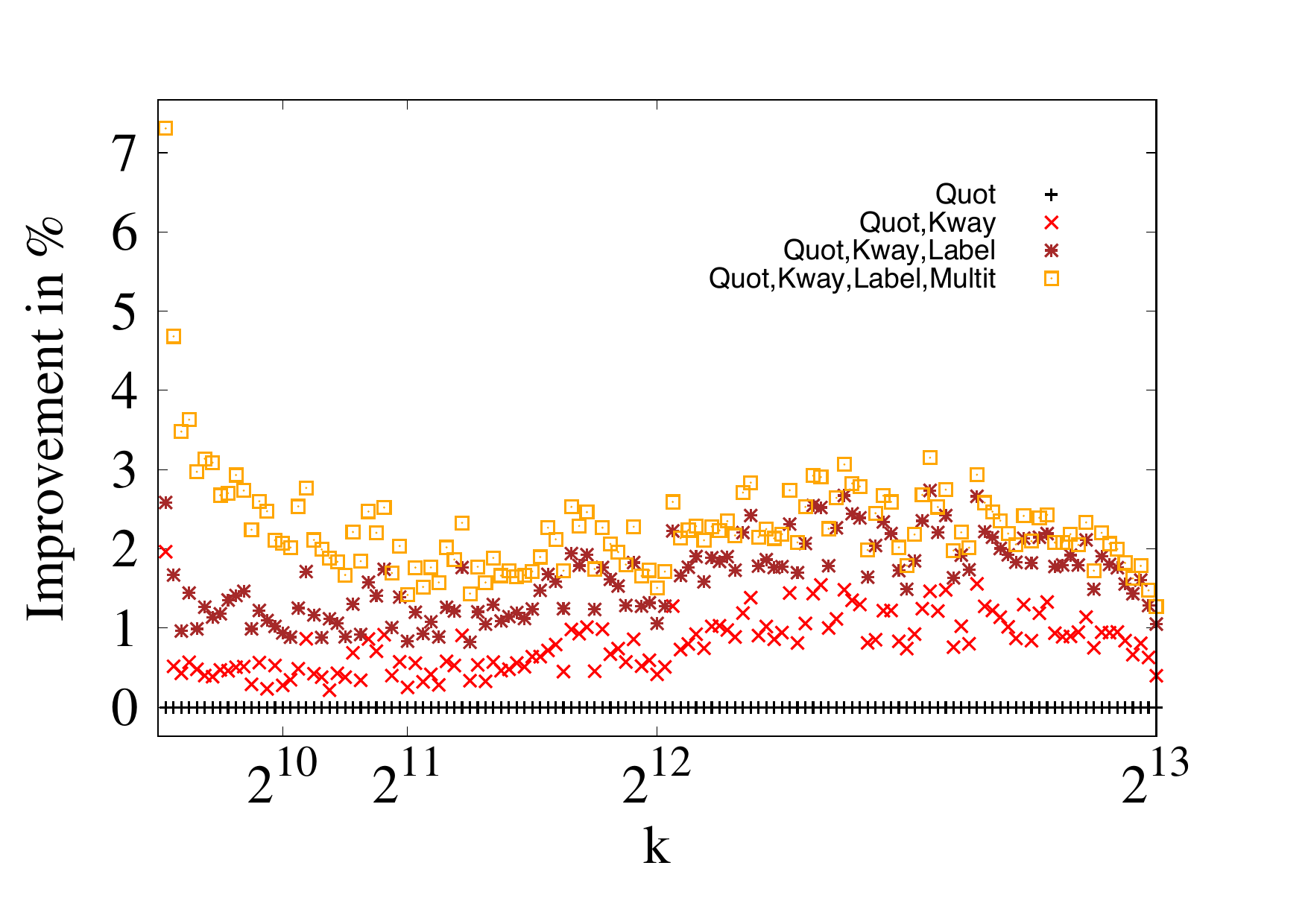}
		\caption{Improvements in objective function over Quot. Higher is better.}
		\label{fig:LSeco_Res}
	\end{subfigure}
	\begin{subfigure}[t]{0.485\textwidth}
		\centering
		\includegraphics[width=\textwidth]{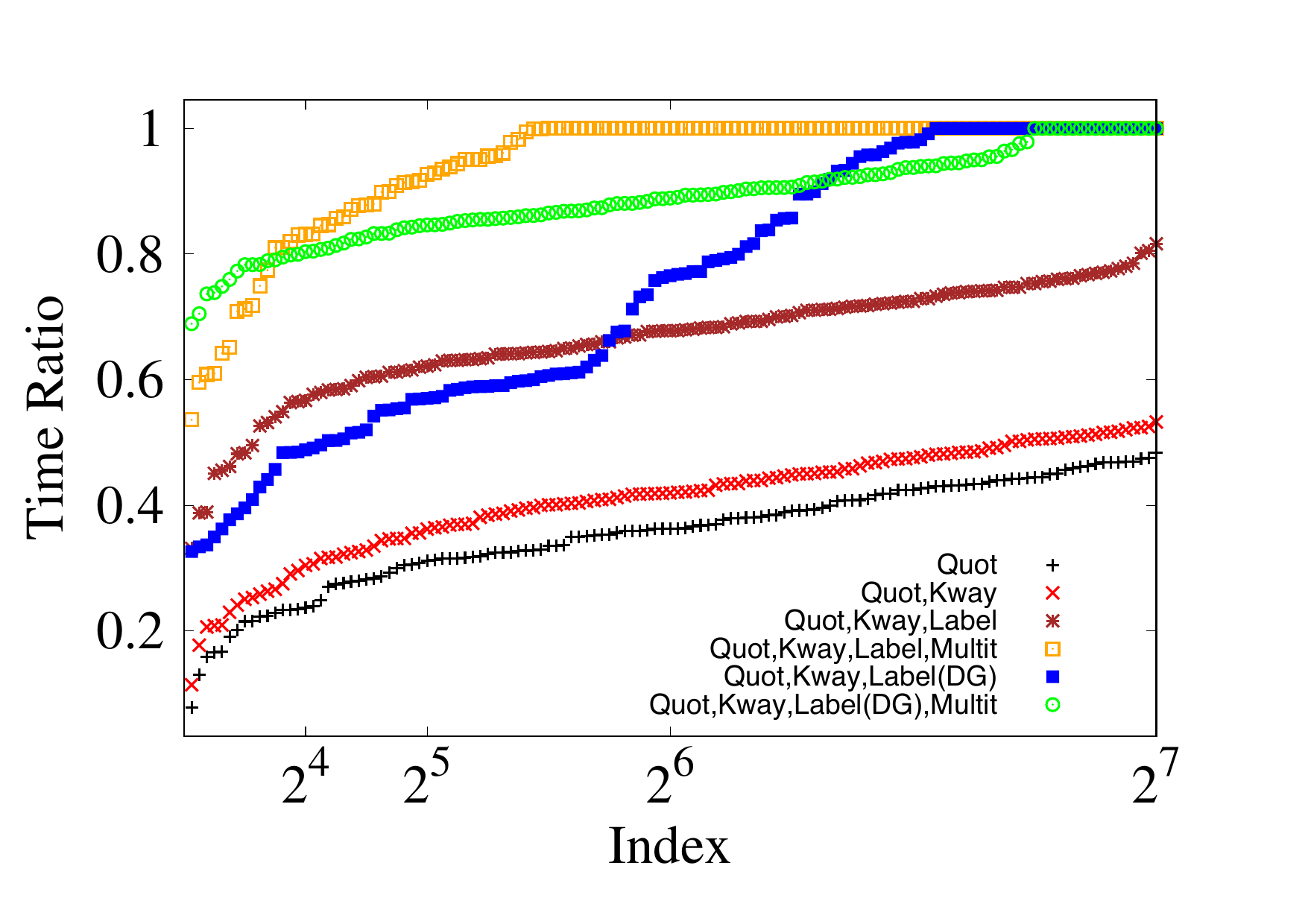}
		\caption{Performance profile for running time (ordered running time ratios). Lower is better.}
		\label{fig:LSeco_Tim}
	\end{subfigure}
	
	\caption{Results for local search experiment of the \emph{eco} algorithm. It comprises four scenarios that represent successive additions of the respective local search methods. 
	In (b), we show two additional scenarios in which delta-gain updates are used (represented by DG).}
	\label{fig:LSeco}
\end{figure}

For the \emph{eco} configuration of our algorithm, we build four configurations by incrementally inserting the local search methods. 
Additionally, we consider two extra configurations equipped with delta-gain updates during label propagation.
Figure \ref{fig:LSeco} summarizes the results concerning these six configurations.
Since the behavior of \emph{strong} in this experiment is equivalent, we omit its results without loss of completeness.

Figure \ref{fig:LSeco} shows that solution quality and running time consistently increase after each consecutive addition of local refinement methods.
Regarding delta gains, running times decrease for some values of $k$ but increase considerably and irregularly for others.
Since this behavior is undesirable for \emph{eco}, we drop delta gains for it.
We also drop delta gains for \emph{strong} since it does not affect solution quality and has negligible influence on running time compared to the $N^{10}_\mathcal{C}$ refinement.
The clear choice for \emph{strong} is the configuration with the four local searches since all of them contribute to incrementally improve solution quality.
For \emph{eco}, we drop only the multi-try FM local search since it adds little to solution quality but significantly increases the running time.

\vspace*{-.5cm}
\paragraph*{Distance Matrix.}

As the objective function is not influenced by the representation of the distance matrix, we only evaluate running time.
We test four configurations: one with each of the three techniques that imply the distance matrix and a reference scenario in which we store the full distance matrix.
Since all configurations of our algorithm display equivalent behavior, we focus on \emph{strong} without loss of generality.
Figure \ref{fig:DistMatrix} plots a running time ratio chart for \emph{strong}.
\begin{figure}[t]
    \captionsetup[subfigure]{justification=centering}
	\centering
        \vspace*{-1cm}
		\includegraphics[width=0.5\textwidth]{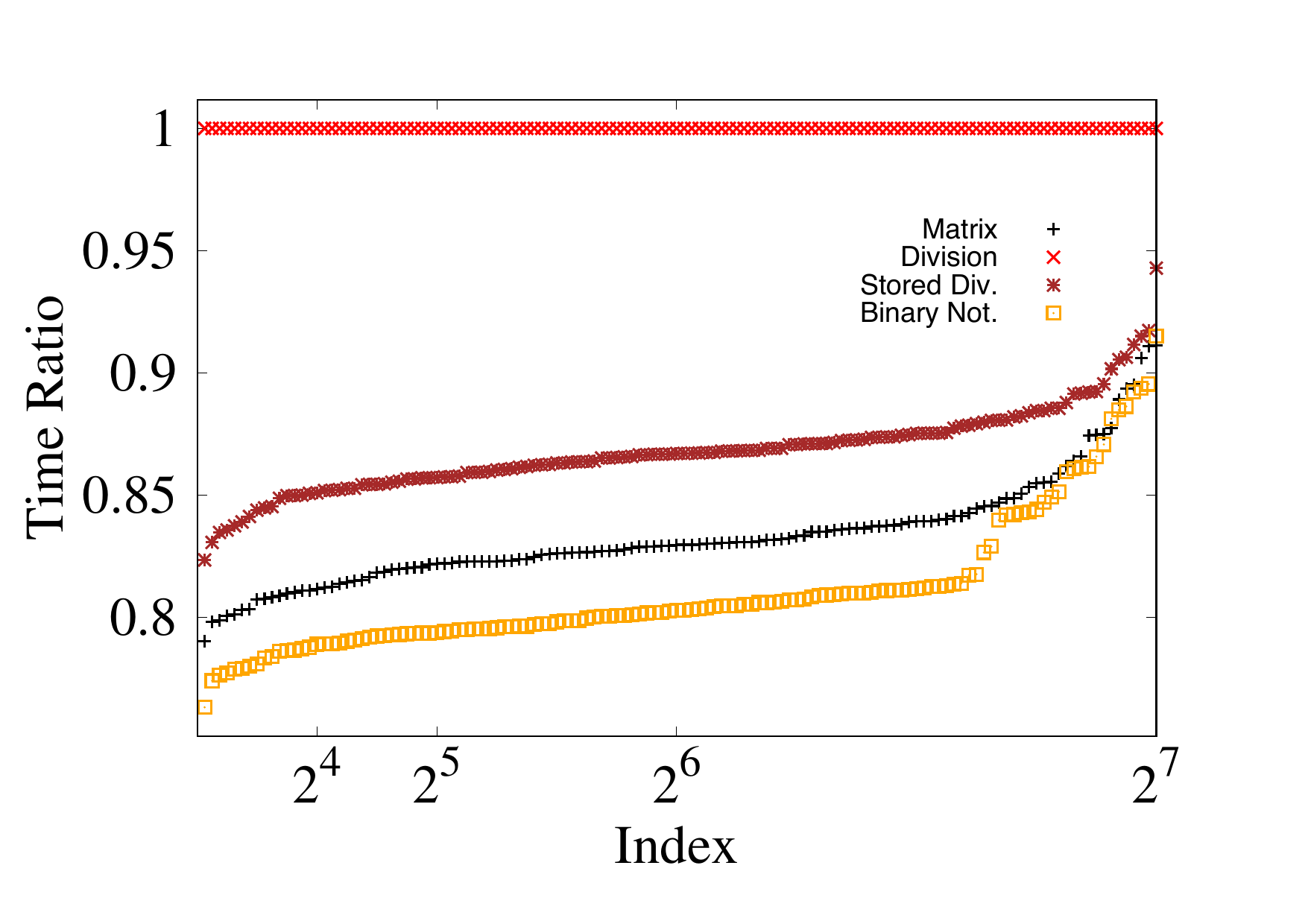}

	\caption{Performance plot for running times of different algorithm configurations (ordered running time ratios for different distance matrix implementations). 
	Four alternative configurations are considered: implicit representations of the distance matrix based on division, stored division, and binary notation, and a reference scenario in which the matrix is explicitly stored. Lower is better.}
	\label{fig:DistMatrix}
\end{figure} 
It is easily understandable that the \emph{binary notation} technique is faster than the \emph{stored division}--based approach, and also that the latter is faster than the \emph{division}--based approach.
On the
   other hand, the \emph{binary notation} outperforms the full-distance matrix approach. While both approaches allow O(1) distance calculations on
   our x86\_64 architecture, accessing the distance matrix incurs a
   memory access.
This leads to frequent cache misses since the $O(k^2)$-sized distance matrix does not fit into the cache of our machine.
Lastly, not using the distance matrix significantly improves the memory footprint of the algorithm. This especially prounced if the number of blocks gets very large. For example, for $2^{15}$ blocks, not using the distance matrix saves roughly an order of magnitude in memory.

After the tuning step, the three configurations of our algorithm ended up as follows:
(i)~\emph{fast} applies MsecI, label propagation with delta-gain updates, and binary notation;
(ii)~\emph{eco} applies MsecI, quotient graph refinement, $k$-way FM, label propagation, and binary notation; and
(iii)~\emph{strong} applies MsecIN, quotient graph refinement, $k$-way FM, label propagation, multi-try FM, and binary notation.
To improve speed even more, we also include a configuration called \emph{fastest} which applies MsecI as initial mapping, does not use any local search during uncoarsening, and never needs to use information from the distance matrix.

\subsubsection{State of the Art}
\label{subsec:State of the Art}
In this section, we compare our algorithms against the best alternative algorithms in the literature.
We report experiments on all graphs listed in Table~\ref{tab:test_instances_walshaw} -- excluding the graphs from the tuning set. 
Figure \ref{fig:intmap_GenSol} gives an overview over our results.

\begin{figure}[t!]
    \captionsetup[subfigure]{justification=centering}
	\centering
        \vspace*{-.75cm}
	\begin{subfigure}[t]{0.495\textwidth}
		\centering
		\includegraphics[angle=-0, width=\textwidth]{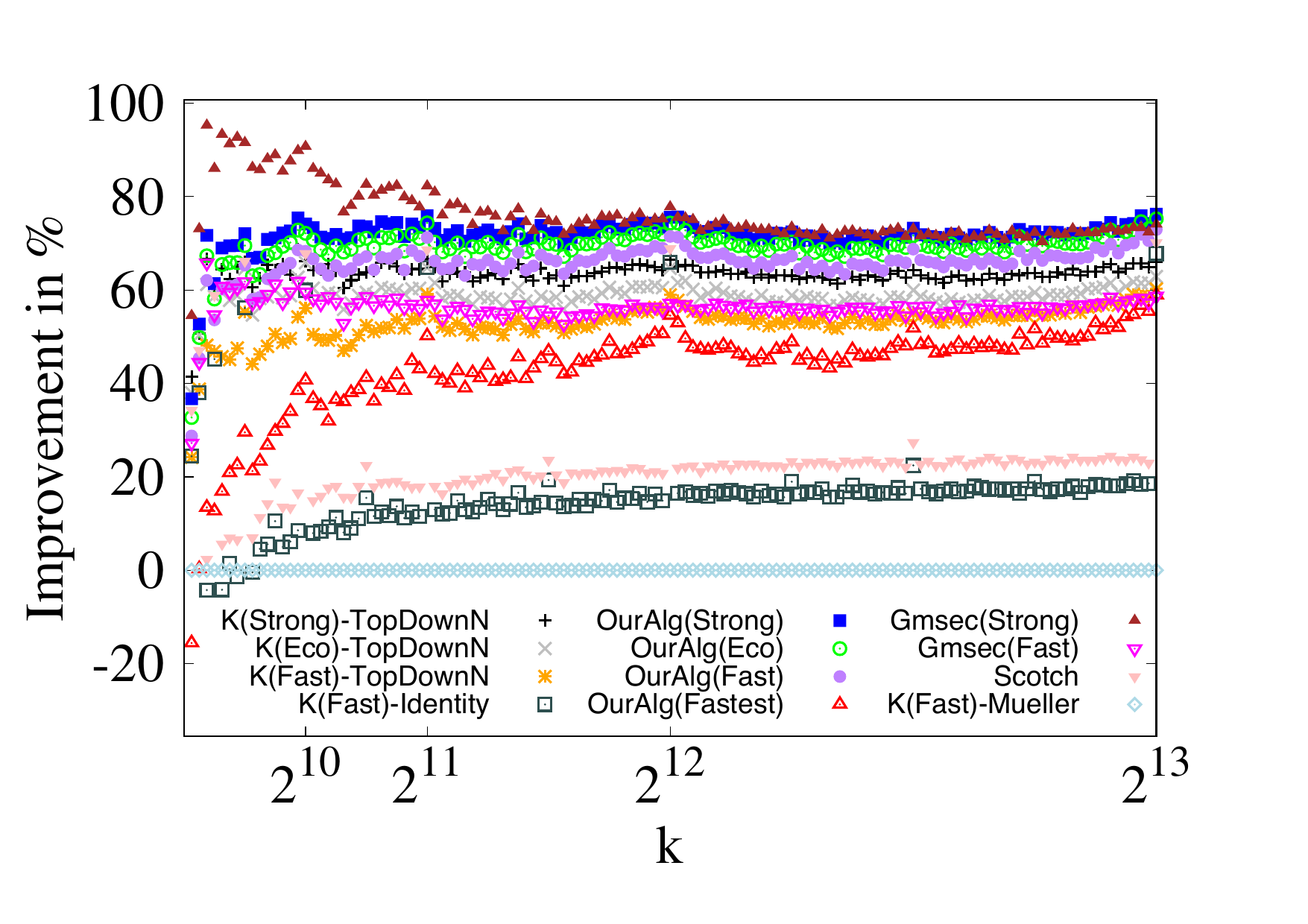}
		\caption{Improvements in objective function over \AlgName{K(Fast)}-Müller-Merbach. Higher is better.}
		\label{fig:intmap_GenSolRes}
	\end{subfigure}
	\begin{subfigure}[t]{0.495\textwidth}
		\centering
		\includegraphics[angle=-0, width=\textwidth]{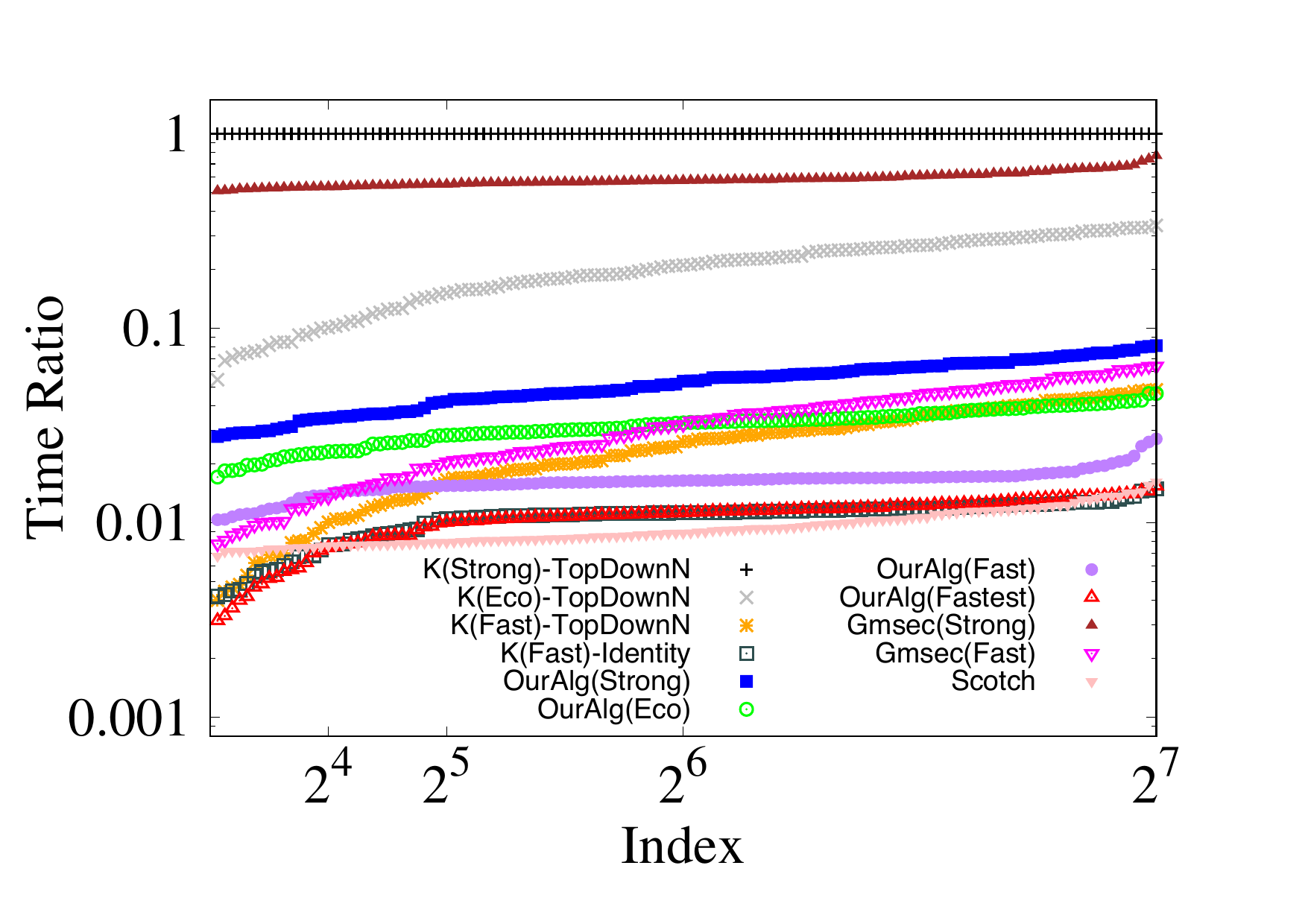}
		\caption{Performance profile for running time.  Lower is better.}
		\label{fig:intmap_GenSolTim}
	\end{subfigure}

        \vspace*{-.25cm}
	\caption{Comparisons against state-of-the-art approaches for process mapping.}
        \vspace*{-.5cm}
	\label{fig:intmap_GenSol}
\end{figure}

Overall, \AlgName{Scotch} has the lowest average running time, directly followed by our algorithm \emph{fastest}, \AlgName{K(Fast)}-\AlgName{Identity}, and our algorithm \emph{fast} (respectively $9\%$, $10\%$, and $73\%$ slower than \AlgName{Scotch} on average).
The average running time of \AlgName{K(Fast)}-TopDownN and \AlgName{Gmsec(Fast)} are respectivelly a factor $2.3$ and $3.1$ higher than \AlgName{Scotch}.
For our algorithms \emph{eco} and \emph{strong}, this factor is respectively $3.3$ and $5.4$.
By definition ScotchTC is also a factor $5.4$ higher than \AlgName{Scotch}. 
Next, \AlgName{Gmsec(Eco)} and \AlgName{K(Eco)}-TopDownN have much higher runtimes ($9.3$ and $20.4$ times slower than \AlgName{Scotch}, respectively).
\AlgName{Gmsec(Strong)} and \AlgName{K(Strong)}-TopDownN are the slowest ones ($62$ and $107$ times slower than \AlgName{Scotch}, respectively).

We now highlight the comparison of various configurations/algorithms.
\AlgName{Gmsec(Strong)} is the algorithm with best overall mapping quality. It is 2.5\% on average better compared to our \emph{strong} configuration, however our \emph{strong} configuration is more than an order of magnitude faster on average -- a factor 11.5 faster.
Better quality of \AlgName{Gmsec(Strong)} stems from the fact that global multi-section itself already takes the system hierarchy into account and hence yields good models to be mapped.
Moreover, the graph partitioning approach itself which is used to compute a communication graph also uses more (time-consuming) sophisticated local search algorithms. This includes methods such as flow-based methods which particularly work well for small values of $k$ as well as global search methods such as V-cycles.
In particular for $k > 2^{11}$, our \emph{strong} has the same average quality as \AlgName{Gmsec(Strong)} but is $9.3$ times faster.
Our algorithm computes similar solutions in much less time since the multilevel algorithm directly optimizes the correct objective. 
For $k > 2^{11}$, our \emph{eco} is $2.4\%$ better and $2.8$ times faster than \AlgName{Gmsec(Eco)}, and our \emph{fast} is $9\%$ better and $2.6$ times faster than \AlgName{Gmsec(Fast)}.
Overall, our \emph{strong} configuration improves solution quality over \AlgName{K(Strong)}-TopDownN by $5.1$\% while being a factor $20$ on average faster.
Our \emph{eco} configuration has roughly $3.6\%$ better quality than \AlgName{K(Strong)}-TopDown but is a factor $32$ faster on average. 
Our \emph{fast} configuration still yields $1.3\%$ better solutions than \AlgName{K(Strong)}-TopDown on average, and is a factor $62$ faster.
Here, improvements stem from the fact that the new algorithms are integrated and not two-phase as well as the fact that these algorithms do not perform multi-sections.
Lastly, our \emph{fastest} algorithm is on average $9\%$ slower than \AlgName{Scotch} but also improves solution quality over \AlgName{Scotch} by $16\%$.
Our \emph{strong} algorithm is $40\%$ better than \AlgName{Scotch} and consumes a factor $5.4$ more running time.

\section{Multilevel Signed Graph Clustering}
\label{sec:sgc_Multilevel Signed Graph Clustering}
\label{chap:sgc_Multilevel Signed Graph Clustering}
\label{chap:sgc_A Distributed Multilevel Memetic Algorithm for Signed Graph Clustering}

In this section, we make a first attempt at solving the signed graph clustering problem by leveraging some of the most effective techniques from graph partitioning that minimize edge-cut.
We introduce and thoroughly engineer a memetic algorithm specifically designed for this problem.
The memetic algorithm features a multilevel approach, which encompasses a coarsening-uncoarsening process and efficient local search methods.
The memetic algorithm is further enhanced with natural multilevel recombination and mutation operations.
We also parallelize our approach using a scalable coarse-grained island-based strategy which has already shown to be scalable in practice. 
Experimental results show that our memetic algorithm is able to converge effectively and demonstrate that it outperforms the state-of-the-art with respect to edge-cut, producing significantly better solutions \hbox{in 1 hour of processing on 16 elements}.

% \vfill \pagebreak
\subsection{Multilevel Algorithm}
\label{sec:sgc_Multilevel Signed Graph Clustering}

We now describe  a novel multilevel algorithm for signed graph clustering which is an important component of our memetic algorithm.
We start with our multilevel strategy, then we discuss \hbox{each of its components}.

\subsubsection{Overall Multilevel Strategy}
\label{subsec:sgc_Overall Multilevel Strategy}

In this section, we describe our overall multilevel strategy for signed graph clustering.
Our multilevel approach for graph partitioning differs from others in that it lacks a separate algorithm for initial solution computation. 
Instead, our scheme starts with coarsening to produce a hierarchy of coarser graphs. 
Upon completion, the coarsest graph is assigned individual clusters for each node, resulting in an initial clustering. 
The uncoarsening phase maps the current best clustering all the way back to the original graph while running local search approaches on each level. 
Two consecutive cycles of this multilevel process are performed, as respectively shown in Figures~\ref{fig:sgc_MSGC} and \ref{fig:sgc_MSGC_second_cycle}. 
In the second cycle, the previously obtained clustering is forced as the initial solution on the coarsest level, to be further refined during uncoarsening.
Our overall multilevel algorithm \hbox{is outlined in Algorithm~\ref{alg:sgc_overall_multilevel_strategy}}.

\begin{algorithm}
	\small
	\textbf{Input} graph $G=(V,E)$ \\
	\textbf{Output} clustering $\Pi: V \rightarrow \MdN$ 
	\begin{algorithmic}[1]  
		\State $G_0 \gets \emptyset$; $G_1 \gets G$; $i \gets 0$ %
		\While{$G_i \neq G_{i+1}$}
			\State $i \gets i+1$
			\State ComputeClustering($G_i$) %
			\State $G_{i+1} \gets$ Contract($G_i$)  %
		\EndWhile
		\State $\Pi_i \gets$ MapNodesToClusters($G_i$) %
		\For{$j\gets i-1,\ldots,1$}
			\State $\Pi_j \gets$ RemapToFiner$(\Pi_{j+1})$ %
			\State $\Pi_j \gets$ RefineClustering$(G_j,\Pi_j)$ %
		\EndFor
		\State $\Pi \gets$ GlobalSearch$(\Pi_{1})$ %
	\end{algorithmic}
	\caption{Our multilevel algorithm}
	\label{alg:sgc_overall_multilevel_strategy}
\end{algorithm}

\subsubsection{Coarsening}
\label{subsec:sgc_Coarsening}

In this section, we explain our \emph{coarsening} phase.
As shown in Algorithm~\ref{alg:sgc_overall_multilevel_strategy}, it consists of two consecutive steps which are iteratively repeated on the current coarsest graph until no more contraction can be performed without increasing the overall edge-cut value. 
In the first step, a graph clustering is computed.
In the second step, each cluster contained in this clustering is contracted \hbox{in order to produce a coarser graph}. 

To be self-contained, we briefly outline our clustering algorithm, which is based on \emph{label propagation}~\cite{labelpropagationclustering}.
In a graph with $n$ nodes and $m$ edges, size-constrained label propagation can be implemented in $O(n+m)$ time. 
The algorithm starts with each node in its own cluster, i.e., its initial cluster ID is set to its node ID.
The algorithm then works in rounds.  
In each round, all nodes are traversed, and each node $u$ is moved to the cluster with the largest connection weight with $u$ in case it is strictly positive, i.e., it is only moved to the cluster $V_i$ with maximum connection weight $\omega(\{(u, v) \mid v \in N(u) \cap V_i \})$ if it is strictly positive. 
Ties are broken randomly. 
This greedy approach ensures that the edge-cut is monotonically decreased. 
The algorithm is repeated at most $L$ times, where $L$ is a tuning parameter.
The intuition of the overall procedure is to force positive edges inside clusters and negative edges between clusters, which is \hbox{favorable for the edge-cut objective}.

After the label propagation, each cluster is replaced by a single node (as exemplified in Figure~\ref{fig:sgc_clustercontraction}), creating a hierarchy of coarser graphs. 
The contraction ensures a clustering of the coarse graph corresponds to the same edge-cut in finer graphs. 
This aggressive graph-contraction strategy allows us to shrink irregular networks, and the coarsening stops when no further contraction can be done without increasing the edge-cut. 
The final output is the coarsest graph, which implies an initial clustering of assigning each node to its own cluster. 
This coarsening phase serves as a method to construct an initial clustering, starting with $n$ clusters and decreasing the number of clusters and \hbox{the edge-cut at each hierarchy level}.

\subsubsection{Uncoarsening}
\label{subsec:sgc_Uncoarsening}

In this section, we explain our \emph{uncoarsening} phase.
It is executed directly after the coarsening phase has been completed and every node of the coarsest graph has been assigned to its own cluster.
As shown in Algorithm~\ref{alg:sgc_overall_multilevel_strategy}, our uncoarsening phase consists of two consecutive steps which are iteratively repeated for each graph in our contraction hierarchy, from the coarsest one to the finest one.
In the first step, the current clustering is mapped to the graph contained in the current level of the contraction hierarchy.
In the second step, we apply a sequence of local refinement methods to optimize our clustering by moving \hbox{nodes between clusters}.

At each level of uncoarsening, we remap the best clustering to the next finer graph using our contraction approach, maintaining edge-cut as in the example shown in Figure \ref{fig:sgc_clustercontraction}. 
We apply two refinement methods: 
First, a \emph{label propagation} refinement randomly visits all nodes and greedily moves each one to the cluster where edge-cut is minimized.
Second, a \emph{$k$-way Fiduccia-Mattheyses} (FM)~\cite{fiduccia1982lth} refinement greedily goes through the boundary nodes trying to relocate them with a more global perspective in order to improve solution quality. 
Both methods are based on the concept of \emph{Gain}, defined as the increase \hbox{in edge-cut caused by node movement}.

\paragraph*{Label Propagation Refinement.}

During uncoarsening, our of our local search procedures is based on \emph{label propagation}~\cite{labelpropagationclustering}. 
The algorithm works in rounds.
In each round, the algorithm visits all nodes in a random order, starting with the labels being the current assignment of nodes to clusters.
When a node $v$ is visited, it is moved to the neighboring cluster with highest positive gain.
Ties are broken randomly and a 0-gain neighboring cluster can be occasionally chosen with $50\%$ probability if there is no neighboring cluster with positive gain.
We perform at most~$\ell$ rounds of the algorithm, \hbox{where $\ell$ is a tuning~parameter}.

\paragraph*{K-Way FM Refinement.}
\label{subsec:sgc_K-Way_FM_Refinement}

Our $k$-way FM refinement was adapted from the implementation in~\cite{kaffpa}, which is an adapted version of the FM algorithm~\cite{fiduccia1982lth} to move nodes between clusters.
Unlike the original algorithm, FM algorithm, our refinement moves nodes globally, rather than restricting movements to pairs of clusters. 
We use a gain-based priority queue initialized with the \emph{complete} partition boundary.
We repeatedly move the highest-gain node to its best neighboring cluster, updating the queue by adding its neighbors not yet moved. 
The search stops when the queue is empty (i.e., each node was moved once) or a random-walk stopping criterion described in~\cite{kaffpa} is met. 
The refinement escapes local optima by allowing negative-Gain movements and rolling back to the lowest cut that satisfies \hbox{the balance criterion}.

\subsubsection{Global Search}
\label{subsec:sgc_Global Search}

After the standard multilevel cycle described in Algorithm~\ref{alg:sgc_overall_multilevel_strategy}, a global search is run using an extra multilevel cycle almost identical to the standard one, with the previous clustering used as the initial solution.
Figure~\ref{fig:sgc_MSGC_second_cycle} illustrates this global search procedure. 
Its difference to the standard multilevel cycle is that the previously computed clustering is forced to be the initial solution in the extra multilevel cycle. 
This is achieved by blocking cut edges from contraction and stopping coarsening when the coarser graph equals the quotient graph of the clustering. 
This global search strategy was introduced in \cite{walshaw2004multilevel} and has been successfully used for many optimization problems.
This approach is extended in Section~\ref{sec:sgc_Distributed Evolutionary Signed Graph Clustering} to \hbox{form the operators of our evolutionary algorithm}.

\begin{figure}[t]
	\centering
	\includegraphics[width=.75\textwidth]{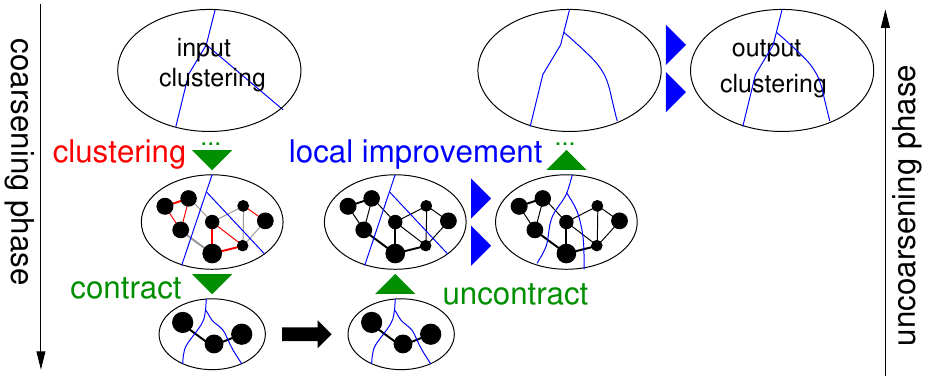}
	\caption{Global Search (adapted from~\cite{kaffpa}).}
	\label{fig:sgc_MSGC_second_cycle}
\end{figure}

\subsection{Distributed Evolutionary Algorithm}
\label{sec:sgc_Distributed Evolutionary Signed Graph Clustering}

In this section, we extend our multilevel algorithm to a distributed evolutionary framework.  
Our approach is directly inspired by the evolutionary algorithm proposed in~\cite{sanders2012distributed} for graph partitioning.
We start by providing a broad description of our overall evolutionary strategy.
Then, we discuss each of its algorithmic components and show how we use our multilevel algorithm as a tool to implement evolutionary operators such as recombination \hbox{and mutation}.

\subsubsection{Evolutionary Strategy}
\label{subsec:sgc_Evolutionary Strategy}

In this section, we describe the overall strategy of our distributed evolutionary algorithm, which is summarized in Algorithm~\ref{alg:sgc_overall_evolutionary_strategy}.
As a typical evolutionary approach, our algorithm is a population-based heuristic which mimics the biological process of evolution to optimize for our problem.
Given a signed graph~$G$, we define an \emph{individual} as a particular clustering on $G$.
Each processing element initializes a population~$P$ of $\alpha$ individuals using different random seeds.
Over many rounds or \emph{generations}, the population evolves via either \emph{mutation} (with probability $\beta=10\%$) or \emph{recombination} (with probability $1-\beta=90\%$).
In particular, the process continues until a desired time limit~$t_f$ is reached.
The \emph{recombination} procedure is executed on two \emph{good} individuals selected via \emph{tournament} and then used as \emph{parents} to produce \emph{offspring} which partially inherits their characteristics.
The \emph{mutation} procedure is run on a randomly-picked individual which is used as reference for building another individual from scratch. 
The next phase is called \emph{replacement}.
New individuals are inserted into the population and worst ones are \emph{evicted} to maintain $\alpha$ population size.
Each running instance of our algorithm generates only one individual per generation, hence it is a \emph{steady-state} evolutionary algorithm~\cite{evolutionary_book}.
After each generation, the best individual found so far by each instance is \hbox{shared among processing elements}.

\begin{algorithm}
        \small
	\textbf{Input} graph $G=(V,E)$ \\
	\textbf{Output} clustering $\Pi: V \rightarrow \MdN$ 
	\begin{algorithmic}[1]  
		\State $P \gets \{\Pi_1, \ldots, \Pi_\alpha\}$ %
		\While{running time $< t_f$}
			\State $i \gets$ Random$([0,1])$ %
			\If{$i > \beta$}
				\State $\Pi_{a}, \Pi_{b} \gets $ Selection($P$) %
				\State $\Pi_c \gets $ Recombination($P$)  %
			\Else
				\State $\Pi_c \gets $ Mutation($\Pi_a$) %
			\EndIf
			\State $P \gets $ Replacement($P,\Pi_c$)  %
			\State CommunicateBest($P$)  %
		\EndWhile
		
		\State $\Pi \gets $ BestIndividualOverall() %
	\end{algorithmic}
	\caption{Our distributed evolutionary algorithm}
	\label{alg:sgc_overall_evolutionary_strategy}
\end{algorithm}

\subsubsection{Notation, Population, and Fitness}
\label{subsec:sgc_Notation, Population, and Fitness}

In this section, we describe how the evolutionary concepts of notation, fitness function, and population are defined and operated within our algorithm.
We define an individual as a particular clustering $\Pi$ on the graph $G$ and represent it by its set of cut edges.
This way, each clustering has a unique notation, i.e., there is no symmetry to represent individuals.
Note that some sets of edges cannot be the set of cut edges of any feasible clustering, e.g., a single edge of a cycle.
However, feasibility is \emph{implicitly} ensured throughout the algorithm, as will become clear later.
As a consequence, we do not need any penalty function, hence the \emph{fitness} function of an individual is its \emph{edge-cut}, which is the objective to be minimized.
The initial population is built by running our multilevel algorithm (Algorithm~\ref{alg:sgc_overall_multilevel_strategy}) from scratch $\alpha$ times, where $\alpha \geq 3$ is automatically determined such that roughly $10\%$ of the allowed time limit $t_f$ is spent \hbox{on the construction of the initial population}.

\subsubsection{Recombination}
\label{subsec:sgc_Recombination}

\begin{figure}[t]
	\centering
	\includegraphics[width=.8\textwidth]{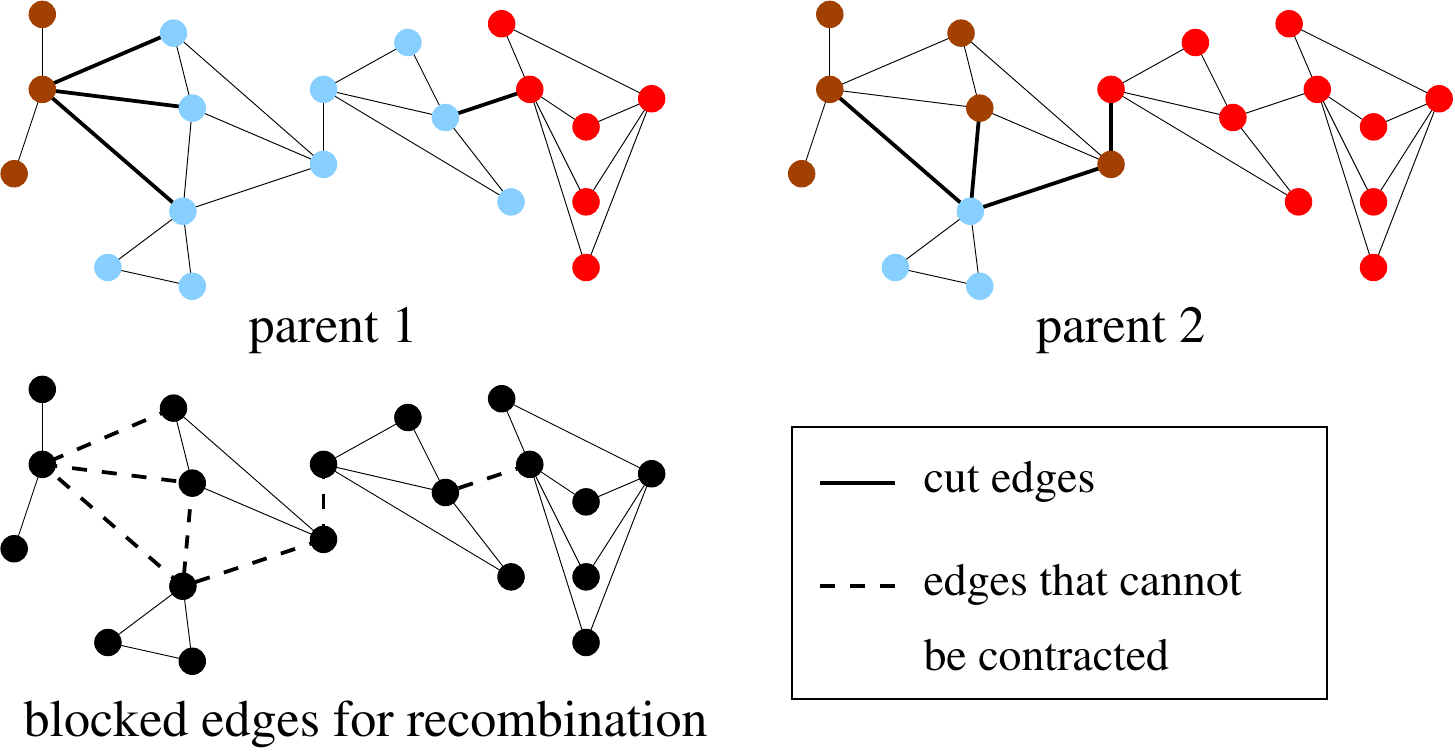}
	\caption{In the recombination, the cut edges of both parents cannot be contracted during the first multilevel cycle to build an offspring.}
	\label{fig:sgc_recombination}
\end{figure}

In this section, we explain the operation of the recombination operator in our evolutionary algorithm. 
We select two \emph{good} individuals, $\Pi_a$ and $\Pi_b$, from the population to serve as parents using a \emph{tournament}~\cite{tournament_selection}, i.e., two distinct individuals are drawn and the one with best edge-cut is selected to be a parent. 
To ensure distinct parents, the second parent may be the loser of the tournament if the winner is already the first parent. 
To create an offspring, we run Algorithm~\ref{alg:sgc_overall_multilevel_strategy} on $G$ and block the cut edges of $\Pi_a$ and $\Pi_b$ from contracting during the first multilevel cycle. 
Figure~\ref{fig:sgc_recombination} demonstrates which edges are blocked for a particular instance.
This leads to a coarsening phase that stops when no further contraction is possible unless the edge-cut increases or blocked edges are contracted. 
We then choose the best edge-cut clustering among three options: the two parents and a clustering where each node of the coarsest graph is assigned its own cluster. 
Our algorithm then uncoarsens this new clustering, applies local searches, and further improves the clustering with the global search in Figure~\ref{fig:sgc_MSGC_second_cycle}. 
This approach directly ensures non-increasing edge-cut compared to $\Pi_a$ and $\Pi_b$ and combines the characteristics of both parents, \hbox{as expected from recombination}.

\subsubsection{Mutation}
\label{subsec:sgc_Mutation}

In this section, we explain the mutation operator in our algorithm. 
We randomly select an individual $\Pi_a$ from the population and create a new individual by running Algorithm~\ref{alg:sgc_overall_multilevel_strategy} on $G$ while blocking $\Pi_a$'s cut edges from contracting only in the \emph{first level} of the first multilevel cycle. 
The initial solution obtained at the finer level of this multilevel cycle can differ from $\Pi_a$ but inherits some of its characteristics due to this edge-blocking constraint. 
The new individual is obtained by optimizing this initial solution during uncoarsening as well as the global search cycle. 
This operator does not guarantee non-increasing edge-cut compared to $\Pi_a$, although it increases population variability and optimizes the new individual \hbox{as much as possible}.

\subsubsection{Replacement}
\label{subsec:sgc_Replacement}

In this section, we explain the replacement phase in our algorithm. 
It inserts the newly generated individual $\Pi_c$ from mutation or recombination into the population $P$ and removes an individual to keep only $\alpha$ individuals alive. 
We first check if $\Pi_c$ has worse edge-cut than all individuals in $P$. 
If so, we do not insert $\Pi_c$ into $P$. 
Otherwise, we remove the individual in $P$ most similar to $\Pi_c$, where a high similarity is defined as a low-cardinality symmetric difference between the two individuals' edge sets. 
This approach helps maintain high diversity in $P$, which is \hbox{beneficial for the evolutionary process}.

% \vfill \pagebreak
\subsubsection{Parallelization}
\label{subsec:sgc_Parallelization}

In this section, we explain how our evolutionary algorithm can be executed across distributed processing elements.
We follow the \emph{coarse-grained} parallelization approach used in~\cite{sanders2012distributed} for a similar problem. 
Each instance follows the same steps outlined in Algorithm~\ref{alg:sgc_overall_evolutionary_strategy}. 
Communication is done through a variation of the \emph{randomized rumor spreading} protocol~\cite{randomgossip}. 
The communication is organized in rounds, in which each process tries to send and receive individuals. 
In each round, a process sends the best individual from its local population to a randomly selected process it has not yet sent this individual to. 
Afterwards, it tries to receive incoming individuals and inserts them into its local population using the described replacement strategy. 
The overall Algorithm~\ref{alg:sgc_overall_evolutionary_strategy} is performed asynchronously, i.e., without global synchronization. 
The communication rounds can also happen only every $x$ iterations to save communication time (we use $x=1$ \hbox{on our very fast communication network)}.

\subsection{Experimental Evaluation}
\label{sec:sgc_Experimental Evaluation}

\paragraph*{Setup.} 
We performed the implementation of our algorithms and competing algorithms inside the KaHIP framework (using C++) and compiled them using gcc 9.3 with full optimization turned on (-O3 flag). 
We have used Machine~C, which has sixteen cores.

\paragraph*{Baselines.}
As the representative of the state-of-the-art for signed graph clustering, we use the general framework \AlgName{GASP}~\cite{bailoni2019generalized}, which is publicly available.
In particular, the two best algorithms implemented in it with respect to the minimization of edge-cut, namely \AlgName{HCC-Sum}~\cite{bailoni2019generalized} and \AlgName{GAEC}~\cite{keuper2015efficient}.
As both of these algorithms have deterministic implementations, we only repeat each experiment with them once.
We requested implementations from the authors of~\cite{HuaRandomWalk2020}, but we have received no response by the time this manuscript was submitted. 
However, based on the $z\_value$ reported in~\cite{HuaRandomWalk2020} and preliminary experiments with \AlgName{GAEC}, \AlgName{GAEC} outperforms this method. 
\hbox{Hence, we do not include~it~here}.

\paragraph*{Instances.}
In our experiments, we use all the real-world signed graphs listed in Table~\ref{tab:sgc_graphs}.
Before running the experiments, we converted each of the graphs in Table~\ref{tab:sgc_graphs} to an undirected signed graph without parallel or self edges.
In particular, we achieve this by simply removing self edges and substituting all parallel and opposite arcs by a single undirected edge whose weight equals the sum of the weights of these arcs.

\paragraph*{Methodology.} 
Depending on the focus of the experiment, we measure running time and/or edge-cut.
An alternative metric for edge-cut is the \emph{z\_value}~\cite{HuaRandomWalk2020}, which can be mathematically defined as $1-\frac{edgecut}{\omega(E^{-})}$ where lower values are better and a value of 0 means a perfectly balanced clustering.
Unless explicitly mentioned otherwise, we run our experiments three times on each of the graphs listed Table~\ref{tab:sgc_graphs} with different random seeds, except for deterministic algorithms, which are run only once per~graph.
Since all values of edge-cut computed in our experiments are negative, we define its geometric mean as the geometric mean of its absolute value multiplied afterwards by $-1$.

We present convergence plots, which are computed in the following way. 
Whenever a processing element creates a clustering it reports a pair ($t$, cut), where the timestamp $t$ is the currently elapsed time (in seconds) on the particular processing element and cut refers to the edge-cut of the computed clustering. 
After the completion of our algorithms on $P$ processing elements, we are left with $P$ sequences of pairs (t, cut) which we now merge into one sequence based on the timestamp.
Then, we scan this sequence in ascending order of timestamp and remove all pairs with an edge-cut higher than then minimum edge-cut scanned so far.
The resulting sequence defines a \emph{step function} which indicates the best edge-cut found up to each timestamp.
Since multiple repetitions are executed for each graph, we simply combine the step functions of all repetitions by computing the arithmetic mean of the best edge-cut found during all repetitions at each timestamp.
When presenting convergence plots over multiple graphs, we normalize the elapsed time for each instance and then combine all step functions in the same way described above, but based on geometric mean rather than arithmetic mean.
The normalization of the elapsed time is done as in~\cite{sanders2012distributed} by dividing all timestamps by a normalization factor representing the approximate running time necessary to compute a clustering for the respective graph in one processing element.
We define this normalization factor as the first (smallest) timestamp reported for an instance, which represents the the actual running time \hbox{invested to create its first clustering}.

\begin{figure}[t!]
	\centering
	\includegraphics[angle=-0, width=\scaleFactorSmall\textwidth]{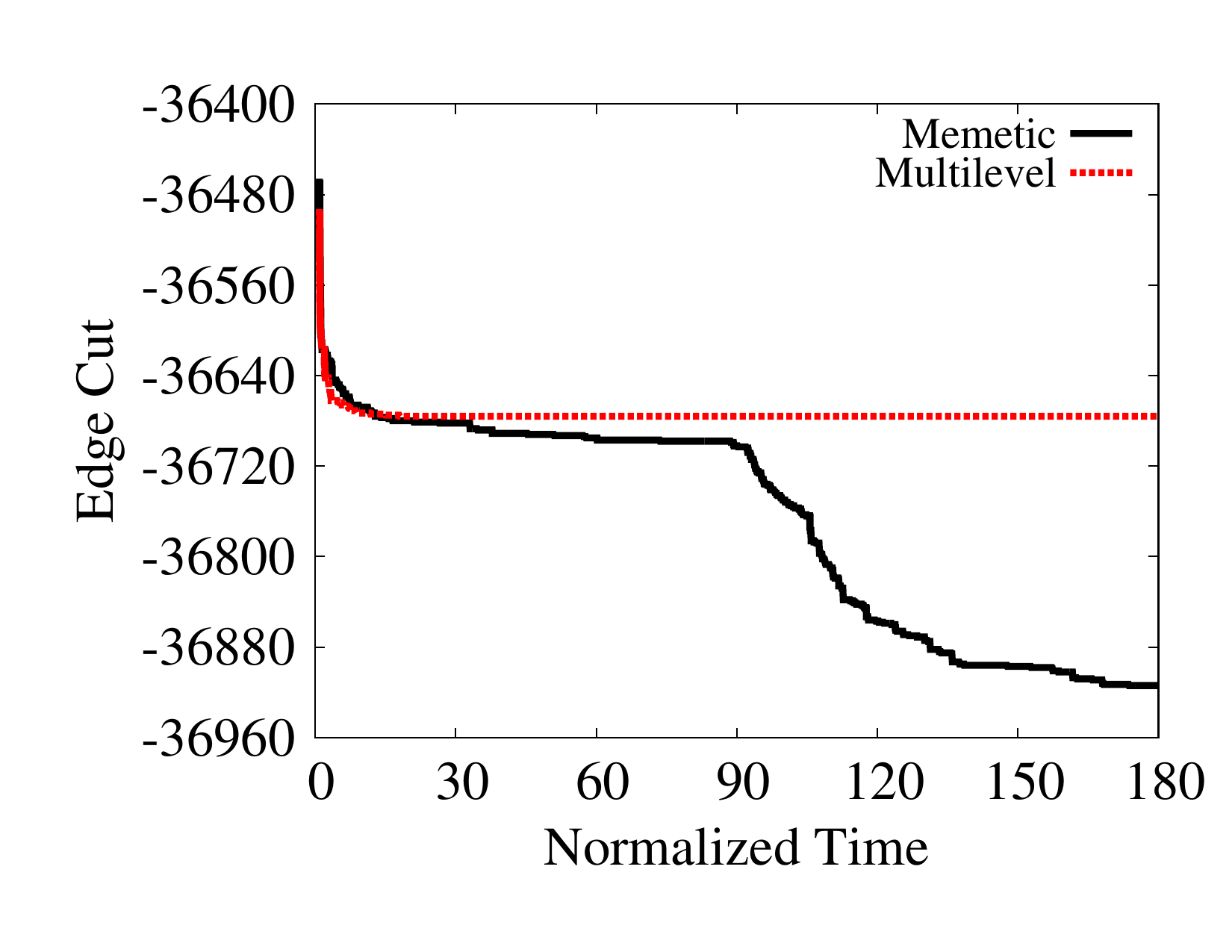}
	\caption{Convergence plot for comparison of our memetic algorithm against repeated executions of our multilevel algorithm for all graphs in Table~\ref{tab:sgc_graphs}.}
	\label{fig:sgc_convergence_overall}
\end{figure}

\begin{figure}[t]
	\captionsetup[subfigure]{justification=centering}
	\centering
	\begin{subfigure}[]{\scaleFactorSmall\textwidth}
		\centering
		\includegraphics[angle=-0, width=\imgScaleFactorSmall\textwidth]{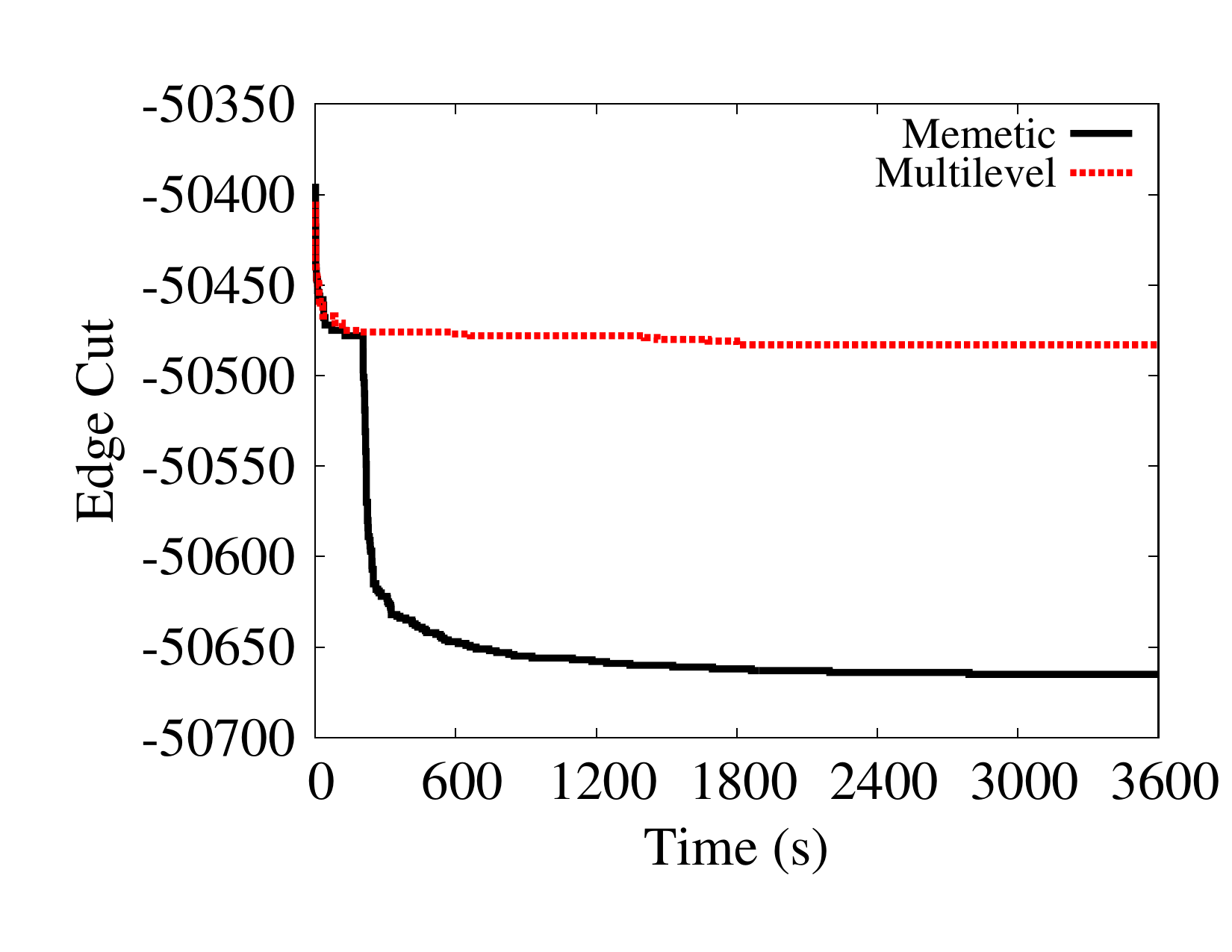}
		\vspace*{\capPositionSmall}
		\caption{\Id{slashdot090216}.}
		\label{fig:sgc_convergence_soc-sign-Slashdot090216}
	\end{subfigure}%
	\begin{subfigure}[]{\scaleFactorSmall\textwidth}
		\centering
		\includegraphics[angle=-0, width=\imgScaleFactorSmall\textwidth]{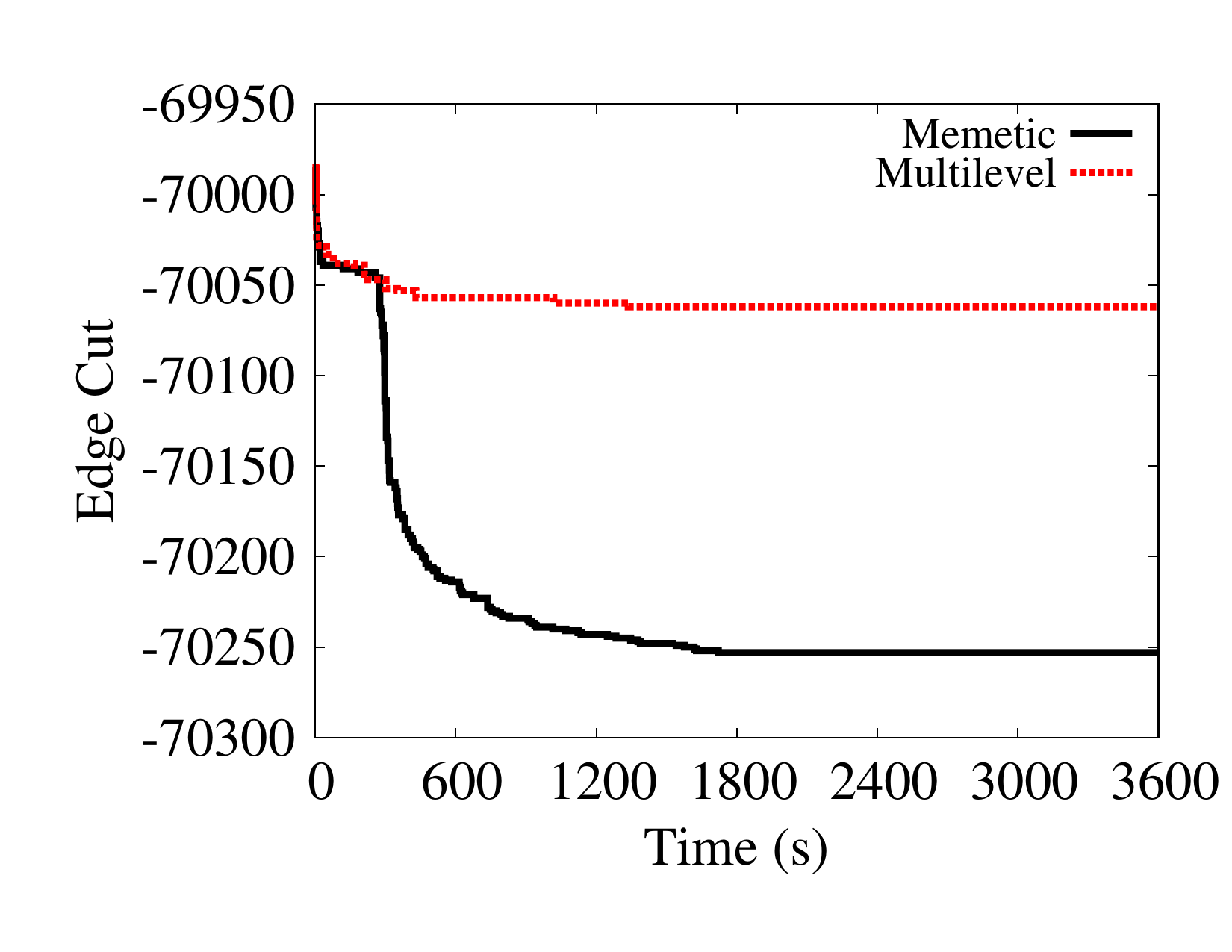}
		\vspace*{\capPositionSmall}
		\caption{\Id{epinions}.}
		\label{fig:sgc_convergence_soc-sign-epinions}
	\end{subfigure}%

	\begin{subfigure}[]{\scaleFactorSmall\textwidth}
		\centering
		\includegraphics[angle=-0, width=\imgScaleFactorSmall\textwidth]{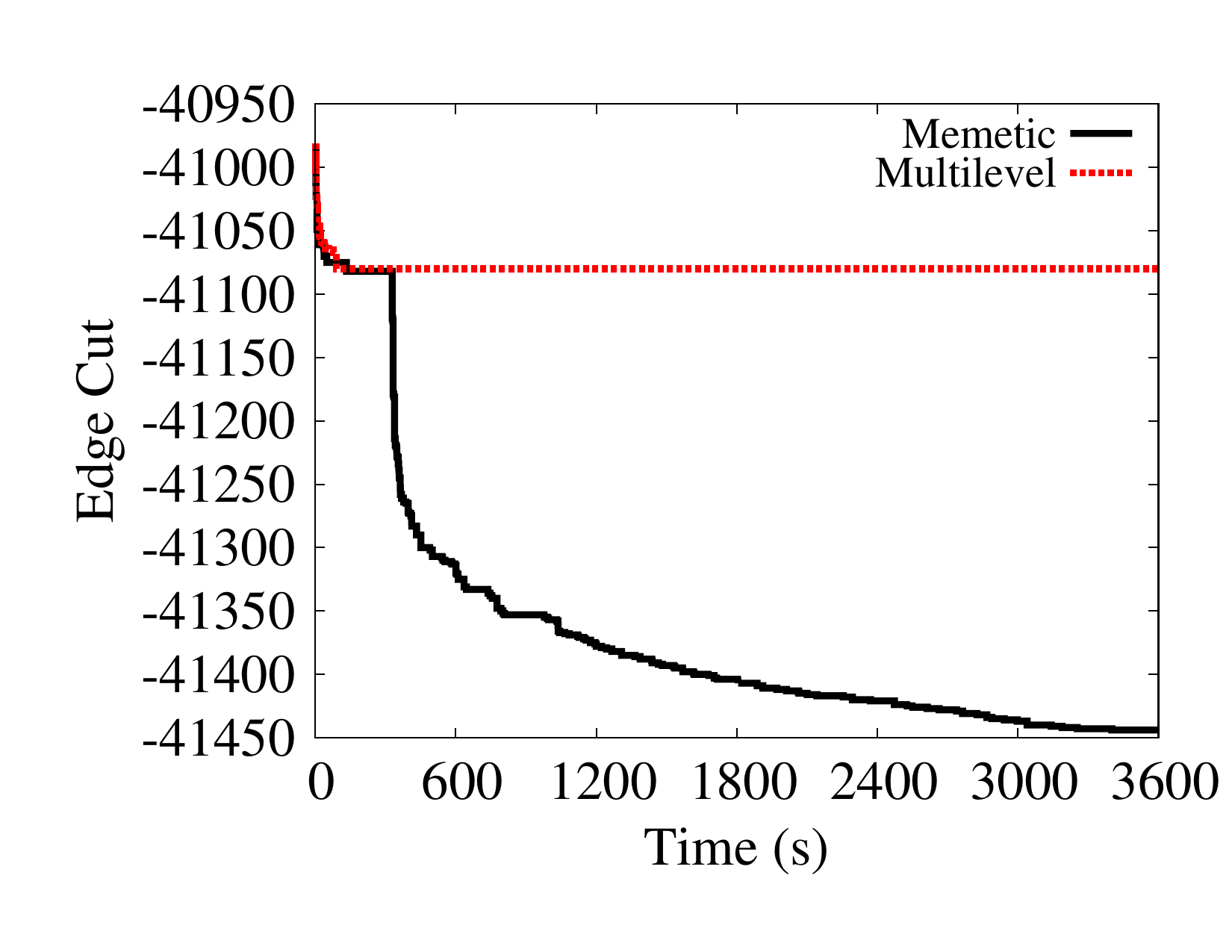}
		\vspace*{\capPositionSmall}
		\caption{\Id{wikisigned-k2}.}
		\label{fig:sgc_convergence_wikisigned-k2}
	\end{subfigure}%
	\begin{subfigure}[]{\scaleFactorSmall\textwidth}
		\centering
		\includegraphics[angle=-0, width=\imgScaleFactorSmall\textwidth]{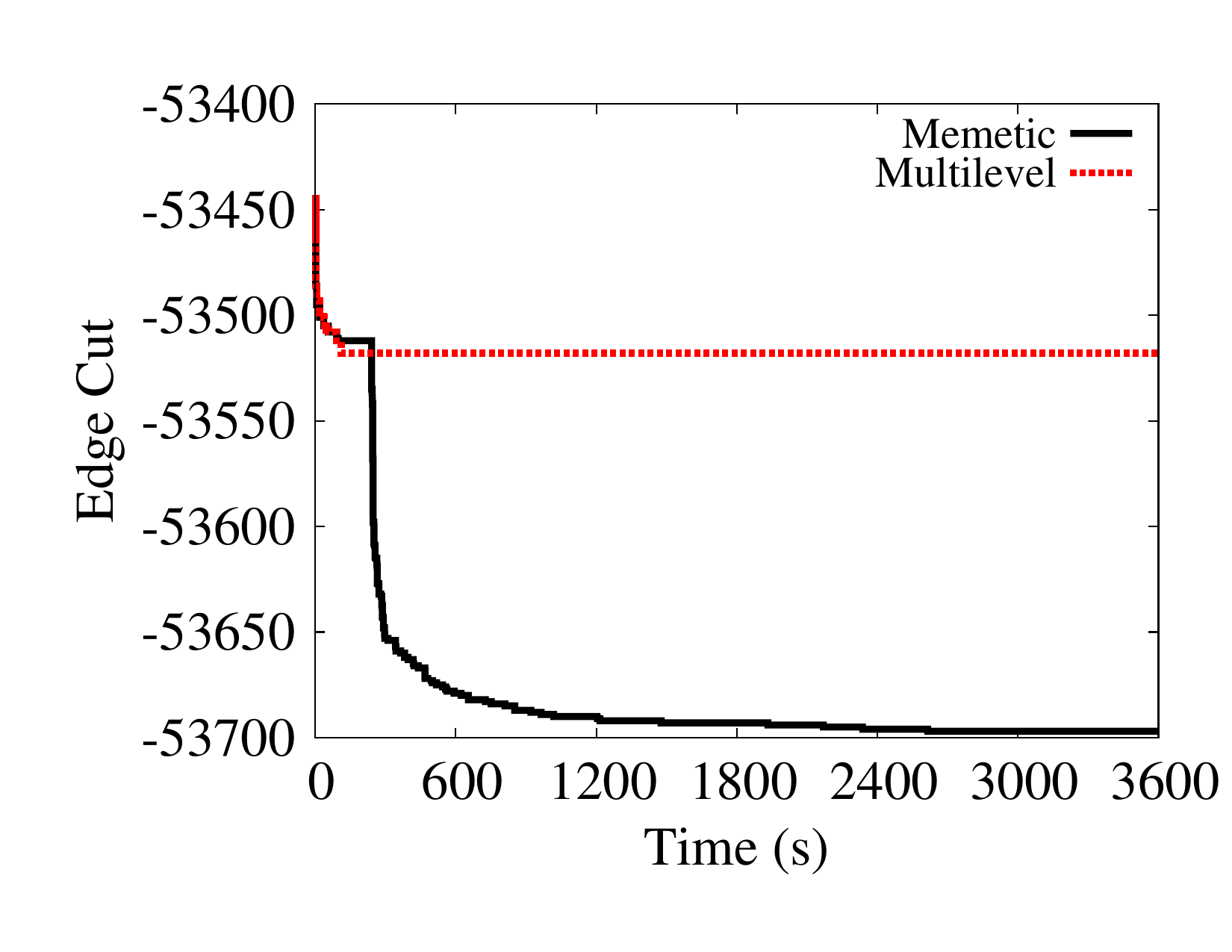}
		\vspace*{\capPositionSmall}
		\caption{\Id{slashdot-zoo}.}
		\label{fig:sgc_convergence_slashdot-zoo}
	\end{subfigure}
	\caption{Convergence plot for comparison of our memetic algorithm against repeated executions of our multilevel algorithm for individual graphs.}
	\label{fig:sgc_convergence_individual}
\end{figure}

\subsubsection{Convergence}

In this section, we evaluate the convergence of our memetic algorithm.
For this experiment, we run our memetic algorithm for $1$~hour on $16$~processing elements without hyper-threading.
As a baseline, we repeatedly run our multilevel algorithm (Algorithm~\ref{alg:sgc_overall_multilevel_strategy}) for the same amount of time in parallel on the same number of processing elements with different random seeds.
Our objective is to check whether our evolutionary strategy is directly improving the global optimization process, or if it has no advantage over simply repeating \hbox{our multilevel algorithm multiple times}.

Figures~\ref{fig:sgc_convergence_overall} shows a convergence plot over all instances.
The baseline improves edge-cut at first, but then stabilizes without major improvement. 
Our memetic algorithm performs similarly to the baseline for the first $10\%$ of the running time, then outperforms it and continues to improve edge-cut. 
This is as expected, as we allocate roughly $10\%$ of the running time to construct an initial population using multiple repetitions of our multilevel algorithm. 
After that, recombination and mutation operators are applied to the population.
Roughly a similar pattern is observed for all graphs, as exemplified in Figure~\ref{fig:sgc_convergence_individual}.
The convergence plots of individual graphs show another pattern not visible in Figure~\ref{fig:sgc_convergence_overall}. 
After about $10\%$ of the running time, our memetic algorithm computes better solutions quickly, then continues to improve edge-cut at a slower pace. 
This shows that our recombination and mutation operators quickly improve solution quality and continue to do so at a slower pace over time even \hbox{after this abrupt improvement}.

\subsubsection{Comparison against State-of-the-Art}

\begin{table*}[]
	\centering
	\scriptsize
	\setlength{\tabcolsep}{10.pt}
	\begin{tabular}{l@{\hskip 60pt}r@{\hskip 40pt}r@{\hskip 40pt}r}
		\toprule	
		Graph & \AlgName{HCC-Sum}  & \AlgName{GAEC}     & \AlgName{Memetic} \\ 
                \midrule
		\Id{bitcoinalpha}           & \numprint{-5476}    & \numprint{-5529}    & \textbf{\numprint{-5564}}                     \\
		\Id{bitcoinotc}             & \numprint{-20318}   & \numprint{-20391}   & \textbf{\numprint{-20440}}                    \\
		\Id{elec}                   & \numprint{-7723}    & \numprint{-7717}    & \textbf{\numprint{-7735}}                     \\
		\Id{chess}                  & \numprint{-4293}    & \numprint{-4312}    & \textbf{\numprint{-4796}}                     \\
		\Id{slashdot081106}         & \numprint{-48442}   & \numprint{-48454}   & \textbf{\numprint{-49805}}                    \\
		\Id{slashdot-zoo}           & \numprint{-51213}   & \numprint{-51806}   & \textbf{\numprint{-53703}}                    \\
		\Id{slashdot090216}         & \numprint{-49831}   & \numprint{-49797}   & \textbf{\numprint{-50668}}                    \\
		\Id{slashdot090221}         & \numprint{-49789}   & \numprint{-49798}   & \textbf{\numprint{-50926}}                    \\
		\Id{wikiconflict}           & \numprint{-2166604} & \numprint{-2166608} & \textbf{\numprint{-2167344}}                  \\
		\Id{epinions}               & \numprint{-68759}   & \numprint{-69156}   & \textbf{\numprint{-70323}}                    \\
		\Id{wikisigned-k2}          & \numprint{-41377}   & \numprint{-41458}   & \textbf{\numprint{-41623}}                    \\ 
                \midrule
		Overall                & \numprint{-36202}   & \numprint{-36320}   & \textbf{\numprint{-37126}}                    \\ 
                \bottomrule
	\end{tabular}
	\caption{Average edge-cut comparison.}
	\label{tab:sgc_results}
\end{table*}

In this section, we experimentally verify how much our evolutionary algorithm can improve over the state-of-the-art. 
For this experiment, we also run our memetic algorithm for $1$~hour on $16$~processing elements.
In order to further diversify our initial population and since \AlgName{HCC-Sum} and \AlgName{GAEC} are fast algorithms, we include the solutions computed by \AlgName{HCC-Sum} and \AlgName{GAEC} in the initial population and compute its remaining individuals using our multilevel algorithm as usual.
As a baseline, we use \AlgName{HCC-Sum} and \AlgName{GAEC}. Note that these algorithms are not randomized and thus can not be run multiple times to improve the result.
Table~\ref{tab:sgc_results} shows our overall results, in which our memetic algorithm computes better solutions than the state-of-the-art for every instance.
Note that our memetic strategy is able to improve the initial population (which also contains results of the competing algorithms) significantly. 
Our best edge-cut gain over \AlgName{HCC-Sum} and \AlgName{GAEC} is obtained for \Id{slashdot-zoo}, for which our average edge-cut is respectively \numprint{2490} and \numprint{1897} units smaller.
In proportional terms, our largest increase in the absolute value of the edge-cut over \AlgName{HCC-Sum} ad \AlgName{GAEC} is obtained for \Id{chess}: $11.9\%$ and $11.4\%$, respectively.
A natural lower bound for the edge-cut of a graph is given by $\omega(E^{-})$.
Based on this bound, our algorithm gets the smallest absolute optimality gap for \Id{bitcoinalpha} (\numprint{3786} units) and the smallest proportional \hbox{optimality gap for \Id{wikiconflict} ($8.35\%$)}.
We conclude that our memetic strategy is effective in improving results of the initial population and overall our algorithm can be seen as the state-of-the-art.

\section{Conclusion}
\label{sec:multilevel_Conclusion}

In this chapter, we proposed two multilevel algorithms for graph decomposition.
Both were engineered in all their details and were subjected to extensive experimental evaluation against \hbox{the state-of-the-art}.

First, we proposed a multilevel algorithm aimed at solving the process mapping problem.
Our algorithm integrates graph partitioning and process mapping and comprises multiple versions with varying speed/quality trade-offs.
Key ingredients of our algorithm include fast label propagation, more localized local search, initial partitioning, as well as a compressed data structure for computing processor distances without the need for storing a distance matrix.
Experimental results demonstrate that our algorithm represents the new state-of-the-art for process mapping.
Specifically, our algorithm generates superior or comparable overall solutions compared to any of the competing algorithms while being orders of magnitude faster than the previous best algorithm in terms of quality.
Our improvements are primarily attributable to the integrated multilevel approach coupled with high-quality local search algorithms and initial mapping algorithms that divide the initial network based on the specified system hierarchy.
Important future work entails parallelization, as well as the integration of global search schemes and different types of coarsening to further enhance solution quality.
Additionally, we plan to investigate the impact of our algorithm on the real performance of applications such as sparse matrix-\hbox{vector multiplications}.

Second, we engineered a memetic algorithm for the signed graph clustering problem.
Building upon a novel multilevel algorithm, we derive highly intuitive recombination and mutation operations.
Experimental results show that our memetic strategy produces substantially superior solutions compared to the present state-of-the-art.
In future work, we aim to enhance our algorithm by incorporating local searches based on \hbox{maximum flows}.

\chapter{Discussion}
\label{chap:Discussion}

\section{Conclusion}
\label{sec:Conclusion}

(Hyper)graph decomposition refers to a set of problems that focus on dividing large (hyper)graphs into smaller sub(hyper)graphs, which facilitates their analysis.
The significance of this lies in its capability to enable efficient computation on large and intricate (hyper)graphs, including chemical compounds, social networks, and computer networks.
In this dissertation, we propose multiple algorithmic contributions in the field of (hyper)graph decomposition, which utilize various techniques like buffered streaming, shared-memory parallelism, and efficient data structures to enhance the performance and quality of solutions.
We provide a thourough description for each algorithm along with experimental results that demonstrate their superiority over existing state-of-the-art algorithms.

The first algorithm we propose is a buffered streaming algorithm for graph partitioning, which loads a batch of nodes and builds a model representing the loaded subgraph and already present partition structure.
This model allows us to apply multilevel algorithms and compute high-quality solutions of massive graphs on inexpensive machines.
We develop a multilevel algorithm to partition the model that optimizes an objective function, which improves on the previous state-of-the-art by removing the dependence on the number of blocks from the running time.
Our algorithm computes better solutions than state-of-the-art using a small buffer size, and for larger numbers of blocks, it becomes faster than state-of-the-art.

The second algorithm is a shared-memory parallel streaming algorithm for the process mapping problem, which maps a streamed communication graph onto a hierarchical topology by performing recursive multi-sections on-the-fly.
This algorithm can also be used to solve the graph partitioning problem as a general tool.
Our algorithm has considerably lower running time complexity than existing state-of-the-art techniques and produces better process mappings.

% \vfill \pagebreak

The third algorithm is a streaming algorithm for hypergraph partitioning, which adapts the state-of-the-art streaming algorithm for graph partitioning by using an efficient data structure that makes the overall running time linearly dependent on the pin-count of the hypergraph and the memory consumption linearly dependent on the numbers of nets and blocks.
Our algorithm outperforms all existing (buffered) streaming algorithms and even an in-memory algorithm with respect to both weighted number of cut hyperedges and connectivity measures.

Next, we propose algorithms for the local motif clustering problem, which builds a (hyper)graph model representing the motif-distribution around the seed node on the original graph.
The first algorithm partitions the (hyper)graph model using a multi-level hypergraph or graph partitioner to minimize the motif conductance of the corresponding partition in the original graph.
The second algorithm transforms the hypergraph model into a flow model based on the max-flow quotient-cut improvement algorithm to obtain a superior solution automatically.
Our algorithms produce better communities than existing state-of-the-art techniques, while being up to multiple orders of magnitude faster.

We also propose multilevel algorithms for the process mapping problem, which include fast label propagation, more localized local search, initial partitioning, and a compressed data structure to compute processor distances without storing a distance matrix.
Our algorithms are able to exploit a given hierarchical structure of the distributed system under consideration and obtain better solutions than existing state-of-the-art techniques.

Finally, we propose algorithms to solve the signed graph clustering problem by using some of the most effective techniques from graph partitioning that minimize edge-cut.
Our multilevel algorithm includes a coarsening-uncoarsening process and efficient local search methods. 
We also introduce a distributed memetic algorithm that utilizes our multilevel algorithm and further enhances it with natural multilevel recombination and mutation operations.
Experimental results demonstrate that our memetic algorithm outperforms the state-of-the-art with respect to edge-cut, producing significantly better solutions.

Overall, our contributions demonstrate the effectiveness of various algorithmic techniques, such as buffered streaming, shared-memory parallelism, and efficient data structures, in solving a range of important problems in graph and hypergraph decomposition.
Our algorithms achieve improved performance and better solutions compared to state-of-the-art methods in various metrics, making them highly promising for practical applications.

\section{Future Work}
\label{sec:Future Work}

We have included specific future work in Chapters \ref{chap:Streaming Algorithms},~\ref{chap:Local Algorithms},~and~\ref{chap:Multilevel Algorithms}.
To enhance the performance of our algorithms, a general approach that could be beneficial is to parallelize them using shared and/or distributed memory.
This would lead to faster execution and enable them to handle even larger instances.
In addition, it would be desirable to conduct experiments with more specialized benchmark sets.
In particular, it would be worthwhile to verify the effectiveness of our streaming, local, and multilevel (hyper)graph decomposition algorithms within specific classes of (hyper)graphs, such as those that are used to model physical simulations, social networks, road networks, and other related phenomena.
For motif clustering, more extensive experiments could be conducted, including larger and more diverse motifs, which could provide new insights into the structure of the problem.
Another important objective is to investigate the impact of our algorithms on real-world applications.
For example, our process mapping and (hyper)graph partitioning algorithms can be tested as preprocessing steps for the distributed computation of sparse matrix vector multiplication.
Furthermore, extending our algorithms to dynamic (hyper)graphs, which is a current trend in the literature, could be accomplished in an easy manner, particularly for streaming algorithms.
Finally, machine learning techniques could be applied to (hyper)graph decomposition problems, given their successful implementation in literature to solve \hbox{other combinatorial problems}.

\cleardoublepage\phantomsection\addcontentsline{toc}{chapter}{\bibname}
\bibliography{bibliography}
\bibliographystyle{unsrtnat}

\appendix
\newcounter{rowcount}
\renewcommand{\therowcount}{\arabic{rowcount}}

\newenvironment{tabenum}
{\setcounter{rowcount}{0}
\begin{tabular}{p{0.15\linewidth}p{0.75\linewidth}}}
{\end{tabular}}

\newcommand{\eduitem}[5]{\multirow{1}{*}{#1} & \texttt{#2} \\
                                             & #3, #4 \\
                                             & #5 \\
                                             &    \\}

\newcommand{\workitem}[4]{\multirow{1}{*}{#1} & \texttt{#2} \\
                                             & #3, #4 \\
                                             &    \\}

\begin{coverpage}{Curriculum Vitae}

\section*{Education}       
\begin{tabenum}
\eduitem{2011-2017}{Bachelor, Computer Engineering}{Federal Center of Technological Education of Minas Gerais}{Brazil}{Thesis: A Local Search Genetic Algorithm Approach to Optimize Roadside Unit Placement in Vehicular Networks based on Gamma Depployment Metric}
\workitem{2014-2015}{Nondegree International Visiting Program}{Colorado State University}{United States}
\eduitem{2017-2019}{Master, Computer Science}{Federal University of Minas Gerais}{Brazil}{Thesis: Gamma Deployment Problem in Grids: Complexity and a new Integer Linear Programming Formulation}
\end{tabenum}

\vfill\pagebreak

\section*{Work Experience}       
\begin{tabenum}
\workitem{2013-2014}{Intern, Technological Initiation}{Federal Center of Technological Education of Minas Gerais}{Brazil}
\workitem{2017}{Intern, Software Engineer}{LYNX Process}{Brazil}
\workitem{2019-2020}{Academic Employee}{University of Vienna}{Austria}
\workitem{2020-}{Academic Employee}{Heidelberg University}{Germany}
\end{tabenum}

\section*{Honors}       
\begin{tabenum}
\workitem{2010}{Silver Medal in Mathematics}{Brazilian Olympics of Mathematics of Public High School}{Brazil}
\workitem{2014}{Scholarship at Colorado State University}{CAPES Foundation}{Brazil}
\workitem{2017}{Honorable Mention for Best Final Bachelor GPA}{Federal Center of Technological Education of Minas Gerais}{Brazil}
\end{tabenum}

\end{coverpage}

\begin{coverpage}{List of Publications}

\section*{Journal Papers}       
\begin{enumerate}

\item[\refstepcounter{enumi}{[\theenumi]}]
Marcelo Fonseca Faraj, Sebasti{\'{a}}n Urrutia, and Jo{\~{a}}o F. M. Sarubbi.
\newblock {Gamma Deployment Problem in Grids: Hardness and New Integer Linear Programming Formulation}.
\newblock {\em International Transactions in Operational Research}, Volume 27, Article No 6, pages 2740--2759, 2020.

\item[\refstepcounter{enumi}{[\theenumi]}]
Marcelo Fonseca Faraj and Christian Schulz.
\newblock {Buffered Streaming Graph Partitioning}.
\newblock {\em ACM Journal of Experimental Algorithmics}, Volume 27, pages 1.10:1--1.10:26, 2022.

\item[\refstepcounter{enumi}{[\theenumi]}]
{\"{U}}mit V. {\c{C}}ataly{\"{u}}rek, Karen D. Devine, Marcelo Fonseca Faraj, Lars Gottesb{\"{u}}ren, Tobias Heuer, Henning Meyerhenke, Peter Sanders, Sebastian Schlag, Christian Schulz, Daniel Seemaier, and Dorothea Wagner
\newblock {More Recent Advances in (Hyper)Graph Partitioning}.
\newblock {\em ACM Computing Surveys}, Volume 55, Issue 12, Article No 253, pages 1--38, 2022.

\end{enumerate}

\section*{Conference Full Papers}       
\begin{enumerate}

\item[\refstepcounter{enumi}{[\theenumi]}]
Marcelo Fonseca Faraj, Jo{\~{a}}o F. M. Sarubbi, Cristiano M. Silva, Marcelo F. Porto, and Nilson T. R. Nunes
\newblock {Estudo de Caso: o Problema do Transporte Escolar Rural em Minas Gerais}.
\newblock In {\em Simp{\'{o}}sio Brasileiro de Pesquisa Operacional (SBPO)}, pages 1889--1900, 2014.

\item[\refstepcounter{enumi}{[\theenumi]}]
Marcelo Fonseca Faraj, Jo{\~{a}}o F. M. Sarubbi, Cristiano M. Silva, Marcelo F. Porto, and Nilson T. R. Nunes
\newblock {A Real Geographical Application for the School Bus Routing Problem}.
\newblock In {\em IEEE Conference on Intelligent Transportation Systems (ITSC)}, pages 2762--2767, 2014.

\item[\refstepcounter{enumi}{[\theenumi]}]
Marcelo Fonseca Faraj, Sebasti{\'{a}}n Urrutia, and Jo{\~{a}}o F. M. Sarubbi.
\newblock {Problema da Deposi{\c{c}}{\~{a}}o Gamma: Prova de NP-Completude e um Novo Modelo de Programa{\c{c}}{\~{a}}o Linear Inteira}.
\newblock In {\em Simp{\'{o}}sio Brasileiro de Pesquisa Operacional (SBPO)}, 2018.

\item[\refstepcounter{enumi}{[\theenumi]}]
Marcelo Fonseca Faraj, Jo{\~{a}}o F. M. Sarubbi, Cristiano M. Silva, and Fl{\'{a}}vio V. C. Martins.
\newblock {A Memetic Algorithm Approach to Deploy {RSU} s Based on the Gamma Deployment Metric}.
\newblock In {\em IEEE Congress on Evolutionary Computation (CEC)}, pages 1--8, 2018.

\item[\refstepcounter{enumi}{[\theenumi]}]
Marcelo Fonseca Faraj, Alexander van der Grinten, Henning Meyerhenke, Jesper L. Träff, and Christian Schulz.
\newblock {High-Quality Hierarchical Process Mapping}.
\newblock In {\em Proceedings of the 18th Symposium on Experimental Algorithms (SEA)}, volume 160 of LIPIcs, pages 4:1--4:15, 2020.

\item[\refstepcounter{enumi}{[\theenumi]}]
Marcelo Fonseca Faraj and Christian Schulz.
\newblock {Recursive Multi-Section on the Fly: Shared-Memory Streaming Algorithms for Hierarchical Graph Partitioning and Process Mapping}.
\newblock In {\em IEEE International Conference on Cluster Computing (CLUSTER)}, volume 9411, pages 473--483, 2022.

\item[\refstepcounter{enumi}{[\theenumi]}]
Adil Chhabra, Marcelo Fonseca Faraj, and Christian Schulz.
\newblock {Local Motif Clustering via (Hyper)Graph Partitioning}.
\newblock In {\em Proceedings of the Symposium on Algorithm Engineering and Experiments (ALENEX)}, pages 96--109. SIAM, 2023.

\item[\refstepcounter{enumi}{[\theenumi]}]
Kamal Eyubov, Marcelo Fonseca Faraj, and Christian Schulz.
\newblock {FREIGHT: Fast Streaming Hypergraph Partitioning}.
\newblock In {\em Proceedings of the 21st Symposium on Experimental Algorithms (SEA)}, to appear, 2023.

\end{enumerate}

\section*{Conference Short Papers and Poster Papers}       
\begin{enumerate}

\item[\refstepcounter{enumi}{[\theenumi]}]
Marcelo Fonseca Faraj, Jo{\~{a}}o F. M. Sarubbi, Cristiano M. Silva, and Fl{\'{a}}vio V. C. Martins.
\newblock {A Hybrid Genetic Algorithm for Deploying RSUs in VANETs based on Inter-Contact Time (Extended Abstract)}.
\newblock In {\em Genetic and Evolutionary Computation Conference Companion (GECCO ’17 Companion)}, pages 193--194, 2017.

\item[\refstepcounter{enumi}{[\theenumi]}]
Marcelo Fonseca Faraj, Sebasti{\'{a}}n Urrutia, and Jo{\~{a}}o F. M. Sarubbi.
\newblock {O Problema da Deposi{\c{c}}{\~{a}}o Gamma {\'{e}} NP-Completo (Extended Abstract)}.
\newblock In {\em Encontro de Teoria da Computa{\c{c}}{\~{a}}o (ETC)}, 2018.

\item[\refstepcounter{enumi}{[\theenumi]}]
Adil Chhabra, Marcelo Fonseca Faraj, and Christian Schulz.
\newblock {Local Motif Clustering via (Hyper)Graph Partitioning (Extended Abstract)}.
\newblock In {\em Proceedings of the 15th International Symposium on Combinatorial Search (SoCS)}, pages 261--263, AAAI Press,~2022.

\item[\refstepcounter{enumi}{[\theenumi]}]
Felix Hausberger, Marcelo Fonseca Faraj, and Christian Schulz.
\newblock {A Distributed Multilevel Memetic Algorithm for Signed Graph Clustering (Short Paper)}.
\newblock In {\em Genetic and Evolutionary Computation Conference Companion (GECCO ’23 Companion)}, to appear, 2023.

\end{enumerate}

\section*{Technical Reports}       
\begin{enumerate}

\item[\refstepcounter{enumi}{[\theenumi]}]
Marcelo Fonseca Faraj, Alexander van der Grinten, Henning Meyerhenke, Jesper L. Träff, and Christian Schulz.
\newblock {High-Quality Hierarchical Process Mapping}.
\newblock Technical Report, University of Vienna, Humboldt Universtät zu Berlin, and Technical University of Vienna, 2020. (arXiv:2001.07134v2)

\item[\refstepcounter{enumi}{[\theenumi]}]
Marcelo Fonseca Faraj and Christian Schulz.
\newblock {Buffered Streaming Graph Partitioning}.
\newblock Technical Report, Heidelberg University, 2021. (arXiv:2102.09384)

\item[\refstepcounter{enumi}{[\theenumi]}]
Marcelo Fonseca Faraj and Christian Schulz.
\newblock {Recursive Multi-Section on the Fly: Shared-Memory Streaming Algorithms for Hierarchical Graph Partitioning and Process Mapping}.
\newblock Technical Report, Heidelberg University, 2021. (arXiv:2202.00394)

\item[\refstepcounter{enumi}{[\theenumi]}]
{\"{U}}mit V. {\c{C}}ataly{\"{u}}rek, Karen D. Devine, Marcelo Fonseca Faraj, Lars Gottesb{\"{u}}ren, Tobias Heuer, Henning Meyerhenke, Peter Sanders, Sebastian Schlag, Christian Schulz, Daniel Seemaier, and Dorothea Wagner
\newblock {More Recent Advances in (Hyper)Graph Partitioning}.
\newblock Technical Report, 2022. (arXiv:2205.13202)

\item[\refstepcounter{enumi}{[\theenumi]}]
Adil Chhabra, Marcelo Fonseca Faraj, and Christian Schulz.
\newblock {Local Motif Clustering via (Hyper)Graph Partitioning}.
\newblock Technical Report, Heidelberg University, 2022. (arXiv:2205.06176)

\item[\refstepcounter{enumi}{[\theenumi]}]
Felix Hausberger, Marcelo Fonseca Faraj, and Christian Schulz.
\newblock {A Distributed Multilevel Memetic Algorithm for Signed Graph Clustering}.
\newblock Technical Report, Heidelberg University, 2022. (arXiv:2208.13618)

\item[\refstepcounter{enumi}{[\theenumi]}]
Kamal Eyubov, Marcelo Fonseca Faraj, and Christian Schulz.
\newblock {FREIGHT: Fast Streaming Hypergraph Partitioning}.
\newblock Technical Report, Heidelberg University, 2023. (arXiv:2302.06259)

\item[\refstepcounter{enumi}{[\theenumi]}]
Adil Chhabra, Marcelo Fonseca Faraj, and Christian Schulz.
\newblock {Faster Local Motif Clustering via Maximum Flows}.
\newblock Technical Report, Heidelberg University, 2023. (arXiv:2301.07145)

\end{enumerate}

\end{coverpage}

\end{document}